\documentclass[a4paper,UKenglish,cleveref, autoref, thm-restate]{lipics-v2021}
\usepackage[linesnumbered,titlenotnumbered,ruled]{algorithm2e}
\usepackage{tikz,pgffor}
\usetikzlibrary{arrows,backgrounds}
\usepackage{tkz-euclide}
\usetikzlibrary{arrows}
\usetikzlibrary{shapes}
\usetikzlibrary{calc}
\usetikzlibrary{automata}
\usetikzlibrary{positioning}
\usepackage{listings}
\usepackage{transparent}
\usepackage[color=pink]{todonotes}
\usepackage{placeins}

\usepackage{booktabs}

\usepackage{nicefrac}

\usepackage[justification=centering]{caption}

\nolinenumbers

\usepackage{calrsfs}
\DeclareMathAlphabet{\pazocal}{OMS}{zplm}{m}{n}

\newcommand{\support}{\ensuremath{Support}}

\newtheorem*{lemma*}{Lemma}
\newtheorem*{theorem*}{Theorem}
\newtheorem*{disclaimer*}{Disclaimer}

\newcommand{\N}{\mathbb{N}}
\newcommand{\R}{\mathbb{R}}
\newcommand{\Z}{\mathbb{Z}}
\newcommand{\A}{\pazocal{A}}

\newcommand{\B}{\mathbb{B}}
\newcommand{\F}{\pazocal{F}}

\newcommand{\calC}{\pazocal{C}}
\newcommand{\calT}{\pazocal{T}}
\newcommand{\calL}{\pazocal{L}}

\newcommand{\calO}{\pazocal{O}}

\newcommand{\calP}{\pazocal{P}}

\newcommand{\pointcomplex}{\calP}

\newcommand{\MEC}{B}
\newcommand{\pointingVASScorrepsonding}[1]{\tilde{#1}}
\newcommand{\stateVASSone}{p}

\newcommand{\stratsMD}[1]{cMD(#1)}
\newcommand{\cMD}{cMD}

\newcommand{\ce}[1]{\left(#1 \right)}

\newcommand{\size}[1]{|\!|#1|\!|}

\newcommand{\RankEff}{\mathit{RankEff}}
\newcommand{\length}{\mathit{len}}

\newcommand{\len}{\textit{len}}

\newcommand{\prob}{\ensuremath{\mathbb{P}}}
\newcommand{\Prob}{\prob}

\newcommand{\E}{\mathbb{E}}

\newcommand{\states}{\ensuremath{Q}}

\newcommand{\bx}{\mathbf{x}}
\newcommand{\by}{\mathbf{y}}
\newcommand{\bz}{\mathbf{z}}
\newcommand{\bu}{\mathbf{u}}
\newcommand{\bv}{\mathbf{v}}

\newcommand{\bn}{\mathbf{n}}

\newcommand{\br}{\mathbf{r}}
\newcommand{\bs}{\mathbf{s}}

\newcommand{\realeffect}{E}

\newcommand{\Term}{\mathit{Term}}

\newcommand{\epsilont}{\ensuremath{\epsilon}}
\newcommand{\epsilonr}{\ensuremath{\epsilon}}

\newcommand{\Q}{\ensuremath{\mathbb{Q}}}
\newcommand{\Nset}{\ensuremath{\mathbb{N}}}

\newcommand{\bigO}{O}
\newcommand{\Exp}{\ensuremath{\mathbb{E}}}
\newcommand{\eexp}{\mathrm{exp}}

\newcommand{\tin}{\mathit{In}}
\newcommand{\tout}{\mathit{Out}}

\newcommand{\countersset}{\mathit{Count}}

\newcommand{\Aset}{\mathit{Cbounds}}
\newcommand{\Bset}{\mathit{Tbounds}}
\newcommand{\aelement}{\mathit{a}}
\newcommand{\belement}{\mathit{b}}

\newcommand{\variable}{\mathit{\tau}}
\newcommand{\variablet}{\mathit{\tau}}

\newcommand{\finitecomputation}{computation }

\newcommand{\M}{\ensuremath{\mathcal{M}}}

\tikzstyle{loop above}=[tran, to path={.. controls +(60:.5) 
	and +(120:.5) .. (\tikztotarget) \tikztonodes}]
\tikzstyle{loop below}=[tran, to path={.. controls +(240:.5) 
	and +(300:.5) .. (\tikztotarget) \tikztonodes}]
\tikzstyle{loop left}=[tran,  to path={.. controls +(150:.5) 
	and +(210:.5) .. (\tikztotarget) \tikztonodes}]
\tikzstyle{loop right}=[tran,  to path={.. controls +(330:.5) 
	and +(30:.5) .. (\tikztotarget) \tikztonodes}]
\tikzstyle{bigstoch}=[circle,draw,minimum size=8ex,inner sep=0pt,font=\Large, very thick,text centered, fill=blue!20, fill opacity=0.2, draw=black!80, text opacity=1]
\tikzstyle{stoch}=[circle,draw,minimum size=4ex,inner sep=0pt,font=\large, very thick,text centered, fill=blue!20, fill opacity=0.2, draw=black!80, text opacity=1]

\tikzstyle{squar}=[regular polygon,regular polygon sides=4,draw,minimum size=4ex,inner sep=0pt,font=\large, very thick,text centered, fill=blue!20, fill opacity=0.2, draw=black!80, text opacity=1]

\tikzstyle{bigstoch2}=[circle,draw,minimum size=8ex,inner sep=0pt,font=\Large, very thick,text centered, fill=blue!20, fill opacity=0.2, draw=black!80, text opacity=1, double]
\tikzstyle{bigmin}=[diamond,draw,minimum size=8ex,inner sep=0pt,font=\Large, very thick,text centered, fill=blue!20, fill opacity=0.2, draw=black!80, text opacity=1]
\tikzstyle{max}=[rectangle,draw,minimum size=4ex,inner sep=0pt,font=\large, very thick,text centered, fill=blue!20, fill opacity=0.2, draw=black!80, text opacity=1]
\tikzstyle{bigtran}=[very thick,draw,-angle 60,font=\scriptsize, inner sep = 6pt]
\tikzstyle{tran}=[thick,draw,font=\scriptsize, inner sep = 6pt,-stealth]

\title{Efficient Analysis of Polynomial Asymptotic Estimates for VASS MDPs} 

\titlerunning{asymptotic estimates for polynomial VASS MDPs} %TODO optional, please use if title is longer than one line
\author{Michal Ajdar{\'{o}}w}{Masaryk University, Czechia \and \url{https://www.muni.cz/lide/422654-michal-ajdarow/}}{xajdarow@fi.muni.cz}{https://orcid.org/0000-0003-0694-0944}{}%TODO mandatory, please use full name; only 1 author per \author macro; first two parameters are mandatory, other parameters can be empty. Please provide at least the name of the affiliation and the country. The full address is optional

\authorrunning{M.~Ajdar{\'{o}}w} %TODO mandatory. First: Use abbreviated first/middle names. Second (only in severe cases): Use first author plus 'et al.'
 %TODO mandatory: Please choose ACM 2012 classifications from https://dl.acm.org/ccs/ccs_flat.cfm 

 %TODO mandatory; please add comma-separated list of keywords

\category{} %optional, e.g. invited paper
\relatedversion{} %optional, e.g. full version hosted on arXiv, HAL, or other respository/website
\EventEditors{John Q. Open and Joan R. Access}
\EventNoEds{2}
\EventLongTitle{42nd Conference on Very Important Topics (CVIT 2016)}
\EventShortTitle{CVIT 2016}
\EventAcronym{CVIT}
\EventYear{2016}
\EventDate{December 24--27, 2016}
\EventLocation{Little Whinging, United Kingdom}
\EventLogo{}
\SeriesVolume{42}
\ArticleNo{23}
\begin{document}

\maketitle

\begin{abstract}
	Markov decision process over vector addition system with states (VASS MDP) is a finite state model combining non-deterministic and probabilistic behavior, augmented with non-negative integer counters that can be incremented or decremented during each state transition. VASS MDPs can be used as abstractions of probabilistic programs with many decidable properties. In this paper, we develop techniques for analyzing the asymptotic behavior of VASS MDPs.
	That is, for every initial configuration of size \(n\), we consider the number of transitions needed to reach a configuration with some counter negative. We show
	that given a strongly connected VASS MDP there either exists an integer \(k\leq 2^d\cdot 3^{|T|} \), where \(d \) is the dimension and \(|T|\) the number of transitions of the VASS MDP, such that for all \(\epsilon>0 \) and all sufficiently large \(n\)  it holds that the complexity of the VASS MDP lies between \(n^{k-\epsilon} \) and \(n^{k+\epsilon} \) with probability at least \(1-\epsilon \), or it holds for all \(\epsilon>0 \) and all sufficiently large \(n\) that the complexity of the VASS MDP is at least \(2^{n^{\frac{1}{2}-\epsilon}} \) with probability at least \(1-\epsilon \).
	We show that it is decidable which case holds and the \(k\) is computable in time polynomial in the size of the considered VASS MDP. We also provide a full classification of asymptotic complexity for  VASS
	Markov chains.
\end{abstract}

%\newpage

\section{Introduction}
\label{sec-intro}

% \begin{figure}[t]
% 	\lstset{basicstyle=\footnotesize}
% 			\begin{tikzpicture}[scale=.9, every node/.style={scale=0.9}, x=1.4cm, y=1.2cm, font=\scriptsize]
				
% 					\node [stoch] (r) at (3,0)  {};
% 					\draw [rounded corners,draw] (0,0.8) rectangle (1.6,-1); 
% 					\node [] (edge) at (1.6,0) {};
% 					\node [] (M) at (0.8,0) {\(M\)};
% 					\draw [tran] (edge) to (r);
% 				\draw [tran,loop right] (r) to node[right] {$ +c(p=\frac{1}{2}),  -c(p=\frac{1}{2})$} (r); 
				 
% 			\end{tikzpicture}
% 	\caption{An example of VASS MDP with \(E(L(n))=\infty\) but with \(L(n)\) being limP bounded by finite function. \(M\) is an arbitrary demonic VASS with counter \(c\). If \(M\) is such that \(C_M [c](n)\in \Theta(f(n))\), then \(L(n)\in \Omega\limP(f^2(n)) \). Additionally, if also \(L_M(n)\in O(f^2(n)) \), then also \(L(n)\in O\limP(f^2(n)) \). } 
% %TODO, note that this is definitely true for polynomial functions, check it holds for superpolynomials as well once we have them properly defined 
	
% 	\label{fig-VASS-model}
% \end{figure}

Vector Addition Systems with States (VASS) \cite{HP:VASS-reachability-TCS} are a model for discrete systems with multiple unbounded resources expressively equivalent to Petri nets \cite{Petri:first-paper}. Intuitively, a VASS with $d \geq 1$ counters is a finite directed graph where the transitions are labeled by $d$-dimensional vectors of integers representing \emph{counter updates}. A computation starts in some state for some initial vector of non-negative counter values and proceeds by selecting transitions non-deterministically and performing the associated counter updates. The computation terminates whenever any of the counters were to become negative. 

%Since the counters cannot assume negative values, transitions that would decrease some counter below zero are disabled.

In program analysis, VASS are used as abstractions for programs operating over unbounded integer variables. Input parameters are represented by initial counter values, and more complicated arithmetical functions, such as multiplication, are modeled by VASS gadgets computing these functions in a weak sense (see, e.g., \cite{LS:Petri-computer}). Branching constructs, such as \textbf{if-then-else}, are usually replaced with non-deterministic choice. VASS are particularly useful for evaluating the \emph{asymptotic complexity} of infinite-state programs, i.e., the dependency of the running time (and other complexity measures)  on the size of the program input \cite{SZV:amortized,SZV:difference-constraints}. Traditional VASS decision problems such as reachability, liveness, or boundedness are computationally hard \cite{CLLLM:VASS-reach-nonelem,Lipton:PN-Reachability,MM:containment-Petri}, and other verification problems such as equivalence-checking \cite{Jancar:PN-bisimilarity-TCS} or model-checking \cite{Esparza:ModelChecking-AI} are even undecidable. In contrast to this, decision problems related to the asymptotic growth of VASS complexity measures are solvable with low complexity and sometimes even in \emph{polynomial time} \cite{BCKNVZ:VASS-linear-termination,Zuleger:VASS-polynomial,KLV:VASS-Grzegorczyk,Leroux:Polynomial-termination-VASS,AK:VASS-polynomial-termination}; see \cite{Kucera:Asymptotic-VASS-Analysis-SIGLOG} for a recent overview. The complexity measures of VASS for which the asymptotic growth is usually considered are  \emph{termination complexity}, which can be seen as an analogy of time complexity, and  \emph{\(c\)-counter complexity}, which is an analogy of space complexity of a single variable \(c\).

The existing results about VASS asymptotic analysis are applicable to programs with non-determinism (in \emph{demonic} or \emph{angelic} form, see \cite{BW:nondet-languages}), but cannot be used to analyze the complexity of \emph{probabilistic programs}. This motivates the study of Markov decision process over VASS (VASS MDP) with both non-deterministic and probabilistic states, where transitions in probabilistic states are selected according to fixed probability distributions. Here, the problems of asymptotic complexity analysis become even more challenging because VASS MDPs subsume infinite-state stochastic models that are notoriously hard to analyze. So far, there are  only two existing results about asymptotic VASS MDP analysis. First is \cite{BCKNV:probVASS-linear-termination} where the linearity of expected termination complexity  is shown decidable in polynomial time for VASS MDPs with DAG-like  maximal end-component (MEC) decomposition, while showing that if the expected termination complexity is not linear then it is at least quadratic. The second is \cite{AKCONCUR23}, which  introduces a new notion of \emph{asymptotic estimates} for analyzing the asymptotic behavior of probabilistic systems which consists of a bound on all but \(\epsilon\)-ratio of runs for all \(\epsilon>0 \). Then it shows that we can fully classify one-dimensional VASS MDPs in terms of asymptotic estimates in polynomial time. Furthermore \cite{AKCONCUR23} also shows that for VASS MDPs with DAG-like MEC decomposition the termination/\(c\)-counter complexity has either a linear tight asymptotic estimate or a quadratic lower asymptotic estimate, and it is decidable in polynomial time which case holds.

Of special relevance to this paper is also  \cite{Zuleger:VASS-polynomial} which shows that given a strongly connected  \(d\)-dimensional VASS with demonic non-determinism the termination/\(c\)-counter complexity is either in \(\Theta(n^k) \) for some \(k\in\mathbb{N}, k\leq 2^{d} \) or in \(2^{\Omega(n)} \), while it is decidable in polynomial time which case holds, and  \(k\) can be computed in polynomial time, both  with respect to the size of the considered VASS.

\textbf{Our Contribution:} 
Our main result (Section~\ref{section-main}) can be seen as an extension of \cite{Zuleger:VASS-polynomial} onto VASS MDPs for asymptotic estimates. That is, we show that given a strongly connected \(d\)-dimensional VASS MDP with demonic non-determinism and \(|T|\) transitions, the termination/\(c\)-counter complexity   either has a tight asymptotic estimate of \(n^k\) for some \(k\in \mathbb{N}\), \(k\leq 2^d\cdot 3^{|T|} \) or a lower asymptotic estimate of \(2^{\sqrt{n}} \), while it is decidable in polynomial time which case holds, and  \(k\) can be computed in polynomial time, both with respect to the size of the considered VASS MDP.

Our next result (Section~\ref{section-markov-chains})  is a full classification of asymptotic complexity for strongly connected VASS
Markov chains. We show that for every VASS Markov chain the termination/\(c\)-counter complexity is either unbounded or one of the functions $n$, $n^2$ is its tight asymptotic estimate.
%\begin{itemize}
%	\item the \(c\)-counter complexity is  either unbounded or one of the functions $n$, $n^2$ is a tight asymptotic estimate.
%	\item termination complexity is either unbounded or one of the functions $n$, $n^2$ is a tight asymptotic estimate.
%	\item the \(t\)-transition complexity is either unbounded, or every increasing unbounded function \(f\) is an upper asymptotic estimate, or one of the functions $n$, $n^2$ is a tight asymptotic estimate.
%\end{itemize}

%Our next result is a full classification of asymptotic complexity for  We show that  

We also present an alternative, more intuitive definition of asymptotic estimates using a natural notion of fixed probability bounds.  (Section~\ref{sec-estimates})

\section{Preliminaries}
\label{sec-prelim}

%\subsection{Probabilistic VASS with Demonic Nondeterminism}
%\label{sec-VASS-def}

We use $\N$, $\Z$, $\Q$, and $\R$ to denote the sets of positive integers,
integers, rational numbers, and real numbers, respectively. 
  We use $\N_\infty$ to denote the set $\N \cup \{\infty\}$ where $\infty$ is treated according to the standard conventions, and \(\N_0=\N\cup \{0\} \).
Given a function $f \colon \R \rightarrow \R$, we use  $\calO(f)$ and $\Omega(f)$ to denote the sets of all $g \colon \N \rightarrow \N$ such that $g(n) \leq a \cdot f(n)$ and $g(n) \geq b \cdot f(n)$ for all sufficiently large $n \in \N$, respectively, where $a,b$ are some positive constants.  If $h \in O(f)$ and $h \in \Omega(f)$, we write $h \in \Theta(f)$. Given a function \(f: A\rightarrow B \) and a set \(X\subseteq A \), we use \(f(X) \) to denote the set \(\{f(x)\mid x\in X \} \).

Let $A$ be a finite index set. The vectors of $\Q^A$ are denoted by bold letters such as $\bu,\bv,\bz,\ldots$. The component of $\bv$ of index $i\in A$ is denoted by $\bv(i)$. 
%For a matrix
%$A \in \Rset^{A\times
	%	B}$ we denote by $A(a,b)$ the element in row of index $a\in A$ and column
%of index by $b\in
%B$, and by $A^{\top}$ the transpose of $A$. 
If the index set is of the form $A=\{1,2,\dots,d\}$ for some positive integer $d$, we write $\Q^d$ instead of $\Q^A$. For every $n \in \R$, we use $\vec{n}$ to denote the constant vector where all
components are equal to~$n$.
%The scalar product of $\bv,\bu \in \R^d$ is denoted by $\bv \cdot \bu$, i.e.,
%$\bv\cdot \bu = \sum_{i=1}^d \bv(i)\cdot\bu(i)$.
The other
standard operations and relations on $\Q$ such as
$+$, $\leq$, or $<$ are extended to $\Q^d$ in the component-wise way. In particular,  $\bv \leq  \bu$ if $\bv(i) \leq \bu(i)$ for every index $i$.

A \emph{probability distribution} over a finite set $A$ is a vector $\nu \in [0,1]^A$ such that $\sum_{a \in A}\nu(a) = 1$. We say that $\nu$ is  $\emph{Dirac}$ if $\nu(a) =1$ for some $a \in A$.
%\emph{rational} if every $\nu(a)$ is rational, and

%We also write $\bv \prec \bu$ if $\bv \leq \bu$ and $\bv(i) < \bu(i)$ for some index $i$. 
%The norm of $\bv$ is defined by $\Norm(\bv) = \sqrt{\bv(1)^2 + \cdots+\bv(d)^2}$.
\vspace{-0.2cm}

\subsection{VASS Markov Decision Processes}

\begin{definition}
	\label{def-VASS} 
	%Let $d \geq 1$. A \emph{$d$-dimensional VASS MDP} is a pair $\V = (S,T)$ where $S \neq \emptyset$ is a finite set of \emph{states} and $T\neq \emptyset$ is a finite set of \emph{transitions} of the form $(p,\mu,\bu)$ where $p \in S$, $\mu$~is a rational probability distribution over~$S$, and $\bu \in \Z^d$.
	
	Let $d \geq 1$. A \emph{VASS MDP} with counters \(\countersset \) is a tuple $\A = \ce{Q, (Q_n,Q_p),T,P}$, where 
	\begin{itemize}
		\item $Q \neq \emptyset$ is a finite set of \emph{states} split into two disjoint subsets $Q_n$ and $Q_p$ of \emph{nondeterministic} and \emph{probabilistic} states,
		\item $T \subseteq Q \times \Z^\countersset\times Q$ is a finite set of \emph{transitions} such that, for every $p \in Q$, the set $\tout(p) \subseteq T$ of all transitions of the form $(p,\bu,q)$ is non-empty.
		\item $P$ is a function assigning to each $t \in \tout(p)$ where $p \in Q_p$ a positive rational probability so that  $\sum_{t \in \tout(p)} P(t) =1$. 
	\end{itemize}
	% %A \emph{$d$-dimensional VASS MDP} is an MDP where $L = Z^d$.
\end{definition}

A  VASS MDP $\A$ is a \emph{VASS Markov chain} if  \(Q_n=\emptyset \). $\A$ is a non-probabilistic VASS if \(Q_p=\emptyset \). We say \(\A\) is \emph{\(d\)-dimensional} if \(|\countersset|=d \). The encoding size of $\A$ is denoted by $\size{\A}$, where the integers representing counter updates are written in binary and probability values are written as fractions of binary numbers. For every $p \in Q$, we use $\tin(p) \subseteq T$ to denote the set of all transitions of the form $(q,\bu,p)$. The update vector $\bu$ of a transition $t = (p,\bu,q)$ is also denoted by $\bu_t$. Let \(C_1,\dots,C_n \) be subsets of \(\countersset \), we use \(\A^{C_1,\dots,C_n}\) to denote the VASS MDP obtained from \(\A \) by removing any counter not in \(\bigcup_{i=1}^n C_i \).

A \emph{configuration} of $\A$ is a pair $p\bv$, where $p \in Q$ and $\bv \in \Z^\countersset$. If some component of $\bv$ is negative, then $p\bv$ is \emph{terminal}.

% The set of all configurations of $\A$ is denoted by $\conf(\A)$. 
%The \emph{size} of  $p\bv \in \conf(\A)$ is $\size{p\bv} = \size{\bv} = \max \{|\bv(i)| : 1\leq i \leq d\}$. Given $n\in \N$, we say that $p\bv$ is \emph{$n$-bounded} if $\size{p\bv}\leq n$. 

A \emph{\finitecomputation}in $\A$ is a finite sequence of the form
$\alpha=p_0\bv_0, p_1 \bv_1,\ldots, p_n \bv_n$ where $(p_{i},\bv_{i+1}-\bv_i,p_{i+1}) \in T$  for all $i<n$.  The \emph{length} of $\alpha$ is defined as $\length(\alpha) = n$. 
%We say that $\V$ is strongly connected  if for each pair of distinct states $p,q$ there is a finite path from $p$ to $q$. 
An \emph{infinite computation} in $\A$ is an infinite sequence $\pi = p_0\bv_0, p_1 \bv_1,p_2 \bv_2,\ldots$ such that every finite prefix of $\pi$ is a \finitecomputation in~$\A$. Let $\Term(\pi)$ be the least $j$ such that $p_j\bv_j$ is terminal. If there is no such $j$, we put  $\Term(\pi) = \length(\pi)$. We say that a computation is \emph{terminal} if it contains a terminal configuration.

We say that $q$ is \emph{reachable} from~$p$ in \(\A\) if there exists a computation \(p\bv,\ldots, q \bv' \) in \(\A \). We say that \(\A \) is \emph{strongly connected} if for each \(p,q\in Q \) \(q\) is reachable from \(p\) in \(\A\).

% A state $q$ is an \emph{immediate successor} of a state $p$ if there is a transition $(p,\ell,q)$ for some $\ell \in L$. A \emph{finite path} in $\A$ of length~$n$ is a finite sequence of the form
% $p_0,\ell_1,p_1,\ell_2,p_2,\ldots,\ell_n,p_n$ where $n \geq 0$ and
% $(p_i,\ell_{i+1},p_{i+1}) \in T$ for all $0 \leq i < n$. If $n \geq 1$ and $p_0 = p_n$, then
% $\pi$ is a \emph{cycle}. An MDP is \emph{strongly connected} if for each pair of distinct states $p,q$ there is a finite path from $p$ to $q$. An \emph{infinite path} in $\A$ is an infinite sequence $p_0,\ell_1,p_1,\ell_2,p_2,\ldots$ such that $p_0,\ell_1,p_1,\ldots,\ell_n,p_n$ is a finite path for every $n \geq 0$. For a finite sequence of the form $\pi = p_0,\ell_1,p_1,\ell_2,p_2,\ldots,\ell_n,p_n$ and a finite or infinite sequence of the form $\varrho = q_0,\kappa_1,q_1,\kappa_2,\ldots$, where $\pi$ and $\varrho$ are not necessarily paths in $\A$, we use $\pi \odot \varrho$ to denote the \emph{concatenated sequence} $p_0,\ell_1,\ldots,\ell_n,p_n,\kappa_1,q_1,\kappa_2,\ldots$ (we do not require $p_n = q_0$). If $\pi,\varrho$ are both paths in $\A$ and $p_n = q_0$, then $\pi \odot \varrho$ is also a path in $\A$.

%Every infinite path $p_0,t_1,p_1,t_2,\ldots$ and every initial vector $\bv \in \Z^d$ determine the corresponding \emph{infinite computation} of $\A$, i.e., the sequence of configurations $p_0\bv_0, p_1 \bv_1, p_2 \bv_2,\ldots$ such that $\bv_0 =\bv$ and $\bv_{i+1} = \bv_i + \bu_{t_{i+1}}$.  A \emph{finite computation} is a finite prefix of a computation. 

 A \emph{strategy} of \(\A \) is a function $\sigma$ assigning to every \finitecomputation $p_0\bv_0,p_1 \bv_1,\ldots,p_n \bv_n$ such that $p_n \in Q_n$ a probability distribution over~$\tout(p_n)$. A strategy is \emph{counterless-memoryless (cM)} if it depends only on the last state $p_n$, and \emph{deterministic (D)} if it always returns a Dirac distribution. We denote by \(\stratsMD{\A} \) the set of all {\cMD}  strategies of \(\A\). 
Note that a {\cMD} strategy \(\sigma \) selects the same outgoing transition in every \(p\in Q_n \) every time \(p\) is visited, and hence we can  “apply” \(\sigma\) to \(\A \) by removing the other outgoing transitions, and declaring all states to be a probabilistic. The resulting VASS Markov chain is denoted by \(\A_\sigma \).  Every initial configuration $p\bv$ and every strategy $\sigma$ determine the probability space over computations initiated in $p\bv$ in the standard way.\footnote{See e.g. \cite{10.5555/1373322} for details on the ``standard way''.} We use $\prob^\sigma_{p\bv}$ to denote the associated probability measure.
 For a measurable function $X$ over computations, we use $\Exp^\sigma_{p \bv}[X]$ to denote the expected value of~$X$.

%Note that every computation uniquely determines its underlying path.

An \emph{end component (EC)} of $\A$ is a pair $(C,L)$ where $C \subseteq Q$ and  $L \subseteq T$ such that the following conditions are satisfied:
\begin{itemize}
	\item $C \neq \emptyset$;
	\item if $p \in C \cap Q_n$, then $\tout(p)\cap L\neq \emptyset$;
	\item if $p \in C \cap Q_p$, then $\tout(p)\subseteq L$;
	\item if $(p,\bu,q) \in L$, then $p,q \in C$;
	\item for all $p,q \in C$ we have that $q$ is reachable from $p$ and vice versa using only transitions from \(L\).
\end{itemize}

Note that if $(C,L)$ and $(C',L')$ are ECs such that $C \cap C' \neq \emptyset$, then $(C\cup C', L \cup L')$ is also an EC. Hence, every $p\in Q$ either belongs to a unique \emph{maximal end component} (MEC), or does not belong to any EC. Also observe that each MEC can be seen as a strongly connected VASS MDP.

%\textbf{Multi-component:} 
A \emph{multi-component} of \(\A \) is a vector \(\bx\in \mathbb{Q}^{T} \) that satisfies all the following conditions:
\begin{itemize}
	\item  \(\bx \geq \vec{0}\);
%	\item \(\bx \neq \vec{0} \);
	\item for each $p\in Q$ it holds \(
		\sum_{t\in \tout(p)}\bx(t)=\sum_{t \in \tin(p)} \bx(t)\);
	\item  for each \(p\in Q_p \), \(t\in \tout(p)\) it holds \(\bx(t)=P(t)\cdot \sum_{t'\in\tin(p)}\bx(t') \).
\end{itemize}
 \begin{remark}Intuitively,  \(\bx \) assigns non-negative flow to each transition, such that what flows into a state also flows out, while for each probabilistic state \(p\in Q_p\) the flow distribution on \(\tout(p) \)  has the the same ratios as the   probability distribution on \(\tout(p) \) given by \(P\).
	\end{remark}

 The \emph{effect} of a multi-component \(\bx \) is defined as  \(\Delta(\bx)=\sum_{t \in T} \bx(t)\cdot \bu_t \).  Each multi-component induces a VASS MDP \(\A_\bx \) created from \(\A \) by removing all transitions \(t\) with \(\bx(t)=0 \).  We say that \(\bx\) is \emph{centered} in a state \(p\) if \(\sum_{t\in \tout(p)}\bx(t)=1 \), and we use \(p_\bx \) to denote some state in which \(\bx \) is centered (if it exists). We say that \(\bx\) is a \emph{component} if \(\bx \) is centered in some state \(p\), and \(\A_\bx \) corresponds to a MEC of \(\A_\sigma \) for some \(\sigma\in \stratsMD{\A}\). 
%  Note that given two components \(\bx,\by \) corresponding to the same MEC \(\MEC\) of \(\A_\sigma \)  for \(\sigma\in \stratsMD{\A} \) but centered in different states, there exists \(a\in \mathbb{Q} \) such that \(\bx=a\by \). 
 
% We say \emph{\(\bx \) is strongly connected} if \(\A_\bx \) is strongly connected.

Given a component \(\by \) centered in \(p\) let \(\realeffect_\by \) be the random variable representing the counters vector in the computation under the only strategy \( \sigma\) of \(\A_\by \) from initial configuration \(p\vec{0} \) upon revisiting \(p\) for the first time. We use \(\support(\by)=\{\bv\mid \prob_{p\vec{0}}^\sigma[\realeffect_\by=\bv]\neq 0 \} \) to denote the support of \(E_\by \), and \(\support^{C_1,\dots,C_n}(\by) \) to denote the restriction of \(\support(\by) \) onto the counters from \(\bigcup_{i=1}^nC_i \). Given a counter \(c \) we say that:
\begin{itemize}
	\item \(\by \) is \emph{increasing on \(c\)} if \(\E_{p\vec{0}}^\sigma[\realeffect_\by(c)]>0 \),
	\item \(\by \) is \emph{decreasing on \(c\)} if \(\E_{p\vec{0}}^\sigma[\realeffect_\by(c)]<0 \),
	\item \(\by \) is \emph{zero-bounded on \(c\)} if \(\E_{p\vec{0}}^\sigma[\realeffect_\by(c)]=0 \) and \(\prob_{p\vec{0}}^\sigma[\realeffect_\by(c)=0]=1 \),
		\item \(\by \) is\emph{ zero-unbounded on \(c\)} if \(\E_{p\vec{0}}^\sigma[\realeffect_\by(c)]=0 \) and \(\prob_{p\vec{0}}^\sigma[\realeffect_\by(c)=0]\neq 1 \).
\end{itemize}

Given a component \(\by \) of \(\A \) centered in \(p\), we use \(\hat{\by} \) to denote the component of \(\A_{\hat{\by}} \) centered in \(p\) where \(\A_{\hat{\by}} \) is created from \(\A_\by \) by replacing every transition \((q,\bu,p)\in \tin(p) \) with \((q,\bu-\E_{p\vec{0}}^\sigma[\realeffect_\by(c)],p) \).
 We use \(co-\hat{\by} \) to denote the component of \(\A_{co-\hat{\by}} \) centered in \(p\) where \(\A_{co-\hat{\by}} \) is created from \(\A_\by \) by replacing every transition \((q,\bu,r)\) with \((q,\E_{p\vec{0}}^\sigma[\realeffect_\by(c)],r) \) if \(r=p \) and with \((q,\vec{0},r) \) otherwise.  Note that  \(\hat\by \) is either zero-bounded or zero-unbounded on every \(c\), whereas \(co-\hat{\by} \) is never zero-unbounded on any \(c\).

\begin{remark}Components are a generalization of simple cycles from non-probabilistic systems onto MDPs, and multi-components are a generalization of multi-cycles. Just as multi-cycles can be seen  as a conical combination of simple cycles, so can multi-cycles be seen as a conical combination of components (see Appending~\ref{app-lemma-decompose-multicomponents-into-components}). Also, just as we can say ``we iterate a cycle \(k\) times'' we also say ``we iterate a component \(k\) times'', where a single iteration of a component \(\by \) represents a computation on \(\A_\by \) started from \(p_\by \) until the first time \(p_\by \) is revisited (hence \(E_\by \) represents the effect of a single iteration of \(\by \)). This is lifted onto multi-components in the same way as iterations of cycles are lifted onto multi-cycles.
	\end{remark}

%We use \((\M_{\hat{\by}},p) \) to denote the pointing pair corresponding to \(\hat{\by} \). 

%%
%\begin{figure}[h]
%%	\vspace{0.1cm}
%	\begin{tabular}{|c|}
%		\hline
%		{\begin{minipage}[c]{0.38\textwidth}\small
%				\vspace{0.2cm}
%				Constraint system~(A):
%				
%				\vspace{0.2cm}
%				$\bx \in \mathbb{Q}^{T}$ such that
%				\begin{align}
%					%	\sum_{t \in T} \bx(t) \bu_t  & \ge \vec{0} \nonumber\\
%					\bx & \ge \vec{0} \nonumber
%					%F \bx & = \vec{0} \nonumber
%				\end{align}
%				and for each $p\in Q$
%				\begin{align}
%					\sum_{t\in \tout(p)}\bx(t)&=\sum_{t \in \tin(p)} \bx(t)\nonumber
%				\end{align}
%				%	
%				%			
%				%		
%				%			\begin{align}
%					%				D \bx & \ge \vec{0} \nonumber\\
%					%				\bx & \ge \vec{0} \nonumber\\
%					%				F \bx & = \vec{0} \nonumber
%					%			\end{align}
%				and for each $p\in Q_p$, $t\in \tout(p)$
%				\begin{align}
%					\bx(t) & =P(t) \cdot \sum_{t'\in \tin(p)} \bx(t') \nonumber
%				\end{align}
%				\vspace*{1em}
%		\end{minipage}}\\
%	\hline
%	\end{tabular}
%	\caption{Constraint system for defining multi-components on VASS MDPs.}
%	\label{fig-multicomp}
%\end{figure}

%
%\FloatBarrier

\subsection{Asymptotic Complexity Measures for VASS MDPs}
\label{sec-asymptotic-measures}
\begin{figure}
	\parbox[c]{.5\textwidth}{
		\begin{tabbing}
			\hspace*{1em} \= \hspace*{2em} \= \hspace*{4em} \= \hspace*{4em} \= \kill
			\> \textbf{input} $N$\\
			\> \textbf{repeat}\\
			\>\> \textbf{random choice:}\\
			\>\>\> $0.5:\ \ N:=N+1;$\\
			\>\>\> $0.5:\ \ N:=N-1;$\\
			\> \textbf{until} $N=0$
	\end{tabbing}}%\hspace*{4em}
	\parbox[c]{.5\textwidth}{\hspace*{4em}
		\begin{tikzpicture}[x=2cm, y=2cm]
			\node [stoch,label={[shift={(0,-1)}]$\A$}] (r) at (3,0)  {$p$};
			\draw [tran,loop right] (r) to node[right] {$0.5: +1$} (r); 
			\draw [tran,loop left]  (r) to node[left] {$0.5: -1$} (r); 
		\end{tikzpicture}
	}
	\caption{A probabilistic program with infinite expected running time for every $N\geq 1$, and its $1$-dimensional VASS MDP model~$\A$.}
	\label{fig-prob-prg}
\end{figure}

Before we introduce asymptotic estimates, let us consider a  simple motivating example. Consider the simple probabilistic program of Fig.~\ref{fig-prob-prg}. The program inputs a positive integer $N$ and then repeatedly increments/decrements~$N$ with probability $0.5$ until $N=0$. One can easily show that for every $N \geq 1$, the program terminates with probability one, and the expected termination time is \emph{infinite}. Based on this, one may conclude that the execution takes a very long time, independently of the initial value of~$N$. However, this conclusion is \emph{not} consistent with practical experience gained from trial runs. The program tends to terminate ``relatively quickly'' for small~$N$, and the termination time \emph{does} depend on~$N$. 
Hence, the function assigning $\infty$ to every $N \geq 1$ is \emph{not} a faithful characterization of the asymptotic growth of termination time. \cite{AKCONCUR23} proposes an alternative characterization based on the observations that:
\begin{itemize}
	\item For every $\varepsilon >0$, the probability of all runs terminating after more than $N^{2+\varepsilon}$ steps approaches \emph{zero} as $N \to \infty$.
	\item For every $\varepsilon >0$, the probability of all runs terminating after more than $N^{2-\varepsilon}$ steps approaches \emph{one} as $N \to \infty$.
\end{itemize}
Since the execution time is ``squeezed'' between $N^{2-\varepsilon}$ and $N^{2+\varepsilon}$ for an arbitrarily small $\varepsilon > 0$ as $N \to \infty$, it can be characterized as ``asymptotically quadratic''. This analysis is in accordance with experimental outcomes (see e.g. \cite{AKN-expnotmeetexp}).

\subsection{Complexity of VASS Computations}
\label{sec-comp-VASS-runs}

Let $\A = \ce{Q, (Q_n,Q_p),T,P}$ be a VASS MDP with counters \(\countersset \), $c \in \countersset$, and $t \in T$. For every computation  $\pi = p_0 \bv_0,p_1 \bv_1,p_2 \bv_2,\ldots$, we put
% \begin{eqnarray*}
	%     \mathit{Length}(\pi)    & = & \Term(\pi)\\
	% 	\mathit{Maxval}[c](\pi) & = & \sup \{\bv_i(c) \mid 0 \leq i < \Term(\pi)\}\\
	%     \mathit{Occur}[t](\pi) & = & \mbox{the total number of all $0 \leq i < \Term(\pi)$ such that $(p_i,\bv_{i{+}1}{-}\bv_i,p_{i+1}) = t$}
	% \end{eqnarray*}
% Furthermore, for every $p \in Q$ we define the \emph{termination}, \emph{$c$-counter}, and \emph{$t$-transition complexity in $p$}, denoted by $\calL$, $\calC[c]$, and $\calT[t]$, respectively, as follows: 
\begin{eqnarray*}
	\calL_\A(\pi)      & = & \Term(\pi)\\
	\calC_\A[c](\pi) & = & \sup \{\bv_i(c) \mid 0 \leq i < \Term(\pi)\}\\
	\calT_\A[t](\pi) & = & \mbox{the total number of all $0 \leq i < \Term(\pi)$ such that $(p_i,\bv_{i{+}1}{-}\bv_i,p_{i+1}) = t$}
\end{eqnarray*}

We refer to the functions $\calL_\A$, $\calC_\A[c]$, and $\calT_\A[t]$ as \emph{termination}, \emph{$c$-counter}, and \emph{$t$-transition complexity}  respectively.

Note that $\calL_\A$, $\calC_\A[c]$, and $\calT_\A[t]$ are the complexity measures for VASS runs used in previous works \cite{BCKNVZ:VASS-linear-termination,Zuleger:VASS-polynomial,KLV:VASS-Grzegorczyk,Leroux:Polynomial-termination-VASS,AK:VASS-polynomial-termination,AKCONCUR23}. These functions can be seen as variants of the standard time/space complexities for Turing machines.

Let $\F$ be one of the complexity functions defined above. In VASS abstractions of computer programs, the input is represented by initial counter values, and the input size corresponds to the maximal initial counter value. The existing works on \emph{non-probabilistic} VASS concentrate on analyzing the asymptotic growth of the functions $\F_{\max} : \Nset \to \Nset_\infty$ where
\begin{eqnarray*}
	\F_{\max}(n)     & = & \sup\{\F(\pi) \mid \pi \mbox { is a computation initiated in $p\vec{n}$ where $p \in Q$}\}
\end{eqnarray*}
For VASS MDP, we can generalize $\F_{\max}$ into $\F_{\eexp}$ as follows:
\begin{eqnarray*}
	\F_{\eexp}(n)     & = & \sup\{\Exp_{p\vec{n}}^{\sigma}[\F] \mid \sigma \mbox { is a strategy of $\A$}, p \in Q\}
\end{eqnarray*}
Note that for non-probabilistic VASS, the values of $\F_{\max}(n)$ and $\F_{\eexp}(n)$ are the same. However, the function $\F_{\eexp}$ suffers from the deficiency illustrated in the motivating example at the beginning of Section~\ref{sec-asymptotic-measures}. To see this, consider the one-dimensional VASS MDP $\A$ modeling the simple probabilistic program (see Fig.~\ref{fig-prob-prg}). For every $n \geq 1$ and the only (trivial) strategy $\sigma$, we have that $\prob^\sigma_{p \bn}[\Term < \infty] = 1$ and $\calL_{\eexp}(n) = \infty$. However, the practical experience with trial runs of $\A$ is the same as with the original probabilistic program (see Section~\ref{sec-asymptotic-measures} above).

\subsection{Asymptotic Estimates}
\label{sec-estimates}

In this section, we introduce asymptotic  estimates allowing for a precise analysis of the asymptotic growth of the termination, $c$-counter, and $t$-transition complexity, especially when their expected values are infinite for a sufficiently large input.

\begin{definition}
	\label{def-estimates}
	Let \(\A \) be a VASS MDP,  $f : \mathbb{R} \to \mathbb{R}$, and $\F$ be one of $\calL_\A$, $\calC_\A[c]$ or $\calT_\A[t]$.
	
	We say that $f$ is a \emph{lower asymptotic estimate of $\F$} if for every $\varepsilon > 0$ there exists $p \in Q$ and a strategy $\sigma$ such that 
	\[
	\liminf_{n \to \infty} \ \prob^{\sigma}_{p \vec{n}}[\F \geq f(n^{1-\varepsilon}) ] \ = \ 1\,
	\]
	Similarly, we say that $f$ is an \emph{upper asymptotic estimate of $\F$} if for every $\varepsilon>0$, every $p \in Q$, and every strategy $\sigma$ it holds
	\[
	\limsup_{n \to \infty} 	\ \prob^{\sigma}_{p \vec{n}}[\F \geq f(n^{1+\varepsilon}) ] \ = \ 0
	\]	
	If there is no upper asymptotic estimate of $\F$, we say that $\F$ is \emph{unbounded from above}. If every \(f:\R\rightarrow \R \) is a lower asymptotic estimate of $\F$, we say that $\F$ is \emph{unbounded from below}. If $\F$ is unbounded from below as well as unbounded from above we say that $\F$ is \emph{unbounded}. Finally, we say that $f$ is a \emph{tight asymptotic estimate of $\F$} if it is both a lower asymptotic estimate and an upper asymptotic estimate of $\F$.
\end{definition}

The above definition is based on the one introduced in \cite{AKCONCUR23}.
An alternative, more intuitive, definition of asymptotic estimates can be obtained using another natural notion, inspired by \cite{AKN-expnotmeetexp}. Consider once again the motivating example from
Section~\ref{sec-asymptotic-measures}.  One can ask the following natural question:

\begin{quote}
	\emph{How many steps of the program in Fig.~\ref{fig-prob-prg} must be executed for a given initial value of $N$ so that the probability of termination is at least \(p\)?}
\end{quote}

This question makes sense for every fixed $p<1$.  Formally, consider the function \(f_p^\F:\mathbb{N}\rightarrow \mathbb{N}_{\infty} \) defined for each \(n\) as the smallest integer such that for every $q \in Q$, and every strategy $\sigma$ it holds \(\Prob_{q\vec{n}}^\sigma[\F\leq  f_p^\F(n)]\geq  p \) (or \(f_p^\F(n)=\infty \) if no such integer exists). We call \(f_p^\F \) \emph{fixed probability bound}. In the next theorem we show that fixed probability bounds are  closely tied with asymptotic estimates.

 \begin{theorem}\label{observation-f-p-f}
 	Let \(f:\mathbb{R}\rightarrow\mathbb{R} \) be such that \(\lim_{n\rightarrow\infty} \frac{f( n)}{f( n^{1+\epsilon})}=0\) for every \(\epsilon>0 \). Then:
 	\begin{itemize}
 		\item \(f\) is a lower asymptotic estimate of \(\F \) iff for every \(\epsilon>0 \) and every \(p<1 \) it holds \(f_p^\F\in \Omega(f(n^{1-\epsilon})) \);
 		\item \(f\) is an upper asymptotic estimate of \(\F \) iff for every \(\epsilon>0 \) and every \(p<1 \) it holds \(f_p^\F\in  \calO(f(n^{1+\epsilon})) \).
 	\end{itemize}
 \end{theorem}

The proof of Theorem~\ref{observation-f-p-f} can be found in Appendix~\ref{app-observation-f-p-f}. Note that the restriction of \(\lim_{n\rightarrow\infty} \frac{f( n)}{f( n^{1+\epsilon})}=0\) for all \(\epsilon>0 \) is not particularly restrictive. For instance, it holds for any polynomial or exponential function.

\section{Strongly Connected VASS MDPs}
\label{section-main}
In this section we present the following main theorem of this paper.

\begin{theorem}\label{theorem-main}
		Let \(\A=\ce{Q, (Q_n,Q_p),T,P} \) be a strongly connected \(d\)-dimensional VASS MDP. Let \(c\) be a counter and \(t\) a transition of \(\A \). Then for each \(\F\in \{\calC_\A[c],\calT_\A[t],\calL_\A \} \) one of the following holds:
	\begin{itemize}
		\item there exists \(k\in \mathbb{N} \), \(k\leq 2^d\cdot 3^{|T|}\) such that  \(n^k \) is a tight asymptotic estimate of \(\F\);
		\item OR \(2^{\sqrt{n}} \) is a lower asymptotic estimate of \(\F \). 
	\end{itemize}
	Furthermore, it is decidable in time polynomial in \(\size{\A} \) which of these cases holds, and for the first case the value of \(k\) can be computed in time polynomial in \(\size{\A}\).
\end{theorem}

%
%Assume that for each counter/transition complexity we have either a tight estimate \(n^{i} \) for some \(i<k\) (sets \(C_1,C_2,\dots,C_{k-1}\), \(T_1,T_2,\dots,T_{k-1} \) ) or a lower estiamte of \(n^{k} \) (sets \(C_{\geq k}, T_{\geq k} \)).
%
%Furthermore assume that for each BSCC \(B' \) of \(\A \) with pointing complexity of \(B' \) having lower estimate of \(n^{i} \) we have a lower esitmate of \(n^i\) for the pointing complexity of \(\hat{B}' \) on \(\A_{+\hat{B}} \). 

%Let \(B \) be a BSCC of an MD strategy on \(\A\) for which we want to show that either the pointing complexity of \(B\) in \(\A \) has upper estimate \(n^{k}\) or a lwoer estiamte of \(n^{k+1} \). 
%
%Let \(p_B\) be some further fixed state of \(B\), and let \(Y_B \) be the random variable denoting the effect of a computation from \(p_B \) under \(B \) until \(p_B\) is revisited for the first time. Let \(D=support(Y_B) \).

Note that it suffices to prove Theorem~\ref{theorem-main} only for \(\calC_\A[c] \) and \(\calT_\A[t] \), as \(\calL_\A \) can be expressed by adding a new step-counter \(sc\) to \(\A \) which is increased by every transition of \(\A\). It then holds that \(\calL_\A(n)=\calC_\A[sc](n)-n \). While we could similarly replace  \(\calT_\A[t] \) with a \(t\)-transition counter that is increased only by \(t\), our approach requires us to analyze  \(\calT_\A[t] \) separately.

The proof of Theorem~\ref{theorem-main} is split as follows: In Section~\ref{sec-informal} we give an informal description of the core ideas used in our approach. In Section~\ref{sec-systems} we describe constraint systems \hyperref[fig-systems]{(I)} and \hyperref[fig-systems]{(II)} that are a key concept to analyzing the asymptotic behavior of VASSes. We then give a formal proof of Theorem~\ref{theorem-main} in Section~\ref{section-complexity}. 

\subsection{Informal Description}
\label{sec-informal}

We combine the approaches of \cite{Zuleger:VASS-polynomial} and \cite{AKCONCUR23} together with novel methods.
The main problem with applying the methods used in \cite{Zuleger:VASS-polynomial} for analyzing strongly connected non-probabilistic  VASS onto VASS MDPs is that this method requires that any cycle/component whose effect on some counter \(c\) is \(0 \) can be iterated sufficient number of times without decreasing \(c\) by more than a constant in total. While this is true for non-probabilistic VASS, on VASS MDPs this holds only for the class of VASS MDPs which  contain no component that is zero-unbounded on some counter. 

The main technique of \cite{AKCONCUR23} we utilize is then that of {\cMD} decomposition, which allows us to view any computation on a VASS MDP \(\A \) as an interweaving of a constant number (that depends only on \(\A \)) of computations on VASS Markov chains \(\A_\sigma \) for various \(\sigma\in \cMD(\A) \), while every strategy of \(\A \) can be seen as if instead of selecting the next transition it selects which one of these VASS Markov chains is to take a single computational step next.\footnote{We refer to \cite{AKCONCUR23} or Appendix~\ref{section-additional-definitions} for detailed description of the {\cMD} decomposition/pointing VASS.} These can in turn  be seen as an interweaving of a constant number of computations on various \(\A_\by \) for components \(\by\) of \(\A\). This allows us to view any computation on a strongly connected VASS MDP as simply ``switching'' between a constant number of entirely separate computations, each of which takes place on the VASS Markov chain \(\A_\by \) for some component \(\by \) of \(\A\). 

%Note that each single of these individual computations take place on a VASS Markov chain, and hence they do not care about counter values. 

The main idea used in this paper is to split the computations on \(\A\) into two parallel parts by utilizing the {\cMD} decomposition. First, using the  {\cMD} decomposition we split the computation into the individual computations on \(\A_\by \) for components \(\by \) of \(\A \), and then we split each of these individual computations into two parallel parts. The computation on \(\A_\by \) is split into a ``probabilistic'' part representing the same computation on \(\A_{\hat{\by}} \), and a second ``deterministic'' part representing the same computation on \(\A_{co-\hat{\by}} \). Note that the sum of the effects of both parts is exactly the effect on \(\A_\by \). An important observation is that the ``probabilistic'' part is either zero-bounded or zero-unbounded on each counter, whereas the ``deterministic'' part can never be zero-unbounded on any counter. Hence we can apply the method  of \cite{AKCONCUR23} (with few modifications) to analyze the asymptotic behavior of the ``deterministic'' part, and it remains only to develop new techniques to analyze the asymptotic behavior of the ``probabilistic'' part.
% We also show that the individual computation on \(\A_\by \) has a lower asymptotic estimate of \(n^{k+1}\) on its length iff both the ``probabilistic'' and the ``deterministic'' parts have a lower asymptotic estimate of \(n^{k+1}\) on their length simultaneously, while if either of the two parallel parts has an upper asymptotic estimate of \(n^k\) then the individual computation on \(\A_\by \) also has an upper asymptotic estimate of \(n^k\).

% We show that for each \(k\in \mathbb{N}\) and each component \(\by\), either both of the parallel parts have simultaneously a lower asymptotic estimate of \(n^{k+1} \) on their length, in which case we show how to construct a strategy that produces a computation whose corresponding computation on  

We do this by analyzing each of the computations on \(\A_{\hat{\by}} \) separately. Since counters on which \(\hat{\by} \) is zero-bounded can never be decreased by more than a constant in \(\A_{\hat{\by}} \), it suffices to analyze the behavior for counters \(c\) on which \(\hat{\by} \) is zero-unbounded. The key notion behind analyzing the behavior of \(\A_{\hat{\by}} \) for such \(c\) is that of a reset. Note that a computation on \(\A_{\hat{\by}} \), when \(c\) is initialized to \(n\), reaches a negative value for \(c\) after roughly \(n^2 \) steps as per \cite{AKCONCUR23} (i.e., the number of steps has a tight asymptotic estimate of \(n^2\)). Hence if we wanted to iterate \(\A_{\hat{\by}}\) for say \(n^{k+2} \) steps we would have to reset \(c\) back to its initial value of \(n\) every time it reaches \(0\), and we would need roughly \(n^k\) such resets. A reset represents a computation whose effect on \(c\) (and potentially other counters) on all of the ``deterministic''  parts summed together  is exactly the opposite of the effect of the  ``probabilistic'' part \(\A_{\hat{\by}}\) since its last reset, hence the two cancel themselves out. We show that either \(\A_{\hat{\by}} \) has an upper asymptotic estimate of \(n^k\) on its length, or the set \(\support_{\hat{\by}} \) restricted only to  counters upper bounded by \(n^{\lfloor\frac{k}{2}\rfloor} \) (note that larger counters do not need resets yet as \((n^{\lfloor\frac{k}{2}\rfloor+1})^2\geq n^{k+1} \)) can be reset in \(\A \) sufficient number of times to iterate \(\A_{\hat{\by}} \) rooughly \(n^{k+1} \) times. We actually show that the set all the counter vectors that can be reset often enough forms a vector space, which allows for efficient computation.

Similarly to the approach of \cite{Zuleger:VASS-polynomial}, our algorithm iteratively classifies each counter and transition as either having an upper asymptotic estimate of \(n^{k}\) or a lower asymptotic estimate of \(n^{k+1}\) for its complexity, starting with \(k=1\) and iteratively classifying for higher and higher \(k\) until it either clasifies every counter and transition with a tight asymptotic estimate, or it obtains a lower asymptotic estimate of \(n^{2^d\cdot 3^{|T|}} \) for the remainder.  Then similarly to \cite{Zuleger:VASS-polynomial}, we show that there exists a  multi-component that is an \emph{exponential iteration scheme}, which is a VASS MDP analogy to an \emph{iteration scheme} for non-probabilistic VASS from \cite{Leroux:Polynomial-termination-VASS}, which can be used to iterate the transitions \(t\) and pump the counters \(c\) with a lower asymptotic estimate of \(n^{2^d\cdot 3^{|T|}} \) for \(\calT_\A[t] \) and \(\calC_\A[c] \) to \(2^{n^{\frac{1}{2}-\epsilon}}\) with high enough probability, thus obtaining the exponential lower asymptotic estimate.

\subsection{Constraint systems \hyperref[fig-systems]{(I)} and \hyperref[fig-systems]{(II)}}
\label{sec-systems}
%Our analysis uses of the systems of linear equations on Fig.~\ref{fig-systems} taken from \cite{AKCONCUR23}.
The starting point of our analysis is the dual constraint system (Fig.~\ref{fig-systems}) used in \cite{Zuleger:VASS-polynomial} and adapted to VASS MDPs in \cite{AKCONCUR23}. 
%, except for the last maximization objective of \hyperref[fig-systems]{(II)} which is new. 

\begin{figure}[ht]
%	\vspace{0.7cm}
	\begin{tabular}{|c|c|}
		\hline
		{\begin{minipage}[c]{0.38\textwidth}\small
%				\vspace{0.2cm}
				\vspace{-0.84cm}
				Constraint system~(I):
				
				\vspace{0.2cm}
				Find $\bx \in \mathbb{Z}^{T}$ such that
				\begin{align}
					\sum_{t \in T} \bx(t)\cdot  \bu_t  & \ge \vec{0} \nonumber\\
					\bx & \ge \vec{0} \nonumber
					%F \bx & = \vec{0} \nonumber
				\end{align}
				
				and for each $p\in Q$
				\begin{align}
					\sum_{t\in \tout(p)}\bx(t)&=\sum_{t \in \tin(p)} \bx(t)\nonumber
				\end{align}
				%	
				%			
				%		
				%			\begin{align}
					%				D \bx & \ge \vec{0} \nonumber\\
					%				\bx & \ge \vec{0} \nonumber\\
					%				F \bx & = \vec{0} \nonumber
					%			\end{align}
				and for each $p\in Q_p$, $t\in \tout(p)$
				\begin{align}
					\bx(t) & =P(t) \cdot \sum_{t'\in \tout(p)} \bx(t') \nonumber
				\end{align}
%				\vspace*{1em}
				
				\textbf{Objective:} \emph{Maximize}\\
				\begin{itemize}
					\item the number of valid inequalities of the form 
					\[\sum_{t \in T} \bx(t)\cdot \bu_t(c) > 0,\] where \(c\in \countersset\).
					\item the number of valid inequalities of the form $\bx(t) > 0$.
				\end{itemize}
				\vspace*{2em}
		\end{minipage}}
		&
		{\begin{minipage}[c]{0.55\textwidth}
					\vspace{0.2cm}
%				\vspace{-0.68cm}
%					\vspace{-0.5cm}
				Constraint system (II):
				
				\vspace{0.2cm}
				Find $\by \in \mathbb{Z}^\countersset,\bz \in \mathbb{Z}^{Q}$ such that
				\begin{align}
					\by & \ge  \vec{0} \nonumber\\
					\bz & \ge  \vec{0} \nonumber
					%	D^T \by + F^T \bz & \le \vec{0} \nonumber
				\end{align}
				and for each $(p,\bu,q)\in T$ where $p\in Q_n$ 			
				\[\bz(q)-\bz(p)+\sum_{c\in\countersset} \bu(c)\cdot \by(c) \leq 0\]
				and for each $p\in Q_p$  
				\[\sum_{t= (p,\bu,q) \in \tout(p)}P(t)\cdot \big(\bz(q)-\bz(p)+\sum_{c\in \countersset}  \bu_t(c)\cdot \by(c)\big)\leq 0 \]
				\vspace*{0.8em}
				
				\textbf{Objective:}
				\emph{Maximize}\\
				\begin{itemize}
					\item the number of valid inequalities of the form $\by(c) > 0$,
					\item the number of transitions $t=(p,\bu,q)\in T$ such that $p\in Q_n$ and \[\bz(q)-\bz(p)+\sum_{c\in\countersset} \bu(c)\cdot \by(c)<0,\]
					\item the number of states $p\in Q_p$ such that 
					\[\sum_{t = (p,\bu,q) \in \tout(p)}P(t)\cdot \big(\bz(q)-\bz(p)+\sum_{c\in \countersset}  \bu(c)\cdot \by(c)\big)< 0\,\]
					%					\item the number of states $p\in Q_p$ such that there exists a transition \((p,\bu,q)\in \tout(p) \) such that
					%					\[\bz(q)-\bz(p)+\sum_{i=1}^{d}  \bu(i)\by(i)\neq 0\,.\]
				\end{itemize}
				\vspace*{1em}
		\end{minipage}}\\
		\hline
	\end{tabular}
	\caption{Constraint systems defined for a given VASS MDP $\A = \ce{Q, (Q_n,Q_p),T,P}$ with counters \(\countersset\).}
	\label{fig-systems}
\end{figure}
%\FloatBarrier

We observe that both~\hyperref[fig-systems]{(I)} and~\hyperref[fig-systems]{(II)} are always satisfiable (set all coefficients to zero) and that the solutions of both constraint systems are closed under addition.
Hence for both~\hyperref[fig-systems]{(I)} and~\hyperref[fig-systems]{(II)}, the set of inequalities for which the maximization objective is satisfied is the same for every optimal solution.
The maximization objectives can be implemented by suitable linear objective functions.
Thus both constraint systems can be solved
in PTIME over the integers as we can use linear programming over the rationales and then scale rational solutions to the integers by multiplying with the least common multiple of the denominators.

%Note that solutions of both~\hyperref[fig-systems]{(I)} and~\hyperref[fig-systems]{(II)} are closed under addition, and addition is conservative towards all of the maximization objectives. Therefore, both~\hyperref[fig-systems]{(I)} and~\hyperref[fig-systems]{(II)} have solutions maximizing the specified objectives, computable in polynomial time. Furthermore, all the maximal solutions satisfy the same maximization objectives.
%
% of the form \( \sum_{(p,\bu,q)\in T} \bx((p,\bu,q))\bu(c)>0\) or $\bx(t) > 0$  for (I), and of form $\by(c) > 0$, $\bz(q)-\bz(p)+\sum_{i=1}^{d} \bu(i)\by(i)<0$ where \( (p,\bu,q)\in T \) and \(p\in Q_n \), and $\sum_{(p,\bu,q)\in T(p)}P((p,\bu,q))(\bz(q)-\bz(p)+\sum_{i=1}^{d}  \bu(i)\by(i))< 0$ where \(p\in Q_p \) for (II). Note that finding a maximal solution can be done in polynomial time, as the number of such inequalities is polynomial, and finding some solution (or decide it's non-existence) can be done in polynomial time using linear programming.
%
For clarity, let us first discuss an intuitive interpretation of solutions of \hyperref[fig-systems]{(I)} and~\hyperref[fig-systems]{(II)}, starting with simplified variants obtained for non-probabilistic VASS. 

In a non-probabilistic VASS, a solution of~\hyperref[fig-systems]{(I)} can be interpreted as a \emph{multi-cycle}, i.e., as a collection of simple cycles $M_1,\dots, M_k$ together with weights $a_1,\ldots ,a_k$ such that the total effect of the multi-cycle $\sum_{i=1}^k a_i \cdot \mathit{effect}(M_i)$ is non-negative on every counter, where $\mathit{effect}(M_i)$ is the effect of $M_i$ on the counters. The objective of~\hyperref[fig-systems]{(I)} ensures that this multi-cycle includes as many transitions as possible, and the total effect of the multi-cycle is positive on as many counters as possible. The VASS MDP analogy of a multi-cycle is that of a multi-component, and the $M_1,\dots, M_k$ should be interpreted as components \(\by_1,\dots,\by_k \) with $\mathit{effect}(\by_i)=\Delta(\by_i)$. The objective of~\hyperref[fig-systems]{(I)} then maximizes the number of transitions included in the multi-component, and the number of counters where the effect of the multi-component is positive.

% To give some intuitive interpretation of solutions to these constraint systems, let us recall what the solutions mean for a non-stochastic VASS (for details see  \cite{Zuleger:VASS-polynomial}). For non-stocahstic VASS, a solution of (I) can be interpreted as a weighted multicycle, that is as cycles \(M_1,\dots, M_k \) along with weights \(a_1,\dots a_k \), such that the effect of the multicycle on all counters is non-negative, where the effect of the multicycle is defined as \(\sum_{i=1}^k a_ieffect(M_i) \), where \(effect(M_i) \) is the effect of \(M_i \). And the maximization objectives of (I) for non-stochastic VASS are then to find such multicycle which includes as many transitions as possible and whose effect is positive on as many counters as possible. For the stochastic VASS, the interpretation only changes such that the cycles \(M_1,\dots, M_k \) are now Markovian strategies on some ECs, and \(effect(M_i) \) is the expected effect per single step for \(M_i\) from the fixed point distribution associated to \(M_i \). Thus the maximization objectives are then to maximize the number of transitions that are included in some of the strategies \(M_1,\dots, M_k \) and to maximize the number of counters for which the expected effect of the "multicycle" is positive.

A solution of~\hyperref[fig-systems]{(II)} for non-probabilistic VASS can be interpreted as a ranking function on configurations defined by $\mathit{rank}(p\bv)=\bz(p)+\sum_{c\in\countersset} \by(c)\cdot \bv(c)$, such that the value of $\mathit{rank}$ cannot increase when moving from a configuration $p\bv$ to a configuration $q\bu$ using a transition $t=(p,\bu-\bv,q)$ (i.e., effect of every transition on \(\mathit{rank} \) is non-positive). The objective of~\hyperref[fig-systems]{(II)} ensures that as many transitions as possible decrease the value of $\mathit{rank}$, and $\mathit{rank}$  depends on as many counters as possible (i.e., \(\by(c)>0 \) for as many counters \(c\) as possible). For VASS MDPs, this interpretation changes only for the outgoing transitions $t$ of probabilistic states. Instead of considering the change of $\mathit{rank}$ caused by such $t\in \tout(p)$,  we now consider the expected change of $\mathit{rank}$ caused by executing a single step from $p$. The objective ensures that $\mathit{rank}$ depends on as many counters as possible, the value of $\mathit{rank}$ is decreased by as many outgoing transitions of non-deterministic states as possible, and the expected change of $\mathit{rank}$ caused by performing a single step is negative in as many probabilistic states as possible.

The key tool for our analysis is the following dichotomy: 

\begin{lemma}[Cited  from~\cite{AKCONCUR23}]\label{lemma:dichotomy}
	Let $\bx$ be a maximal solution to the constraint system~\hyperref[fig-systems]{(I)} and $\by,\bz$ be a maximal solution to the constraint system~\hyperref[fig-systems]{(II)}. Then, for each counter $c\in \countersset$ we have that either $\by(c)>0$ or $\sum_{t \in T} \bx(t)\cdot \bu_t(c)>0$, and for each transition $t = (p,\bu,q)\in T$ we have that
	\begin{itemize}
		\item if $p\in Q_n$ then either $\bx(t)>0$ or $\bz(q)-\bz(p)+\sum_{c\in\countersset} \bu(c)\cdot \by(c)<0$;
		\item if $p\in Q_p$ then either  $\bx(t)>0$ or \[\sum_{t'=(p,\bu',q') \in \tout(p)}P(t')\cdot \big(\bz(q')-\bz(p)+\sum_{c\in\countersset}  \bu'(c)\cdot \by(c)\big)< 0\] 
	\end{itemize}%states \(p\in Q\) and all transitions \(t\in T(p) \)
\end{lemma}

\subsection{Formal Description}
\label{section-complexity}

In this section we formally describe the procedure that performs the analysis as per Theorem~\ref{theorem-main}.  For the rest of the section we fix a strongly connected VASS MDP $\A = \ce{Q, (Q_n,Q_p),T,P}$ with counters \(\countersset \).

We use \(\tilde{C}_i\subseteq  \countersset \) to  denote the set of all counters with tight asymptotic estimate of \(n^i \) for \(\calC_\A[c] \), \(\tilde{C}_{k+}\subseteq  \countersset \) to denote the set of all counters with lower asymptotic estimate of \(n^k\) for \(\calC_\A[c] \), and we use \(T_i\subseteq T \) to denote the set of all transitions with lower asymptotic estimate of \(n^i \) for \(\calT_\A[t] \). We denote by \(\A_{T_i} \) the VASS MDP \(\A\) restricted only to transitions from \(T_i \). 

Let VASS MDPs \(\A_1,\A_2,\dots \) be such that each \(\A_i \) is obtained from \(\A_{T_i} \) by ``creating local copies'' for each counter \(c\in \bigcup_{j=1}^{i-1}\tilde{C}_j\) in the following way: let \(1\leq j\leq  i-1 \), \(c\in \tilde{C}_j \), and let \(B_1,\dots,B_w \) be all the MECs of \(\A_{T_{i-j}} \). Then \(\A_i \) instead of the counter \(c\) contains counters \(c_1,\dots,c_w \) such that for every transition \(t=(p,\bu,q) \) of \(\A_{T_i} \), in \(\A_i \) this transition is changed into \((p,\bu_i,q) \) where \(\bu_i(c_x)=\bu(c) \) if \(t\) is a transition of \(B_x \), and \(\bu_i(c_x)=0 \) otherwise.  Formally, \(\A_i=\ce{Q, (Q_n,Q_p),T',P'} \) is a VASS MDP with counters \(\countersset_i \) where:
\begin{itemize}
	\item \(\countersset_i=\tilde{C}_{i+}\cup \bigcup_{j=1}^{i-1} \bigcup_{\MEC; \MEC \textit{ is a MEC of }\A_{i-j}} \bigcup_{c\in \tilde{C}_j} \{c_\MEC \}\),
	\item for each \((p,\bu,q)\in T \) it holds \((p,\bu_i,q)\in T' \) and \(P'((p,\bu_i,q))=P((p,\bu,q)) \) where \(\bu_i\) is defined as \(\bu_i(c)=\bu(c) \) for all \(c\in \tilde{C}_{i+} \), and for each \(1\leq j\leq i-1 \), each \(c\in \tilde{C}_{j} \), and each MEC \(\MEC \) of \(\A_{i-j} \) it holds \(\bu_i(c_\MEC)=\bu(c) \) if \(\MEC \) contains \((p,\bu,q)\), and \(\bu_i(c_\MEC)=0\) otherwise.
 	\item \(T'\) contains no transitions other than the ones defined by the previous step.
\end{itemize}  
Note that each transition of \(\A_i \) corresponds to a transition from \(\A\) and vice versa, hence we will use transitions of \(\A \) (\(\A_i\)) to also reference the corresponding transitions from \(\A_i \) (\(\A\)). Also note  that \(\A_1=\A \).

We use \(C_i\subseteq \bigcup_{i=1}^\infty \countersset_i \) to  denote the set of all counters with tight asymptotic estimate of \(n^i \) for \(\calC_{\A_{i}}[c] \) where \(i\in \mathbb{N}\), and \(C_{k+}\subseteq \bigcup_{i=1}^\infty \countersset_i \) to denote the set of all counters with lower asymptotic estimate of \(n^k\) for \(\calC_{\A_i}[c] \) where \(i\in \mathbb{N}\). We use \(\A^{C_1,\dots,C_j}_{i,T} \) to denote the VASS MDP \(\A_i \) restricted only to counters from \(\countersset_i\cap \bigcup_{l=1}^{j} C_l \) and only to transition from \(T \).

We say that we have \emph{classified} \(\A \) up to \( k\) if for each counter \(c\) (transition \(t\)) of \(\A \)  either we have that \(n^{a} \) is a tight asymptotic estimate of \(\calC_\A[c] \) (\(\calT_\A[t] \)) where \(1\leq a\leq k; a\in \mathbb{N} \) or we have a lower asymptotic estimate  \(n^{k+1} \) for \(\calC_\A[c] \) (\(\calT_\A[t] \)). Note that  classifying \(\A\) up to \(0\) is trivial since each \(\calC_\A[c] \) (\(\calT_\A[t] \)) has a trivial lower asymptotic estimate of \(n\). 

Let us begin with the following Lemma. \begin{lemma}\label{lemma-polynomial-for-one-k}
	If we have  classified \(\A \) up to \( k-1\), then  classifying \(\A\) up to \(k\) can be done in time polynomial in \(\size{\A} \). 
\end{lemma}
	First step towards  classifying \(\A \) up to \(k\) is the following Lemma.
	\begin{lemma}\label{lemma-hatB-zero-unbounded-rankl-give-upper-estimate-nk}
			Let \(1\leq l \leq  \lfloor\frac{k}{2}\rfloor \). Let  \(\by^{C_1,\dots,C_l}_{k-l,T_{k-l+1}},\bz^{C_1,\dots,C_l}_{k-l,T_{k-l+1}} \) be   a maximal solution of \hyperref[fig-systems]{(II)} for \(\A^{C_1,\dots,C_l}_{k-l,T_{k-l+1}} \), and let \(rank^{C_1,\dots,C_l}_{k-l,T_{k-l+1}} \) be the resulting ranking function defined by \(\by^{C_1,\dots,C_l}_{k-l,T_{k-l+1}},\bz^{C_1,\dots,C_l}_{k-l,T_{k-l+1}} \). There exists a set of transitions \(R^{C_1,\dots,C_l}_{k-l} \) of \(\A^{C_1,\dots,C_l}_{k-l,T_{k-l+1}} \) such that both of these hold:
		\begin{itemize}
			\item for each component  \(\by \) of \(\A_{k-l,T_{k-l+1}} \) it holds that \(\hat{\by} \) is zero-bounded on \(rank^{C_1,\dots,C_l}_{k-l,T_{k-l+1}} \) iff \( R^{C_1,\dots,C_l}_{k-l}\cap \{t\mid \by(t)>0 \}= \emptyset \);
			\item \(\calT_\A[t] \) has an upper asymptotic estimate of \(n^k\) for each \(t\in R^{C_1,\dots,C_l}_{k-l}\).	
		\end{itemize}
		Furthermore, assuming we have a classification of \(\A\) up to \(k-1\),  \(R_{k-l}\) can be computed in time polynomial in \(\size{\A} \).
	\end{lemma}

	Proof of Lemma~\ref{lemma-hatB-zero-unbounded-rankl-give-upper-estimate-nk} can be found in Appendix~\ref{section-proof-lemma-hatB-zero-unbounded-rankl-give-upper-estimate-nk}. The main idea behind this proof is that for each component  \(\by \) of \(\A^{C_1,\dots,C_l}_{k-l,T_{k-l+1}} \) we show  that \(\hat{\by} \) is zero-bounded on \(rank^{C_1,\dots,C_l}_{k-l,T_{k-l+1}}\) iff the effect of every transition of \(\by\) on \(rank^{C_1,\dots,C_l}_{k-l,T_{k-l+1}}\) is \(0\). For the second part, we show that \(\hat{\by} \) is not zero-bounded on \(rank^{C_1,\dots,C_l}_{k-l,T_{k-l+1}} \) iff the number of times \(\hat{\by}\) can be iterated has an upper asymptotic estimate of \(n^k \). Then we   divide the computations on \(\A\) into segments such that the transitions from \(R^{C_1,\dots,C_l}_{k-l}\) in  each segment either decrease or increase \(rank^{C_1,\dots,C_l}_{k-l,T_{k-l+1}} \) by it's maximal possible value. We then show that if 
	it were possible to iterate any component \(\by \) with  \( R^{C_1,\dots,C_l}_{k-l}\cap \{t\mid \by(t)>0 \}\neq \emptyset \) at least  \(n^{k+\epsilon} \) times in \(\A\) with sufficiently large probability, then the probability of all computations that contain ``too many'' segments that decrease \(rank^{C_1,\dots,C_l}_{k-l,T_{k-l+1}} \) is larger than the maximal possible probability of all such computations, obtained by a bound on how much \(rank^{C_1,\dots,C_l}_{k-l,T_{k-l+1}} \) can be increased in total with any non-negligible probability in \(\A\).\qed
	
	  Towards  classifying \(\A \) up to \(k\) we first compute for each \(1\leq l \leq \lfloor \frac{k}{2}\rfloor \) the sets \(R^{C_1,\dots,C_l}_{k-l}  \) using Lemma~\ref{lemma-hatB-zero-unbounded-rankl-give-upper-estimate-nk}, thus finding the set \(T_k'=\bigcup_{l=1}^{\lfloor \frac{k}{2}\rfloor} R^{C_1,\dots,C_l}_{k-l} \) of all transitions \(t\) for which Lemma~\ref{lemma-hatB-zero-unbounded-rankl-give-upper-estimate-nk} gives an upper asymptotic estimate of \(n^k \) for \(\calT_\A[t] \).  
	  
	Note that here \(k\) can be arbitrarily large, thus naively computing the sets \(R^{C_1,\dots,C_l}_{k-l} \) for  every \(1\leq l\leq \lfloor\frac{k}{2}\rfloor\)  may not finish in polynomial time. However, it suffices to compute \(R^{C_1,\dots,C_l}_{k-l}  \) only for polynomially many \(l \), as there are only polynomially many different VASS MDPs \(\A^{C_1,\dots,C_l}_{k-l,T_{k-l+1}} \). To see this, consider the sets \(\Aset=\{\aelement< k\mid C_\aelement\neq \emptyset \} \) and \(\Bset=\{\belement< k\mid T_b\setminus T_{b+1}\neq \emptyset  \} \). Notice that \(\A^{C_1,\dots,C_{l}}_{k-l,T_{k-l+1}}\) and \(\A^{C_1,\dots,C_{l-1}}_{k-l-1,T_{k-l}}\) differ only if one of these three holds:	\begin{itemize}
		\item  \(\bigcup_{i=1}^lC_i \neq \bigcup_{i=1}^{l-1}C_i  \) which is equivalent to \(l\in\Aset \);
		\item  or there exists a transition \(t\in T_{k-l}\setminus T_{k-l+1}\) which is equivalent to \(k-l\in\Bset \);
		\item or  there exists a counter \(c\in C_{a}\subseteq \bigcup_{i=1}^lC_i \) such that there exists a transition \(t\in T_{k-l-a}\setminus T_{k-l-a+1} \), which is equivalent to  \(k-l-a\in\Bset \) for some  \(a\in \Aset;a\leq l \).
	\end{itemize}  Thus \(\A^{C_1, \dots,  C_l}_{k-l,T_{k-l+1}}\) and \(\A^{C_1, \dots, C_l}_{k-l-1,T_{k-l}}\) can differ only for \(l\in X=\Aset\cup \{k-b\mid b\in \Bset \} \cup \{k-b-a \mid a\in \Aset, b\in \Bset \} \). Since both \(\Aset\) and \(\Bset\) are polynomially large, \(X\) is also polynomially large. Thus we can simply compute the set \(X\) and then compute \(R^{C_1,\dots,C_l}_{k-l}\) only for \(1\leq l \leq \lfloor \frac{k}{2}\rfloor \) such that \(l\in X\). This allows us to compute the set \(T_k'=\bigcup_{l=1}^{\lfloor \frac{k}{2}\rfloor} R^{C_1,\dots,C_l}_{k-l} \) in polynomial time.
	
%	(in this case \(\A^{C_1\cup \dots \cup C_l}_{k-l,T_{k-l+1}}\) has more local copies of \(c\) compared to \(\A^{C_1\cup \dots \cup C_l}_{k-l-1,T_{k-l}}\))
	
%	Therefore we can obtain in polynomial time a set of transitions \(T_{k}' \) that have an upper asymptotic estimate of \(n^{k}\) for \(\calT_\A[t] \) as per the Lemma~\ref{lemma-hatB-zero-unbounded-rankl-give-upper-estimate-nk}.
	
	After this, we compute maximal solutions \(\bx_k,\by_k\), and \(\bz_k \) for the systems \hyperref[fig-systems]{(I)} and \hyperref[fig-systems]{(II)} for \(\A_{k,\hat{T}_{k+1}} \) where \(\hat{T}_{k+1}=T_{k}\setminus T_k' \). These can be computed in polynomial time as discussed in Section~\ref{sec-systems}. 
	\begin{remark}
		Note that \(\A_{k,\hat{T}_{k+1}} \) contains no component that is zero-unbounded on any counter \(c\in \bigcup_{i=1}^{\lfloor\frac{k}{2}\rfloor} C_i \). While this does not necessarily  mean that there is no component of \(\A_{k,\hat{T}_{k+1}} \) that is zero-unbounded on some counter, this does imply that for every component \(\by\) of \(\A_{k,\hat{T}_{k+1}} \), \(\hat{\by} \)  can  be iterated at least \( n^{k+1} \) times with high probability (see Appendix~\ref{app-sec-lowerest-nkplus1-hatby}). This is sufficient to apply a modification of the method used in \cite{Zuleger:VASS-polynomial} to obtain the remaining upper/lower asymptotic estimates towards  classifying \(\A\) up to \(k\). 
		\end{remark}
	
	We obtain the additional upper asymptotic estimates of \(n^k\) from the following Lemma.\begin{lemma}\label{lemma-linear-upper-bound}\label{lemma-upper-bound-T}
		Let \[\variablet=    \begin{cases}
			1 & \text{if } k=1  \\
			k & \text{if } k>1 \text{ and }T_k'\neq \emptyset \\
			\max\{a+b\mid a\in \Aset,b\in \Bset; a+b\leq k \} & \text{if } k>1 \text{ and }T_k'= \emptyset
		\end{cases}\]
	
	Then for every counter $c\in \countersset$ such that $\by_k(c)>0$ it holds that \(n^\variablet\) is an upper asymptotic estimate of \( \calC_\A[c]\). Furthermore each transition \(t=(p,\bu,q) \)  of \(\A_{k,\hat{T}_{k+1}} \) has an upper asymptotic estimate of \(n^\variablet \) for \(\calT_\A[t] \) if one of the following holds:
		\begin{itemize}
			\item  \(p\in Q_n \) and \(\bz_k(q)-\bz_k(p)+\sum_{c\in\countersset} \bu(c)\cdot \by_k(c)<0\),
			\item  \(p\in Q_p \) and \(\sum_{t' = (p',\bu',q') \in \tout(p)}P(t')\cdot \big(\bz_k(q')-\bz_k(p')+\sum_{c\in\countersset}  \bu'(c)\cdot \by_k(c)\big)< 0\).
		\end{itemize}  
	\end{lemma}	
	
%	\begin{lemma}\label{lemma-upper-bound-T}
%	
%		
%		where \[l=    \begin{cases}
%			1 & \text{if } k=1  \\
%			k & \text{if } k>1 \text{ and }T_k'\neq \emptyset \\
%			\max\{a+b\mid a\in A,b\in B; a+b\leq k \} & \text{if } k>1 \text{ and }T_k'= \emptyset
%		\end{cases}\]		
%	\end{lemma}
	
	The proof of Lemma~\ref{lemma-linear-upper-bound} can be found in the Appendix~\ref{app-lemma-linear-upper-bound}. The idea is that we can design a supermartingale that is upper bounded by \(n^{\variablet+\epsilon} \), and which with very high probability upper bounds the ranking  function \(rank_k \) defined by \(\by_k,\bz_k \) and extended onto computations of \(\A\). Thus obtaining an upper asymptotic estimate of \(n^\variablet\) for counters \(rank_k \) considers, as well as for transitions that strictly decrease the supermartingale on average.  \qed
	
	We then obtain the remaining lower asymptotic estimates of \(n^{k+1} \) from the following Lemma.	
	\begin{lemma}\label{lemma-something-something-k+1-estiamtes}
		For every transition \(t\) with \(\bx_k(t)>0 \) and every counter \(c\in \countersset\) with \(\sum_{t' \in T} \bx_k(t')\cdot \bu_{t'}(c)>0 \) it holds that \(n^{k+1} \) is a lower asymptotic estimate of \(\calT_\A[t] \) and \(\calC_\A[c] \).
	\end{lemma}	
	Proof of Lemma~\ref{lemma-something-something-k+1-estiamtes} can be found in Appendix~\ref{section-proof-of-lemma-something-something-k+1-estiamtes}. The idea here is that we design a strategy, which iterates \(\bx_k \)  for a total of \(n^{k+1-\epsilon} \) times with very high probability. This is achieved by splitting the computation into two parallel computations, the ``probabilistic'' part as well as the ``deterministic'' part (see Section~\ref{sec-informal}). We show that for every component \(\by \) of \(\A_{k,\hat{T}_{k+1}} \) we can iterate \(\hat{\by} \) at least \(n^{k+1-\epsilon} \) times with very high probability. This ensures that with high probability the ``probabilistic'' part cannot terminate before taking at least \(n^{k+1-\epsilon}\) steps in \(\A_{\hat{T}_{k+1}} \). This then allows us to use the same method used in \cite{Zuleger:VASS-polynomial} of ``switching'' between components of \(\A\) in such a way that   that the ``deterministic'' part is guaranteed to not terminate before at least \(n^{k+1-\epsilon} \) steps as well, while simultaneously iterating \(\bx_k \) roughly \(n^{k+1-\epsilon} \) times with very high probability.	\qed
	
	From Lemma~\ref{lemma:dichotomy} we have for each counter \(c\in C_{k+} \) that either $\by_k(c)>0$ or \(\sum_{t \in T} \bx(t)\cdot \bu_t(c)>0 \). In the former case we obtain from Lemma~\ref{lemma-upper-bound-T} an upper asymptotic estimate of \(n^\variablet \) for \(\calC_\A[c] \) and in the latter case we obtain from Lemma~\ref{lemma-something-something-k+1-estiamtes} a lower asymptotic estimate of \(n^{k+1}\) for \(\calC_\A[c] \). Similarly, from Lemma~\ref{lemma:dichotomy} we also have that for each transition \(t\in \hat{T}_{k+1} \) either Lemma~\ref{lemma-linear-upper-bound} gives an upper asymptotic estimate \(n^\variablet\) for \(\calT_\A[t] \) or Lemma~\ref{lemma-something-something-k+1-estiamtes} gives a lower asymptotic estimate \(n^{k+1} \) for \(\calT_\A[t] \). Thus  classifying \(\A \) up to \(k\) can be done in polynomial time assuming we already have  classified \(\A\) up to \(k-1 \). This finishes the proof of Lemma~\ref{lemma-polynomial-for-one-k}. \qed
	%\end{proof}

%The key here is be that if \(\variablet\neq k\) then for every counter/transition for which we have a lower asymptotic estimate of \(n^k \) we have to obtain a lower asymptotic estimate of \(n^{k+1} \), as otherwise they would have both an upper asymptotic estimate of \(n^\variablet \) and a lower asymptotic estimate of \(n^k \) for \(\variablet<k \) which is a contradiction. 	

Now we will show that it suffices to perform the classification for at most polynomially many \(k\).  To see this, notice for any counter/transition for which we do not yet have an upper asymptotic estimate we can obtain an upper asymptotic estimate of \(n^{k} \) only either from Lemma~\ref{lemma-hatB-zero-unbounded-rankl-give-upper-estimate-nk} or from Lemma~\ref{lemma-linear-upper-bound}, and if we do not obtain an upper asymptotic estimate this way then we obtain a lower asymptotic estimate of \(n^{k+1}\) from Lemma~\ref{lemma-something-something-k+1-estiamtes}. Furthermore, from  Lemma~\ref{lemma-linear-upper-bound} we can obtain an upper asymptotic estimate \(n^k \) only if either there exist \(a\in \Aset \) and \(b\in \Bset\) such that \(a+b=k \), or if \(T_k'\neq \emptyset \), or if \(k=1\). Similarly, from the Lemma~\ref{lemma-hatB-zero-unbounded-rankl-give-upper-estimate-nk} we can obtain an upper asymptotic estimate \(n^k \) only if \(R^{C_1,\dots,C_l}_{k-l} \neq \emptyset\) for some \(1\leq l\leq \lfloor\frac{k}{2}\rfloor \).

%But as argued above,
But for each \(r\in \mathbb{N} \) there exists \( a\in \Aset\) such that  \(C_1\cup \dots\cup C_r\) is the same as \(C_1\cup \dots\cup C_a \). Furthermore, \(\A_{s,T_{s+1}}^{C_1,\dots,C_r} \) differs from \(\A_{s-1,T_{s}}^{C_1,\dots,C_r} \) only if at least one of the following holds: \begin{itemize}
	\item there exists a transition \(t\in T_{s}\setminus T_{s+1} \); (different set of transitions)
	\item there exists a counter \(c\in C_a\subseteq C_{1}\cup \dots\cup C_{r} \) such that there exists a transition \(t\in T_{s-a}\setminus T_{s-a+1} \); (different local copies of \(c\))
%	\item there exists a counter \(c\in C_{r} \);
\end{itemize}

%  Let \(k'\) be the largest \(k\) for which we have already performed the classification of \(\A \).
  
  Therefore each of the  VASS MDPs \(\A_{s,T_{s+1}}^{C_1,\dots,C_r} \) is represented in the set \(M=\{\A_{s,T_{s+1}}^{C_1,\dots,C_r} \mid r\in \Aset, s\in S_r \} \) where \( S_r=\Bset\cup \{a+b \mid a\in \Aset, b\in \Bset, a\leq r \}\). 
%  \todo{dunno, is this even needed?}
%  And for each \(\A_{s,T_{s+1}}^{C_1,\dots,C_r}\in M \) it holds that \(\A_{s,T_{s+1}}^{C_1,\dots,C_r}\) differs from \(\A_{s',T_{s'+1}}^{C_1,\dots,C_{r'}}\) for each \(s',r'\) where either \(s'<s \) or \(r'<r \).
  
  For each \(r\in \Aset, s\in S_r \) let \(k_{s,r} \) denote the smallest \(k\) such that we can obtain an upper asymptotic estimate of \( n^k\) from \(\A_{s,T_{s+1}}^{C_1,\dots,C_r} \) using the Lemma~\ref{lemma-hatB-zero-unbounded-rankl-give-upper-estimate-nk}. From the following Lemma it holds \(k_{s,r}=\max(s+r,2\cdot r) \).

  \begin{lemma}\label{lemma-I-dont-lknooow-but-it-is-somehtingmain}
  	We can only obtain an upper asymptotic estimate of \(n^k\) from Lemma~\ref{lemma-hatB-zero-unbounded-rankl-give-upper-estimate-nk} for values of  \(k \) satisfying \(k=\max(s+r,2\cdot r)\) where \(r\in \Aset\) and \( s\in S_r \).
  \end{lemma}
  
  The proof of Lemma~\ref{lemma-I-dont-lknooow-but-it-is-somehtingmain} can be found in the Appendix~\ref{section-proof-lemma-hatB-zero-unbounded-rankl-give-upper-estimate-nk}.\qed

Thus it suffices to perform classification of \(\A \) for only the polynomially many values \(k\in X_1\cup X_2\) where \(X_1=\Bset\cup \{a+b\mid a\in \Aset,b\in \Bset \}\cup \{1\} \) and \(X_2=\{\max(s+r,2\cdot r)\mid r\in \Aset, s\in S_r  \} \) which can also be written as \(X_2= \{\max(s+r,2\cdot r)\mid r\in \Aset,a\in \Aset,b\in \Bset; \textit{ either }s=b \textit{ or }s=a+b\textit{ and }a\leq r \} \). 

Therefore we can perform the full analysis in polynomial time as follows: First we perform the classification of \(\A  \) up to \(k=1\), and each time we finish the classification of \(\A \) up to some \(k\) we recompute the sets \(\Aset,\Bset,X_1,X_2 \) and then perform  the  classification of \(\A  \) for the smallest \(k'>k \) with \(k'\in X_1\cup X_2 \). Note that we add new elements to \(\Aset,\Bset,X_1,X_2 \) only if we find a new upper estimate \(n^k \) for some \(\calC_\A[c] \) or \(\calT_\A[t] \), which can happen at most polynomially many times, and every time we add only polynomially many elements. We proceed this way until \(X_1\cup X_2 = \emptyset \) which happens in time polynomial in \(\size\A\) at which point the algorithm stops. 

Whenever we add a new element to \(X_1\cup X_2 \) the largest element of \(X_1\) (\(X_2\)) is at most double (triple) of the largest element of \(\Aset\cup \Bset \). Also we can add new elements to \(\Aset\) at most \(|\countersset|\) times and to \(\Bset\) at most \(|T| \) times. Furthermore, we can only obtain an upper asymptotic estimate of \(n^k\) for \(\calC_\A[c] \) if either \(k\in X_1 \) or there is a transition \(t\) with a tight asymptotic estimate of \(n^k\) for \(\calT_\A[t] \). Thus we obtain the following Lemma.

\begin{lemma}
 Given a counter \(c\in\countersset\) (a transition \(t\in T\)) either \(\calC_\A[c]\) (\(\calT_\A[t]\)) has an upper asymptotic estimate of \(n^{2^{|\countersset|}\cdot 3^{|T|}} \) or \(\calC_\A[c]\) (\(\calT_\A[t]\)) has a lower asymptotic estimate of \(n^k\) for every \(k\in \mathbb{N} \).
\end{lemma}

We now show that the counters/transitions for which we do not obtain an upper asymptotic estimate using the above have a lower asymptotic estimate \(2^{\sqrt{n}} \).

\begin{lemma}
	Given a counter \(c\in\countersset\) (a transition \(t\in T\)) if \(\calC_\A[c]\) (\(\calT_\A[t]\)) has a lower asymptotic estimate of \(n^{k} \) for all \(k\in \N\) then \(\calC_\A[c]\) (\(\calT_\A[t]\)) has a lower asymptotic estimate \(2^{\sqrt{n}} \).
\end{lemma}

\section{VASS Markov Chains}\label{section-markov-chains}

In this section, we give a full and effective classification of \(\calL_\A\), \(\calC_\A[c]\), and
\(\calT_\A[t]\) for VASS Markov chains. More precisely, we have the following:

\begin{theorem}
	\label{thm-VASS-Markov-chain-analysis}

	%	Let $\A$ be a VASS Markov chain. Let \(c\) be a counter and \(t\) a transition of \(\A \). Then for each \(\mathbb{X}\in \{\calC_\A[c],\calT_\A[t],\calL_\A \} \) one of the following holds:
	%	\begin{itemize}
		%		\item \(\mathbb{X} \) is unbounded.
		%		\item $n^2$ is a tight asymptotic estimate of \(\mathbb{X} \).
		%		\item $n$ is a tight asymptotic estimate of \(\mathbb{X} \).
		%	\end{itemize}
	%	It is decidable in polynomial time which of the above cases hold.
	%	
	Let $\A$ be a strongly connected VASS Markov chain. Let \(c\) be a counter and \(t\) a transition of \(\A \).  Then for each \(\F\in \{\calC_\A[c],\calT_\A[t],\calL_\A \} \) exactly one of the following holds:	
	\begin{itemize}
		\item $\F$ is unbounded.
		\item $n^2$ is a tight asymptotic estimate of $\F$.
		\item $n$ is a tight asymptotic estimate of $\F$.
	\end{itemize}
	It is decidable in time polynomial in \(\size{\A} \) which of the above cases holds.
\end{theorem}

The proof of Theorem~\ref{thm-VASS-Markov-chain-analysis} can be found in the Appendix~\ref{app-markov-chains}.  The main idea is to analyze for each individual counter \(c\) the \(1\)-dimensional VASS Markov chain restricted only to \(c\), by utilizing the results about \(1\)-dimensional VASS MDPs from \cite{AKCONCUR23}, and then combine these into analysis of \(\A\). \qed

Note that in general (i.e., not strongly connected) VASS Markov chains, the computation ``very quickly'' falls into a MEC. As the expected effect of the computation before reaching a MEC is upper bounded by a constant, the effect of the computation before reaching a MEC is asymptotically negligible with very high probability.\footnote{The probability of not yet being in a MEC decreases exponentially fast with the number of steps. Therefore the probability of the effect of these steps being over \(n^\epsilon \) on any counter decreases exponentially fast as \(n\) increases.} Hence Theorem~\ref{thm-VASS-Markov-chain-analysis} can be extended onto general VASS Markov chains by analyzing all of its MECs individually, with the only difference being that for transitions \(t\) that are not part of any MEC, \(\calT_\A[t] \) would be ``asymptotically  constant'' (note that this does not mean that \(\calT_\A[t] \) would have a constant function as its upper asymptotic estimate, however the expectation  of \(\calT_\A[t] \) would in \(\Theta(1) \)).

\section{Conclusions}

We presented a precise complexity classification for polynomial asymptotic estimates on strongly connected VASS MDPs, and we have shown that  on this sub-class polynomial tight asymptotic estimates can be computed efficiently. We also presented full classification of asymptotic complexity  for strongly connected VASS Markov chains, and showed an alternative definition of  asymptotic estimates using a natural notion of fixed probability bounds. These results are especially encouraging given that the study of multi-dimensional VASS MDPs is notoriously difficult.

While our main result is only for strongly connected VASS MDPs, we hypothesize that it can can be extended onto general VASS MDPs by utilizing the notion of types introduced in \cite{AKCONCUR23} combined with the methods used in \cite{AK:VASS-polynomial-termination} for extending results about analysis of strongly connected non-probabilistic VASS  onto general non-probabilistic VASS and VASS games. This approach might also be able to extend our results onto VASS MDP games which combine probabilistic states with both angelic and demonic non-determinism.

\section{Conclusions}

We presented a precise complexity classification for polynomial asymptotic estimates on strongly connected VASS MDPs, and we have shown that  on this sub-class polynomial tight asymptotic estimates can be computed efficiently. We also presented full classification of asymptotic complexity  for strongly connected VASS Markov chains, and showed an alternative definition of  asymptotic estimates using a natural notion of fixed probability bounds. These results are especially encouraging given that the study of multi-dimensional VASS MDPs is notoriously difficult.

While our main result is only for strongly connected VASS MDPs, we hypothesize that it can can be extended onto general VASS MDPs by utilizing the notion of types introduced in \cite{AKCONCUR23} combined with the methods used in \cite{AK:VASS-polynomial-termination} for extending results about analysis of strongly connected non-probabilistic VASS  onto general non-probabilistic VASS and VASS games. This approach might also be able to extend our results onto VASS MDP games which combine probabilistic states with both angelic and demonic non-determinism.

\bibliography{str-long,concur}
 \appendix
 \newpage
 
 \section{Additional Definitions}
\label{section-additional-definitions}
In this Section we introduce additional definitions that are used thorough the appendix.

Given subsets \(C_1,\dots,C_n \) of \(\countersset\) and a multi-component \(\bx\) we use \(\Delta^{C_1,\dots,C_n}(\bx) \) to denote the vector \(\Delta(\bx)\) restricted only to the counters from \(\bigcup_{i=1}^n C_i \).

\textbf{Concatenation:} 
Given two computations \(\alpha=p_0\bv_0, p_1 \bv_1,\dots, p_a \bv_a\) and \(\beta=p_0'\bv_0', p_1' \bv_1',\dots, p_b' \bv_b' \) we define  their \emph{concatenation} \(\alpha\cdot\beta  \) as \(\alpha\cdot\beta= p_0\bv_0,\dots, p_a \bv_a,p_1' \bv_1',\dots, p_b' \bv_b'\). Note that if \(p_a\bv_a=p_0'\bv_0' \) then \(\alpha\cdot \beta \) is also a computation.

We use the symbol \(\epsilon \) to denote an empty computation.

\textbf{Pointing VASS:} 
A \emph{pointing pair} is a pair \((\M,p) \) where \(\M \) is a VASS Markov chain and \(p\) is a state of \(\M\).  A \emph{pointing VASS} with counters \(\countersset\) is defined as \(\B=\big((\M_1,p_1),\dots,(\M_n,p_n)\big) \) where each \((\M_i,p_i) \) is a pointing pair such that \(\M_i \) is a  VASS Markov chain with counters \(\countersset\). Note that the VASS Markov chains \(\M_1,\dots,\M_n \) may share states. A \emph{pointing configuration} of \(\B \) is a tuple \((q_1,\dots,q_n)a\bv  \) where each \(q_i \) is a state of \(\M_i \), \(a\in \{1,\dots,n \} \) and \(\bv\in \mathbb{Z}^\countersset \). We say that in the pointing configuration \((q_1,\dots,q_n)a\bv  \), each VASS Markov chain  \(\M_i \) is in the state \(q_i \), \(\M_a \) is the VASS Markov chain that was pointed to last, and \(\bv \) is the counters vector. The dynamics of a pointing VASS are such that in the pointing configuration \((q_1,\dots,q_n)a\bv  \) the allowed actions are \(\{b\in \{1,\dots,n\}\mid q_a=q_{b} \} \). Upon taking the action \(b\) from \((q_1,\dots,q_n)a\bv  \), the probability of the next pointing configuration being \((q_1,\dots,q_{b-1},q_b',q_{b+1},\dots,q_n)b\bv'  \) such that \(\bv'=\bv+\bu \) is equal to the probability of the transition \((q_b,\bu,q_b') \) in \(\M_b \) (and any  pointing configuration that is not assigned a probability this way has probability \(0\)).  

If some component of \(\bv \) is negative then \((q_1,\dots,q_n)a\bv  \) is \emph{terminal}.  A \emph{pointing computation} is a sequence of pointing configurations \(\pi=(p_1,\dots,p_n)a_0\bv_0,(q_1^1,\dots,q_n^1)a_1\bv_1,(q_1^2,\dots,q_n^2)a_2\bv_2,\dots \). Let $\Term(\pi)$ be the least $j$ such that $(q_1^j,\dots,q_n^j)a_j\bv_j$ is terminal.  
%\(\pi=(p_1,\dots,p_n)a_0\bv_0,(q_1^1,\dots,q_n^1)a_1\bv_1,(q_1^2,\dots,q_n^2)a_2\bv_2,\dots \)
%\(\pi=(q_1^0,\dots,q_n^0)a_0\bv_0,(q_1^1,\dots,q_n^1)a_1\bv_1,(q_1^2,\dots,q_n^2)a_2\bv_2,\dots \)

%  such that: \begin{itemize}
	%	\item \((q_1^0,\dots,q_n^0)=(p_1,\dots,p_n) \);
	%%	\item \(q_{a_{i}}^i=q_{a_{i+1}}^i\);
	%	\item \(q_{j}^i=q_{j}^{i+1} \) if \(j\neq a_{i+1} \);
	%	\item \(\bv_{i+1}=\bv_i+\bu_{i+1} \) where \(\bu_{i+1} \) is such that \((q_{a_{i+1}}^i,\bu_{i+1},q_{a_{i+1}}^{i+1}) \) is a transition of \(\M_{a_{i+1}} \). 
	%\end{itemize}

	A \emph{pointing strategy} is a function $\sigma$ assigning to every finite pointing computation $(p_1,\dots,p_n)a_0\bv_0,(q_1^1,\dots,q_n^1)a_1\bv_1,\dots,(q_1^m,\dots,q_n^m)a_m\bv_m$  a probability distribution over  \(\{b\in \{1,\dots,n\}\mid q_{b}^m=q_{a_m}^m \} \). Every initial counter vector $\bv$, every \(a\in\{1,\dots,n \} \), and every pointing strategy $\sigma$ determine the probability space over pointing computations initiated in $(p_1,\dots,p_n)a\bv$ in the standard way. We use $\prob^\sigma_{p_a\bv}$ to denote the associated probability measure.
	For a measurable function $X$ over computations, we use $\Exp^\sigma_{p_a \bv}[X]$ to denote the expected value of~$X$.

	Note that from \cite[Lemma~31 of the full paper]{AKCONCUR23}, for each VASS MDP \(\A \) there exists a pointing VASS \(\pointingVASScorrepsonding{\A} \) that is bisimulation equivalent to \(\A \), and such that  \(\pointingVASScorrepsonding{\A} \)  contains all the possible pointing pairs \((\A_\sigma,p) \) where \(p \) is a state of \(\A \) and \(\sigma\in \stratsMD{\A} \). Thus to each computation \(\pi \) on \(\A \) we can assign a corresponding  pointing computation \(\tilde{\pi} \) on \(\tilde{\A} \) and vice versa. Especially, this allows us to assign to each computation \(\pi=p_0\bv_0,\ldots, p_n \bv_n\) on \(\A\) and a transition \(t\in \tout(p_n) \) a unique pointing pair \((\M,pone) \) that is being pointed at in \(\pointingVASScorrepsonding{\A} \) at this step.\footnote{Note that in general this pointing pair assignment is not necessarily unique. However as discussed in \cite{AKCONCUR23} we can fix some function assigning the pointing pair and under this function the assignment is unique. Thorough this paper we assume we have fixed some such function and thus we assume that all such assignments are unique.} We say that this step points to a MEC \(\MEC\) of \(\A_\sigma \) for \(\sigma\in \stratsMD{\A} \) if \(\MEC\) is entirely contained in \(\M \) and \(p\) is a state of \(\MEC\). Note that since \(\M \) is a VASS Markov chain equal to \(\A_\sigma \) for some \(\sigma\in \cMD(\A) \), it either holds that no such \(\MEC \) exists (if \(p\) is not a state of any MEC of \(\M \)) or \(\MEC \) is determined uniquely.\footnote{We refer to \cite{AKCONCUR23} for more detailed  description of pointing VASS.}

Given a VASS MDP \(\A \), as for a component \(\by \) of \(\A\) it holds that \(\A_\by \) is a VASS Markov chain, each component \(\by \) of \(\A\) defines uniquely a pointing pair \((\A_\by,p_\by) \).  Given a component \(\by \) of \(\A\) we use \(\A_{+\by} =\big( (\M_1,p_1),\dots,(\M_w,p_w),(\A_\by,p_\by) \big) \) to denote the pointing VASS where \(\tilde{\A}=\big( (\M_1,p_1),\dots,(\M_w,p_w) \big)  \).
%	\footnote{Remember that \(\tilde{\A} \) is the pointing VASS corresponding to \(\A \)}

\textbf{Pointing complexity:}
Let  $\B = \{(\M_1,p_1),\dots,(\M_n,p_n)\} $ be a pointing VASS, and \(a\in\{1,\dots,m \} \). For every pointing computation \(\pi=(q_1^0,\dots,q_n^0)a_0\bv_0,(q_1^1,\dots,q_n^1)a_1\bv_1,\dots \), we put
\begin{eqnarray*}
	\calP_{\B}[\M_b](\pi) & = & \mbox{the total number of all $0 \leq i < \Term(\pi)$ such that $a_i = b$}
\end{eqnarray*}

We refer to the function \(\calP_{\B}[\M_a] \) as \emph{\(\M_a\)-pointing complexity}.

We extend the notion of lower/upper/tight asymptotic estimates to \(\calP_{\B}[\M_a] \) in the natural way.

	\textbf{Bin:} We often use the term \emph{bin} to simplify the definition of strategies. Intuitively, a bin contains a counters vector that is assigned to it (virtually) by the strategy. Note that this assignment is always deterministic, the strategy \(\sigma \) can always compute the counters vector it assigned to each bin for any finite computation generated by \(\sigma \), hence no memory is needed for \(\sigma \) to ``remember'' the bins. This allows us to define the behavior of \(\sigma \) based on the current counter vectors assigned to the bins. Note that we always ensure that the actual counters vector is always greater or equal than the sum of all the (virtual) counter vectors stored in all of the bins. Unless stated otherwise, we assume that if any counter becomes negative in any of the bins considered by \(\sigma\), then from that point on \(\sigma \) is undefined. Note that if all the counter vectors assigned to all of the bins are positive then this implies that also the actual counters vector is positive. We assign names to the bins to simplify referencing the individual bins.

When we say that \(\sigma \) plays/simulates some strategy \(\pi \) on the simulation-bin whose current counters vector is \(\bv' \), what we actually mean, is that \(\sigma \) plays exactly the same as the strategy \(\pi\) initiated from the initial history \(p\bv' \), where \(p\) is the current state. Furthermore, the effect on counters when simulating \(\pi\) on the simulation-bin is always added fully to the simulation-bin. That is, say after \(10\) steps the strategy \(\pi \) initiated in \(p\bv' \) would have reached the configuration \(q\bu \), then after \(\sigma\) simulates \(\pi \) on the simulation-bin as above for \(10\) steps, the counters vector in the simulation-bin will become \(\bu \) (and all the other bins will have their counter vectors untouched). If we say that \(\sigma\) pauses the simulation of \(\pi \) on the simulation-bin, we mean that \(\sigma \) actually remembers (again, it can be computed from the history) the computation \(\alpha\) taken by \(\pi \) during this simulation so far, and if at any point in the future we say that \(\sigma \) unpauses/resumes  the simulation of \(\pi \) on the simulation-bin then \(\sigma \) ``resumes'' by playing as \(\pi\) for the computation \(\alpha \) as above, as if no pause happened. Note that we will never modify the ``paused'' bin during the pause (except potentially if the strategy being simulated on this bin is \(\cMD \)).

\section{Technical Lemmas}

In this Section we prove several technical Lemmas we need thorough the paper.
First we shall prove that all counters can be pumped ``almost'' to their lower asymptotic estimate simultaneously with very high probability.

\begin{lemma}\label{counters-pumpable-all-at-once}
	Let \(\A \) be a strongly connected VASS MDP. Let \(c_1,\dots,c_d \) be all the counters of \(\A \) and let \(k_1,\dots,k_d \) be such that \(n^{k_i} \) is a lower asymptotic estimate of \( \calC_\A[c_i]\) for all \(1\leq i\leq d\). Then for each \(\epsilon>0\) there exists a strategy \(\sigma \) which from any initial configuration with counter values \(\vec{n} \) reaches with probability \(p_n\), such that \(\lim_{n\rightarrow\infty} p_n=1 \), a configuration with counters vector \(\bv_n \) with \(\bv_n(c_i)\geq n^{k_i-\epsilon} \) for each \(1\leq i \leq d\).
\end{lemma}

\begin{proof}

 For each \(n\in\mathbb{N}\) and \(1\leq i\leq d \), let \(\sigma_n^i \) the strategy which from initial configuration \(q^i_n\vec{n}\) for some state \(q^i_n\) reaches the counter value \(n^{k_i-\epsilon'} \) on counter \(c_i \) with probability \(p^i_n \) such that \(\lim_{n\rightarrow\infty} p^i_n=1 \), where \(\epsilon>\epsilon'>0 \). These strategies  exist from \(  \calC_\A[c_i]\) having a lower asymptotic estimate of \(n^{k_i} \). 
 
 We will now give a high level description of \(\sigma \). When started in an  initial configuration \(q\vec{n} \) the computation under \(\sigma \) behaves as follows. First, \(\sigma\) virtually divides the counters vector \(\vec{n} \) into \(d+1 \) equally sized "bins", each ``bin'' therefore contains the vector \(\lfloor \frac{\vec{n}}{d+1} \rfloor \). Then  \(\sigma\) repeats the following for each \(1\leq i\leq d \): Use the \((d+1)\)-st ``bin'' to move to configuration \(q^i_{\lfloor \frac{n}{d+1} \rfloor } \) using a strategy that minimizes the expected number of steps needed to reach it, and then run on the \(i\)-th ``bin'' a computation under \(\sigma^i_{\lfloor \frac{n}{d+1} \rfloor} \) for initial configuration \(q^i_{\lfloor \frac{n}{d+1} \rfloor }\vec{\lfloor \frac{n}{d+1} \rfloor }\) 
 until this computation reaches at least the value \( (\lfloor \frac{n}{d+1} \rfloor)^{k_i-\epsilon'} \) 
in the counter \(c_i\). Once this value is reached (or exceeded), the computation for this \(i\) ends, and the computation for the next \(i\) starts.

 Computation in each ``bin'' never touches the counter values from other ``bins'' unless the computation in the given ``bin'' terminated, in which case \(\sigma\) behaves arbitrarily. After \(\sigma \) performs the above for all  \(1\leq i\leq d\), unless one of the computations on the ``bin'' terminated we are guaranteed to reach a configuration where we have at least \((\lfloor \frac{n}{d+1} \rfloor )^{k_i-\epsilon'} \) for the value of counter \(c_i\) in the \(i \)-th ``bin''. The probability of \(\sigma \) reaching  such configuration is at least \(p_n=1-p_{d+1}-\sum_{i=1}^{d} (1-p^i_{\lfloor \frac{n}{d+1} \rfloor }) \), where \(p_{d+1} \) is the probability the \((d+1)\)-st bin becomes negative on any counter at any point. Note that \(\lim_{n\rightarrow\infty} \sum_{i=1}^{d}(1-p^i_{\lfloor \frac{n}{d+1} \rfloor })=0 \), and since \(\A\) is strongly connected, the only way \((d+1)\)-st ``bin'' gets depleted is if any of the \(d\) computations taking place on the \((d+1)\)-st ``bin'' take at least \(\frac{\lfloor \frac{n}{d+1} \rfloor}{d\cdot |u|}\)  steps, where \(u\) is the maximal counter decrease per single transition in \(\A\). But the expected number of steps for each of these computations is constant. Therefore from Markov inequality we obtain that \(p_{d+1}\leq \frac{a\cdot d\cdot|u|}{\lfloor \frac{n}{d+1} \rfloor} \) for some constant \(a\). Thus \(\lim_{n \to \infty}p_{d+1}=0 \). This means that \(\lim_{n\rightarrow\infty} p_n =1 \).

All that remains is to show that \((\lfloor \frac{n}{d+1} \rfloor )^{k_i-\epsilon'} \geq n^{k_i-\epsilon} \) for each \(1\leq i\leq d \). As \(k_i-\epsilon<k_i-\epsilon' \) the left side of the inequality is dominated by the term \(n^{k_i-\epsilon'} \) which grows asymptotically faster than the right side term \(n^{k_i-\epsilon} \). Therefore for all sufficiently large \(n\) it holds \((\lfloor \frac{n}{d+1} \rfloor )^{k_i-\epsilon'} \geq n^{k_i-\epsilon} \).
\end{proof}

\subsection{Operations on Multi-components}
\label{app-lemma-decompose-multicomponents-into-components}
In this section we show basic results about arithmetic operations on multi-components used in this paper. 

We begin by showing that the following operations on multi-components of \(\A \) produce another multi-component of \(\A \): addition (Lemma~\ref{lemma-addition-multicomponents}), multiplication by non-negative number (Lemma~\ref{lemma-multiplication-multicomponents}), and subtraction assuming the result is non-negative on every component (Lemma~\ref{lemma-subtraction-multicomponents}). Then we show that multi-components can be decomposed into a conical sum of components (Lemma~\ref{lemma-decompose-multicomponents-into-components}), and finally that components can be turned into a component centered in any other state using multiplication by constant (Lemma~\ref{lemma-component-recentering}).

\begin{lemma}\label{lemma-addition-multicomponents}
	Let \(\bx_1,\bx_2 \) be multi-components of \(\A\). Then \(\bx'=\bx_1+\bx_2 \) is also a multi-component of \(\A\).
\end{lemma}

\begin{proof}
	%	We have to show that \(\bx'\) is a multi-component. 
	Since both \(\bx_1 \) and \(\bx_2\) are multi-components, it holds that \(\bx_1\geq \vec{0} \) and  \(\bx_2\geq \vec{0}\). Hence also \(\bx'=\bx_1+\bx_2\geq \vec{0} \).
	
	Since both \(\bx_1 \) and \(\bx_2\) are multi-components, for each \(p\in \states(\A) \) both of these hold \[\sum_{t\in \tout(p)}\bx_1(t)= \sum_{t \in \tin(p)} \bx_1(t)\] \[\sum_{t\in \tout(p)}\bx_2(t)= \sum_{t \in \tin(p)} \bx_2(t)\] Therefore 
	\begin{gather*}
		\sum_{t\in \tout(p)}\bx'(t) = \sum_{t\in \tout(p)}\big(\bx_1(t)+\bx_2(t)\big) = \sum_{t\in \tout(p)}\bx_1(t)+\sum_{t\in \tout(p)}\bx_2(t)=\\ \sum_{t \in \tin(p)} \bx_1(t)+\sum_{t \in \tin(p)} \bx_2(t)=\sum_{t \in \tin(p)} \big(\bx_1(t)+ \bx_2(t)\big)=\sum_{t \in \tin(p)} \bx'(t)
	\end{gather*}
	Hence \(
	\sum_{t\in \tout(p)}\bx'(t)= \sum_{t \in \tin(p)} \bx'(t)\) for each \(p\in \states(\A) \).
	
	Since both \(\bx_1 \) and \(\bx_2\) are multi-components, it holds for each \(p\in \states_p(\A) \) and each \(t\in\tout(p)\) that  both of these hold \[\bx_1(t)=P(t)\cdot \sum_{t'\in\tout(p)}\bx_1(t') \] \[\bx_2(t)=P(t)\cdot \sum_{t'\in\tout(p)}\bx_2(t') \] Therefore 	\begin{gather*}\bx'(t)=\bx_1(t)+\bx_2(t) =P(t)\cdot \sum_{t'\in\tout(p)}\bx_1(t') + P(t)\cdot \sum_{t'\in\tout(p)}\bx_2(t') =\\
		P(t)\cdot \sum_{t'\in\tout(p)}\big(\bx_1(t') + \bx_2(t')\big)=P(t)\cdot \sum_{t'\in\tout(p)}\bx'(t') \end{gather*} Hence \(
	\bx'(t)=P(t)\cdot \sum_{t'\in\tout(p)}\bx'(t')\)  for each \(p\in \states_p(\A) \), \(t\in\tout(p)\).
	
	\(\bx' \) is a multi-component of \(\A \).\qed
\end{proof}

\begin{lemma}\label{lemma-multiplication-multicomponents}
	Let \(\bx \) be a multi-component of \(\A\) and let \(a\geq 0 \). Then \(\bx'=a\cdot \bx \) is also a multi-component of \(\A\).
\end{lemma}

\begin{proof}
	Since \(\bx \) is a multi-component, it holds \(\bx\geq \vec{0} \), therefore also \(\bx'=a\cdot \bx\geq 0 \).
	
	Since \(\bx \) is a multi-component, for each \(p\in \states(\A) \) it holds \[\sum_{t\in \tout(p)}\bx(t)= \sum_{t \in \tin(p)} \bx(t)\]  Therefore 
	\begin{gather*}
		\sum_{t\in \tout(p)}\bx'(t) 
		=
		\sum_{t\in \tout(p)}a\cdot \bx(t) 
		=
		a\cdot \sum_{t\in\tout(p)} \bx(t)
		=\\ a\cdot\sum_{t\in \tin(p)} \bx(t) = \sum_{t\in \tin(p)} a\cdot\bx(t)=\sum_{t \in \tin(p)} \bx'(t)
	\end{gather*}
	Hence \(
	\sum_{t\in \tout(p)}\bx'(t)= \sum_{t \in \tin(p)} \bx'(t)\) for each \(p\in \states(\A) \).

	Since \(\bx \) is a multi-component, it holds each \(p\in \states_p(\A) \) and each \(t\in\tout(p)\) that  \[\bx(t)=P(t)\cdot \sum_{t'\in\tout(p)}\bx(t') \]  Therefore 	\begin{gather*}\bx'(t)=a\cdot \bx(t)=a\cdot P(t)\cdot \sum_{t'\in\tout(p)}\bx(t') 
		=\\
		P(t)\cdot \sum_{t'\in\tout(p)}a\cdot \bx(t')=P(t)\cdot \sum_{t'\in\tout(p)}\bx'(t') \end{gather*} Hence \(
	\bx'(t)=P(t)\cdot \sum_{t'\in\tout(p)}\bx'(t')\)  for each \(p\in \states_p(\A) \), \(t\in\tout(p)\).
	
	\(\bx' \) is a multi-component of \(\A \).\qed
\end{proof}

\begin{lemma}\label{lemma-subtraction-multicomponents}
	Let \(\bx_1,\bx_2 \) be multi-components of \(\A\) such that \(\bx_1-\bx_2\geq \vec{0} \). Then \(\bx'=\bx_1-\bx_2 \) is also a multi-component of \(\A\).
\end{lemma}

\begin{proof}
	We have \(\bx'\geq \vec{0} \) straight from the definition of \(\bx_1 \) and \(\bx_2 \).
	
	Since both \(\bx_1 \) and \(\bx_2\) are multi-components, it holds that \(\bx_1\geq \vec{0} \) and  \(\bx_2\geq \vec{0}\). Hence also \(\bx'=\bx_1+\bx_2\geq \vec{0} \).
	
	Since both \(\bx_1 \) and \(\bx_2\) are multi-components, for each \(p\in \states(\A) \) both of these hold \[\sum_{t\in \tout(p)}\bx_1(t)= \sum_{t \in \tin(p)} \bx_1(t)\] \[\sum_{t\in \tout(p)}\bx_2(t)= \sum_{t \in \tin(p)} \bx_2(t)\] Therefore 
	\begin{gather*}
		\sum_{t\in \tout(p)}\bx'(t) = \sum_{t\in \tout(p)}\big(\bx_1(t)-\bx_2(t)\big) = \sum_{t\in \tout(p)}\bx_1(t)-\sum_{t\in \tout(p)}\bx_2(t)=\\ \sum_{t \in \tin(p)} \bx_1(t)-\sum_{t \in \tin(p)} \bx_2(t)=\sum_{t \in \tin(p)} \big(\bx_1(t)- \bx_2(t)\big)=\sum_{t \in \tin(p)} \bx'(t)
	\end{gather*}
	Hence \(
	\sum_{t\in \tout(p)}\bx'(t)= \sum_{t \in \tin(p)} \bx'(t)\) for each \(p\in \states(\A) \).
	
	Since both \(\bx_1 \) and \(\bx_2\) are multi-components, it holds for each \(p\in \states_p(\A) \) and each \(t\in\tout(p)\) that  both of these hold \[\bx_1(t)=P(t)\cdot \sum_{t'\in\tout(p)}\bx_1(t') \] \[\bx_2(t)=P(t)\cdot \sum_{t'\in\tout(p)}\bx_2(t') \] Therefore 	\begin{gather*}\bx'(t)=\bx_1(t)-\bx_2(t) =P(t)\cdot \sum_{t'\in\tout(p)}\bx_1(t') - P(t)\cdot \sum_{t'\in\tout(p)}\bx_2(t') =\\
		P(t)\cdot \sum_{t'\in\tout(p)}\big(\bx_1(t') - \bx_2(t')\big)=P(t)\cdot \sum_{t'\in\tout(p)}\bx'(t') \end{gather*} Hence \(
	\bx'(t)=P(t)\cdot \sum_{t'\in\tout(p)}\bx'(t')\)  for each \(p\in \states_p(\A) \), \(t\in\tout(p)\).
	
	\(\bx' \) is a multi-component of \(\A \).\qed
\end{proof}

\begin{lemma}\label{lemma-decompose-multicomponents-into-components}
	A multi-component \(\bx\) of \(\A \) can be decomposed into a conical sum of components, that is \(\bx=\sum_{\by} a_\by \cdot\by\) where \(\by \) ranges over all components of \(\A \) and \(a_\by\geq 0 \) for all \(\by \). Furthermore, the decomposition can be done in such a way that \(a_\by>0 \) iff \(\by \) is a component of \(\A_\bx \).
\end{lemma}
\begin{proof}
	Let \(\by_1,\dots,\by_l \) be all the components of \(\A \). Let \(\bx_0=\bx \). For each \(1\leq i \leq l \) let \(\bx_i=\bx_{i-1}-a_i\cdot\by_i \) where \(a_i=\min_t a_i^t \) where \(t\) ranges over all transitions of \(\A_{\by_i}\) and \(a_i^t=\frac{\bx_{i-1}(t)}{\by_i(t)} \). 
	
	First let us prove that \(\bx_i\geq \vec{0} \) and \(a_i\geq 0 \) for all \(1\leq i\leq l\). We will show this by induction on \(i \). For base case, consider \(i=0 \), and let us set \(a_0=0 \). Then it holds \(\bx_0=\bx\geq \vec{0} \) from \(\bx\) being a multi-component, and \(a_0\geq 0 \) by our definition. Now assume the induction holds up \(i-1\), and we want to show it works for \(i\). From induction assumption we have that \(\bx_{i-1}\geq \vec{0} \), and since \(\by_i \) is a component it also holds \(\by\geq \vec{0} \). Hence it holds \(a_i^t\geq 0 \) for each transition \(t\) of \(\A_{\by_i} \), implying \(a_i\geq 0 \). Assume now towards contradiction that \(\bx_i(t)<0 \) for some transition \(t\) of \(\A \). If \(t \) is not a transition of \(\A_{\by_i} \), then \(\bx_i(t)=\bx_{i-1}(t)\geq 0 \). Hence \(t\) is a transition of \(\A_{\by_i} \). Notice that \(\bx_{i-1}(t)-a_i^t\cdot \by_i(t) = \bx_{i-1}(t)-\frac{\bx_{i-1}(t)}{\by_i(t)}\cdot \by_i(t)=0  \). Since both \(\bx_{i-1}(t)\geq 0 \) and \(\by_i(t)\geq 0 \), the only way for \(\bx_i(t)<0 \) to hold is if \(a_i> a_i^t  \). But this is a contradiction with our definition of \(a_i \). Therefore \(\bx_i\geq \vec{0}\).  
	
	Note that from from Lemma~\ref{lemma-multiplication-multicomponents} we have that \(a_i\cdot \by_i \) is a multi-component of \(\A \), thus from Lemma~\ref{lemma-subtraction-multicomponents} we then have that \(\bx_i\) is also a multi-component of \(\A\).
	
	Therefore, if \(\bx_l=\vec{0} \) then it holds \(\bx=\sum_{i=1}^{l}a_i\cdot\by_i \) with \(a_i\geq 0  \) for all \(1\leq i\leq l\). Assume therefore that \(\bx_l\neq \vec{0} \).  Since \(\bx_l\) is a multi-component of \(\A\) it induces a VASS Markov chain \(\A_{\bx_l} \) which contains at least one MEC \(B \) of \(\A_\sigma \) for some \(\sigma\in \cMD \). Since \(B\) is included also in \(\A \) there exists \(1\leq i\leq l\) such that \(\A_{\by_i} \) corresponds to \(B\), and therefore \(\by_i(t)>0 \) implies \(\bx_l(t)>0 \) for each \(t\). But since \(\bx_i\geq \bx_l \) this also means that \(\by_i(t)>0 \) implies \(\bx_{i}(t)>0 \) for each \(t\). Hence it holds for each transition \(t\) of \(\A_{\by_i} \) that \(a_i<a_i^t \), which is a contradiction with how \(a_i \) is defined. Hence it must hold that \(\bx_l=\vec{0} \).
	
	It remains to show that we can do this decomposition in such a way that \(a_\by>0 \) iff \(\by \) is a component of \(\A_\bx \). It holds trivially that \(a_\by=0 \) if \(\by \) is not a component of \(\A_\bx \) as that means there exists a transition \(t\) such that \(\by(t)>0\) while \(\bx(t)=0 \). Hence it suffices to show that we can do this decomposition in such a way that \(a_\by>0 \) for every component \(\by \) of \(\A_\bx \).  We will show that if there exists a decomposition \(\bx=\sum_{\by} a_\by\cdot \by\) for which there are exactly \(i\geq 1 \) components \(\by \) of \(\A_\bx \) with \(a_\by=0 \) then there also exists a decomposition \(\bx=\sum_{\by} a_\by' \cdot\by\) such that there are at most \(i-1 \) components \(\by \) of \(\A_\bx \) with \(a_\by'=0 \).
	
	Note that this would finish the proof of the Lemma, as there are only finitely many components in \(\A_\bx \) and we have already shown that at least one decomposition does exist.
	
	Let \(\bx=\sum_{\by} a_\by\cdot \by\) be such that there are exactly \(i\geq 1 \) components \(\by \) of \(\A_\bx \) with \(a_\by=0 \). Let \(\by' \) be a component of \(\A_{\bx} \) such that \(a_{\by'}=0 \).
	
	From Lemma~\ref{lemma-multiplication-multicomponents} it holds that \(\frac{\bx}{2} \) is also a multi-component of \(\A \), and clearly one possible decomposition of \(\frac{\bx}{2} \) is \(\frac{\bx}{2}=\sum_{\by} \frac{a_\by}{2}\cdot \by \). Furthermore, Since \(\by' \) is a component of \(\A_{\bx} \) it must hold for every transition \(t\) of \(\A\) that \(\by'(t)>0 \) implies \(\bx(t)>0 \) which implies  \(\frac{\bx}{2}(t)>0 \). Hence there exists \(b>0 \) such that \(\frac{\bx}{2}-b\cdot \by'\geq \vec{0}  \). From Lemma~\ref{lemma-subtraction-multicomponents} it holds that \(\frac{\bx}{2}-b\cdot \by' \) is a multi-component of \(\A \). Therefore we can apply the first part of this Lemma to obtain a decomposition \(\frac{\bx}{2}-b\cdot \by' =\sum_{\by} b_\by \cdot\by \) of \(\frac{\bx}{2}-b\cdot \by'  \) with \(b_\by\geq 0 \) for all \(\by \). We can thus decompose  \(\bx \) as \(\bx= \frac{\bx}{2}+(\frac{\bx}{2}-b\cdot \by') + b\cdot \by'=\sum_{\by} \frac{a_\by}{2}\cdot \by +\sum_{\by} b_\by\cdot  \by  + b\cdot \by'=\sum_{\by} a_\by'\cdot \by\) where \(a_\by'=\frac{a_\by}{2}+b_\by \) if \(\by\neq \by' \) and \(a_{\by'}'=\frac{a_{\by'}}{2}+b_{\by'}+b \). Clearly \(a_\by>0 \) implies \(a_{\by}'>0 \), while \(a_{\by'}'>0 \). Hence there are at most \(i-1\) components \(\by \) of \(\A_\bx \) for which it holds \(a_{\by}'=0 \). Lemma Holds. 
\end{proof}
  
 \begin{lemma}\label{lemma-component-recentering}
 	Let \(\by \) be a component of \(\A \) centered in \(p \). Let \(q \) be a state of the MEC corresponding to \(\by \). Then there exists \(a>0 \) such that \(a\cdot \by \) is a component of \(\A \) centered in \(q \). 
 \end{lemma}

 \begin{proof}
 	Let \(a=\frac{1}{\sum_{t \in \tout(q)} \by(t)}\). Note that \(a>0\) as \(\sum_{t \in \tout(q)} \by(t)>0 \) from the definitions of \(\by \) and \(q\).  By Lemma~\ref{lemma-multiplication-multicomponents} we have that \(a\cdot \by \) 
 	is a multi-component of \(\A \). It holds \(\sum_{t \in \tout(q)} a\cdot \by(t) =a\cdot \sum_{t \in \tout(q)}  \by(t) =\frac{1}{\sum_{t \in \tout(q)} \by(t)}\cdot \sum_{t \in \tout(q)} \by(t)=1 \), hence \(a\cdot \by \) is 
 	centered in \(q \). Since \(\by(t)>0 \) iff \(a\cdot \by(t)>0 \), the MECs corresponding to \(\by \) and \(a\cdot \by \) are the same. Hence \(a\cdot \by \) is a component of \(\A\) centered in \(q\).
 \end{proof}

 \section{Proof for VASS Markov Chains}
\label{app-markov-chains}

In this section we prove the Theorem~\ref{thm-VASS-Markov-chain-analysis} from Section~\ref{section-markov-chains}. We begin by restating the theorem.

\begin{theorem*}[\textbf{\ref{thm-VASS-Markov-chain-analysis}}]
	Let $\A$ be a strongly connected VASS Markov chain. Let \(c\) be a counter and \(t\) a transition of \(\A \).  Then for each \(\F\in \{\calC_\A[c],\calT_\A[t],\calL_\A \} \) exactly one of the following holds:	
	\begin{itemize}
		\item $\F$ is unbounded.
		\item $n^2$ is a tight asymptotic estimate of $\F$.
		\item $n$ is a tight asymptotic estimate of $\F$.
	\end{itemize}
	It is decidable in time polynomial in \(\size{\A} \) which of the above cases holds.
\end{theorem*}
\begin{proof}
	Note that for analyzing \(\calL_\A \), we can simply add a new step-counter \(sc\) to \(\A \), that is a counter that is increased by \(+1\) by every single transition of \(\A \). Then it holds \(\calL_\A=\calC_\A[sc]-n \), while \(\calL_\A\) has a trivial lower asymptotic estimate of \(n\). As such the asymptotic behavior of \(\calL_\A\) is the same as that of \(\calC_\A[sc] \). Similarly, we can express \(\calT_\A[t] \) as \(\calC_\A[c_t]-n \) where \(c_t\) is a fresh transition counter which is increased only by \(t\) and unchanged by any transition other than \(t\). It thus suffices to analyze \(\calC_\A[c] \).
	
	Let \(\countersset \) be the set of counters of \(\A \).	Given a counter \(c\in \countersset\) we denote by \(\A_c \) the \(1\)-dim VASS Markov chain obtained from \(\A \) by removing all counters but \(c\). Let \(\by \) be some component corresponding to \(\A \).\footnote{Note that \(\by \) exists since there is only a single strategy in \(\A\) which is in \(\cMD \). Thus every multi-component of \(\A\) that is centered in some state is a component.} Let \(\sigma\) be the only strategy on \(\A \).

	%Also note that every computation on \(\A \) will ``very quickly'' fall into some MEC.\footnote{the probability of not being in a MEC decreases exponentially with the number of steps.}

	%For each transition \(t\) that is not contained in any BSCC \(M_1,\dots, M_w \) it holds that each time \(t\) is used there is a positive probability \(p>0\) of the computation falling into some BSCC before iterating \(t\) again. Thus it holds \(\prob_{p\vec{n}}^\sigma(\cal_\A[t]>k)\leq p^{k-1} \). And therefore for any function \(f:\mathbb{R}\rightarrow\mathbb{R}  \) it holds \(\prob_{p\vec{n}}^\sigma(\cal_\A[t]\geq f(n))\leq p^{f(n)-1}  \). And if \(f\) is unbounded then it holds \(\lim_{n\rightarrow\infty}\prob_{p\vec{n}}^\sigma(\cal_\A[t]\geq f(n))\leq \lim_{n\rightarrow\infty} p^{f(n)-1}=0 \). Therefore \(f \) is an upper asymptotic estimate of $\calT_\A[t]$ for every unbounded function \(f:\mathbb{R}\rightarrow\mathbb{R} \).

	Let \(C_+=\{c\mid \Delta(\by)(c)>0 \} \). There are three possibilities.
	
	\begin{itemize}
		\item If \(c\notin C_+\) or \(\Delta(\by)\ngeq \vec{0} \): in the former case we get from \cite[Theorem~11]{AKCONCUR23}  that \(n\) is a tight asymptotic estimate of \(\calC_{\A_{c}}[c] \). In the latter case we get that there exists a counter \(c'\) with \(\Delta(\by)(c')<0 \) which gives us from \cite[Theorem~11]{AKCONCUR23}  that \(n\) is a tight asymptotic estimate of \(\calL_{\A_{c'}}\). Thus in both cases  \(n\) is a tight asymptotic estimate of \(\calC_{\A}[c] \), as it holds \(\calC_{\A}[c]\leq \calC_{\A_c}[c] \) 		as well as \(\calC_{\A}[c]\leq n+u\cdot \calL_{\A_{c'}}\) where \(u \) is the maximal counter change per single transition in \(\A \).
		\item If \(c\in C_+\), \(\Delta(\by)\geq \vec{0} \), and \(\by \) is zero-bounded on every counter \(c'\notin C_+\):  then each \(\A_{c'} \) is either increasing or zero-bounded. Thus from \cite[Theorem~11]{AKCONCUR23} we have that \(\calL_{\A_{c'}} \) is unbounded for every counter \(c'\). Furthermore, as \(\A_{c} \) is increasing, from \cite[Theorem~11]{AKCONCUR23} we have that also \(\calC_{\A_{c}}[c] \) is unbounded. 
		
			Assume that there exists \(f:\R\rightarrow\R \) such that \(f\) is not lower asymptotic estimate of \(\calC_\A[c] \).   Let \(g:\N\rightarrow\N \) be such that			 \(\lim_{n\rightarrow\infty} \prob_{p\vec{n}}^\sigma[\calC_\A[c]\geq f(n^{1-\epsilon})\textit{ and this happens within at 				most  }g(n) \textit{ steps}]=1\). Note that the existence of \(g(n)\) follows from \(\calC_{\A_{c}}[c] \) being unbounded. Then it holds
			 \begin{gather*}
			 \lim_{n\rightarrow\infty} \prob_{p\vec{n}}^\sigma[\calC_\A[c]\geq f(n^{1-\epsilon})]
			 \geq\\
			  \lim_{n\rightarrow\infty} \prob_{p\vec{n}}^\sigma[\calC_\A[c]\geq f(n^{1-\epsilon})\mid \textit{ for each } c'\in \countersset,  \calL_{\A_{c'}}\geq g(n)]\cdot \Pi_{c'\in \countersset} \prob_{p\vec{n}}^\sigma[\calL_{\A_{c'}}\geq g(n)]=\\=1
			 \end{gather*}

		Hence \(\calC_\A[c] \) is also unbounded.
		
		\item If \(c\in C_+\), \(\Delta(\by)\geq \vec{0} \), and there exists a counter \(c'\) such that \(\by \) is zero-unbounded on \(c'\): then from \cite{AKCONCUR23} we have a tight asymptotic estimate of \(n^2\) for \(\calL_{\A_{c'}} \). This gives us an upper asymptotic estimate of \(n^2 \) for \(\calC_\A[c] \).
		
		%	If for each \(1\leq i\leq w \) there either exists a counter \(c'\) such that \(\Delta(\bx_i)(c')<0 \), or \(\Delta(\bx_i)=0 \): then every single MEC of \(\A_\sigma \) for \(\sigma\in\Sigma_{MD} \) is either decreasing, zero-bounded, or zero-unbounded on \(c \). Thus from \cite{AKCONCUR23} it holds that \(\calC_{\A_c}[c] \) has an upper asymptotic estimate of \(n^2\). Therefore also \(\calC_\A[c] \) has an upper asymptotic estimate of \(n^2 \). 
		
		%	In addition to this:
		%	\begin{itemize}
			%		 if there exists \(1\leq i\leq w \) such that \(\Delta(\bx_i)\geq \vec{0} \), there exists counter \(c'\) such that \(\bx_i \) is zero-unbounded on \(c' \), and \(\Delta(\bx_i)(c)>0 \):
			
			Furthermore, for each counter \(r\) we have from \cite{AKCONCUR23} that  \(\calL_{\A_{r}} \) has a lower asymptotic estimate of \(n^2\). Thus it holds that \(\lim_{n\rightarrow\infty} \prob_{p_\by\vec{n}}^\sigma(\calL_\A\geq n^{2-\epsilon})=1 \). Let \(\epsilon_1>0 \), and let \(R_{\epsilon_1} \) be the set of all computations on \(\A \) that visit \(p_\by \) at 		least once every \(n^{\epsilon_1} \) steps within  the first \(n^2 \) steps.  Let \(X \) be the random variable denoting the number of sub-paths, within the prefix of length \(n^2\) of the computation, of length \(n^{\epsilon_1}\) that do not contain \(p_\by\).
			Since every step there is a positive probability \(\kappa\) 		 of reaching \(p_\by \) within the next \(a\) steps, it holds \(\E_{p_\by\vec{n}}^\sigma(X)\leq  n^2\cdot (1-\kappa)^{n^{\epsilon_1}} \). From Markov inequality we obtain that			 	 \(\prob_{p_\by\vec{n}}^\sigma(X\geq 1)\leq n^2\cdot (1-\kappa)^{n^{\epsilon_1}}  \). Thus \(\lim_{n\rightarrow\infty}\prob_{p_\by\vec{n}}^\sigma(R_{\epsilon_1})=1 \). 
			
			Therefore \(\lim_{n\rightarrow\infty} \prob_{p_\by\vec{n}}^\sigma(\calL_\A\geq n^{2-\epsilon}\textit{ and }R_{\epsilon_1})=1 \). 
			Thus we can consider a (virtual) strategy \(\sigma' \) which  on \(\A \)  first splits the counter vector into three bins equally, and then performs the computation on \(\A \) in such way that for each counter other than \(c\) it puts the effect into 
			the first bin, while the effect on the counter \(c\) is split between the second and third bin such that the effect of \(\hat{\by} \) on \(c\) is put into the second bin, while every time the computation revisits \( p_\by\) it adds \(\Delta(\by)(c) \) to the third bin (note that the sum of all the three bins always sums up to the actual counters vector). Since \(\hat{\by} \) is either
			zero-bounded or zero-unbounded on \(c\), from \cite{AKCONCUR23} we have in both cases a lower asymptotic estimate of  \((\lfloor\frac{n}{3}\rfloor)^2\) on the number of steps before the second bin depletes \(c\). Furthermore, for any computation in \(R_{\epsilon_1} \) the third bin reaches the value of at least \(n^{2-\epsilon_1}\cdot \Delta(\by_i)(c) \) for the value of \(c\). From \(\lim_{n\rightarrow\infty} \prob_{p_\by\vec{n}}^{\sigma'}(\calL_\A\geq n^{2-\epsilon}\textit{ and }R_{\epsilon_1})=1 \) it therefore holds  \(\lim_{n\rightarrow\infty} \prob_{p_\by\vec{n}}^{\sigma'}(\calC_\A[c]\geq n^{2-\epsilon-\epsilon_1-\epsilon_2})=1 \) for some \(\epsilon_2>0 \). Thus \(n^2 \) is lower asymptotic estimate of \(\calC_\A[c] \). Combined with the upper asymptotic estimate form above this gives a tight asymptotic estimate of  \(n^2\) for  \(\calC_\A[c] \).
	\end{itemize}

	\textbf{The decidability in polynomial time:} 
	When classifying the asymptotic estimate of \(\calC_\A[c] \) for the counter \(c\), we can do so in polynomial time as follows: First we compute some component \(\by\). Then we ask if \(\Delta(\by)(c)>0 \), if not then \(\calC_\A[c] \) has a tight asymptotic estimate of \(n\). If \(\Delta(\by)(c)>0 \) then we ask whether \(\by \) is either increasing or zero-bounded on every counter, and if yes then \(\calC_\A[c] \) is unbounded, otherwise \(\calC_\A[c] \) has a tight asymptotic estimate of \(n^2\). Note that deciding whether \(\by \) is increasing or zero-bounded can be done in polynomial time as per \cite{AKCONCUR23}.
\end{proof}

 \section{Proofs for Fixed Probability Bounds}
\label{app-observation-f-p-f}

In this section we prove the Theorem~\ref{observation-f-p-f} from Section~\ref{sec-comp-VASS-runs}. Let us start by restating the theorem.

\begin{theorem*}[\textbf{\ref{observation-f-p-f}}]

Let \(f:\mathbb{R}\rightarrow\mathbb{R} \) be such that \(\lim_{n\rightarrow\infty} \frac{f( n)}{f( n^{1+\epsilon})}=0\) for every \(\epsilon>0 \). Then:
\begin{itemize}
	\item \(f\) is a lower asymptotic estimate of \(\F \) iff for every \(\epsilon>0 \) and every \(p<1 \) it holds \(f_p^\F\in \Omega(f(n^{1-\epsilon})) \);
	\item \(f\) is an upper asymptotic estimate of \(\F \) iff for every \(\epsilon>0 \) and every \(p<1 \) it holds \(f_p^\F\in \calO(f(n^{1+\epsilon})) \).
\end{itemize}
\end{theorem*}

This follows from the following four Lemmas, each proving one direction of one of the bullet points. 

\begin{lemma}
	Let \(f:\mathbb{R}\rightarrow\mathbb{R} \) be such that \(\lim_{n\rightarrow\infty} \frac{f( n)}{f( n^{1+\epsilon})}=0\) for every \(\epsilon>0 \).
	If for every \(\epsilon>0 \) and every \(p<1 \) it holds \(f_p^\F\in \calO(f(n^{1+\epsilon})) \) then \(f\) is an upper asymptotic estimate of \(\F \).
\end{lemma}

\begin{proof}
 In such case, for each \(p<1 \) and each \(\epsilon>0 \) there exists  \(a_{p,\epsilon}\in \mathbb{N} \) and \(n_{p,\epsilon}\in \mathbb{N}\) such that \(f_p^\F(n)\leq a_{p,\epsilon}\cdot  f( n^{1+\epsilon} )\) for all \(n\geq n_{p,\epsilon} \). Hence for all \(n\geq n_{p,\epsilon} \), every $q \in Q$, and every strategy $\sigma$ it holds \(\Prob_{q\vec{n}}^\sigma[\F\leq  a_{p,\epsilon}\cdot  f( n^{1+\epsilon} )] \geq \Prob_n[\F\leq  f_p^\F(n)]\geq  p\). Since \(\lim_{n\rightarrow\infty} \frac{f( n)}{f( n^{1+\epsilon})}=0\) for every \(\epsilon>0 \), it holds for each \(p<1 \) and each \(\epsilon>0 \) that \(f( n^{1+2\epsilon} )\geq a_{p,\epsilon}\cdot f( n^{1+\epsilon} ) \) for all sufficiently large \(n\). Thus for each \(p<1 \) and each \(\epsilon>0 \) there exists \(n_{p,\epsilon}'\in \mathbb{N} \) such that for each \(n>n_{p,\epsilon}'\), every $q \in Q$, and every strategy $\sigma$ it holds \( \Prob_{q\vec{n}}^\sigma[\F\leq  f( n^{1+2\epsilon} )]  \geq  \Prob_{q\vec{n}}^\sigma[\F\leq  a_{p,\epsilon}\cdot  f( n^{1+\epsilon} )] \geq  p \).
	
	Hence 
	\begin{multline*}	
	\limsup_{n\rightarrow\infty} \Prob_{q\vec{n}}^\sigma[\F\geq f( n^{1+3\epsilon} )] = \limsup_{n\rightarrow\infty} 1-\Prob_{q\vec{n}}^\sigma[\F< f( n^{1+3\epsilon} )] \leq\\
	\limsup_{n\rightarrow\infty} 1-\Prob_{q\vec{n}}^\sigma[\F\leq f( n^{1+2\epsilon} )]\leq  \limsup_{n\rightarrow\infty} 1-\max \{p<1 \mid n> n_{p,\epsilon}' \}=0 \end{multline*}
	
	Hence \(f\) is an upper asymptotic estimate of \(\F \).
\end{proof}

\begin{lemma}
	
		Let \(f:\mathbb{R}\rightarrow\mathbb{R} \) be such that \(\lim_{n\rightarrow\infty} \frac{f( n)}{f( n^{1+\epsilon})}=0\) for every \(\epsilon>0 \).
	If for every \(\epsilon>0 \) and every \(p<1 \) it holds \(f_p^\F\in \Omega(f(n^{1-\epsilon})) \) then \(f\) is a lower asymptotic estimate of \(\F \).
	
%	Assume that \(f_p(n)\in \Omega(f(\lceil n^{1-\epsilon} \rceil)) \) for each \(\epsilon>0 \) and \(p<1\), and that it holds \(\lim \frac{f(\lceil n^{1-2\epsilon}}{f(\lceil n^{1-\epsilon}}=0\). Then \(f\) is a lower asymptotic estimate of \(\F \).
\end{lemma}

\begin{proof}
	In such case, for each \(p<1 \) and each \(\epsilon>0 \) there exists  \(a_{p,\epsilon}\in \mathbb{N} \) and \(n_{p,\epsilon}\in \mathbb{N}\) such that \(f_p^\F(n)\geq a_{p,\epsilon} \cdot f( n^{1-\epsilon} )\) for all \(n\geq n_{p,\epsilon} \).

	Hence for each \(p<1\) and each \(\epsilon>0 \), as long as \(f^\F_p(n)\neq \infty \), it holds for every $q \in Q$, and every strategy $\sigma$ that \( \Prob_{q\vec{n}}^\sigma[\F\leq a_{p,\epsilon}\cdot  f( n^{1-\epsilon} )-1]\leq \Prob_{q\vec{n}}^\sigma[\F\leq  f^\F_p(n)-1]<  p\) for all \(n\geq n_{p,\epsilon}\). From \(\lim_{n\rightarrow\infty} \frac{f( n)}{f( n^{1+\epsilon})}=0\) we get that for all sufficiently large \(n\) it holds \(a_{p,\epsilon}\cdot  f( n^{1-\epsilon} )-1\geq  f( n^{1-2\epsilon} ) \).

	Let \(p'<1 \) be the largest value such that \(f_{p'}^\F\neq \infty \).
	
	Then for each \(p<p' \), each \(\epsilon>0 \), every $q \in Q$, and every strategy $\sigma$  there exists \(n_{p,\epsilon}'\in\mathbb{N} \) such that for each \(n>n_{p,\epsilon}'\) it holds \(p\geq  \Prob_{q\vec{n}}^\sigma[\F\leq a_{p,\epsilon}\cdot  f( n^{1-\epsilon} )-1]  \geq \Prob_{q\vec{n}}^\sigma[\F\leq f( n^{1-2\epsilon} )] \).
	
%Let \( \) \(p>1\) such that \(f_p^\F\neq \infty \) it holds 
	
%	Hence for each \( \epsilon>1\)  and each \(n\in \mathbb{N}\) it holds \(\min \{p<1\mid n> n_{p,\epsilon} \} \leq \Prob_n[X< f(\lceil n^{1-2\epsilon} \rceil)] \)
	
	If \(p'=0 \) then it holds \(\Prob_{q\vec{n}}^\sigma[\F\geq f( n^{1-2\epsilon} )]=0\) trivially from the definition of \(f_{p'}^\F \). 
	
	If \(p'\neq 0 \) then it holds
	\begin{multline*}
	\liminf_{n\rightarrow\infty} \Prob_n[\F\geq f( n^{1-2\epsilon} )]=
	\liminf_{n\rightarrow\infty} 1-\Prob_n[\F< f( n^{1-2\epsilon} )]\geq 
	\liminf_{n\rightarrow\infty} 1-\min \{p<p'\mid n> n_{p,\epsilon}' \} = 1
		\end{multline*}
	
	Hence \(f\) is a lower asymptotic estimate of \(\F \).
\end{proof}

\begin{lemma}
	Let \(f \) be an upper asymptotic  estimate of \(\F\). Then for each \(\epsilon>0 \) and \(p<1\) it holds \(f^\F_p(n)\in \calO(f( n^{1+\epsilon} )) \).
\end{lemma}
\begin{proof}
	Let us fix an \(\epsilon>0\) and \(p<1 \).
	From the definition of the upper asymptotic estimate, we obtain \(\limsup_{n\rightarrow\infty} \Prob_{q\vec{n}}^\sigma[\F\geq f( n^{1+\epsilon} )]=0 \) for any \(q\in Q\) and any strategy \(\sigma \). Therefore, for each \(\gamma>0  \) there exists \(n_\gamma\in \mathbb{N} \) such that for all \(n\geq n_\gamma \) it holds \(\Prob_{q\vec{n}}^\sigma[\F\geq f( n^{1+\epsilon} )]<\gamma \). Thus, for \(\gamma=1-p\), we obtain \(\Prob_{q\vec{n}}^\sigma[\F\geq f( n^{1+\epsilon})]<1-p \) which implies \(\Prob_{q\vec{n}}^\sigma[\F< f( n^{1+\epsilon})]\geq p \), for all \(n\geq n_{1-p} \), any \(q\in Q\), and any strategy \(\sigma\). Hence, for all \(n\geq n_{1-p} \),  it holds \(f_p^\F(n)\leq \lceil f( n^{1+\epsilon} ) \rceil\) which implies \(f^\F_p(n)\in \calO(f( n^{1+\epsilon} )) \).\qed
\end{proof}

\begin{lemma}
	Let \(f:\mathbb{R}\rightarrow\mathbb{R} \) be such that \(\lim_{n\rightarrow\infty} \frac{f( n)}{f( n^{1+\epsilon})}=0\), and let \(f\) be a lower asymptotic estimate of \(\F\). Then for each \(\epsilon>0 \) and \(p<1\) it holds \(f^\F_p(n)\in \Omega( f( n^{1-\epsilon} )) \).
\end{lemma}
\begin{proof}
	Let us fix an \(\epsilon>0\) and \(p<1\).
	From the definition of the lower asymptotic estimate, we obtain that there exists \(q\in Q\) and a strategy \(\sigma\)  such that \(\liminf_{n\rightarrow\infty} \Prob_{q\vec{n}}^\sigma[\F\geq f( n^{1-\nicefrac{\epsilon}{2}} )]=1 \). Hence, for each \(\gamma>0  \) there exists \(n_\gamma\in \mathbb{N} \) such that for all \(n\geq n_\gamma \) it holds \(\Prob_{q\vec{n}}^\sigma[\F\geq f( n^{1-\nicefrac{\epsilon}{2}} )]>1-\gamma \) which implies \(\Prob_{q\vec{n}}^\sigma[\F> f( n^{1-\epsilon} )]>1-\gamma \) for all sufficiently large \(n\) as \(f\) is increasing. This can be rewritten as \(\Prob_{q\vec{n}}^\sigma[\F\leq  f( n^{1-\epsilon} )]\leq \gamma \). Therefore, if we put \(\gamma<p \), we obtain
	\(\Prob_{q\vec{n}}^\sigma[\F\leq  f( n^{1-\epsilon} )]<\gamma<p \) for all sufficiently large \(n\),  which implies \(f^\F_p(n)>f( n^{1-\epsilon}  )\).
	
	 Hence \(f^\F_p(n)\in \Omega( f( n^{1-\epsilon} )) \).
\end{proof}

\section{Proof of Lemma~\ref{lemma-linear-upper-bound}}
\label{app-lemma-linear-upper-bound}
\begin{lemma*}[\textbf{\ref{lemma-linear-upper-bound}}]
		Let \[\variablet=    \begin{cases}
	1 & \text{if } k=1  \\
	k & \text{if } k>1 \text{ and }T_k'\neq \emptyset \\
	\max\{a+b\mid a\in \Aset,b\in \Bset; a+b\leq k \} & \text{if } k>1 \text{ and }T_k'= \emptyset
\end{cases}\]

Then for every counter $c\in \countersset$ such that $\by_k(c)>0$ it holds that \(n^\variablet\) is an upper asymptotic estimate of \( \calC_\A[c]\). Furthermore each transition \(t=(p,\bu,q) \)  of \(\A_{k,\hat{T}_{k+1}} \) has an upper asymptotic estimate of \(n^\variablet \) for \(\calT_\A[t] \) if one of the following holds:
\begin{itemize}
	\item  \(p\in Q_n \) and \(\bz_k(q)-\bz_k(p)+\sum_{c\in\countersset} \bu(c)\cdot \by_k(c)<0\),
	\item  \(p\in Q_p \) and \(\sum_{t' = (p',\bu',q') \in \tout(p)}P(t')\cdot \big(\bz_k(q')-\bz_k(p')+\sum_{c\in\countersset}  \bu'(c)\cdot \by_k(c)\big)< 0\).
\end{itemize}  
\end{lemma*}
%
%\begin{lemma*}[\textbf{\ref{lemma-upper-bound-T}}]
%	Let \(t=(p,\bu,q) \) be a transition of \(\A_{k,\hat{T}_{k+1}} \), if one of the following holds then \(\calT_\A[t] \) has an upper asymptotic estimate of \(n^l \):
%	\begin{itemize}
%		\item if \(p\in Q_n \) and \(\bz_k(q)-\bz_k(p)+\sum_{i=1}^{d} \bu(i)\by_k(i)<0\),
%		\item if \(p\in Q_p \) and \(\sum_{t' = (p',\bu',q') \in \tout(p)}P(t')\big(\bz_k(q')-\bz_k(p')+\sum_{i=1}^{d}  \bu'(i)\by_k(i)\big)< 0\).
%	\end{itemize} 
%	
%	where \[l=    \begin{cases}
%		1 & \text{if } k=1  \\
%		k & \text{if } k>1 \text{ and }T_k'\neq \emptyset \\
%		\max\{a+b\mid a\in A,b\in B; a+b\leq k \} & \text{if } k>1 \text{ and }T_k'= \emptyset
%	\end{cases}\]
%\end{lemma*}

%
%	
%	\begin{lemma}\label{lemma-linear-upper-bound}
	%		For every counter $c\in C_{k+}$ such that $\by_k(c)>0$ it holds that \(n^k\) is an upper asymptotic estimate of \( \calC_\A[c]\).
	%	\end{lemma}	
%
%	\begin{lemma}\label{lemma-upper-bound-T}
	%		Let \(t=(p,\bu,q) \) be a transition of \(\A_{k}' \), if one of the following holds then \(\calT_\A[t] \) has an upper asymptotic estimate of \(n^k \):
	%		\begin{itemize}
		%			\item if \(p\in Q_n \) and \(\bz_k(q)-\bz_k(p)+\sum_{i=1}^{d} \bu(i)\by_k(i)<0\),
		%			\item if \(p\in Q_p \) and \(\sum_{t' = (p',\bu',q') \in \tout(p)}P(t')\big(\bz_k(q')-\bz_k(p')+\sum_{i=1}^{d}  \bu'(i)\by_k(i)\big)< 0\).
		%		\end{itemize} 
	%		
	%	\end{lemma}

Thorough this Section let us fix some strategy \(\sigma \) on \(\A\) and \(\epsilon>0  \). Note that if \(\hat{T}_{k+1}=\emptyset \) the Lemma holds trivially, hence in the rest of this section we will assume \(\hat{T}_{k+1}\neq\emptyset \). Note that we can also wlog. assume that \(\by_k(c) \) is either \(0\) or at least \(1\) for each counter \(c\).
%		, and \(a>0 \) such that \(\limsup_{n\rightarrow\infty} \prob_{p\vec{n}}[\calC[c]\geq n^{k+\epsilon}]\geq a \).

 Let \(rank_{k,\hat{T}_{k+1}}(p\bv')=\bz_k(p)+\sum_{c\in \countersset_k} \bv'(c)\cdot \by_k(c) \) be a ranking function given by \(\by_k,\bz_k \) on \(\A_{k,\hat{T}_{k+1}} \), where \(\countersset_k\) are the counters and  \(p\bv' \) a configuration of \(\A_{k,\hat{T}_{k+1}} \). We will now extend \(rank_{k,\hat{T}_{k+1}} \) onto \(\A \) as follows: For a configuration \(p\bv\) of \(\A \), let \(rank_k(p\bv)=\bz_k(p)+\sum_{c\in \countersset_k} \bv_{k}(c)\cdot \by_k(c) \) where \(\bv_k\in \mathbb{Q}^{d_k}\) is a counters vector on \(\A_{k,\hat{T}_{k+1}} \) such that:
 \begin{itemize}
 	\item if \(c\in C_{k+}\) then \(\bv_k(c)=\bv(c) \),
 	\item if \(c\in \bigcup_{i=1}^{k-1} C_i\) and \(c_p \) is the local copy of \(c\) in \(\A_k \) in the MEC containing \(p\) then \(\bv_k'(c_p)=\bv(c) \),
 	\item otherwise \(\bv_k'(c)=0\).
 \end{itemize} (intuitively \(rank_k\) corresponds to the value of \(rank_{k,\hat{T}_{k+1}} \) but it accounts for the fact that the counter \(c\) of \(\A \) might have different weights associated to it in different states, as \(\A_k \) may contain multiple copies of \(c\) and each copy \(c_k' \) may have different value of \(\by_k(c_k') \). Thus \(rank_k\) uses in configuration \(p\bv \) for \(c\) the weight that is associated with the copy of \(c\) in \(\A_k \) in the state \(p \))

Let \(P_0V_0,P_1V_1,P_2V_2,\dots \) be the random variables encoding the computation under  \(\sigma \) in \(\A\), and let \(\bar{T}_1,\bar{T}_2,\dots \) be the random variables encoding the transitions taken (i.e. \(P_iV_i \) represents the configuration at \(i\)-th step of the computation and \(\bar{T}_i \) the \(i\)-th transition taken by this computation). Let \(R_0,R_1,R_2,\dots \) represent the value of \(rank_k \) at \(i-\)th step (i.e. \(R_i=rank_k(P_iV_i) \)).

Let us define technical constants \(0<\epsilon_1,\epsilont_2,\dots  \). As their exact values are not important we leave the assignment of their exact values to Table~\ref{Table-app-lemma-linear-upper-bound} at the end of the section where we also show that our assignment satisfies all the assumptions we make on \(\epsilon_1,\epsilon_2,\dots\) thorough this section.

Let \(E_{1}\) be the set of all computations \(\pi \) on \(\A \), such that:
\begin{itemize}
	\item \(\calT_\A[t](\pi)\leq n^{i+\epsilont_1}  \) for each transition \(t\in T_{i}\setminus T_{i+1}\) where \(i\in \{1,\dots,k-1 \} \),
	\item AND \(\calT_\A[t](\pi)\leq n^{k+\epsilont_1}  \) for each transition \(t\in T_k' \),
	\item AND \(\calC_\A[c](\pi)\leq n^{i+\epsilont_1}  \) for each counter \(c\in C_i \) where \(i\in \{1,\dots,k-1 \} \).
\end{itemize} 
Note that \(\lim_{n\rightarrow\infty} \prob_{p\vec{n}}^\sigma(E_1)=1 \) as every counter and transition we restrict in \(E_{1} \) has a corresponding upper asymptotic estimate for \(\calC_\A[c] \) and \(\calT_\A[t] \).

Let \(R_0',R_1',R_2',\dots \) be random variables defined as follows: \[R_0'=R_0+\sum_{t\in \hat{T}_{k+1}}u\cdot n^{l+\epsilont_2} \]
\[R_i'=    \begin{cases}
	R_{i-1}'-R_{i-1}+R_i & \text{if } \bar{T}_i\in \hat{T}_{k+1} \text{ and } P_0V_0,P_1V_1,\dots,P_iV_i \in E_{1} \\
	R_{i-1}' & \text{else }
\end{cases}\]
where \(u\) is the maximal change of a counter per single transition in \(\A \).

Furthermore, let \(X_1,X_2,\dots\) and \(X_1',X_2',\dots\) be random variables defined by \( R_i= R_{i-1}+X_i\) and \(R_{i}'=R_{i-1}'+X_{i}'\). That is \(X_i=R_i-R_{i+1} \) and \(X_i'=R_i'-R_{i+1}' \). Also let \(S_0,S_1,\dots \) and \(S_0',S_1',\dots \) be defined as \(S_0=S_0'=0 \), \(S_i=S_{i-1}+X_i \), and \(S_i'=S_{i-1}'+X_i' \). Note that it holds \(R_i=R_0+S_i \) and \(R_i'=R_0'+S_i' \).

First let us show that 
\(R_0',R_1',\dots \) is a supermartingale.

\begin{lemma}\label{lemma-Rprimes-supermartinagel}
	\(R_0',R_1',\dots \)  is a supermartingale.
\end{lemma}
\begin{proof}
		It holds \(\Exp^\sigma_{p\vec{n}}(R_{i}'|R_{i-1}')=\Exp^\sigma_{p\vec{n}}(R_{i-1}'|R_{i-1}')+\Exp^\sigma_{p\vec{n}}(X_{i}'|R_{i-1}')=R_{i-1}'+\Exp^\sigma_{p\vec{n}}(X_{i}'|R_{i-1}') \). Thus we want to show that \(\Exp^\sigma_{p\vec{n}}(X_{i}'|R_{i-1}')\leq 0 \).	It holds
	
	\(X_{i}'=\begin{cases}
		-R_{i-1}+R_i &\text{if } \bar{T}_i\in \hat{T}_{k+1} \text{ and } P_0V_0,P_1V_1,\dots,P_iV_i \in E_{1}  \\
		0 & \text{else }
	\end{cases} \)
	
	Thus it holds \begin{gather*}
	\Exp^\sigma_{p\vec{n}}(X_{i}'|R_{i-1}')
	= \sum_{t\in \hat{T}_{k+1}} \prob_{p\vec{n}}^\sigma(\bar{T}_{i}=t \textit{ and } P_0V_0,P_1V_1,\dots,P_iV_i \in E_{1}|R_{i-1}')\cdot\RankEff(t)+\\+ \sum_{t\in \hat{T}_{k+1}} \prob_{p\vec{n}}^\sigma(\bar{T}_{i}=t \textit{ and } P_0V_0,P_1V_1,\dots,P_iV_i \notin E_{1} |R_{i-1}')  \cdot 0 + \sum_{t\notin \hat{T}_{k+1}} \prob_{p\vec{n}}^\sigma(\bar{T}_{i}=t|R_{i-1}')\cdot 0  
	=\\= \sum_{t\in \hat{T}_{k+1}} \prob_{p\vec{n}}^\sigma(\bar{T}_{i}=t \textit{ and } P_0V_0,P_1V_1,\dots,P_iV_i \in E_{1}|R_{i-1}')\cdot\RankEff(t) 
	 \end{gather*}
	
	where \(\RankEff(t) \) represents the effect of \(t\)  on \(rank_k\).

	Let \(T_n=\hat{T}_{k+1}\cap\bigcup_{p\in Q_n} \tout(p)\) and \(T_p=\hat{T}_{k+1}\cap \bigcup_{p\in Q_p} \tout(p) \). We can write		
	\begin{gather*}
		\Exp^\sigma_{p\vec{n}}(X_{i}'|R_{i-1}')
		=
		\sum_{t\in \hat{T}_{k+1}} \prob_{p\vec{n}}^\sigma(\bar{T}_{i}=t \textit{ and } P_0V_0,P_1V_1,\dots,P_iV_i \in E_{1}|R_{i-1}')\cdot\RankEff(t)
		=\\=
	\sum_{t\in T_n} \prob_{p\vec{n}}^\sigma(\bar{T}_{i}=t\textit{ and } P_0V_0,P_1V_1,\dots,P_iV_i \in E_{1}|R_{i-1}')\cdot \RankEff(t)+\\+\sum_{t\in T_p} \prob_{p\vec{n}}^\sigma(\bar{T}_{i}=t\textit{ and } P_0V_0,P_1V_1,\dots,P_iV_i \in E_{1}|R_{i-1}')\cdot \RankEff(t)
	\end{gather*}

	Note that for each \(t\in \hat{T}_{k+1}\)  the effect of \(t \) on \(rank_k\) is the same as the effect of \(t \) on \(rank_{k,\hat{T}_{k+1}} \).  Hence for each \(t=(p,\bu,q)\in \hat{T}_{k+1} \) it holds that \(\RankEff(t)=\bz_k(q)-\bz_k(p)+\sum_{c\in \countersset_k} \by_k(c)\cdot \bu_k(c)\). Thus from \( \by_k,\bz_k\) being a solution of \hyperref[fig-systems]{(II)} for \(\A_k' \) we have for each \(t\in T_n \) that  \(\RankEff(t)\leq 0 \),  and for each \(p\in Q_n \) that \(\sum_{t\in \tout(p)} P(t)\cdot \RankEff(t) \leq 0 \) (See Section~\ref{sec-systems} for details). Therefore it holds
	\[\sum_{t\in T_n} \prob_{p\vec{n}}^\sigma(\bar{T}_{i}=t\textit{ and } P_0V_0,P_1V_1,\dots,P_iV_i \in E_{1}|R_{i-1}')\cdot \RankEff(t)\leq 0
	\]
	and 
	\begin{gather*}
	 \sum_{t\in T_p} \prob_{p\vec{n}}^\sigma(\bar{T}_{i}=t\textit{ and } P_0V_0,P_1V_1,\dots,P_iV_i \in E_{1}|R_{i-1}')\cdot \RankEff(t)
	=\\=
	\sum_{p\in Q_p} \prob_{p\vec{n}}^\sigma(P_{i-1}=p\textit{ and } P_0V_0,P_1V_1,\dots,P_iV_i \in E_{1}\mid R_{i-1}') \cdot  \sum_{t\in \tout(p)} P(t)\cdot \RankEff(t) 
	\leq 0		
\end{gather*}
	Thus \(E_{p\vec{n}}^\sigma[X_{i}'|R_{i-1}']\leq 0\) and so \(R_0',R_1',\dots \) is a supermartingale.
\end{proof}

Now let us prove the following Lemma.

\begin{lemma}\label{lemma-R-lessthan-Rprime-cond-Reps2}
	For all sufficiently large \(n\) it holds \(\prob_{p\vec{n}}^\sigma(R_i> R_i'|E_{1})=0 \) for all \(i\in \mathbb{N}_0 \).
\end{lemma}

From \(R_0=R_0'-\sum_{t\in \hat{T}_{k+1}}u\cdot n^{l+\epsilont_2} \) we have  \(\prob_{p\vec{n}}^\sigma(R_i> R_i'|E_{1})=\prob_{p\vec{n}}^\sigma(S_i- S_i'>\sum_{t\in \hat{T}_{k+1}}u\cdot n^{l+\epsilont_2}|E_{1}) \). 

Therefore Lemma~\ref{lemma-R-lessthan-Rprime-cond-Reps2} follows from the following Lemma.

\begin{lemma}
	For all sufficiently large \(n\), conditioned on \(E_1 \) it holds  it holds \(|S_i- S_i'|<\sum_{t\in \hat{T}_{k+1}}u\cdot n^{l+\epsilont_2} \) for all \(i\in \mathbb{N}_0 \).
\end{lemma}
\begin{proof}
	Note that conditioned on \(E_{1}\) it holds that \(X_i\neq X_i' \) iff \(\bar{T}_i\notin \hat{T}_{k+1}  \). And since there is an upper limit on how many times any transition \(t \notin \hat{T}_{k+1}\) can appear in \(E_{1} \), conditioned on \(E_{1}\) it holds that  \(|S_i-S_i'|\leq   \sum_{t \notin \hat{T}_{k+1}} \max_t(E_{1})\cdot |\RankEff(t)| \) where \(\max_t(E_{1})\) is the maximal number of times \(t\) can appear along any computation from \(E_{1} \). We can rewrite this as 
	\begin{gather*}
		|S_i-S_i'|
		\leq 
		\sum_{t\in T_k'}\max_t(E_{1})\cdot |\RankEff(t)| + \sum_{i=1}^{k-1}\sum_{t \in T_{i}\setminus T_{i+1}} \max_t(E_{1})\cdot |\RankEff(t)|
		\leq\\
		\sum_{t\in T_k'}n^{k+\epsilont_1}\cdot r + \sum_{i=1}^{k-1}\sum_{t \in T_{i}\setminus T_{i+1}} n^{i+\epsilont_1}\cdot |\RankEff(t)|
	\end{gather*}
	for some constant \(r\).
	Notice that for any \(t \in T_{i}\setminus T_{i+1}\) the value of \(\RankEff(t)\) consists of two parts, the first is the actual effect of the transition on the counters + the change of rank from changing the state, which can be bounded by a constant, and the second part consists in potentially changing the weights assigned to individual counters in \(rank_k \). But \(t \in T_{i}\setminus T_{i+1}\) can change the weights only for counters from \(C_1,\dots,C_{k-i} \) which in \(E_{1} \) are all upper bounded by \(n^{max(A\cap \{1,\dots,k-i \})+\epsilont_1} \). Thus  \(|\RankEff(t)|\leq r\cdot (1+n^{max(A\cap \{1,\dots,k-i \})+\epsilont_1}) \).
	
	%	 of  Notice that both of these MECs always have the same local copies for counters from \(C_{k-i-1},\dots,C_{k-1},C_{k+} \). Thus along any path of \( E_{\epsilon_2}\) it always holds \(|\RankEff(t)|\leq b'(1+n^{k-i+\epsilon_2}) \) for some constant \( b'\).
	
	Therefore assuming 
	\begin{equation}\label{equation-epsilon-bounds-something1}
		\epsilont_2>2\cdot \epsilont_3> \epsilont_3>2\cdot \epsilont_1
	\end{equation}
 we can write for all sufficiently large \(n\)
	\begin{gather*}
		\sum_{t\in T_k'}n^{k+\epsilont_1}\cdot r + \sum_{i=1}^{k-1}\sum_{t \in T_{i}\setminus T_{i+1}} n^{i+\epsilont_1}\cdot |\RankEff(t)|
		\leq \\
		\sum_{t\in T_k'}n^{k+\epsilont_1}\cdot r  + \sum_{i=1}^{k-1}\sum_{t \in T_{i}\setminus T_{i+1}} n^{i+\epsilont_1}\cdot r\cdot (1+n^{max(A\cap \{1,\dots,k-i \})+\epsilont_1})
		\leq \\
		|T_k'|\cdot n^{k+\epsilont_3} + \sum_{i=1}^{k-1}|T_i\setminus T_{i+1}|\cdot  n^{i+\epsilont_3+max(A\cap \{1,\dots,k-i \})}
%		\leq \\
%		n^{k+2\cdot \epsilont_3} + \sum_{i=1}^{k-1} n^{i+2\cdot \epsilont_3+max(A\cap \{1,\dots,k-i \})}
	\end{gather*}
	%		 
	%		=
	%		n^{k+\epsilon_3} + (k-1) n^{k+\epsilon_3}
	%		\leq\\
	%		n^{k+\epsilon_1}
	%		\leq 
	%		\sum_{t\in T_{k}\setminus T_k'}un^{k+\epsilon_1}
	%	\end{gather*}
%	

Consider now the three cases considered in the definition of \(l\) in Lemma~\ref{lemma-linear-upper-bound}.

In the case that \(k=1 \) it holds 
\begin{gather*}
	|T_k'|\cdot n^{k+\epsilont_3} + \sum_{i=1}^{k-1}|T_i\setminus T_{i+1}|\cdot  n^{i+\epsilont_3+max(A\cap \{1,\dots,k-i \})}
	= 
	|T_k'|\cdot n^{k+\epsilont_3} 
	\leq \\
	n^{k+2\cdot \epsilont_3} < n^{k+\epsilont_2} \leq 	\sum_{t\in \hat{T}_{k+1}}u\cdot n^{k+\epsilont_2} = \sum_{t\in \hat{T}_{k+1}}u\cdot n^{l+\epsilont_2}
\end{gather*}

In the case that \(k> 1 \) and \(T_k'\neq \emptyset \) it holds 
\begin{gather*}
	|T_k'|\cdot n^{k+\epsilont_3} + \sum_{i=1}^{k-1}|T_i\setminus T_{i+1}|\cdot  n^{i+\epsilont_3+max(A\cap \{1,\dots,k-i \})}
	\leq \\
	|T_k'|\cdot n^{k+\epsilont_3}  + \sum_{i=1}^{k-1}|T_i\setminus T_{i+1}|\cdot  n^{i+\epsilont_3+k-i}
	=
	|T_k'|\cdot n^{k+\epsilont_3}  + \sum_{i=1}^{k-1}|T_i\setminus T_{i+1}|\cdot  n^{k+\epsilont_3}
	\leq \\ 
	n^{k+2\cdot \epsilon_3} 
	 < n^{k+\epsilont_2} \leq 	\sum_{t\in \hat{T}_{k+1}}u\cdot n^{k+\epsilont_2} = \sum_{t\in \hat{T}_{k+1}}u\cdot n^{l+\epsilont_2}
\end{gather*}
 we have \(l=\max\{a+b\mid a\in A,b\in B; a+b\leq k \}\) then
And in the case that \(k>1 \) and \(T_k'=\emptyset \) it holds 
\begin{gather*}
	|T_k'|\cdot n^{k+\epsilont_3} + \sum_{i=1}^{k-1}|T_i\setminus T_{i+1}|\cdot  n^{i+\epsilont_3+max(A\cap \{1,\dots,k-i \})}
	=\\
	0\cdot n^{k+\epsilont_3} + \sum_{i\in B\cap \{1,\dots,k-1 \}} |T_i\setminus T_{i+1}|\cdot  n^{i+\epsilont_3+max(A\cap \{1,\dots,k-i \})} +\\+ \sum_{i\in  \{1,\dots,k-1 \}\setminus B} |T_i\setminus T_{i+1}|\cdot  n^{i+\epsilont_3+max(A\cap \{1,\dots,k-i \})}
	=\\
	\sum_{i\in B\cap \{1,\dots,k-1 \}} |T_i\setminus T_{i+1}|\cdot  n^{i+\epsilont_3+max(A\cap \{1,\dots,k-i \})} + \sum_{i\in  \{1,\dots,k-1 \}\setminus B} 0\cdot  n^{i+\epsilont_3+max(A\cap \{1,\dots,k-i \})}
	=\\
	\sum_{i\in B\cap \{1,\dots,k-1 \}} |T_i\setminus T_{i+1}|\cdot  n^{i+\epsilont_3+max(A\cap \{1,\dots,k-i \})} 
	\leq \\
	n^{2\cdot \epsilont_3+\max\{a+b\mid a\in A,b\in B; a+b\leq k \}}
	= 
	n^{l+2\cdot \epsilon_3}
	< n^{l+\epsilont_2} \leq 	\sum_{t\in \hat{T}_{k+1}}u\cdot n^{l+\epsilont_2}
\end{gather*}

Thus conditioned on \(E_{1} \) in all three cases it holds for all sufficiently large \(n\) that \(|S_i-S_i'|<    \sum_{t\in \hat{T}_{k+1}}u\cdot n^{l+\epsilont_2} \). 
\end{proof}
% and therefore \(\prob_{p\vec{n}}^\sigma(|S_i-S_i'|>    \sum_{t\in \hat{T}_{k+1}}u\cdot n^{l+\epsilont_2} |E_{1})=0 \).

Now let us consider the stopping rule \(\tau \) that stops when either any counter becomes negative, or any counter \(c\in C_{k+}\) with \(\by_k(c)>0 \) becomes larger then \(n^{l+\epsilon} \) for the first time.	(i.e. either \(V_\tau(c')<0 \) for any counter \(c' \), or \(V_\tau(c)\geq n^{l+\epsilon} \) for any \(c\in C_{k+} \) with \(\by_k(c)>0 \)). It holds for all \(i\in \mathbb{N}_0\) that \[R_{min(i,\tau)} \leq max_{p\in Q} \bz_k(p) + max_{c; \by_k(c)>0} \by_k(c) \cdot d_k\cdot  (n^{l+\epsilon}+u)  \]

%Since it also holds \(R_i'=R_0'+S_i' \), \(R_0=R_0'-\sum_{t\in T_{k}\setminus T_k'}u\cdot n^{l+\epsilont_2}  \), \(R_i=R_0+S_i \), and  \(|(S_i\mid E_{\epsilon_2})-(S_i'\mid E_{\epsilon_2})|\leq    \sum_{t\in T_{k}\setminus T_k'}un^{k+\epsilon_1} \) 

Furthermore, conditioned on \(E_{1} \) we can write 
\begin{gather*}
R_i'
=
R_0'+S_i'
=
R_0+\sum_{t\in \hat{T}_{k+1}}u\cdot n^{l+\epsilont_2} +S_i'
=
R_0+\sum_{t\in \hat{T}_{k+1}}u\cdot n^{l+\epsilont_2} +S_i'+S_i-S_i
= \\
R_i+\sum_{t\in \hat{T}_{k+1}}u\cdot n^{l+\epsilont_2} +S_i'-S_i
\leq 
R_i+\sum_{t\in \hat{T}_{k+1}}u\cdot n^{l+\epsilont_2} +|S_i'-S_i|
<
R_i+2\cdot \sum_{t\in \hat{T}_{k+1}}u\cdot n^{l+\epsilont_2} 
\end{gather*} Thus conditioned on \(E_{1} \) we also have \[R_{min(i,\tau)}' \leq max_{p\in Q} \bz_k(p) + max_{c; \by_k(c)>0} \by_k(c) \cdot d_k\cdot  (n^{l+\epsilon}+u) +2\cdot \sum_{t\in \hat{T}_{k+1}}u\cdot n^{l+\epsilont_2} \]
Since \(P_0V_0,P_1,V_1,\dots,P_iV_i\notin E_1 \) implies that \(R_i'=R_{i+1}'=\dots \), the above holds also without the conditioning on \(E_1\). And as from Lemma~\ref{lemma-Rprimes-supermartinagel} \(R_0',R_1',\dots \) is a supermartingale, we can apply the optional stopping theorem on \(R_0',R_1',\dots \) to obtain: 
%\[E(R_1)\geq E(R_T) \]
\begin{gather*}
max_{p\in Q} \bz_k(p) + max_{c; \by_k(c)>0} \by_k(c) \cdot d_k\cdot  (n+u) + \sum_{t\in \hat{T}_{k+1}}u\cdot n^{l+\epsilont_2}
 \geq\\
\E_{p\vec{n}}^\sigma(R_0')
\geq
 \E_{p\vec{n}}^\sigma(R_{\tau}')
\geq   p\cdot X_{n^{l+\epsilon}}+(1-p)\cdot X_0 
\end{gather*}
where \(X_{n^{l+\epsilon}} \) represents the minimal possible value of \(R_\tau' \) if any counter \(c\in C_{k+}\) with \(\by_k(c)>0 \) has  \(V_\tau(c)\geq n^{l+\epsilon} \), \(p \) is the probability of any such counter being at least \(n^{l+\epsilon} \) upon stopping, and \(X_0 \) represents the minimal value of \(R_\tau' \) if no such counter reached \(n^{l+\epsilon} \).

Assuming \begin{equation}\label{equation-epsilons-fsdfsgh}
	\epsilont_4>\epsilont_2
\end{equation}  We can simplify this for all sufficiently large \(n\) as
\[max_{p\in Q} \bz_k(p) + max_{c; \by_k(c)>0} \by_k(c) \cdot d_k\cdot  (n+u) + \sum_{t\in \hat{T}_{k+1}}u\cdot n^{l+\epsilont_2}  \geq   p\cdot X_{n^{l+\epsilon}}+(1-p)\cdot X_0 \]
\[max_{p\in Q} \bz_k(p) + max_{c; \by_k(c)>0} \by_k(c) \cdot d_k\cdot  (n+u) + \sum_{t\in \hat{T}_{k+1}}u\cdot n^{l+\epsilont_2}   \geq   p\cdot X_{n^{l+\epsilon}}-(1-p)\cdot max_{c\in\{1,\dots,d_k\}} \by_k(c)\cdot u\cdot d_k \]
\[max_{p\in Q} \bz_k(p) + max_{c; \by_k(c)>0} \by_k(c) \cdot d_k\cdot  (n+u) + \sum_{t\in \hat{T}_{k+1}}u\cdot n^{l+\epsilont_2} +(1-p)\cdot max_{c\in\{1,\dots,d_k\}} \by_k(c)\cdot u\cdot d_k  \geq   p\cdot X_{n^{l+\epsilon}} \]
\[n^{l+\epsilont_4} \geq   p\cdot X_{n^{l+\epsilon}} \]
\[n^{l+\epsilont_4} \geq   p\cdot n^{l+\epsilon} \]
\[n^{\epsilont_4-\epsilon} \geq   p \]
and assuming \begin{equation}\label{equation-epsilons-somewheresomiubibh}\epsilon>\epsilont_4\end{equation} it holds \(\lim_{n\rightarrow\infty}n^{\epsilon_4-\epsilon}=0 \).

Thus for any  counter \(c\in C_{k+} \) with \(\by_k(c)>0 \) it holds for any \(\epsilon>0 \), any strategy \(\sigma \), any state \(p \), and all sufficiently large \(n\) that  \(\prob_{p\vec{n}}^\sigma[\calC_\A[c]\geq n^{l+\epsilon}]\leq p + (1-\prob_{p\vec{n}}^\sigma[E_{1}]) \) and thus \(\lim_{n\rightarrow\infty}\prob_{p\vec{n}}^\sigma[\calC_\A[c]\geq n^{l+\epsilon}]\leq  \lim_{n\rightarrow\infty}p + (1-\prob_{p\vec{n}}^\sigma[E_{1}])= 0 \). Thus \(\calC_\A[c] \) has an upper asymptotic estimate of \(n^l \). This finishes the proof of Lemma~\ref{lemma-linear-upper-bound} for counters. We proceed to extend the proof onto transitions as well.

Consider the \(1\)-dim VASS MDP \(\A_R' \) created from \(\A \) by replacing all counters with a single counter which ``almost'' corresponds to \(R_i' \). That is \(\A_R' \) has only one counter, the same set of states as \(\A \), and each transition \(t=(p,\bu,q) \) of \(\A \) is in \(\A_R' \) replaced with \(t'=(p,\bu',q) \) where \[\bu'(1)=\begin{cases}
\RankEff(t) & \text{if } t\in \hat{T}_{k+1} \\
0 & \text{else }
\end{cases}\]

Note that conditioned on \(E_{1} \) the counter of \(\A_R' \), when initialized in \(rank_k(p\vec{n})+\sum_{t\in \hat{T}_{k+1}}u\cdot n^{l+\epsilont_2}\), is after \(i\)-steps equal to the value of \(R_i' \) (assuming the equivalent computation on \(\A \)). And since \(R_i< 0 \) implies the computation on \(\A \) has already terminated (since the only way to add a negative value to \(rank_k\) is if some counter is negative) we get from Lemma~\ref{lemma-R-lessthan-Rprime-cond-Reps2} that if \(\A_R' \) terminated then the equivalent computation on \(\A \) either also terminated or it is not in \(E_{1} \). 

Let \(t=(p,\bu,q)\) be a transition of \(\A_{k,\hat{T}_{k+1}}\) such that one of the following holds:
\begin{itemize}
\item \(p\in Q_n \) and \(\bz_k(q)-\bz_k(p)+\sum_{c\in \countersset_k} \bu(c)\cdot \by_k(c)<0\),
\item \(p\in Q_p \) and \(\sum_{t' = (p',\bu',q') \in \tout(p)}P(t')\cdot \big(\bz_k(q')-\bz_k(p')+\sum_{c\in \countersset_k}  \bu'(c)\cdot \by_k(c)\big)< 0\).
\end{itemize} 
Notice that from the way \hyperref[fig-systems]{(II)} is designed, it holds that every single MEC of \(\A'_{R,\sigma} \) with \(\sigma\in\stratsMD{\A} \) that contains \(t\) is decreasing (this is since such \(t\) either decreases \(rank_k\) if \(p\in Q_n \) or if \(p\in Q_p \) then \(rank_k\) decreases on average in a single computation step from \(p\), and this ``average decrease'' can never be compensated in \(\A_{R}'\)). Also, there is no MEC of \(\A'_{R,\sigma} \) with \(\sigma\in\stratsMD{\A} \) that is increasing (again, this is due to there being no transitions that ``increase \(rank_k \) on average'' in \(\A_{k,\hat{T}_{k+1}} \) and transitions not in \(\hat{T}_{k+1} \) have effect \(0\) in \(\A_{R}' \), we refer to Section~\ref{sec-systems} for more details). Thus from the results about \(1\)-dim VASS MDPs from \cite{AKCONCUR23} it holds that \(\calT_{\A_R'}[t] \) has an upper asymptotic estimate of \(n\).  Therefore \[\prob_{p\vec{n}}^\sigma[\calT_\A[t]\geq n^{l+\epsilon}]\leq (1-\prob_{p\vec{n}}^\sigma[E_{1}])+\prob_{p\vec{1}\cdot (rank_k(p\vec{n})+\sum_{t\in \hat{T}_{k+1}}u\cdot n^{l+\epsilont_2})}^\sigma[\calT_{\A_R'}[t]\geq n^{l+\epsilon}] \]

But  assuming \begin{equation}\label{equation-epsilon-khviviguo}
	\epsilont_5>\epsilont_2
\end{equation} it holds for all sufficiently large \(n\) that \(rank_k(p\vec{n})+\sum_{t\in \hat{T}_{k+1}}u\cdot n^{l+\epsilont_2}\leq n^{l+\epsilont_5}\). Thus assuming  \begin{equation}\label{equation-epsilon-khvivifsdfsguo}
l+\epsilon>(l+\epsilont_5)(1+\epsilont_6) \end{equation} it holds

%and for \(x=rank_k(p\vec{n})+\sum_{t\in \hat{T}_{k+1}}un^{l+\epsilon_1}\leq n^{l+\epsilon_5} \) that	

\begin{gather*}
\limsup_{n\rightarrow\infty} \prob_{p\vec{n}}^\sigma[\calT_\A[t]\geq n^{l+\epsilon}] 
\leq\\
\limsup_{n\rightarrow\infty}(1-\prob_{p\vec{n}}^\sigma[E_{1}])+\prob_{p\vec{1}\cdot (rank(p\vec{n})+\sum_{t\in \hat{T}_{k+1}}u\cdot n^{l+\epsilont_2})}^\sigma[\calT_{\A_R'}[t]\geq	 n^{l+\epsilon}] 
\leq   \\
\limsup_{n\rightarrow\infty}(1-\prob_{p\vec{n}}^\sigma[E_{1}])+\prob_{p\vec{1}\cdot n^{l+\epsilont_5}}^\sigma[\calT_{\A_R'}[t]\geq (n^{l+\epsilont_5})^{1+\epsilont_6}] 
= 	0 \end{gather*}

 Thus \(\calT_\A[t] \) has an upper asymptotic estimate of \(n^l \) proving Lemma~\ref{lemma-linear-upper-bound} for transitions. 

It remains to show there exist values for \(\epsilon_1,\epsilon_2,\dots \) that satisfy all of our assumptions. We do this in Table~\ref{Table-app-lemma-linear-upper-bound}.

\begin{table*}[h]
	\caption{Values of \(\epsilon_1,\epsilon_2,\dots \) for Section~\ref{app-lemma-linear-upper-bound}}
	\centering
	%	\begin{center}
		\begin{tabular}{|l|| c c| c|} 
			\hline
			\(\epsilon\) assignment &  \multicolumn{2}{|c|}{restrictions}  & After substitution \\ 
			\hline\hline
			%		\multirow{10}{*}{ \begin{tabular}{c}
					%			 \(\epsilon_1=\frac{9}{10} \)
					%				\\ \(\epsilon_3=\frac{1}{8} \)
					%				\\ \(\epsilon_{4}=\frac{1}{10} \)
					%				\\ \(\epsilon_{5}=\frac{5}{12} \)
					%				\\ \(\epsilon_{6}=\frac{1}{2} \)
					%				\\ \(\epsilon_7=\frac{1}{7} \)
					%		\end{tabular} }
			\(\epsilon_1=\nicefrac{\epsilon}{10000} \) & \(0<\epsilon_1,\epsilon_2,\dots  \) & &  \\ 
			\hline
			\(\epsilon_2=\nicefrac{\epsilon}{100} \) & \(\epsilont_2>2\cdot \epsilont_3> \epsilont_3>2\cdot \epsilont_1 \) & \eqref{equation-epsilon-bounds-something1} & \(\nicefrac{\epsilon}{100}>2\cdot \nicefrac{\epsilon}{1000}> \nicefrac{\epsilon}{1000}>2\cdot \nicefrac{\epsilon}{10000}  \) \\
			\hline
			\(\epsilon_{3}=\nicefrac{\epsilon}{1000} \) & 	\(\epsilont_4>\epsilont_2\) & \eqref{equation-epsilons-fsdfsgh} & \(\nicefrac{\epsilon}{10}>\nicefrac{\epsilon}{100}\) \\
			\hline
			\(\epsilon_{4}=\nicefrac{\epsilon}{10} \) & \(\epsilon>\epsilont_4\) & \eqref{equation-epsilons-somewheresomiubibh} & \(\epsilon>\nicefrac{\epsilon}{10}\) \\
			\hline
			\(\epsilon_{5}=\nicefrac{\epsilon}{90} \) & \(\epsilont_5>\epsilont_2 \) & \eqref{equation-epsilon-khviviguo} & \(\nicefrac{\epsilon}{90}>\nicefrac{\epsilon}{100} \)  \\ 
			\hline
			\(\epsilon_{6}=\nicefrac{\min(\epsilon,1)}{(2\cdot l)} \) & \(l+\epsilon>(l+\epsilont_5)(1+\epsilont_6) \) & \eqref{equation-epsilon-khvivifsdfsguo} & \(l+\epsilon>(l+\nicefrac{\epsilon}{90})(1+\nicefrac{\min(\epsilon,1)}{(2\cdot l)}) \) \\
%			\hline
%			& \(	\epsilont_{2}<\epsilon \) & \eqref{restraion-eq-exp-6} & \(\min(\nicefrac{9\cdot \epsilon}{10},\nicefrac{9}{10})<\epsilon \) \\
%			\cline{2-4}
			%				& \(2\epsilon_{5}-2\epsilon_{7}-\epsilon_{6}>0 \) & 18744 & 7560 \\
			%				\cline{2-4}
			%				& \(\epsilon_{1}>\epsilon \) & 18744 & 7560 \\
			%				\cline{2-4}
			%				& 		 \( 0<\epsilon<\frac{1}{10}\) & 18744 & 7560 \\
			\hline
		\end{tabular}
		\par
		\label{Table-app-lemma-linear-upper-bound}	
		%	\end{center}
\end{table*}

\section{Lemma~\ref{lemma-main-for-all-k}}

%We define \(T_i=\{t \in \transitions \mid T[t] \textit{ has a lower asymptotic estimate of } n^i \} \) and \(C_i=\{c \mid C[c] \textit{ has a tight asymptotic estimate of } n^i \} \), and \(C_{i+}=\{c \mid C[c] \textit{ has a lower asymptotic estimate of } n^i \} \)

Let us begin by stating the following technical definition:

\begin{definition}
	We say that a computation \(\alpha\) is \(k\)-level \(m\)-cyclic if there exists a tree graph \(G=(V,E) \) along with a function \(f \) that assings to each vertex of \(G\) a path on \(\A\), and a function \(g\) that assigns to each vertex of \(G\) a MEC of some \(\A_i \) such that:
	\begin{itemize}
		\item the root \(v_0 \) of \(G\) has \(f(v_0)=\alpha \) and \(g(v_0)=\A \);
		\item for each vertex \(v\) of \(G \), let \(v_1,\dots,v_i \) be all the children of \(v\), then it holds that \(f(v)=f(v_1)\cdot \ldots \cdot f(v_i) \);
		\item every leaf \(v\) of \(G\) has distance exactly \(k\) from the root and \(\length(f(v))>1 \).  (root has distance \(0\) to itself)
		\item every non-leaf vertex \(v \) whose distance from root is \(i \) has exactly \(m\cdot l_i  \) children \(v_1,\dots,v_{m\cdot l_i } \), where \(l_i \) is the number of MECs in \(\A_{i}^{g(v)} \), where \(\A_{i}^{g(v)} \) is the VASS MDP obtained by restricting \(\A_{i} \) only to the transitions included in  \(g(v) \). For each each MEC \(B \) of \(\A_{i}^{g(v)} \) there exist exactly \(m \) children \(v_j\) of \(v \) such that \(f(v_j) \) contains every single state of \(B \) at least once (but not counting the very first state of \(f(v_j)\)), and such that \(g(v_j)=B \). 
	\end{itemize}
	
	We call a tuple \((V,E,f,g) \) that satisfies the above conditions a \((k,m)\)-tree of \(\alpha \).
\end{definition}

Note that if there exists a \((k,m)\)-tree of \(\alpha \) then  we can decompose \(\alpha \) into "levels" such that at the \(i\)-th level we are guaranteed to cycle between all the states of any given MEC of \(A_i \) at least \(m^i \) times, and we do so in such a way that we ``interweave'' these cycles at different levels so that for every \(i<j \) we can only count at most \(m^{j-i}\) cycles on a MEC \(B \) of \(\A_{j} \) within a single cycle on a MEC \(B' \) of \(\A_i \) where \(B \) is contained in \(B'\).\footnote{Cycle as in visit them in order, if we were to say visit \(p,q,q,q,q,p,p,q,p \) then this are only two cycles as we starting from \(p\) reach \(q\) before returning to \(p\) only twice.}

Most importantly, if \(\alpha \) has a \((k,n^{1-\epsilon})\)-tree, then \(\length(\alpha)\geq n^{k-k\cdot \epsilon} \), and as we will show later, we will be able to use the \((k,n^{1-\epsilon})\)-tree to define a new strategy that with high enough probability produces a computation with a \((k+1,n^{1-\epsilon})\)-tree.

Let \(\alpha_{p\vec{n}}^{\sigma} \) be the random variable that represents the computation under the strategy \(\sigma \) in \(\A \) from initial configuration \(p\vec{n} \), let \(\alpha_{p\vec{n},..i}^{\sigma} \) be the prefix of \(\alpha_{p\vec{n}}^{\sigma}\) of length \(i\), let  \(\alpha_{p\vec{n},i..j}^{\sigma} \) be the suffix of \(\alpha_{p\vec{n},..j}^{\sigma} \)  of length \(j-i\) (and if \(j-i<0 \) then we put \(\alpha_{p\vec{n},..j}^{\sigma}=\epsilon \)), and let \(\alpha_{p\vec{n},i..}^{\sigma} \) be the suffix of \(\alpha_{p\vec{n}}^{\sigma}\) obtained by removing the first \(i \) configurations from \(\alpha_{p\vec{n}}^{\sigma}\).

We will prove the following technical Lemma.

\begin{lemma}\label{lemma-main-for-all-k}
	Let \(\A \) be a strongly connected VASS MDP, then  for each \(k\in \mathbb{N}_0 \) all of the following hold:
	\begin{enumerate}
		\item \label{enum-main-1} Each transition \(t\) of \(\A\) has either an upper asymptotic estimate of \( n^{k}\) or a lower asymptotic estimate of \(n^{k+1} \) for \(\calT_\A[t] \);
		\item \label{enum-main-2} Each counter \(c\) of \(\A\)  has either an upper asymptotic estimate of \( n^{k}\) or a lower asymptotic estimate of \(n^{k+1} \) for \(\calC_\A[c] \);
		\item \label{enum-main-3} There exists a multi-component \(\bx \) on \( \A_k \) satisfying  \(\Delta^{C_1,\dots,C_k}(\bx)=\vec{0} \) as well as \(\bx(t)>0 \) iff \(t\in T_{k+1} \);
		\item \label{enum-main-4} For each component \(\by\) of \(\A\), either there exists a transition \(t\) with \(\by(t)>0 \) and with an upper asymptotic estimate of \(n^{k} \) for \(\calT_\A[t] \), or \(\calP_{\A_{+\hat\by}}[\M_{\hat{\by}}] \) has a lower asymptotic estimate of \(n^{k+1}\);

		\item \label{enum-main-5}  Let \(1>\epsilon>0 \) and let \(h_{\epsilon}:\mathbb{N}\rightarrow\mathbb{N}_0 \) be such that \(h_{\epsilon}(n)\leq n^{1-\epsilon} \) for all \(n\in \mathbb{N} \) and \(\lim_{n\rightarrow\infty}h_{\epsilon}(n)=\infty \). Then  there exist
		\begin{itemize}
			\item a strategy \(\sigma_k^\epsilon \) on \(\A\),
			\item a function \(g^{k,\epsilon}\)  which assigns to each  \(\alpha_{p\vec{n},..i}^{\sigma_k^\epsilon}  \) either a MEC of \(\A_{k+1} \) or the symbol \(\bot\),
%			\item functions \(f_1^{k,\epsilon},\dots,f_{k+1}^{k,\epsilon}\) that to each  each \(\alpha_{p\vec{n},..i}^{\sigma_k^\epsilon}  \) assign a computation,
			\item functions \(
			r_1^{k,\epsilon},\dots,r_{k+1}^{k,\epsilon}\) that to each  each \(\alpha_{p\vec{n},..i}^{\sigma_k^\epsilon}  \) assign either \(NEXT\) or \(SAME\), where \(0\leq i\leq \length(\alpha_{p\vec{n}}^{\sigma_k^\epsilon} )\)
		\end{itemize}
		 Which satisfy the following conditions.		
		\begin{itemize}
			\item \(r_{1}^{k,\epsilon}(\alpha_{p\vec{n},..0}^{\sigma_k^\epsilon} )=NEXT \);
			\item for all \(1\leq i\leq k+1\) and all \(0\leq j\leq \length(\alpha_{p\vec{n}}^{\sigma_k^\epsilon} )\): if \(r_i^{k,\epsilon}(\alpha_{p\vec{n},..j}^{\sigma_k^\epsilon} )=NEXT \) then \(r_{l}^{k,\epsilon}(\alpha_{p\vec{n},..j}^{\sigma_k^\epsilon} )=NEXT \) for all \(i\leq l\leq k+1 \);
			\item for all \(0\leq i< \length(\alpha_{p\vec{n}}^{\sigma_k^\epsilon} )\): if \(r_{k+1}^{k,\epsilon}(\alpha_{p\vec{n},..i}^{\sigma_k^\epsilon} )=SAME \) and  \(g^{k,\epsilon}(\alpha_{p\vec{n},..i}^{\sigma_k^\epsilon})\neq \bot \) then \(g^{k,\epsilon}(\alpha_{p\vec{n},..i+1}^{\sigma_k^\epsilon})= g^{k,\epsilon}(\alpha_{p\vec{n},..i}^{\sigma_k^\epsilon}) \)			
%			\item for all \(i\): if \(r_{k+1}(\alpha_{p\vec{n},..i}^{\sigma_k^\epsilon} )="next" \) then \(f(\alpha_{p\vec{n},..j}^{\sigma_k^\epsilon} )="next" \) for all \(i<l\leq k+1 \);
		\end{itemize}

%		For each \(0\leq l\leq k+1 \) let \(v_1,\dots,v_a \) be all the vertices of \(G_n^\epsilon \) at the distance \(l \) from the root, then it holds for all \(i\) that \(\alpha_{p\vec{n},..i}^{\sigma_k^\epsilon}=f_{k}^\epsilon(\alpha_{p\vec{n},..i}^{\sigma_k^\epsilon},v_1)\cdot \ldots \cdot f_{k}^\epsilon(\alpha_{p\vec{n},..i}^{\sigma_k^\epsilon},v_a) \), that \(f_{k}^\epsilon(\alpha_{p\vec{n},..i}^{\sigma_k^\epsilon},v_j)=f_{k}^\epsilon(\alpha_{p\vec{n},..i-1}^{\sigma_k^\epsilon},v_j)\) for each \(v_j\) such that \(v_{k}^\epsilon(\alpha_{p\vec{n},..i-1}^{\sigma_k^\epsilon}) \) is not a descendant of \(v_j \), and that	 \(f_{k}^\epsilon(\alpha_{p\vec{n},..i-1}^{\sigma_k^\epsilon},v_j) \) is a prefix of \(f_{k}^\epsilon(\alpha_{p\vec{n},..i}^{\sigma_k^\epsilon},v_j) \) for the \(v_j\) such that \(v_{k}^\epsilon(\alpha_{p\vec{n},..i-1}^{\sigma_k^\epsilon}) \) is a descendant of \(v_j \). (i.e., the \((i-1)\)-st transition is put in \(f_{k}^\epsilon(\alpha_{p\vec{n},..i}^{\sigma_k^\epsilon},v_j) \) where \(v_{k}^\epsilon(\alpha_{p\vec{n},..i-1}^{\sigma_k^\epsilon}) \) is a descendant of \(v_j \) while for all the other \(v_j \) the strings do not change)
		
		Furthermore with probability \(p_{p\vec{n}}^{\sigma_k^\epsilon}\) such that \(\lim_{n\rightarrow\infty}p_{p\vec{n}}^{\sigma_{k}^\epsilon}=1 \)  there exists a  \((k+1,h_\epsilon(n))\)-tree \(G=(V,E,f,g) \) of \(\alpha_{p\vec{n}}^{\sigma_k^\epsilon} \)  such that
			\begin{itemize}
			\item for each \(1\leq l\leq k+1 \): let \(v_1,\dots,v_{a} \) be all the vertices of \(G\) at distance \(l \) from root, and let \(i_0<i_1<i_2<\dots<i_b \) be all the indexes \(i \) such that \(r_l(\alpha_{p\vec{n},..i}^{\sigma_{k}^\epsilon})=NEXT \). Then
			\begin{itemize}
%				\item \(f(v_1)=\alpha_{p\vec{n},i_0..i_1}^{\sigma_{k}^\epsilon} \),
\item \(a=b\),
				\item \(f(v_j)=\alpha_{p\vec{n},i_{j-1}\dots i_j}^{\sigma_{k}^\epsilon} \) for all \(0<j<a \),
				\item \(f(v_a)=\alpha_{p\vec{n},i_b..}^{\sigma_{k}^\epsilon} \).
			\end{itemize} 
		\item Let \(v_1,\dots,v_a \) be all the leaves of \(G\), and let \(i_0<i_1<i_2<\dots<i_b \) be all the indexes \(i \) such that \(r_{k+1}(\alpha_{p\vec{n},..i}^{\sigma_{k}^\epsilon})=NEXT \). Then for each \(0\leq x \leq  b  \) there exists \(j_x\) such that \( i_{x} \leq j_x<i_{x+1}\) (here we put  \(i_{b+1}=n^{k+2} \)) such that
		\begin{itemize}
			\item \(g(v_x)=g^{k,\epsilon}(\alpha_{p\vec{n},..j_x}^{\sigma_k^\epsilon})\neq \bot \),
			\item \(\alpha_{p\vec{n},j_x..i_{x+1}}^{\sigma_k^\epsilon}\) contains every state of \(g(v_x) \) at least once,
			\item for all \(j_x\leq y< i_{x+1} \) it holds \(g^{k,\epsilon}(\alpha_{p\vec{n},..y}^{\sigma_k^\epsilon})=g^{k,\epsilon}(\alpha_{p\vec{n},..j_x}^{\sigma_k^\epsilon}) \).
		\end{itemize}  
		\end{itemize}

%		 for each vertex \(v\) of \(G_n^\epsilon \) it holds \(f(v)= f_{k}^\epsilon(\alpha_{p\vec{n}}^{\sigma_k^\epsilon},v)\) and for each leaf \(v\) of \(G_n^\epsilon\) such that \(v=v_k^\epsilon(\alpha_{p\vec{n},..i}^{\sigma_k^\epsilon}) \) it holds either \(g_k^\epsilon(\alpha_{p\vec{n},..i}^{\sigma_k^\epsilon})=\bot \) or \(g(v)= g_k^\epsilon(\alpha_{p\vec{n},..i}^{\sigma_k^\epsilon}) \), and for each leaf \(v\) of \(G_n^\epsilon\) there exist \(i<j \) such that:
%	 \(f_{k}^\epsilon(\alpha_{p\vec{n},i+1\dots j}^{\sigma_k^\epsilon},v) \) contains every state of \(g(v) \) at least once and for each \(i\leq a \leq j \) it holds \(g_k^\epsilon(\alpha_{p\vec{n},..a}^{\sigma_k^\epsilon})\neq \bot \).

	\end{enumerate}
	
\end{lemma}

The point \ref{enum-main-5} deserves an explanation. Essentially it says that for each \(\epsilon>0 \) and each function \(h_\epsilon \) with \(h_\epsilon(n)\leq n^{1-\epsilon} \) there exists a strategy \(\sigma_k^\epsilon \) on \(\A \) that with high enough probability produces a computation that is \((k+1)\)-level \(h_\epsilon(n)\)-cyclic, and furthermore, a \( (k+1,h_\epsilon(n))\)-tree \(G \) of this computation can be created ``along the computation'', that is at every step of the computation we know exactly to which leaf of \(G \) the next step will belong along with the MEC that corresponds to this leaf (except at some steps we are allowed to not know the MEC, but the states visited by these steps do not count toward the given leaf visiting each state of its MEC, i.e., the leaf is required to visit each state of the given MEC after we know the MEC). And we can do this in a such way that at the end of the computation the \(G\)  created this way is with high enough probability a valid \( (k+1,h_\epsilon(n))\)-tree of the resulting computation. 

\begin{proof}

	We prove Lemma~\ref{lemma-main-for-all-k} by induction over \(k\).
	
	\textbf{The base case \(k=0\):} 
	\begin{enumerate}
	
		\item Since \(\A \) is strongly connected, the strategy that at each step chooses the next transition uniformly at random induces a strongly connected VASS Markov chain \(\M \) that has single MEC that contains all transitions of \(\A\). Thus from Theorem~\ref{thm-VASS-Markov-chain-analysis}  \(\calT_\M[t] \) has a lower asymptotic estimate of \(n\) for every \(t\). Thus also \(\calT_\A[t] \) has a lower asymptotic estimate of \(n\) for every transition \(t\).
			\item  \(\calC_\A[c] \) has a trivial lower asymptotic estimate of \(n\) for each counter \(c\).
		
		\item We define \(\A_0=\A \). Let \(\bx \) be an arbitrary multi-component on \(\A \) that contains every transition of \(\A \) (for example we can take one that includes every single component of \(\A \) which exists from Lemma~\ref{lemma-decompose-multicomponents-into-components}). Then this \(\bx\) satisfies both \(\Delta^{C_0}(\bx)=\vec{0} \) and \(\bx(t)>0 \) iff \(t\in T_{1} \). This is due to \(C_0=\emptyset \) and \(T_1=T \).
		\item For any component \(\by \) we can see \(\M_{\hat{\by}} \) as a strongly connected VASS MDP, and thus we have a trivial lower asymptotic estimate of \(n\) for \(\calL_{\M_{\hat\by}} \). As the pointing strategy  on \(\A_{+\hat{\by}} \) can simply ignore any VASS Markov chain other than \(\M_{\hat{\by}} \) (i.e., it always outputs \(\hat{\by} \)) this gives us a lower asymptotic estimate of \(n\) for \(\calP_{\A_{+\hat{\by}}}[\M_{\hat\by}] \) for every component \(\by\) of \(\A\).	
		\item Let us have \(1>\epsilon>0 \) and \(h_\epsilon:\mathbb{N}\rightarrow\mathbb{N}_0 \) with \(h_\epsilon(n)\leq n^{1-\epsilon} \) and \(\lim_{n\rightarrow\infty}h_{\epsilon}(n)=\infty \). Consider the strategy \(\sigma^0_{\epsilon} \) on \(\A \) which at each step chooses the next transition uniformly at random.  
		
		Let \(p_1,\dots,p_w \) (and let \(p_0=p_w \)) be all the states of \(\A \), and assume that \(\alpha_{p\vec{n},..n^2}^{\sigma_0^\epsilon} \) can be divided into \(w\cdot h_\epsilon(n) \) sub-computations \(\alpha_{p\vec{n},..n^2}^{\sigma_0^\epsilon}=\alpha_{1}^1\cdot \ldots\cdot \alpha_1^w\cdot \alpha_{2}^1\cdot \ldots\cdot \alpha_2^w\cdot  \ldots\cdot \alpha_{h_\epsilon(n)}^1\cdot \ldots\cdot \alpha_{h_\epsilon(n)}^w  \) such that for all \(i,j\) it holds that \(\alpha_{i}^j \) starts in \(p_{j-1} \), ends in \(p_{j} \), and does not otherwise contain \(p_{j-1} \), with the only exceptions being that \(\alpha_{1}^1 \) is not required to start in \(p_0 \) and  \(\alpha_{h_\epsilon(n)}^w\) is not required to end in \(p_w \) while being permitted to contain multiple instances of \(p_w \) (it is still required to contain at least one \(p_w \)).
		If such division exists then  \(\alpha_{p\vec{n}}^{\sigma_0^\epsilon} \) is a \(1\)-level \(h_\epsilon(n)\)-cyclic computation, as we could simply take \(G=(V,E,f,g) \) with \(V=\{v_{root},v_1,\dots,v_{h_\epsilon(n)}\}\), \(E=\{(v_{root},v_1),\dots,(v_{root},v_{h_\epsilon(n)}) \} \) such that \(g(v)=\A \) for all \(v\in V \) and \(f(v_{root})=\alpha_{p\vec{n}}^{\sigma_1^\epsilon} \), \(f(v_i)=\alpha_{i}^1\cdot \ldots \cdot \alpha_{i}^w  \) for \(i<h_\epsilon(n)\), and \(f(v_{h_\epsilon(n)})=\alpha_{h_\epsilon(n)}^1\cdot \ldots \cdot \alpha_{h_\epsilon(n)}^w \cdot \alpha_{p\vec{n},n^3..}^{\sigma_0^\epsilon} \). 
		Furthermore, the division into the sub-paths is unique and given \(\alpha_{p\vec{n},..i}^{\sigma_0^\epsilon}\) we can compute the division of this prefix \(\alpha_{p\vec{n},..i}^{\sigma_0^\epsilon}=\alpha_{1,..i}^1\cdot \ldots\cdot\alpha_{1,..i}^w\cdot\alpha_{2,..i}^1\cdot\ldots\cdot\alpha_{2,..i}^w\cdot\ldots\cdot\alpha_{h_\epsilon(n),..i}^1\cdot\ldots\cdot\alpha_{h_\epsilon(n),..i}^w  \)  such that \(\alpha_{j,..n^2}^l=\alpha_{j}^l\) in the following way: 
		
		Let \(q_i\bu_i \) be the last configuration of \(\alpha_{p\vec{n},..i}^{\sigma_0^\epsilon}\), let \(t_i\) be the last transition of \(\alpha_{p\vec{n},..i+1}^{\sigma_0^\epsilon} \), and let \(j_i,l_i \) be indexes such that \(\alpha_{j_i,..i}^{l_i}\neq \epsilon \) is the last non-empty computation of the division of \(\alpha_{p\vec{n},..i}^{\sigma_0^\epsilon}\) (i.e., the smallest  \(j_i\leq h_\epsilon(n) \) such that there exists a smallest value of \(l_i\leq w \) for which it holds \[\alpha_{j_i+\lfloor\frac{l_i}{w}\rfloor,..i}^{(l_i\mod w)+1}=\epsilon \]). Then we can define:
		\[\alpha_{j,..i+1}^l=\begin{cases}
			\alpha_{j,..i}^l & \text{if } j<j_i \\
			\alpha_{j,..i}^l & \text{if } j=j_i \textit{ and }l<l_i \\
			\alpha_{j,..i}^l, q_{i+1}\bu_{i+1} & \text{if } j=j_i \textit{ and }l=l_i\textit{ and } q_i\neq p_{l_i} \\
			\alpha_{j,..i} & \text{if } j=j_i \textit{ and }l=l_i\textit{ and } q_i= p_{l_i} \\
			q_{i}\bu_i,q_{i+1}\bu_{i+1	} & \text{if } j=j_i+\lfloor\frac{l_i}{w}\rfloor \textit{ and }l=(l_i \mod w)+1\textit{ and } q_i= p_{l_i} \\
			\epsilon & \text{else }
		\end{cases}\]
		
		Thus we can define
		\[r_1^{0,\epsilon}(\alpha_{p\vec{n},..i}^{\sigma_0^\epsilon})=\begin{cases}
			NEXT & \text{if }   i=0\\
			NEXT & \text{if }   q_i= p_{l_i} \textit{ and }(l_i,j_i)\neq (w,h_\epsilon(n))\\
			SAME & \text{else }
		\end{cases} \]

%		 \[f_0^\epsilon(\alpha_{p\vec{n},..i}^{\sigma_0^\epsilon},v_{root})=\alpha_{p\vec{n},..i}^{\sigma_0^\epsilon} \] \[f_0^\epsilon(\alpha_{p\vec{n},..i}^{\sigma_0^\epsilon},v_j)=\alpha_{j,..i}^{1}\cdot \ldots \cdot \alpha_{j,..i}^{w} \]  \[v_0^\epsilon(\alpha_{p\vec{n},..i}^{\sigma_0^\epsilon})=v_{j_i} \]
		\[g^{0,\epsilon}(\alpha_{p\vec{n},..i}^{\sigma_0^\epsilon})=\A \]
		
		Let \(E \) denote the event that \(\alpha_{p\vec{n},..n^2}^{\sigma_0^\epsilon}\) can be divided into sub-computations as described above, and let \(\pi_{p\vec{n}}^{\sigma_0^\epsilon}  \) be the random variable denoting the length of the shortest prefix of \(\alpha_{p\vec{n},..n^2}^{\sigma_0^\epsilon}\) that can be divided in such a way (or \(\pi_{p\vec{n}}^{\sigma_0^\epsilon}=n^2\) if no such prefix exists). It remains to show that \(\lim_{n\rightarrow\infty}\prob_{p\vec{n}}^{\sigma_0^\epsilon}[E]=1 \). To see this, notice that the expected length of \(\alpha_i^j \) is a constant (as the sub-path ends the moment \(p_j \) is reached, which happens in expected constant time), therefore \(\E_{p\vec{n}}^{\sigma_0^\epsilon}[\pi_{p\vec{n}}^{\sigma_0^\epsilon}]\leq a\cdot w\cdot h_\epsilon(n) \) for some constant \(a\). We can thus apply Markov inequality to obtain \(\prob_{p\vec{n}}^{\sigma_0^\epsilon}[\pi_{p\vec{n}}^{\sigma_0^\epsilon}\geq n^{1-\epsilon_1} ]\leq \frac{a\cdot w\cdot h_\epsilon(n)}{n^{1-\epsilon_1} }\leq \frac{a\cdot w\cdot n^{1-\epsilon}}{n^{1-\epsilon_1} }  \) and for \(0<\epsilon_1<\epsilon \) it holds \(\lim_{n\rightarrow\infty} \frac{a\cdot w\cdot n^{1-\epsilon}}{n^{1-\epsilon_1} } =0\). But since every transition of \(\A \) can decrease the counters by at most a constant \(u\) it holds \(\length(\alpha_{p\vec{n}}^{\sigma_0^\epsilon})\geq \lfloor\frac{n}{u}\rfloor\geq n^{1-\epsilon_1} \)  (for all sufficiently large \(n\)). 
		Thus it holds \(\lim_{n\rightarrow\infty}\prob_{p\vec{n}}^{\sigma_0^\epsilon}[E]
		=1 \). 
		
		%	Thus from the definition of \( E_{p\vec{n},..i}^{\sigma_0^\epsilon}\) it holds that \((E_{p\vec{n},..i}^{\sigma_0^\epsilon}\mid \neg E_{p\vec{n}}^{\sigma_0^\epsilon})\geq n^{1-\epsilon_1}\).
		
		Therefore the base case for case \ref{enum-main-5} holds.\qed
	\end{enumerate}

		\textbf{	The induction step:}  Assume Lemma~\ref{lemma-main-for-all-k} holds for all \(k'<k\), we will now prove it holds for for \(k\) as well.
	
	\begin{enumerate}
		\item  From Lemma~\ref{lemma-hatB-zero-unbounded-rankl-give-upper-estimate-nk} we obtain a set \(T_{k}' \) of transition with upper asymptotic estimate of \(n^k\) for \(\calT_A[t] \) for each \(t\in T_{k}'\). Let \(\bx_k,\by_k,\bz_k \) be maximal solutions of systems \hyperref[fig-systems]{(I)} and \hyperref[fig-systems]{(II)} for \(\A_{k,\hat{T}_{k+1}} \). For every transition \(t=(p,\bu,q)\in \hat{T}_{k+1}=T_{k}\setminus T_{k}' \) 		 we get from Lemma~\ref{lemma-something-something-k+1-estiamtes} that $\bx_k(t)>0$ implies a lower asymptotic estimate of \(n^{k+1}\) for \(\calT_\A[t] \), while Lemma~\ref{lemma-upper-bound-T} together with Lemma~\ref{lemma:dichotomy} gives that \(\bx_k(t)=0 \) implies an upper asymptotic estimate of \(n^k \) for \(\calT_\A[t] \).

	%states \(p\in Q\) and all transitionseither  transition \(t\) has either upper asymptotic estimate of \( n^{k}\) or lower asymptotic estiamte of \(n^{k+1} \) for \(\T[t] \) in \(\A \); TODO
		\item  Let \(\bx_k,\by_k,\bz_k \) be maximal solutions of systems \hyperref[fig-systems]{(I)} and \hyperref[fig-systems]{(II)} for \(\A_{k,\hat{T}_{k+1}} \). Lemma~\ref{lemma-upper-bound-T} gives an upper asymptotic estimate of \(n^k \) for \(\calC_\A[c] \) for each counter \(c\) with \(\by_k(c)>0 \) which from Lemma~\ref{lemma:dichotomy} implies \(\Delta(\bx_k)(c)=0 \), while Lemma~\ref{lemma-something-something-k+1-estiamtes} gives   a lower asymptotic estimate of \(n^{k+1} \) for \(\calC_\A[c] \) for each counter \(c\) with \(\Delta(\bx_k)(c)>0 \).

		\item Let \(\bx_k \) be a maximal solution of system \hyperref[fig-systems]{(I)} for \(\A_{k,\hat{T}_{k+1}} \).  From point~\ref{enum-main-1} for \(k\) we have that \(t\in T_{k+1} \) iff \(\bx_k(t)>0 \) and from point~\ref{enum-main-2} for \(k\) we have \(c\in C_{1}\cup\dots\cup C_{k} \) iff \(\Delta(\bx_k)(c)=0 \). Thus \(\bx_k \) has the desired property.

		\item 
		Let \(\br \) be a component of \( \A\). From Lemma~\ref{lemma-D-in-Rl-for-all-l-imply-D-in-R} we have that if for each \(1\leq l<\frac{k}{2} \) it holds \(\support^{C_1,\dots,C_l}(\hat{\br})\subseteq \Delta^{C_1,\dots,C_l}(X_{k-l,T_{k+1-l}}) \) then \(\support^{C_1,\dots,C_{\lfloor \frac{k}{2}\rfloor}}(\hat{\br})\subseteq R^{B_\bx,k}  \), which from Lemma~\ref{lemma-D-in-R-lower-estimate-nk+1-for-hatB} implies a lower asymptotic estimate of \(n^{k+1} \) for \(\calP_{\A_{+\hat{\br}}}[\M_{\hat{\br}}] \).
		
		If on the other hand there exists \(1\leq l<\frac{k}{2} \) such that \(\support^{C_1,\dots,C_l}(\hat{\br})\nsubseteq \Delta^{C_1,\dots,C_l}(X_{k-l,T_{k+1-l}}) \) then from Lemma~\ref{lemma-D-not-in-Rl-implies-B-zero-unbounded-rankl-dsada} we have that \(\hat{\br} \) is not zero-bounded on \(rank_{k-l,T_{k-l+1}}^{C_1,\dots,C_l} \), which then gives us from Lemma~\ref{lemma-hatB-zero-unbounded-rankl-give-upper-estimate-nk} that there exists a transition \(t\) with \(\bx(t)>0 \) such that \(\calT_\A[t] \) has an upper asymptotic estimate of \(n^k\).

		\item   This is proven in the next Section. Specifically in  Lemma~\ref{lemma-something-something-k+1}.	
	\end{enumerate}
\end{proof}

\subsection{Lower Asymptotic Estimates \(n^{k+1} \)}
\label{section-proof-of-lemma-something-something-k+1-estiamtes}

\begin{lemma*}[\textbf{\ref{lemma-something-something-k+1-estiamtes}}]
	For every transition \(t\) with \(\bx_k(t)>0 \) and every counter \(c\in \countersset\) with \(\sum_{t' \in T} \bx_k(t')\cdot \bu_{t'}(c)>0 \) it holds that \(n^{k+1} \) is a lower asymptotic estimate of \(\calT_\A[t] \) and \(\calC_\A[c] \).
\end{lemma*}

In the following, we first describe for any \(\epsilon>0 \) and function \(h_\epsilon:\mathbb{N}\rightarrow\mathbb{N}_0 \) with \(h_\epsilon(n)<n^{1-\epsilon} \) a strategy \(\sigma^\epsilon_k \) and then in  Lemma~\ref{lemma-something-something-k+1} we prove this strategy iterates any transition \(t\) with \(\bx_k(t)>0 \) at least \((h_\epsilon(n))^{k+1} \) times and increases every counter \(c \) with \(\sum_{t \in T} \bx(t)\cdot \bu_t(c)>0 \) to at least \((h_\epsilon(n))^{k+1} \) with very high probability, thus for \(h_\epsilon(n)=\lfloor n^{1-\epsilon} \rfloor \) this gives us Lemma~\ref{lemma-something-something-k+1-estiamtes} .

%Since \(\bx_k \) is a maximal solution of \hyperref[fig-systems]{(I)} for \(\A_{k,\hat{T}_{k+1}} \), it holds that \(\bx_k \) is a multi-component on \(\A_{k,\hat{T}_{k+1}} \) with \(\Delta(\bx_k)\geq \vec{0} \).

Let us define  \(\A_{k,\hat{T}_{k+1}}^{\bx_k} \) as the VASS MDP created from \(\A_{k,\hat{T}_{k+1}} \) by removing any transition \(t\) with \(\bx_k(t)=0 \), and let \(B_1,\dots,B_w \) be all the MECs of \(\A_{k,\hat{T}_{k+1}}^{\bx_k} \).

From Lemma~\ref{lemma-decompose-multicomponents-into-components} we can express \(\bx_k \) as a conical sum of components, that is \(\bx_k=\sum_{\by} a_\by\cdot \by \) where \(\by \) ranges over all the components of \(\A \), and \(a_\by>0 \) iff \(\by \) is a component of \(\A_{k,\hat{T}_{k+1}}^{\bx_k} \). We can wlog. assume that \(a_\by\in \mathbb{N}_0\) for all \(\by\), as if it is not we can simply take a multiple of \(\bx_k\) instead of \(\bx_k\).

For each \(1\leq i\leq w \) let \(\by_1^i,\dots,\by_{l_i}^i \) be all the components \(\by\) of \(B_i \). Let us consider the graph \(G^i=(V^i,E^i) \) where \(V^i=\{\by_1^i,\dots,\by_{l_i}^i \} \) and \(\{\by_a^i,\by_b^i\}\in E^i \) iff the MECs corresponding to \(\by^i_a \) and \(\by^i_b \) share any states. Clearly \(G^i\) is a strongly connected graph, so there exists a spanning tree \(G_T^i \) of \( G^i\).\footnote{\(G^i \) is strongly connected since \(B_i \) is a MEC, and thus for each two states \(p,q \) of \(B_i \) there exists a simple path from \(p \) to \(q\), and hence there also exists a strategy \(\sigma\in \cMD \) which reaches \(q\) from \(p\), and thus there also exists a component \(\by_{p,q} \) corresponding to \(\sigma\) whose corresponding MEC contains both \(p\) and \(q\). Let \(\by \) and \(\by'\) be two components of \(B_i \), and let \(p,q \) be some states of the MECs corresponding to \(\by,\by' \), respectively. Then it holds \((\by,\by_{p,q}),(\by_{p,q},\by')\in E^i \). Hence \(G^i \) is strongly connected.} Wlog. we can declare \(\by_1^i \) the root of \(G_T^i \). For each \(\by_j^i \) let \(parent(\by_j^i) \) denote the parent of \(\by_j^i \) in \(G_T^i \) and \(children(\bx_j^i)=\{\bx_a^i \mid parent(\bx_a^i)=\bx_j^i \}\). We can wlog. also assume that \(p_{\by_j^i} \) is a state of the MEC corresponding to \(parent(\by_j^i) \) (this follows from Lemma~\ref{lemma-component-recentering}).

%For each \(1\leq i\leq w \) let \(\by_1^i,\dots,\by_{l_i}^i \) be all the components \(\by\) of \(B_i \). Let us consider the graph \(G^i=(V^i,E^i) \) where \(V^i=\{\by_1^i,\dots,\by_{l_i}^i \} \) and \((\by_a^i,\by_b^i)\in E \) iff the MECs corresponding to \(\by_i \) and \(\by_j \) share any states. Clearly \(G^i\) is a strongly connected graph, and as such there exists a tree covering \(G_T^i \) of \( G^i\). Wlog. we can declare \(\by_1^i \) as root of \(G_T^i \). For each \(\by_j^i \) let \(parent(\by_j^i) \) denote the parent of \(\by_j^i \) in \(G_T^i \) and \(children(\bx_j^i)=\{\bx_b^i \mid parent(\bx_b^i)=\bx_j^i \}\). Note that we can wlog. assume that \(p_{\by_j^i} \) is also a state of \(parent(\by_j^i) \). (this follows from different components corresponding to same MEC but rooted in different state being just a constant multiple of one another, thus we wlog. can assume that for each MEC of a {\cMD} strategy all of the corresponding components \(\by\) either have \(a_\by>0 \) or all of them have \(a_\by=0 \))

Let us fix \(1>\epsilon>0 \) and a function \(h_\epsilon:\mathbb{N}\rightarrow\mathbb{N}_0 \) such that \(h_\epsilon(n)\leq n^{1-\epsilon} \) and \(\lim_{n\rightarrow\infty}h_{\epsilon}(n)=\infty \).

Let us define technical constants \(0<\epsilont_1,\epsilont_2,\dots  \). As their exact values are not important we leave the assignment of their exact values to Table~\ref{Table-eps-section-another-proof-k+1}, where we also show that our assignment satisfies all the assumptions we make on \(\epsilont_{1},\epsilont_{2},\dots\).

Let \(h_{\epsilont_{1}}:\N\rightarrow \N_0\) be such that \(h_{\epsilont_{1}}(n)\leq n^{1-\epsilont_{1}} \) and \(\lim_{n\rightarrow\infty}h_{\epsilont_1}(n)=\infty \). Let \(\sigma_{k-1}^{\epsilont_1} \), \(g_{k-1}^{k-1,\epsilont_1} \), \(r_{1}^{k-1,\epsilont_1},\dots,r_{k}^{k-1,\epsilont_1} \) be the strategy and functions from the induction assumption on point  \ref{enum-main-5} of Lemma~\ref{lemma-main-for-all-k} for \(k-1 \), for \(\epsilont_1 \) and \(h_{\epsilont_{1}}\). Note that this requires \begin{equation}\label{eq-epsbound-bvuycrvubh}
	\epsilont_1<1
\end{equation}

%such that \(h_{\epsilont_1}(\lfloor \frac{n^{1-\epsilont_{2}}}{m}\rfloor))=h_{\epsilon}(n) \) for all sufficiently large \(n\). Note that such \(h_{\epsilont_1}\) that satisfies \(h_{\epsilont_1}(n)\leq n^{1-\epsilont_1}\) always exists if \(\epsilont_1+\epsilont_{2}+\epsilont_1\epsilont_{2}< \epsilon   \).\todo{nope, such \(h_{\epsilont_1} \) does not exist, which one do we need?}

For each component \(\by \) of \(\A_{k,\hat{T}_{k+1}} \) let \(\sigma_{\hat\by}^{\epsilont_3} \) be a pointing strategy on \(\A_{+\hat{\by}} \) such that \(\lim_{n\rightarrow\infty}\prob_{p\vec{n}}^{\sigma_{\hat\by}^{\epsilont_{3}}}[\calP_{\A_{+\hat{\by}}}[\M_{\hat{\by}}]\geq n^{k+1-\epsilont_{3}}]=1 \). Note that the existence of \(\sigma_{\hat\by}^{\epsilont_{3}} \) follows from the following: from Lemma~\ref{lemma-D-in-Rl-for-all-l-imply-D-in-R} we have that if for each \(1\leq l<\frac{k}{2} \) it holds \(\support^{C_1,\dots,C_l}(\hat{\by})\subseteq \Delta^{C_1,\dots,C_l}(X_{k-l,T_{k+1-l}}) \) then \(\support^{C_1,\dots,C_{\lfloor \frac{k}{2}\rfloor}}(\hat{\by})\subseteq R^{B_\bx,k}  \), which from Lemma~\ref{lemma-D-in-R-lower-estimate-nk+1-for-hatB} implies a lower asymptotic estimate of \(n^{k+1} \) for \(\calP_{\A_{+\hat{\by}}}[\M_{\hat{\by}}] \). If on the other hand there exists \(1\leq l<\frac{k}{2} \) such that \(\support^{C_1,\dots,C_l}(\hat{\by})\nsubseteq \Delta^{C_1,\dots,C_l}(X_{k-l,T_{k+1-l}}) \) then from Lemma~\ref{lemma-D-not-in-Rl-implies-B-zero-unbounded-rankl-dsada} we have that \(\hat{\by} \) is not zero-bounded on \(rank_{k-l,T_{k-l+1}}^{C_1,\dots,C_l} \), which then gives us from Lemma~\ref{lemma-hatB-zero-unbounded-rankl-give-upper-estimate-nk} that there exists a transition \(t\) with \(\by(t)>0 \) such that \(t\notin \hat{T}_{k+1} \) which is a contradiction with \(\by \) being a component of \(\A_{k,\hat{T}_{k+1}} \). Also note that this implies that \(\hat{\by} \) is zero-bounded for every component \(\by \) of \(\A_{k,\hat{T}_{k+1}} \).

\textbf{Description of \(\sigma_k^\epsilon \):}
We will now describe a strategy \(\sigma_k^\epsilon \) and functions \(g^{k,\epsilon}\), \(r_1^{k,\epsilon},\dots,r_{k+1}^{k,\epsilon}\)  that satisfy point \ref{enum-main-5} of Lemma~\ref{lemma-main-for-all-k} for \(k\), \(\epsilon \) and \(h_\epsilon \), and furthermore
with probability \(p^{\sigma_k^\epsilon}_{p\vec{n}}\) such that \(\lim_{n\rightarrow\infty} p^{\sigma_k^\epsilon}_{p\vec{n}}=1 \) it holds that \(\alpha_{p\vec{n}}^{\sigma_k^\epsilon} \) reaches at least \(n^{k+1-\epsilon} \) in every counter \(c\) with \(\Delta(\bx_k)(c)>0 \) and contains each transition \(t\) with \(\bx_k(t)>0 \) at least \(n^{k+1-\epsilon} \) times.

%Note that this would finish proofs for \(k\) for the points  \ref{enum-main-1} (together with Lemmas~\ref{lemma:dichotomy} and \ref{lemma-linear-upper-bound}), \ref{enum-main-2} (together with Lemmas~\ref{lemma:dichotomy} and \ref{lemma-upper-bound-T}), \ref{enum-main-3}, and \ref{enum-main-5}. \todo{for 3 we probably also want to argue why \(\delta(\bx_k)(c)=0 \) for the remaining coutners, its cause we already have upper estiamt efor them}

The strategy \(\sigma_k^\epsilon \) works as follows (we give only high level description):

\(\sigma_k^\epsilon\) remembers (note that all of these can always be computed from the history) a set of states \(P \) that is initialized to \(P=\emptyset \), a variable \(g_k^\epsilon \) initialized to \(g_k^\epsilon=\bot \), and variable \(procccessed\)  that is initialed to \(procccessed=FALSE\). At each point let \(\alpha \) denote the current computation generated by \(\sigma_k^\epsilon\) so far. At each step we put \(g^{k,\epsilon}(\alpha)=g_k^\epsilon(\alpha) \) where \(g_k^\epsilon(\alpha)\) is the value of the variable \(g_k^\epsilon\) right before this step was taken by \(\sigma_k^\epsilon \). We also put \(r^{k,\epsilon}_1(\alpha)=\dots=r^{k,\epsilon}_{k+1}(\alpha)=SAME \) at each step unless stated otherwise.

\textbf{Bins initialization:} From Lemma~\ref{counters-pumpable-all-at-once} we have a strategy \(\pi \) that simultaneously pumps each counter \(c\), such that \(c\in C_i\) for \(i\leq k \) (here we use \(C_{k}=C_{k+} \)), to \(n^{i-\epsilont_{2}}\), and this is achieved with probability \(p_n'\) such that \(\lim_{n\rightarrow\infty} p_n'=1\). \(\sigma_k^\epsilon \) therefore starts by using \(\pi \) to reach such configuration, that is until the current counters vector \(\bv_0 \) is such that \(\bv_0(c)\geq  n^{i-\epsilont_{2}} \) for every counter \(c\in C_i \) for each \(i\leq k \) (here \(n\) is the initial size of the counters as usual). At this point, \(\sigma_k^\epsilon \) divides the counters vector into \(m \) bins of equal size (here \(m\) is some properly chosen constant that does not depend on \(n\)), that is each bin contains the vector \(\lfloor \frac{\bv_0}{m} \rfloor \). If at any point any counter in any bin becomes negative, then \(\sigma_k^\epsilon \) is from that point on undefined, but we assume that from that point on all of the counter vectors in all of the bins are left untouched.

%Let \(G_n^{k+1}=(V_n^{k+1},E_n^{k+1}) \) be a tree graph where each non-leaf node has exactly \(h_\epsilon(n) \) children and the distance of leafs to root is \(k+1 \). 
%At each point we will put \(f_k^{\epsilon}(\alpha)=\beta \), \(g_k^{\epsilon}(\alpha)=B \), and \(v_k^{\epsilon}(\alpha)=v_{\li} \), where \(v_{\li}\) is the \(\li\)-th leaf of \(G_n^{k+1}\). Also note that every time \(\sigma_k^\epsilon\) takes a step then unless stated otherwise \(\sigma_k^\epsilon\) also extends \(\beta\) by this step (in the case that \(\sigma_k^\epsilon\) otherwise modifies \(\beta \) then we take for \(f_k^{\epsilon}(\alpha)\) the value of \(\beta \) as it was right after the last step of \(\alpha \) was performed and before these modifications of \(\beta \) other than appending the last step of \(\alpha \) took place).

At this point \(\sigma_k^\epsilon\) virtually set every counter in every \(\hat\by\)-bin as well as the \((k-1)\)-tree-bin to \(\lfloor \frac{n^{1-\epsilont_{2}}}{m} \rfloor \) (note that this new value is smaller than the previous counter value for all of these bins) and then \(\sigma_k^\epsilon\) starts playing as per \(\sigma_{k-1}^{\epsilont_1} \) on the \((k-1)\)-tree-bin, and it will occasionally pause this computation to perform something else.

After each step performed as per \(\sigma_{k-1}^{\epsilont_1}\) in the \((k-1)\)-tree-bin (and also right before taking the very first step according to \(\sigma_{k-1}^{\epsilont_1}\)) let \(\alpha_{k-1}^{\epsilont_1} \) denote the computation taken by \(\sigma_{k-1}^{\epsilont_1}\) in the \((k-1)\)-tree-bin so far, and let \(p\) be the current state. First \(\sigma_k^\epsilon\) checks whether \(r_k^{k-1,\epsilont_1}(\alpha_{k-1}^{\epsilont_1})=NEXT \), and if yes then \(\sigma_k^\epsilon \) sets \(processed=FALSE \), \(g_k^\epsilon=\bot \), as well as \(r_l^{k,\epsilon}(\alpha)=r_{l+1}^{k,\epsilon}(\alpha)=\dots=r_{k+1}^{k,\epsilon}(\alpha)=NEXT \), where \(1\leq l\leq k\) is the smallest index such that \(r_l^{k-1,\epsilont_1}(\alpha_{k-1}^{\epsilont_1})=NEXT \). 

Second, \(\sigma_k^\epsilon\) checks whether  \(g^{k-1,\epsilont_1}(\alpha_{k-1}^{\epsilont_1})= \bot \), and if yes then \(\sigma_k^\epsilon\)  immediately performs another step according to \(\sigma_{k-1}^{\epsilont_1}\) in the \((k-1)\)-tree-bin, thus skipping steps three and four.

Third, \(\sigma_k^\epsilon\) checks if  \(processed=FALSE \), and if yes then it sets \(P \) to the set of all states contained in the MEC \(g^{k-1,\epsilont_1}(\alpha_{k-1}^{\epsilont_1}) \) of \(\A_k \), and also sets \(processed=TRUE \) (note that the second step ensures that \(g^{k-1,\epsilont_1}(\alpha_{k-1}^{\epsilont_1}) \neq \bot\)).
%Furthermore in such case \(\sigma_k^\epsilon\) also sets \(B=\bot \), \(\beta=p \), increases \(\li \) by \(+1\), and sets \(\kappa=g_{k-1}^{\epsilont_1}(\alpha_{k-1}^{\epsilont_1}) \).

% , and furthermore in such case \(\sigma_{k-1}^{\epsilont_{3}}\)  also for each component \(\by \) of \(\A \) moves the accrued effect from the \(\by\)-bin to a main bin (that is let \(\bu \) be the current counter vector in the \(\by\)-bin, then \(\sigma_{k-1}^{\epsilont_{3}}\) adds \(-\bu+\lfloor \frac{\bv_0}{m} \rfloor \) to the \(\by\)-bin while also adding \(+\bu \) to the main bin). 

Fourth, \(\sigma_k^\epsilon\)  asks if there exists \(1\leq i\leq w \) such that \(p=p_{\by_1^i} \) and \(p\in P \). If yes then \(\sigma_k^\epsilon\)  pauses the computation under \(\sigma_{k-1}^{\epsilont_1}\) in the \((k-1)\)-tree-bin, removes \(p\) from \(P \), and performs a single \(B_i \)-cycling-procedure before unpausing the computation under \(\sigma_{k-1}^{\epsilont_1}\) in the \((k-1)\)-tree-bin. (if such \(i\) doesn't exist then it continues by playing next step of \(\sigma_{k-1}^{\epsilont_1} \) in the \((k-1)\)-tree-bin)

\textbf{\(B_i \)-cycling-procedure:}
A single \(B_i \)-cycling-procedure consists of first setting \(g_k^\epsilon=B_i \), and then repeating the \(\lfloor n^{\epsilont_4} \rfloor\)-iteration-of-\(B_i \)-procedure exactly  \(h_\epsilon(n) \) times.

In the following, whenever we say that ``\(\sigma_k^\epsilon\) plays a component \(\by \) for \(x\) times'' (we will only do so when the current state is \(p_\by \)) what we mean is that \(\sigma_k^\epsilon\) actually simulates \(\sigma_{\hat\by}^{\epsilont_{3}} \) on the \(\hat\by\)-bin from the initial pointing configuration  	$(p_1,\dots,p_\kappa,p_\by)\by\lfloor \frac{n^{1-\epsilont_{2}}}{m} \rfloor$ (where \(\A_{+\hat{\by}}=\big( (\M_1,p_1),\dots,(\M_\kappa,p_\kappa),(\M_\by,p_\by) \big)\)) in such a way  that whenever \(\sigma_{\hat\by}^{\epsilont_{3}}\) points to a Markov chain \(\M \) then \(\sigma_k^\epsilon\) actually chooses the next transition according to the {\cMD} strategy corresponding to \(\M \) while adding the effect of this transition to the \(\hat{\by}\)-bin. Additionally, if the new state becomes \(p_\by \) right after \(\sigma_{\hat\by}^{\epsilont_{3}}\) pointed to \(\M_{\hat{\by}} \) then in addition to this \(\sigma_k^\epsilon\) also adds \(-\Delta(\by) \) to the \(\hat{\by}\)-bin and \(+\Delta(\by) \) to the \(\by\)-bin (thus the effect on \(\hat{\by}\)-bin will be exactly the same as the effect of the same step on \(\M_{\hat{\by}} \)). This proceeds until the state \(p_\by \) is revisited for exactly the \(x\)-th time right after \(\sigma_{\hat\by}^{\epsilont_{3}}\) last pointed to \(\M_{\hat{\by}} \) (if \(\sigma_{\hat\by}^{\epsilont_{3}}\) pointed elsewhere then visiting \(p_\by \) does not count), at which point this simulation of \(\sigma_{\hat\by}^{\epsilont_{3}} \) on the \(\hat\by\)-bin is paused. And if in the future ``\(\sigma_k^\epsilon\) plays a component \(\by \) for \(x\) times'' again, then it unpauses the previously paused simulation instead of starting a new one. Again until revisiting \(p_\by\) for the \(x\)-th time as above before pausing once again.

\textbf{\(\lfloor n^{\epsilont_{4}} \rfloor\)-iteration-of-\(B_i \)-procedure:}
A single \(\lfloor n^{\epsilont_{4}} \rfloor\)-iteration-of-\(B_i \)-procedure consists of first setting \(r_{k+1}^{k,\epsilon}(\alpha)=NEXT\) and then of \(\sigma_k^\epsilon\) playing each component \(\by_1^i,\dots,\by^i_{l_i} \) of \(B_i \) exactly \(a_\by\cdot \lfloor n^{\epsilont_{4}} \rfloor\) times as follows:

%First, if \(first\neq true \) then  \(\sigma_k^\epsilon\) sets \(r_{k+1}^{k,\epsilon}(\alpha)=NEXT \)\(\beta=p_{\by_1^i} \) and increases \(\li \) by \(+1\) (note that at this point  \(p_{\by_1} \) is the current state). 

%Second, \(\sigma_k^\epsilon\) sets \(B=B_i \) and \(first=false \).

First \(\sigma_k^\epsilon\) initializes variables \(P_{1},\dots,P_{l_i} \) where each \(P_j=children(\by_j^i) \).

Second \(\sigma_k^\epsilon\) starts playing \(\by_1^i \) exactly \(a_{\by_{1}^i}\cdot \lfloor n^{\epsilont_{4}} \rfloor\) times, except that this computation can be temporarily paused in favor of playing different components as described below.

Let \(\by^i_{j} \) be the component that is currently being played by \(\sigma_k^\epsilon\), then after  every step taken \(\sigma_k^\epsilon\) asks whether there exists \(\by_a^i\in P_j\) such that the current state \(p\) satisfies \(p=p_{\by_a^i} \). If such \(\by_a^i\) exists then the playing of  \(\by^i_{j} \) is  paused and \(\sigma_k^\epsilon\) first sets  \( P_j= P_j\setminus \{\by_a^i\} \) and then \(\sigma_k^\epsilon\) starts playing \(\by_a^i \) exactly \(a_{\by_{a}^i}\cdot \lfloor n^{\epsilont_{4}} \rfloor\) times. Note that once the playing of  \(\by_a^i \) finishes  \(\sigma_k^\epsilon\) moves to the step of "check whether there exists \(\by_a^i\in P_i\) such that..." (and this check is made as if \(\by^i_{j}\) was being played, not \(\by_a^i\)) before resuming playing of \(\by^i_{j} \).

The \(\lfloor n^{\epsilont_{4}} \rfloor\)-iteration-of-\(B_i \)-procedure ends once the playing of \(\by_1^i \) exactly \(a_{\by_{1}^i}\cdot \lfloor n^{\epsilont_{4}} \rfloor\) times finishes. At this point \(\sigma_k^\epsilon\) for each \(1\leq j\leq l_i \) adds \(-\Delta(\by_j^i) \cdot a_{\by_{j}^i}\cdot \lfloor n^{\epsilont_{4}} \rfloor\) to the \(\by_{j}^i\)-bin and \(+\Delta(\by_j^i) \cdot a_{\by_{j}^i}\cdot \lfloor n^{\epsilont_{4}} \rfloor\) to the main-bin. We say that this \(\lfloor n^{\epsilont_{4}} \rfloor\)-iteration-of-\(B_i \)-procedure succeeded if for each \(1\leq j\leq l_i \) the playing of \(\by_j^i\) visited at least once every single state of the MEC induced by \(\by_j^i \), otherwise it failed. Note that the procedure can end (unless the computation terminated during the procedure) only in the state \(p_{\by_1^i} \), and if it succeeded then its total effect on each \(\by_{j}^i\)-bin is exactly \(\vec{0}\).

\subsubsection{Analysis of \(\sigma_k^\epsilon \)}

Let \(\alpha_{p\vec{n}}^{\sigma_{k}^\epsilon} \) be the random variable that represents the computation generated by \(\sigma_{k}^\epsilon \) in \(\A \) from initial configuration \(p\vec{n} \), and let \(\alpha_{p\vec{n},..i}^{\sigma_k^\epsilon} \) be the prefix of \(\alpha_{p\vec{n}}^{\sigma_{k}^\epsilon}\) of length \(i\).

Let \(\alpha_{p\vec{n}}^{\sigma_{k-1}^{\epsilont_1}} \) be the random variable that represents the computation generated by \(\sigma_{k-1}^{\epsilont_1} \) in the \((k-1)\)-tree-bin when \(\sigma_{k}^\epsilon \) is started from the initial configuration \(p\vec{n}\), and let \(\alpha_{p\vec{n},..i}^{\sigma_{k-1}^{\epsilont_1}} \) be the prefix of \(\alpha_{p\vec{n}}^{\sigma_{k-1}^{\epsilont_1}}\) of length \(i\).

We will now prove the following Lemma.

\begin{lemma}\label{lemma-something-something-k+1}
	For \(\sigma_k^\epsilon \) and the functions \(g^{k,\epsilon} \), \(r_1^{k,\epsilon},\dots,r_{k+1}^{k,\epsilon}\) all of the following hold:
	\begin{itemize}
		\item \(\lim_{n\rightarrow\infty} \prob_{p\vec{n}}^{\sigma_k^\epsilon}[\calC_\A[c]\geq (h_\epsilon(n))^{k+1} ]=1 \) for every counter \(c\) with \(\Delta(\bx_k)(c)>0 \);
		\item \(\lim_{n\rightarrow\infty} \prob_{p\vec{n}}^{\sigma_k^\epsilon}[\calT_\A[t]\geq (h_\epsilon(n))^{k+1} ]=1 \) for every transition \(t\) with \(\bx_k(t)>0 \);
		\item  \(\sigma_k^\epsilon \), \(g^{k,\epsilon}\), \(r_{1}^{k,\epsilon},\dots,r_{k+1}^{k,\epsilon} \) satisfy point \ref{enum-main-5} of Lemma~\ref{lemma-main-for-all-k} for \(k\), \(\epsilon \), and \(h_\epsilon \). 
	\end{itemize}
\end{lemma}
\begin{proof}
		\textbf{Event \(E_0\):} Let \(E_0 \) be the set of all computations \(\alpha \) generated by \(\sigma_k^\epsilon \) in \(\A\), initiated in a configuration with counters vector \(\vec{n}\), which does not terminate before the bins are initialized, that is it succesfully reaches a configuration with counters vector \(\bv_0 \). As discussed in the ``bins initialization'' above, it holds \(\lim_{n\rightarrow\infty}\prob_{p\vec{n}}^{\sigma_k^\epsilon}[E_0]=1 \).

		\textbf{Event \(E_1\):}	Let \(E_1 \) be the set of all computations generated by \(\sigma_k^\epsilon \), initiated in a configuration with counters vector \(\vec{n}\), for which  \(\sigma_{k-1}^{\epsilont_1}\) on the  \((k-1)\)-tree-bin generates a computation that has the property from point~\ref{enum-main-5} of Lemma~\ref{lemma-main-for-all-k} for \(k-1\), \(\epsilont_1\), and \(h_{\epsilont_1} \). Additionally, if the computation under \(\sigma_k^\epsilon \) terminates before any counter in the  \((k-1)\)-tree-bin becomes negative then we include such computation in \(E_{1} \) even if it does not meet the previous requirement. Note that from the way we chose \(\sigma_{k-1}^{\epsilont_1}\) it holds \(\lim_{n\rightarrow\infty} \prob_{p\vec{n}}^{\sigma_k^\epsilon}(E_{1})=1 \).
		
		% If for some reason \(\sigma_k^\epsilon \) terminates while all counters in the \((k-1)\)-tree ``bin'' are positive then the corresponding path is included in \(E_{k-1}^{\epsilont_1} \) as well. Note that from how we chose \(\sigma_{k-1}^{\epsilont_1} \) (i.e., the induction assumption on point~\ref{enum-main-5} from Lemma~\ref{lemma-main-for-all-k}) it holds \(\lim_{n\rightarrow\infty} \prob_{p\vec{n}}^{\sigma_k^\epsilon}(E_{k-1}^{\epsilont_1})=1 \).

		\textbf{Event \(E_2\):}	Let \(E_2 \) be the set of all computations generated by \(\sigma_k^\epsilon \) \(\A\), initiated in a configuration with counters vector \(\vec{n}\), which have the following property: for each component \(\by \) with \(a_\by>0 \) no counter in the \(\hat{\by}\)-bin becomes negative before the pointing computation under \(\sigma^{\epsilont_{3}}_{\hat\by}\) in the \(\hat{\by}\)-bin points at \(\hat{\by} \) at least \(n^{k+1-\epsilont_{3}} \) times (note that computations which either are not in \(E_0\) or for which any counter became negative in any other bin are also in \(E_2\)). From how we chose the pointing strategies \(\sigma^{\epsilont_{3}}_{\hat\by}\) it holds for each \(\by \) that \(\lim_{n\rightarrow\infty}\prob_{p\vec{1}\cdot \lfloor \frac{n^{1-\epsilont_{2}}}{m} \rfloor }^{\sigma_{\hat{\by}}^{\epsilont_3}}[\pointcomplex_{\A_{+\hat{\by}}}[\M_{\hat{\by}}]\geq (\lfloor \frac{n^{1-\epsilont_{2}}}{m} \rfloor)^{k+1-\epsilont_3}]=1 \). Assuming  \begin{equation}\label{eq-epsbound-bhcftxfdsfsw}
			\epsilont_3\leq k+1
		\end{equation} it holds \((\lfloor \frac{n^{1-\epsilont_{2}}}{m} \rfloor)^{k+1-\epsilont_3}
		\leq 
		( \frac{n^{1-\epsilont_{2}}}{m} )^{k+1-\epsilont_3} 
		=
		 \frac{n^{k+1-\epsilont_3-(k+1-\epsilont_3)\cdot \epsilont_{2}}}{m^{k+1-\epsilont_3}}  
		 \leq 
		   n^{k+1-\epsilont_3}
		 \). Hence 
		  \[\lim_{n\rightarrow\infty} \prob_{p\vec{n}}^{\sigma_k^\epsilon}(E_2)\geq \lim_{n\rightarrow\infty} 1- \sum_{\by} \prob_{p\vec{1}\cdot \lfloor \frac{n^{1-\epsilont_{2}}}{m} \rfloor }^{\sigma_{\hat{\by}}^{\epsilont_3}}[\pointcomplex_{\A_{+\hat{\by}}}[\M_{\hat{\by}}]< (\lfloor \frac{n^{1-\epsilont_{2}}}{m} \rfloor)^{k+1-\epsilont_3}] =1  \]

		\textbf{Event \(E_3\):}
		Let \(E_3 \) be the set of all computations \(\alpha \) generated by \(\sigma_k^\epsilon \) in \(\A\), initiated in a configuration with counters vector \(\vec{n}\), for which for each \(1\leq i\leq w\) the first \(h_\epsilon(n)\cdot (h_{\epsilont_{1}}(\lfloor \frac{n^{1-\epsilont_{2}}}{m} \rfloor))^k \) \(\lfloor n^{\epsilont_{4}} \rfloor\)-iteration-of-\(B_i \)-procedures succeed (here we consider procedures which did not even start to have succeeded). Note that this condition is equivalent to saying that for each \(1\leq i\leq w\) and each \(1\leq j \leq l_i \), during each of the first \(h_\epsilon(n)\cdot (h_{\epsilont_{1}}(\lfloor \frac{n^{1-\epsilont_{2}}}{m} \rfloor))^k \) \(\lfloor n^{\epsilont_{4}} \rfloor\)-iteration-of-\(B_i \)-procedures, the pointing strategy \(\sigma_{\hat{\by}_{j}^i}^{\epsilont_3} \) simulated on the \(\hat{\by}_j^i \)-bin during the  playing of \(\by_j^i \)  visits, right after pointing to \(\hat{\by}_j^i \), every state of the MEC corresponding to \(\by_j^i \) during this procedure. 
		
		During a single \(\lfloor n^{\epsilont_{4}} \rfloor\)-iteration-of-\(B_i \)-procedure each counter \(c\) in a \(\by_j^i\)-bin can be changed by at most \(a_{\by}\cdot \Delta(\by_j^i)(c)\cdot \lfloor n^{\epsilont_{4}} \rfloor < \lfloor \frac{n^{1-\epsilont_{2}}}{m}\rfloor \) for all sufficiently large \(n\), assuming \begin{equation}\label{eq-epsilonbound-jhjvgcgfxdysfsdfstjh}
			1-\epsilont_{2}>\epsilont_{4}
		\end{equation}  before being set back to \(\lfloor \frac{\bv_0(c)}{m}\rfloor\geq \lfloor \frac{n^{1-\epsilont_{2}}}{m}\rfloor \). Hence for all sufficiently large \(n\) the procedure cannot fail by terminating unless a counter becomes negative on some \(\hat{\by}_{j}^i\)-bin. Conditioned on \(E_2\) it holds that a counter can only become negative in a \(\hat{\by}_{j}^i\)-bin after \(\sigma^{\epsilont_{3}}_{\hat\by}\) points at \(\hat{\by}_{j}^i \) at least \(n^{k+1-\epsilont_{3}} \) times.

			All of the steps at which \(\sigma_{\hat{\by}_{j}^i}^{\epsilont_3} \)  points at \(\hat{\by}_j^i \) during the first \(h_\epsilon(n)\cdot (h_{\epsilont_{1}}(\lfloor \frac{n^{1-\epsilont_{2}}}{m} \rfloor))^k \) \(\lfloor n^{\epsilont_{4}} \rfloor\)-iteration-of-\(B_i \)-procedures produce a computation \(\pi\) on \(\M_{\hat{\by}_{j}^i} \) (see description of pointing VASS in Appendix~\ref{section-additional-definitions} for details). We will now show that with very high probability, \(\pi \) never takes more than \(\lfloor n^{\epsilont_5} \rfloor\) steps without visiting every single state of the MEC corresponding to \(\by_j^i\) at least once. Note that this would imply  \(\lim_{n\rightarrow\infty}\prob_{p\vec{n}}^{\sigma^\epsilon_k}[E_3]=1 \) as that would imply that during each of the first \(h_\epsilon(n)\cdot (h_{\epsilont_{1}}(\lfloor \frac{n^{1-\epsilont_{2}}}{m} \rfloor))^k \) \(\lfloor n^{\epsilont_{4}} \rfloor\)-iteration-of-\(B_i \)-procedures  \(\sigma^{\epsilont_{3}}_{\hat\by}\) points at \(\hat{\by}_{j}^i \) at most \(a_{\by_j^i}\cdot \lfloor n^{\epsilont_4}\rfloor \cdot \lfloor n^{\epsilont_5}\rfloor \) times, which implies \begin{gather*}
			\length(\pi)
			\leq
			 h_\epsilon(n)\cdot (h_{\epsilont_{1}}(\lfloor \frac{n^{1-\epsilont_{2}}}{m} \rfloor))^k\cdot a_{\by_j^i}\cdot \lfloor n^{\epsilont_4}\rfloor \cdot  \lfloor n^{\epsilont_5}\rfloor \leq\\
		    n^{1-\epsilon}\cdot \big((\lfloor \frac{n^{1-\epsilont_{2}}}{m} \rfloor)^{1- \epsilont_{1}}\big)^{k}\cdot a_{\by_j^i}\cdot \lfloor n^{\epsilont_4}\rfloor \cdot  \lfloor n^{\epsilont_5}\rfloor 
		    \leq\\
		    n^{1-\epsilon}\cdot  \frac{n^{k-k\cdot \epsilont_{1}-k\cdot (1- \epsilont_{1})\cdot\epsilont_{2}}}{m^{k-k\cdot \epsilont_{1}}} \cdot a_{\by_j^i}\cdot  n^{\epsilont_4} \cdot   n^{\epsilont_5}
		    = \\
		    n^{1+k-\epsilon-k\cdot \epsilont_{1}-k\cdot (1- \epsilont_{1})\cdot\epsilont_{2}+\epsilont_4+\epsilont_5}\cdot  \frac{1}{m^{k-k\cdot \epsilont_{1}}} \cdot a_{\by_j^i}
		    < n^{k+1-\epsilont_{3}} 
		\end{gather*} where the last inequality holds for all sufficiently large \(n\) assuming \begin{equation}\label{eq-eps-bound-jhcyrxfdsfsdfstycuv}
			-\epsilont_{3} > -\epsilon-k\cdot \epsilont_{1}-k\cdot (1- \epsilont_{1})\cdot\epsilont_{2}+\epsilont_4+\epsilont_5
			\end{equation} hence conditioned on \(E_2 \) the procedure cannot fail by depleting any counter on any bin, and assuming \begin{equation}\label{eq-eps-bound-jhcyrxtycuv}
			\epsilont_4 > \epsilont_5
			\end{equation} this would also imply that  each of the first \(h_\epsilon(n)\cdot (h_{\epsilont_{1}}(\lfloor \frac{n^{1-\epsilont_{2}}}{m} \rfloor))^k \) \(\lfloor n^{\epsilont_{4}} \rfloor\)-iteration-of-\(B_i \)-procedures that are initiated succeed, assuming sufficiently large \(n\) (since \(\lfloor n^{\epsilont_4}\rfloor >  \lfloor n^{\epsilont_5}\rfloor \) for all sufficiently large \(n\)).

	 Let \(r \) be the probability of \(\pi\) avoiding some state \(p'\) for a total of \(\lfloor n^{\epsilont_5}\rfloor \) steps within the first at most \(n^{k+1-\epsilont_{3}} \) steps. Since at every step, the probability of reaching \(p'\) within the next at most constant number of steps is bounded from below by a constant (remember that \(\pi\) is a computation on a strongly connected VASS Markov chain), it holds \(r\leq a^{\lfloor n^{\epsilont_5}\rfloor}\cdot n^{k+1-\epsilont_{3}} \) for some constant \(a<1\). 	Hence
		\begin{align*}
			\lim_{n\rightarrow\infty}\prob_{p\vec{n}}^{\sigma_k^\epsilon}[E_3\mid E_2]
			\geq
			\lim_{n\rightarrow\infty} 1-r
			\geq
			\lim_{n\rightarrow\infty} 1-a^{\lfloor n^{\epsilont_5}\rfloor}\cdot n^{k+1-\epsilont_{3}}
			=
			1	 
		\end{align*}
		And since \(\lim_{n\rightarrow\infty}\prob_{p\vec{n}}^{\sigma^\epsilon_k}[E_2]=1 \) this also gives us   \(\lim_{n\rightarrow\infty}\prob_{p\vec{n}}^{\sigma^\epsilon_k}[E_1]=1 \).

	\textbf{Number of \(\lfloor n^{\epsilont_{4}} \rfloor\)-iteration-of-\(B_i \)-procedures:}	We will now show that for each computation in \(E_0\cap E_1 \cap E_2 \cap E_3  \), for each \(1\leq i\leq w \), the number of times \(\lfloor n^{\epsilont_{4}} \rfloor\)-iteration-of-\(B_i \)-procedure is performed is exactly \(h_\epsilon(n)\cdot (h_{\epsilont_1}(\lfloor \frac{n^{1-\epsilont_{2}}}{m} \rfloor))^{k}\). Assume towards contradiction that there exists a computation \(\alpha\in E_0\cap E_1 \cap E_2 \cap E_3  \) generated by \(\sigma_k^\epsilon \) from an initial configuration \(p\vec{n}\) that does not have this property, that is there exists \(1\leq i\leq w\) such that the \(\lfloor n^{\epsilont_{4}} \rfloor\)-iteration-of \(B_i \)-procedure is not performed exactly  \( h_\epsilon(n)\cdot (h_{\epsilont_1}(\lfloor \frac{n^{1-\epsilont_{2}}}{m} \rfloor))^{k} \) times along \(\alpha \).

%	 Notice that \(h_\epsilon(n)\cdot (h_{\epsilont_1}(n))^{k} \leq  n^{1-\epsilon}\cdot (n^{1-\epsilont_1})^{k}= n^{k+1-k\cdot \epsilont_1-\epsilon}\) 
	
	Since \(\alpha\in E_1 \) there exists a \((k,h_{\epsilont_1}(\lfloor \frac{n^{1-\epsilont_{2}}}{m} \rfloor))\)-tree \(G=(V,E,f,g) \) of \(\alpha_{p\vec{n}}^{\sigma_{k-1}^{\epsilont_1}} \) satisfying the conditions of point~\ref{enum-main-5} of Lemma~\ref{lemma-main-for-all-k}. Let \(v_1,\dots,v_a \) be all the leaves of \( G\). Notice that for each leaf \(v_j\)  there exist indexes \(a_j,b_j,c_j \), such that all of the following hold:
%	\todo{would be nice to properly prove these as lemmas somewhere}
	\begin{itemize}
		\item \(\alpha_{p\vec{n},a_j,\dots,c_j}^{\sigma_{k-1}^{\epsilont_1}}=f(v) \),
		\item \(\alpha_{p\vec{n},b_j,\dots,c_j}^{\sigma_{k-1}^{\epsilont_1}} \) contains every single state of \(g(v) \) at least once,
		\item \(r_{k}^{k-1,\epsilont_1}(\alpha_{p\vec{n},..a_j}^{\sigma_{k-1}^{\epsilont_1}})=NEXT \),
		\item for each \(a_j< l< c_j \) it holds \(r_{k}^{k-1,\epsilont_1}(\alpha_{p\vec{n},..l}^{\sigma_{k-1}^{\epsilont_1}})=SAME \),
		\item for each \(b_j\leq l<c_j \) it holds  \(g^{k-1,\epsilont_1}(\alpha_{p\vec{n},..l}^{\sigma_{k-1}^{\epsilont_1}})=g(v)\),
		\item for each \(a_j\leq l<b_j \) it holds  \(g^{k-1,\epsilont_1}(\alpha_{p\vec{n},..l}^{\sigma_{k-1}^{\epsilont_1}})=\bot\).
	\end{itemize}
	Therefore for each leaf of \(G \) \(\sigma_k^{\epsilon} \) sets the set \(P\) to the set of all states of \(g(v) \) exactly once, and then for each \(B_i \) that is included in \(g(v) \) perform exactly one \(B_i\)-cycling-procedure which consists of performing the \(\lfloor  n^{\epsilont_{4}} \rfloor \)-iteration-of-\(B_i\)-procedure exactly \(h_\epsilon(n) \) times. Since each MEC \(\MEC\) of \(\A_{k} \) satisfies \(\MEC=g(v_j) \) for exactly \((h_{\epsilont_{1}}(\lfloor \frac{n^{1-\epsilont_{2}}}{m} \rfloor))^{k} \) leaves, this means that assuming no counter becomes negative in the  the main-bin then \(\alpha \) iterated exactly \(h_{\epsilon}(n)\cdot (h_{\epsilont_1}(\lfloor \frac{n^{1-\epsilont_{2}}}{m} \rfloor))^{k} \) \(\lfloor n^{\epsilont_{4}} \rfloor\)-iteration-of \(B_i \)-procedures for each \(B_i \).  Notice that from conditioning on \(E_3 \) this also implies it is not possible for any \(\lfloor n^{\epsilont_{4}} \rfloor\)-iteration-of \(B_i \)-procedure to fail as long as no counter becomes negative in the main-bin. 

	Hence it must hold that  some counter \(c\in C_\variable \) of \(\A_k \) for some \(1\leq \variable\leq k \) (here we put \(C_k=C_{k+} \)) becomes negative in the  main-bin along \(\alpha \) before this happens. Notice that the  main-bin is modified only by the \(\lfloor n^{\epsilont_{4}}\rfloor \)-iteration-of-\(B_i\)-procedures, and thus at each point along \(\alpha\)  it holds that the value of \(c \) in the main-bin  is equal to \(\lfloor \frac{n^{\variable-\epsilont_{2}}}{m} \rfloor + \sum_{i=1}^{w}\#_{i}\cdot \sum_{j=1}^{l_i}  \Delta(\by_j^i)(c)\cdot a_{\by_j^i}\cdot  \lfloor n^{\epsilont_{4}} \rfloor  \) where \(\#_{i} \) represents the number of \(\lfloor n^{\epsilont_{4}}\rfloor \)-iteration-of-\(B_i\)-procedures performed so far.
	
	Given two components \(\by_1 \) and \(\by_2 \) of \(\A_k \), we define the \emph{level-distance} of \(\by_1 \) and \(\by_2 \) as the smallest value \(0\leq l\) such that both \(\by_1 \) and \(\by_2 \) belong to the same MEC of \(\A_{k-l} \) (note that the maximal possible level-distance is \(k-1 \)). Also notice for any two components \(\by_1 \) and \(\by_2 \) of \(\A \) whose level distance is at least \(\variable\) it holds that either \(\Delta(\by_1)(c)=0 \) or \(\Delta(\by_2)(c)=0 \) since both \(\by_1 \) and \(\by_2 \) contain a different local copy of the counter of which \(c \) is a local copy of in \(\A_k \) (note that if \(\variable=k \) then it is not possible for the level-distance to be at least \(\variable \) hence in such case this implication holds trivially). Also notice that for any two \(B_i,B_j \) the level-distance between \(\by_a^i \) and \(\by_b^j \) is the same regardless of \(a \) and \(b\), allowing us to define the level-distance of \(B_i\) and \(B_j \) as the level-distance between \(\by_1^i \) and \(\by_1^j \).
	
	But notice that for each vertex \(v\) of \(G \) at distance \(k-\variable+1 \) from the root such that the sub-tree \(G_v\) rooted in \(v\) contains a leaf \(v' \) with the MEC \(g(v') \) that contains the counter \(c\), it holds    for each MEC \(\MEC\) of \(\A_k \) containing \(c\) that \(G_v \) has exactly \((h_{\epsilont_1}(\lfloor \frac{n^{1-\epsilont_{2}}}{m} \rfloor))^{\variable-1} \) leafs \(v'\) with \(g(v')=\MEC \). And since in each such \(v'\) \(\sigma_k^\epsilon \) performs the \(\lfloor n^{\epsilont_{4}}\rfloor \)-iteration-of-\(B_i\)-procedure exactly \(h_\epsilon(n) \) times, at each point it holds for each two  \(B_i, B_j \) whose level-distance is at most \(\variable\)  that \(|\#_{i}-\#_{j}|\leq 2\cdot h_\epsilon(n)\cdot (h_{\epsilont_1}(\lfloor \frac{n^{1-\epsilont_{2}}}{m} \rfloor))^{\variable-1} \). Let us fix some  \(B_b \) of \(\A_{k,\hat{T}_{k+1}}^{\bx_k} \) that contains \(c \) and let \(M_{\variable}=\{i\in \{1,\dots,w   \}\mid \textit{the level-distance of } B_b\textit{ and } B_i \textit{ is at most } \variable-1\}  \). Let \(r=\min_{i\in M_\variable}\#_{i}   \), then it holds for all \(i\in M_\variable\) that \(\#_i-\#_r\leq |\#_i-\#_r|\leq   2\cdot h_\epsilon(n)\cdot (h_{\epsilont_1}(\lfloor \frac{n^{1-\epsilont_{2}}}{m} \rfloor))^{\variable-1}\).
	
	Therefore the value of the counter \(c\) on the main-bin (after being initialized) at each point along \(\alpha \) satisfies 
	\begin{gather*}
		\lfloor \frac{n^{\variable-\epsilont_{2}}}{m} \rfloor + \sum_{i=1}^{w}\#_{i}\cdot \sum_{j=1}^{l_i}  \Delta(\by_j^i)(c)\cdot a_{\by_j^i}\cdot \lfloor n^{\epsilont_{4}} \rfloor 
		=\\
		\lfloor \frac{n^{\variable-\epsilont_{2}}}{m} \rfloor + \sum_{i\in M_{\variable}}\#_{i}\cdot \sum_{j=1}^{l_i}  \Delta(\by_j^i)(c)\cdot a_{\by_j^i}\cdot  \lfloor n^{\epsilont_{4}} \rfloor
		+ \sum_{i\notin M_{\variable}}\#_{i}\cdot \sum_{j=1}^{l_i}  \Delta(\by_j^i)(c)\cdot a_{\by_j^i}\cdot  \lfloor n^{\epsilont_{4}} \rfloor
		=\\
		\lfloor \frac{n^{\variable-\epsilont_{2}}}{m} \rfloor + \sum_{i\in M_{\variable}}\#_{i}\cdot \sum_{j=1}^{l_i}  \Delta(\by_j^i)(c)\cdot a_{\by_j^i}\cdot  \lfloor n^{\epsilont_{4}} \rfloor
		+ \sum_{i\notin M_{\variable}}\#_{i}\cdot \sum_{j=1}^{l_i}  0\cdot a_{\by_j^i}\cdot  \lfloor n^{\epsilont_{4}} \rfloor
		=\\
		\lfloor \frac{n^{\variable-\epsilont_{2}}}{m} \rfloor + \sum_{i\in M_{\variable}}\#_{i}\cdot \sum_{j=1}^{l_i}  \Delta(\by_j^i)(c)\cdot a_{\by_j^i}\cdot  \lfloor n^{\epsilont_{4}} \rfloor
		=\\
		\lfloor \frac{n^{\variable-\epsilont_{2}}}{m} \rfloor + \sum_{i\in M_{\variable}}(\#_i-\#_r+\#_r)\cdot \sum_{j=1}^{l_i}  \Delta(\by_j^i)(c)\cdot a_{\by_j^i} \cdot \lfloor n^{\epsilont_{4}} \rfloor
		=\\
		\lfloor \frac{n^{\variable-\epsilont_{2}}}{m} \rfloor + \sum_{i\in M_{\variable}}\#_r\cdot \sum_{j=1}^{l_i}  \Delta(\by_j^i)(c)\cdot a_{\by_j^i} \cdot \lfloor n^{\epsilont_{4}} \rfloor
		+
		\sum_{i\in M_{\variable}} (\#_i-\#_r)\cdot \sum_{j=1}^{l_i}  \Delta(\by_j^i)(c)\cdot a_{\by_j^i}\cdot  \lfloor n^{\epsilont_{4}} \rfloor   
		\geq \\
		\lfloor \frac{n^{\variable-\epsilont_{2}}}{m} \rfloor + \sum_{i\in M_{\variable}}\#_r\cdot \sum_{j=1}^{l_i}  \Delta(\by_j^i)(c)\cdot a_{\by_j^i}\cdot  \lfloor n^{\epsilont_{4}} \rfloor
		-
		\sum_{i\in M_{\variable}} (\#_i-\#_r)\sum_{j=1}^{l_i}  |\Delta(\by_j^i)(c)|\cdot a_{\by_j^i}\cdot  \lfloor n^{\epsilont_{4}} \rfloor 
		\geq \\
		\lfloor \frac{n^{\variable-\epsilont_{2}}}{m} \rfloor + \#_r\cdot \lfloor n^{\epsilont_{4}} \rfloor\cdot \sum_{i\in M_{\variable}}\sum_{j=1}^{l_i}  \Delta(\by_j^i)(c)\cdot a_{\by_j^i} 
		\\-
		\sum_{i\in M_{\variable}} 2\cdot h_\epsilon(n)\cdot (h_{\epsilont_1}(\lfloor \frac{n^{1-\epsilont_{2}}}{m} \rfloor))^{\variable-1}\cdot \sum_{j=1}^{l_i}  |\Delta(\by_j^i)(c)|\cdot a_{\by_j^i} \cdot \lfloor n^{\epsilont_{4}} \rfloor 
	\end{gather*}
	As it also holds that \(0\leq \Delta(\bx_k)(c)=\sum_{\by} a_\by\cdot \Delta(\by)(c)
	=
	\sum_{i\in M_{\variable}}\sum_{j=1}^{l_i}a_{\by_j^i}\cdot \Delta(\by_j^i)(c)
	\) 
	we can write 
	\begin{gather*}
		\lfloor \frac{n^{\variable-\epsilont_{2}}}{m} \rfloor + \#_r\cdot \lfloor n^{\epsilont_{4}} \rfloor\cdot \sum_{i\in M_{\variable}}\sum_{j=1}^{l_i}  \Delta(\by_j^i)(c)\cdot a_{\by_j^i} 
		\\-
		\sum_{i\in M_{\variable}} 2\cdot h_\epsilon(n)\cdot (h_{\epsilont_1}(\lfloor \frac{n^{1-\epsilont_{2}}}{m} \rfloor))^{\variable-1}\cdot \sum_{j=1}^{l_i}  |\Delta(\by_j^i)(c)|\cdot a_{\by_j^i} \cdot \lfloor n^{\epsilont_{4}} \rfloor 
		=\\
		\lfloor \frac{n^{\variable-\epsilont_{2}}}{m} \rfloor + \#_r\cdot \lfloor n^{\epsilont_{4}} \rfloor\cdot  \Delta(\bx_k)(c) 
		-
		\sum_{i\in M_{\variable}} 2\cdot h_\epsilon(n)\cdot (h_{\epsilont_1}(\lfloor \frac{n^{1-\epsilont_{2}}}{m} \rfloor))^{\variable-1}\cdot  \sum_{j=1}^{l_i} |\Delta(\by_j^i)(c)|\cdot a_{\by_j^i} \cdot \lfloor n^{\epsilont_{4}} \rfloor 
		\geq \\
		\lfloor \frac{n^{\cdot -\epsilont_{2}}}{m} \rfloor -
		\sum_{i\in M_{\variable}} 2\cdot h_\epsilon(n)\cdot (h_{\epsilont_1}(\lfloor \frac{n^{1-\epsilont_{2}}}{m} \rfloor))^{\variable-1}\cdot\sum_{j=1}^{l_i}  |\Delta(\by_j^i)(c)|\cdot a_{\by_j^i}\cdot  \lfloor n^{\epsilont_{4}} \rfloor 
		\geq\\
		\lfloor \frac{n^{\variable-\epsilont_{2}}}{m} \rfloor
		-
		\sum_{i\in M_{\variable}} 2\cdot h_\epsilon(n)\cdot (h_{\epsilont_1}(\lfloor \frac{n^{1-\epsilont_{2}}}{m} \rfloor))^{\variable-1}\cdot l_i \cdot  \max_{\by}(a_{\by}\cdot |\Delta(\by)(c)|)\cdot  \lfloor n^{\epsilont_{4}} \rfloor
		=\\
		\lfloor \frac{n^{\variable-\epsilont_{2}}}{m} \rfloor
		-
		|M_{\variable}|\cdot  2\cdot h_\epsilon(n)\cdot (h_{\epsilont_1}(\lfloor \frac{n^{1-\epsilont_{2}}}{m} \rfloor))^{\variable-1}\cdot l_i \cdot  \max_{\by}(a_{\by}\cdot |\Delta(\by)(c)|)\cdot  \lfloor n^{\epsilont_{4}} \rfloor
		\geq\\
		\lfloor \frac{n^{\variable-\epsilont_{2}}}{m} \rfloor
		-
		w \cdot  2\cdot n^{1-\epsilon}\cdot \big((\lfloor \frac{n^{1-\epsilont_{2}}}{m} \rfloor)^{1-\epsilont_1}\big)^{\variable-1}\cdot l_i \cdot  \max_{\by}(a_{\by}\cdot |\Delta(\by)(c)|)\cdot   n^{\epsilont_{4}} 
		\geq\\
		\lfloor \frac{n^{\variable-\epsilont_{2}}}{m} \rfloor
		-
		w \cdot  2\cdot n^{1-\epsilon}\cdot \big(( \frac{n^{1-\epsilont_{2}}}{m} )^{1-\epsilont_1}\big)^{\variable-1}\cdot l_i \cdot  \max_{\by}(a_{\by}\cdot |\Delta(\by)(c)|)\cdot   n^{\epsilont_{4}} 
		=\\
		\lfloor \frac{n^{\variable-\epsilont_{2}}}{m} \rfloor
		-
		w \cdot  2\cdot n^{1-\epsilon}\cdot  \frac{
			n^{(1-\epsilont_{2})\cdot(1-\epsilont_1)\cdot(\variable-1)
		}}{m^{(1-\epsilont_1)\cdot(\variable-1)}} \cdot l_i \cdot  \max_{\by}(a_{\by}\cdot |\Delta(\by)(c)|)\cdot   n^{\epsilont_{4}} 
		=\\
		\lfloor \frac{n^{\variable-\epsilont_{2}}}{m} \rfloor
		-
		w \cdot  2\cdot n^{1-\epsilon}\cdot  \frac{
			n^{
				\variable-1-(\variable-1)\cdot\epsilont_1-(\variable-1)\cdot\epsilont_2\cdot(1-\epsilont_1)
		}}{m^{(1-\epsilont_1)\cdot(\variable-1)}} \cdot l_i \cdot  \max_{\by}(a_{\by}\cdot |\Delta(\by)(c)|)\cdot   n^{\epsilont_{4}} 
		=\\
%		\lfloor \frac{n^{\variable-\epsilont_{2}}}{m} \rfloor
%		-
%		w \cdot  2\cdot n^{\variable-\epsilon+(\variable-1)\cdot \epsilont_1+\epsilont_{4}}\cdot l_i \cdot  \max_{\by}(a_{\by}\cdot |\Delta(\by)(c)|) 
%		=\\
		\lfloor \frac{n^{\variable-\epsilont_{2}}}{m} \rfloor
		-
		w \cdot  2\cdot \frac{1}{m^{(1-\epsilont_{2})\cdot (1-\epsilont_{2}) }} \cdot l_i \cdot  \max_{\by}(a_{\by}\cdot |\Delta(\by)(c)|) \cdot n^{    \variable-\epsilon-(\variable-1)\cdot\epsilont_1-(\variable-1)\cdot\epsilont_2\cdot(1-\epsilont_1)+\epsilont_4}
	\end{gather*}

	And since \(w \cdot  2\cdot\frac{1}{m^{(1-\epsilont_{2})\cdot (1-\epsilont_{2}) }}\cdot  l_i \cdot  \max_{\by}(a_{\by}\cdot |\Delta(\by)(c)|)  \)  is a constant, assuming  \begin{equation}\label{eq-epsbound-fsdhivuttrxtrx}
		-\epsilont_2>-\epsilon-(\variable-1)\cdot\epsilont_1-(\variable-1)\cdot\epsilont_2\cdot(1-\epsilont_1)+\epsilont_4
	\end{equation} it  holds for all sufficiently large \(n\) that \[\lfloor \frac{n^{\variable-\epsilont_{2}}}{m} \rfloor
	-
	w \cdot  2\cdot \frac{1}{m^{(1-\epsilont_{2})\cdot (1-\epsilont_{2}) }} \cdot l_i \cdot  \max_{\by}(a_{\by}\cdot |\Delta(\by)(c)|) \cdot n^{    \variable-\epsilon-(\variable-1)\cdot\epsilont_1-(\variable-1)\cdot\epsilont_2\cdot(1-\epsilont_1)+\epsilont_4}\geq 0\] Therefore for all sufficiently large \(n\) no counter can get depleted on the main-bin during a computation \(\alpha \in E_0\cap E_1 \cap E_2 \cap E_3   \) under \(\sigma_k^\epsilon \). Therefore assuming sufficiently large \(n\) and conditioned on \(E_0\cap E_1 \cap E_2 \cap E_3 \) it holds for any computation produced by \(\sigma_k^\epsilon \) that \(\#_r=h_\epsilon(n)\cdot (h_{\epsilont_1}(\lfloor \frac{n^{1-\epsilont_{2}}}{m} \rfloor))^{k}\) for each \(1\leq i \leq w\) which gives us that each counter \(c\) with \(\Delta(\bx_k)(c)>0 \) reaches in the main-bin at least  the value
	\begin{gather*}
		\lfloor \frac{n^{k-\epsilont_{2}}}{m} \rfloor + \#_r\cdot \lfloor n^{\epsilont_{4}} \rfloor\cdot  \Delta(\bx_k)(c) \\-
		|M_{k}|\cdot  2\cdot h_\epsilon(n)\cdot (h_{\epsilont_1}(\lfloor \frac{n^{1-\epsilont_{2}}}{m} \rfloor))^{k-1}\cdot l_i \cdot  \max_{\by}(a_{\by}\cdot |\Delta(\by)(c)|)\cdot  \lfloor n^{\epsilont_{4}} \rfloor
		\geq\\
		 \#_r\cdot \lfloor n^{\epsilont_{4}} \rfloor\cdot  \Delta(\bx_k)(c) -
		|M_{k}|\cdot  2\cdot h_\epsilon(n)\cdot (h_{\epsilont_1}(\lfloor \frac{n^{1-\epsilont_{2}}}{m} \rfloor))^{k-1}\cdot l_i \cdot  \max_{\by}(a_{\by}\cdot |\Delta(\by)(c)|)\cdot  \lfloor n^{\epsilont_{4}} \rfloor
		=\\
		 h_\epsilon(n)\cdot (h_{\epsilont_1}(\lfloor \frac{n^{1-\epsilont_{2}}}{m} \rfloor))^{k}\cdot \lfloor n^{\epsilont_{4}} \rfloor\cdot  \Delta(\bx_k)(c) \\-
		 |M_{k}|\cdot  2\cdot h_\epsilon(n)\cdot (h_{\epsilont_1}(\lfloor \frac{n^{1-\epsilont_{2}}}{m} \rfloor))^{k-1}\cdot l_i \cdot  \max_{\by}(a_{\by}\cdot |\Delta(\by)(c)|)\cdot  \lfloor n^{\epsilont_{4}} \rfloor
		 =\\
		 h_\epsilon(n)\cdot (h_{\epsilont_1}(\lfloor \frac{n^{1-\epsilont_{2}}}{m} \rfloor))^{k-1}
		 \cdot
		  \big( h_{\epsilont_1}(\lfloor \frac{n^{1-\epsilont_{2}}}{m} \rfloor)\cdot \lfloor n^{\epsilont_{4}} \rfloor\cdot  \Delta(\bx_k)(c)-|M_{k}|\cdot  2\cdot l_i \cdot  \max_{\by}(a_{\by}\cdot |\Delta(\by)(c)|)\cdot  \lfloor n^{\epsilont_{4}} \rfloor  \big)
		  \geq \\
		 h_\epsilon(n)\cdot (h_{\epsilont_1}(\lfloor \frac{n^{1-\epsilont_{2}}}{m} \rfloor))^{k-1}
		 \cdot
		  \big(h_{\epsilont_1}(\lfloor \frac{n^{1-\epsilont_{2}}}{m} \rfloor)\big)^{1-\epsilont_1}
		 \geq \\
		 h_\epsilon(n)\cdot (h_{\epsilont_1}(\lfloor \frac{n^{1-\epsilont_{2}}}{m} \rfloor))^{k-\epsilont_1}
	\end{gather*}
	
	Where the second to last inequality holds for all sufficiently large \(n\) from \(\lim_{n\rightarrow\infty}h_{\epsilont_1}(n)=\infty\). 
	
	Towards proving \(\lim_{n\rightarrow\infty} \prob_{p\vec{n}}^{\sigma_k^\epsilon}[\calC_\A[c]\geq (h_\epsilon(n))^{k+1} ]=1 \) for every counter \(c\) with \(\Delta(\bx_k)(c)>0 \) it now suffices to show that \((h_\epsilon(n))^{k+1}\leq h_\epsilon(n)\cdot (h_{\epsilont_1}(\lfloor \frac{n^{1-\epsilont_{2}}}{m} \rfloor))^{k-\epsilont_1}  \) for some properly chosen \(h_{\epsilont_1}\). It holds \((h_\epsilon(n))^{k+1}\leq h_\epsilon(n)\cdot (n^{1-\epsilon})^{k}=h_\epsilon(n)\cdot n^{k-k\cdot\epsilon} \), whereas for  \(h_{\epsilont_1}(n)=\lfloor n^{1-2\cdot \epsilont_1}\rfloor\) it holds  for all sufficiently large \(n\) that \begin{align*}
		h_\epsilon(n)\cdot (h_{\epsilont_1}(\lfloor \frac{n^{1-\epsilont_{2}}}{m} \rfloor))^{k-\epsilont_1}
	\geq
	 h_\epsilon(n)\cdot \big(\lfloor (\lfloor \frac{n^{1-\epsilont_{2}}}{m} \rfloor)^{1-2\cdot \epsilont_1}\rfloor\big)^{k-\epsilont_1}
	\geq\\
	 h_\epsilon(n)\cdot \big( ( n^{1-\epsilont_{2}} )^{1-2\cdot \epsilont_1}\big)^{k-2\cdot\epsilont_1}
	 =
	  h_\epsilon(n)\cdot n^{k-2\cdot\epsilont_1-(k-2\cdot\epsilont_1)\cdot2\cdot \epsilont_1-(k-2\cdot\epsilont_1)\cdot(1-2\cdot \epsilont_1)\cdot\epsilont_{2}} 
	\end{align*}
	Hence assuming \begin{equation}\label{eq-epsbound-jnhvgycukvhkb}
		-2\cdot\epsilont_1-(k-2\cdot\epsilont_1)\cdot2\cdot \epsilont_1-(k-2\cdot\epsilont_1)\cdot(1-2\cdot \epsilont_1)\cdot\epsilont_{2}>-k\cdot \epsilon
	\end{equation}
	it holds \(\lim_{n\rightarrow\infty} \prob_{p\vec{n}}^{\sigma_k^\epsilon}[\calC_\A[c]\geq (h_\epsilon(n))^{k+1} ]=1 \).
	
	This also proves \item \(\lim_{n\rightarrow\infty} \prob_{p\vec{n}}^{\sigma_k^\epsilon}[\calT_\A[t]\geq (h_\epsilon(n))^{k+1} ]=1 \) for every transition \(t\) with \(\bx_k(t)>0 \), as we can simply add a new \(t\)-transition-counter \(c_t\) to \(\A\) which is increased only by \(t\). Then it clearly holds \(\Delta(\bx_k)(c_t)>0 \) and \(\calT_\A[t]=\calC_\A[c_t]-n\).
	
	It remains to show there exist values for \(\epsilon_1,\epsilon_2,\dots \) that satisfy all of our assumptions. We do this in Table~\ref{Table-eps-section-another-proof-k+1}.

		\begin{table*}[h]
			\caption{Values of \(\epsilon_1,\epsilon_2,\dots \) for Section~\ref{section-proof-of-lemma-something-something-k+1-estiamtes}. Remember that \(0<\epsilon<1 \) and \(1\leq \variable\leq k\).}
			\centering
			%	\begin{center}
				\begin{tabular}{|l|| c c|} 
					\hline
					\(\epsilon\) assignment &  \multicolumn{2}{|c|}{restrictions}   \\ 
					\hline\hline
					%		\multirow{10}{*}{ \begin{tabular}{c}
							%			 \(\epsilon_1=\frac{9}{10} \)
							%				\\ \(\epsilon_3=\frac{1}{8} \)
							%				\\ \(\epsilon_{4}=\frac{1}{10} \)
							%				\\ \(\epsilon_{5}=\frac{5}{12} \)
							%				\\ \(\epsilon_{6}=\frac{1}{2} \)
							%				\\ \(\epsilon_7=\frac{1}{7} \)
							%		\end{tabular} }
					\(\epsilon_1=\nicefrac{\epsilon}{1000} \) & \(0<\epsilon_1,\epsilon_2,\dots  \) &   \\ 
					\hline
					\(\epsilon_2=\nicefrac{\epsilon}{2} \) & \(\epsilont_1<1 \) & \eqref{eq-epsbound-bvuycrvubh}  \\
					\hline
					\(\epsilon_{3}=\epsilon \) & 	\(\epsilont_3\leq k+1\) & \eqref{eq-epsbound-bhcftxfdsfsw}  \\
					\hline
					\(\epsilon_{4}=\nicefrac{\epsilon}{5000} \) & \(1-\epsilont_{2}>\epsilont_{4}\) & \eqref{eq-epsilonbound-jhjvgcgfxdysfsdfstjh}  \\
					\hline
					\(\epsilon_{5}=\nicefrac{\epsilon}{10000} \) & \(-\epsilont_{3} > -\epsilon-k\cdot \epsilont_{1}-k\cdot (1- \epsilont_{1})\cdot\epsilont_{2}+\epsilont_4+\epsilont_5\) & \eqref{eq-eps-bound-jhcyrxfdsfsdfstycuv}   \\ 
					\hline
					& \(\epsilont_4 > \epsilont_5 \) & \eqref{eq-eps-bound-jhcyrxtycuv}  \\
					\cline{2-3}
					& \(	-\epsilont_2>-\epsilon-(\variable-1)\cdot\epsilont_1-(\variable-1)\cdot\epsilont_2\cdot(1-\epsilont_1)+\epsilont_4 \) & \eqref{eq-epsbound-fsdhivuttrxtrx}  \\
					\cline{2-3}
					& \(	-2\cdot\epsilont_1-(k-2\cdot\epsilont_1)\cdot2\cdot \epsilont_1-(k-2\cdot\epsilont_1)\cdot(1-2\cdot \epsilont_1)\cdot\epsilont_{2}>-k\cdot \epsilon \) & \eqref{eq-epsbound-jnhvgycukvhkb}  \\
%					\cline{2-3}
					
					%				& \(2\epsilon_{5}-2\epsilon_{7}-\epsilon_{6}>0 \) & 18744 & 7560 \\
					%				\cline{2-4}
					%				& \(\epsilon_{1}>\epsilon \) & 18744 & 7560 \\
					%				\cline{2-4}
					%				& 		 \( 0<\epsilon<\frac{1}{10}\) & 18744 & 7560 \\
					\hline\hline
					&\multicolumn{2}{|c|}{after substitution} \\
					\hline\hline
					& \(\nicefrac{\epsilon}{1000}<1 \) & \eqref{eq-epsbound-bvuycrvubh} \\
					\cline{2-3}
					& \(\nicefrac{\epsilon}{5000}\leq k+1\) & \eqref{eq-epsbound-bhcftxfdsfsw} \\
					\cline{2-3}
					& \(1-\nicefrac{\epsilon}{2}>\nicefrac{\epsilon}{5000}\) & \eqref{eq-epsilonbound-jhjvgcgfxdysfsdfstjh} \\
					\cline{2-3}
					& \(-\epsilon > -\epsilon-k\cdot \nicefrac{\epsilon}{1000}-k\cdot (1- \nicefrac{\epsilon}{1000})\cdot\nicefrac{\epsilon}{2}+\nicefrac{\epsilon}{5000}+\nicefrac{\epsilon}{10000} \) & \eqref{eq-eps-bound-jhcyrxfdsfsdfstycuv} \\
					\cline{2-3}
					& \(\nicefrac{\epsilon}{5000} > \nicefrac{\epsilon}{10000} \) & \eqref{eq-eps-bound-jhcyrxtycuv} \\
					\cline{2-3}
					& \(	-\nicefrac{\epsilon}{2}>-\epsilon-(\variable-1)\cdot\nicefrac{\epsilon}{1000}-(\variable-1)\cdot\nicefrac{\epsilon}{2}\cdot(1-\nicefrac{\epsilon}{1000})+\nicefrac{\epsilon}{5000}\) & \eqref{eq-epsbound-fsdhivuttrxtrx} \\
					\cline{2-3}
					& \(	-2\cdot\nicefrac{\epsilon}{1000}-(k-2\cdot\nicefrac{\epsilon}{1000})\cdot2\cdot \nicefrac{\epsilon}{1000}-(k-2\cdot\nicefrac{\epsilon}{1000})\cdot(1-2\cdot \nicefrac{\epsilon}{1000})\cdot\nicefrac{\epsilon}{2}>-k\cdot \epsilon\) & \eqref{eq-epsbound-jnhvgycukvhkb} \\
%					\cline{2-3}
					
					\hline

				\end{tabular}
				\par
				\label{Table-eps-section-another-proof-k+1}	
				%	\end{center}
		\end{table*}

%	\geq\\
%	(h_\epsilon(n))^{k+1} \lfloor n^{\epsilont_{4}} \rfloor \Delta(\bx_k)(c)
%	\geq 
%	(h_\epsilon(n))^{k+1}  n^{\epsilont_{4}-\epsilon_{334}} 
%	\geq 
%	(h_\epsilon(n))^{k+1}
	
\end{proof}

\section{Proof of point \ref{enum-main-4} of Lemma~\ref{lemma-main-for-all-k}}
\label{sec-pointingupperbound-app}
In this Section we prove the induction step for \(k\) for point \ref{enum-main-4} of Lemma~\ref{lemma-main-for-all-k}.
That is we want to prove for a given component \(\br\) of \(\A\) that either there exists a transition \(t\) with \(\br(t)>0 \) and with an upper asymptotic estimate of \(n^{k} \) for \(\calT_\A[t] \), or \(\calP_{\A_{+\hat\br}}[\M_{\hat{\br}}] \) has a lower asymptotic estimate of \(n^{k+1}\).

Given a counters vector \(\bv\) and subsets of counters \(C_1,\dots, C_m \) we use \(\bv^{C_1,\dots,C_m}=[\bu_1,\dots,\bu_m]^{C_1,\dots, C_m} \) to denote the vector \(\bv \) restricted to the counters from \(C_1\cup\dots\cup C_m \) such that  \(\bu_i \) is the vector \(\bv\) restricted to the counters from \(C_i \).  (i.e., for each \(1\leq i\leq m \), for each \(c\in C_i\) we put \(\bu_i(c)=\bv(c)\) and \(\bu_i \) is a \(|C_i| \)-dimensional vector. Also we treat \([\bu_1,\dots,\bu_m]^{C_1,\dots, C_m} \) as a \((|C_1|+\dots+|C_m|) \)-dimensional vector such that for each \(c\in C_i\) it holds \([\bu_1,\dots,\bu_m]^{C_1,\dots, C_m}(c)=\bv^{C_1,\dots,C_m}(c)=\bu_i(c)\))

We denote by \(\support^{C_1,\dots,C_m}(\bx)=\{\bv^{C_1,\dots,C_m}\mid \bv \in \support(\bx) \} \)  the set of vectors \(\support(\bx) \) restricted to the counters from \(C_1\cup \dots \cup C_m \).

%\item \label{enum-main-4} For each component \(\bx\) of \(\A\), either there exists a transition \(t\) with \(\bx(t)>0 \) and with an upper asymptotic estimate of \(n^{k} \) for \(\calT_\A[t] \), or \(\calP_{\A_{+\hat\bx}}[\M_{\hat{\bx}}] \) has a lower asymptotic estimate of \(n^{k+1}\);

%(notice that this is not neccesarily the same as the MECs of \(\A_{i,T_j} \)).  
Given a   MEC \(\MEC\) of \(\A_\sigma \) for  \(\sigma\in\stratsMD{\A} \), we use \(X_{i,T_j}^{B,l} \) to denote the set of all multi-components \(\bx\) on \(\A_{i,T_j} \) such that \(\bx(t)=0\) for all transitions \(t\) that are not contained in the same MEC of \(\A_l \) as \(\MEC\). 

Given a   MEC \(\MEC\) of \(\A_\sigma \) for  \(\sigma\in\stratsMD{\A} \), we define
\begin{multline*}
	R^{B,k}=\big\{[\bv_1,\dots,\bv_{\lfloor \frac{k}{2}\rfloor}]^{C_1,\dots,C_{\lfloor \frac{k}{2}\rfloor}} \mid
	\exists [\bv_1^1,\dots,\bv^1_{\lfloor \frac{k}{2}\rfloor}]^{C_1,\dots,C_{\lfloor \frac{k}{2}\rfloor}}  \in \Delta^{C_1,\dots,C_{\lfloor \frac{k}{2}\rfloor}}(X^{B,k-2+1}_{k-1,T_k}), 
	\\
	[\bv_1^2,\dots,\bv^2_{\lfloor \frac{k}{2}\rfloor}]^{C_1,\dots,C_{\lfloor \frac{k}{2}\rfloor}}  \in \Delta^{C_1,\dots,C_{\lfloor \frac{k}{2}\rfloor}}(X^{B,k-2\cdot 2+1}_{k-2,T_{k-1}}),\dots
	\\	
	\dots,[\bv_1^{\lfloor \frac{k}{2}\rfloor},\dots,\bv^{\lfloor \frac{k}{2}\rfloor}_{\lfloor \frac{k}{2}\rfloor}]^{C_1,\dots,C_{\lfloor \frac{k}{2}\rfloor}}  \in \Delta^{C_1,\dots,C_{\lfloor \frac{k}{2}\rfloor}}(X^{B,k-2\cdot {\lfloor \frac{k}{2}\rfloor}+1}_{k-{\lfloor \frac{k}{2}\rfloor},T_{k-\lfloor \frac{k}{2}\rfloor+1}});
	\\
	\forall j\in \{1,2,\dots,\lfloor \frac{k}{2}\rfloor \}: \sum_{i=1}^{{\lfloor \frac{k}{2}\rfloor}} \bv_j^i=\bv_j; \forall m<j: \bv_m^j=\vec{0}  \big\} 
\end{multline*}
where we assume that for each \(i\) the local copies of each counter  \(c\in C_i \), for each transition \(t\in T_{k-i+1}\) the local copy of \(c\) used by \(t\) in \(\A_{k- i,T_{k-i+1}} \)  is the same as the local copy of \(c \) used by \(t\) in \(\A_{k- j,T_{k-j+1}} \) for for each \(j\geq i \) (note that as \(j\) increases so does the number of transitions in \(\A_{k- j,T_{k-j+1}} \) and the number of local copies in \(\A_{k- j,T_{k-j+1}} \) decreases as \(j \) increases. But every multi-component from \(X^{B,k-2\cdot j+1}_{k-j,T_{k-j+1}} \) ever modifies only a single local copy of each counter \(c\in C_i \) for \(i\geq j \).) 

We use \(\MEC_\br \) to denote the MEC corresponding to \(\br \) in \(\A_\br \).

The induction step for \(k\) for point \ref{enum-main-4} of Lemma~\ref{lemma-main-for-all-k} follows from the following Lemma.

\begin{lemma}
	If for each \(1\leq l \leq \frac{k}{2} \) it holds \(\support^{C_1,\dots,C_l}(\hat{\br})  \subseteq \Delta^{C_1,\dots,C_l}(X_{k-l,T_{k-l+1}})\) then  then \(\calP_{\A_{+\hat{\br}}}[\M_{\hat\br}] \) has a lower asymptotic estimate of \(n^{k+1} \).
	
	If on the other hand there exists \(1\leq l \leq \frac{k}{2} \) such that \(\support^{C_1,\dots,C_l} (\hat{\br}) \nsubseteq \Delta^{C_1,\dots,C_l}(X_{k-l,T_{k-l+1}})\) then there exists a transition \(t\) with \(\br(t)>0 \) and an upper asymptotic estimate of \(n^{k} \) for \(\calT_\A[t] \).
\end{lemma}
\begin{proof}
	In the first case we get from Lemma~\ref{lemma-D-in-Rl-for-all-l-imply-D-in-R} that \(\support^{C_1,\dots,C_{\lfloor\frac{k}{2}\rfloor}}(\hat{\br})\subseteq R^{B_\br,k}  \) and therefore from Lemma~\ref{lemma-D-in-R-lower-estimate-nk+1-for-hatB} we get that \(\calP_{\A_{+\hat{\br}}}[\M_{\hat\br}] \) has a lower asymptotic estimate of \(n^{k+1} \).
	
	Whereas in the second case we get from Lemma~\ref{lemma-D-not-in-Rl-implies-B-zero-unbounded-rankl-dsada} that \(\hat\br \) is not zero-bounded on \(rank_{k-l,T_{k-l+1}}^{C_1,\dots,C_l} \) which then gives us from Lemma~\ref{lemma-hatB-zero-unbounded-rankl-give-upper-estimate-nk} that there exists a transition \(t\) with \(\br(t)>0 \) such that  \(\calT_\A[t] \) has an upper asymptotic estimate of \(n^{k} \).
\end{proof}

\subsection{Lower Asymptotic Estimate \(n^{k+1} \) for \(\calP_{\A_{+\hat{\br}}}[\M_{\hat{\br}}] \)}\label{app-sec-lowerest-nkplus1-hatby}

In this Section we prove the following Lemma.

\begin{lemma}\label{lemma-D-in-R-lower-estimate-nk+1-for-hatB}
	Let \(\br \) be a component of \(\A\). If \(\support^{C_1,\dots,C_{\lfloor \frac{k}{2}\rfloor}}(\hat\br)\subseteq R^{B_\br,k} \) then \(\calP_{\A_{+\hat{\br}}}[\M_{\hat{\br}}] \) has a lower asymptotic estimate of \(n^{k+1} \). 
\end{lemma}

Let us begin with the following technical Lemmas.

\begin{lemma}\label{lemma-v-in-XBiTi+1-implies-minus-v-also-in-XBiTi+1sss}
	Let \(1\leq i\leq \lfloor\frac{k}{2}\rfloor \) and let \(\bx\in X_{k-i,T_{k-i+1}} \), then there exists \(\by\in X_{k-i,T_{k-i+1}}  \) such that \(\Delta^{C_1,\dots,C_{k-i}}(\bx)=-\Delta^{C_1,\dots,C_{k-i}}(\by) \).
\end{lemma}
\begin{proof}
	From the induction assumption for \(k-i \) on point \ref{enum-main-3} of Lemma~\ref{lemma-main-for-all-k} there exists a multi-component \(\bz \) on \(\A_{k-i} \) such that \(\bz(t)>0 \) iff \(t\in T_{k-i+1} \), and \(\Delta^{C_1,\dots,C_{k-i}}(\bz )=\vec{0} \). 
	
	%	Let \(\bz' \) be the multi-component created from \(\bz \) by setting \(\bz'(t)=0 \) for each \(t\) contained in different MEC than \(B\) in \(\A_{k-2i+1} \). Note that for each counter \(c\in C_1\cup \dots\cup C_i \) there are different copies of \(c\) in each MEC of \(\A_{k-2i+1} \) in \(\A_{k-i} \), therefore it holds \(\Delta^{C_1,\dots,C_i}(\bz')=\Delta^{C_1,\dots,C_i}(\bz)=\vec{0} \).
	
	As \(A_{k-i,T_{k-i+1}}\) contains only transitions from \(T_{k-i+1} \) it thus holds that \(\bx(t)>0 \) implies \(\bz(t)>0 \), and thus there exists \(a>0\) such that \(\by=a\cdot \bz-\bx\geq \vec{0} \), thus from  Lemmas~\ref{lemma-multiplication-multicomponents}~and~\ref{lemma-subtraction-multicomponents} \(\by \) is a multi-component of \(\A_{k-i,T_{k-i+1}} \). Furthermore, it holds \(\Delta^{C_1,\dots,C_{k-i}}(\by)=a\cdot \Delta^{C_1,\dots,C_{k-i}}(\bz)-\Delta^{C_1,\dots,C_{k-i}}(\bx)=a\cdot \vec{0}-\Delta^{C_1,\dots,C_{k-i}}(\bx)=-\Delta^{C_1,\dots,C_{k-i}}(\bx)  \). \end{proof}

For each \(1\leq i\leq \lfloor\frac{k}{2} \rfloor  \) let  \(Y_{k-i,T_{k-i+1}}\subseteq X_{k-i,T_{k-i+1}} \) be such that \(\bx\in Y_{k-i,T_{k-i+1}}\) if  \(\Delta^{C_1,\dots,C_{i-1}}(\bx)= \vec{0} \).
We have the following:
%\[\Delta^{C_1,\dots,C_{k-i}}(\by)= [\vec{0},\dots,\vec{0},\bv_{i},\dots,\bv_{k-i}]^{C_1,\dots,C_{k-i}} \] for some \(\bv_{i},\dots,\bv_{k-i} \).
\begin{itemize}
	\item  \(\Delta^{C_1,\dots,C_{k-i}}(Y_{k-i,T_{k-i+1}}) \) is closed under multiplication by \(-1 \) from  Lemma~\ref{lemma-v-in-XBiTi+1-implies-minus-v-also-in-XBiTi+1sss} since \(i\leq k-i \);
	\item  \(\Delta^{C_1,\dots,C_{k-i}}(Y_{k-i,T_{k-i+1}}) \) is closed under addition from  Lemma~\ref{lemma-addition-multicomponents};
	\item  \(\Delta^{C_1,\dots,C_{k-i}}(Y_{k-i,T_{k-i+1}}) \) is closed under multiplication by non-negative constant from  Lemma~\ref{lemma-multiplication-multicomponents}.
\end{itemize} Therefore \(\Delta^{C_1,\dots,C_{k-i}}(Y_{k-i,T_{k-i+1}}) \) is a vector space and there exists an orthogonal basis \(BASIS_i \) of \(\Delta^{C_1,\dots,C_{k-i}}(Y_{k-i,T_{k-i+1}}) \). Furthermore, for each \(\bu\in BASIS_i \) there exist multi-components \(\bx_{i,\bu}^{+},\bx_{i,\bu}^- \in Y_{k-i,T_{k-i+1}}\) such that \(\Delta^{C_1,\dots,C_{k-i}}(\bx_{i,\bu}^+)=\bu \) and   \(\Delta^{C_1,\dots,C_{k-i}}(\bx_{i,\bu}^-)=-\bu \). Let  \(XBASIS_i=\{\bx_{i,\bu}^{+},\bx_{i,\bu}^- \mid \bu\in BASIS_i \}  \) and let \(XBASIS=\bigcup_{i=1}^{\lfloor \frac{k}{2}\rfloor }  XBASIS_i \)	 and
\(BASIS=\bigcup_{i=1}^{\lfloor \frac{k}{2}\rfloor }  \Delta^{C_1,\dots,C_{\lfloor \frac{k}{2}\rfloor}}(XBASIS_i) \). Note that we can wlog. assume that for each \(\bu\in BASIS\) there exists exactly one \(1\leq i\leq \lfloor \frac{k}{2}\rfloor \) such that \(\bu\in BASIS_i^{C_1, \dots, C_{\lfloor\frac{k}{2}\rfloor} }\cup (-BASIS_i^{C_1, \dots, C_{\lfloor\frac{k}{2}\rfloor} }) \), as if there exists multiple such indexes we can simply replace one of the sets \(BASIS_i \) with a multiple of itself.
as a conical sum  \[ 	[\bv_1^i,\dots,\bv^i_{\lfloor \frac{k}{2}\rfloor}]^{C_1,\dots,C_{\lfloor \frac{k}{2}\rfloor}} = \sum_{\by\in Z_i } b_{i,\by}\cdot \Delta^{C_1,\dots,C_{\lfloor \frac{k}{2}\rfloor}}(\by) \]

\begin{lemma}\label{lemma-bvuiuobvuiguolbj}
	Let \(\bv=[\bv_1,\dots,\bv_{\lfloor \frac{k}{2}\rfloor}]^{C_1,\dots,C_{\lfloor \frac{k}{2}\rfloor}}\in R^{B_\br,k} \), then \(\bv \) can be expressed as a linear sum of elements from \(BASIS \).
\end{lemma}
\begin{proof}
	From the definition of \(  R^{B_\br,k}\) we can write \[[\bv_1,\dots,\bv_{\lfloor \frac{k}{2}\rfloor}]^{C_1,\dots,C_{\lfloor \frac{k}{2}\rfloor}}=\sum_{i=1}^{\lfloor\frac{k}{2}\rfloor} [\bv_1^i,\dots,\bv^i_{\lfloor \frac{k}{2}\rfloor}]^{C_1,\dots,C_{\lfloor \frac{k}{2}\rfloor}} \] where \( [\vec{0},\dots,\vec{0},\bv^i_{i},\dots,\bv^i_{\lfloor \frac{k}{2}\rfloor}]^{C_1,\dots,C_{\lfloor \frac{k}{2}\rfloor}}=[\bv_1^i,\dots,\bv^i_{\lfloor \frac{k}{2}\rfloor}]^{C_1,\dots,C_{\lfloor \frac{k}{2}\rfloor}}\in \Delta^{C_1,\dots,C_{\lfloor \frac{k}{2}\rfloor}}(X^{B,k-2\cdot i+1}_{k-i,T_{k-i+1}}) \). Therefore for each \(1\leq i\leq \lfloor \frac{k}{2}\rfloor\) there exists a multi-component \(\bx_i\in X^{B,k-2\cdot i+1}_{k-i,T_{k-i+1}} \) such that \([\bv_1^i,\dots,\bv^i_{\lfloor \frac{k}{2}\rfloor}]^{C_1,\dots,C_{\lfloor \frac{k}{2}\rfloor}}= \Delta^{C_1,\dots,C_{\lfloor \frac{k}{2}\rfloor}}(x_i) \). Since \(X^{B,k-2\cdot i+1}_{k-i,T_{k-i+1}}\subseteq X_{k-i,T_{k-i+1}} \) it holds \(\bx_i\in Y_{k-i,T_{k-i+1}}  \). Hence we can write \[[\bv_1^i,\dots,\bv^i_{\lfloor \frac{k}{2}\rfloor}]^{C_1,\dots,C_{\lfloor \frac{k}{2}\rfloor}}=\sum_{\bu\in BASIS_i} a_{i,\bu}\cdot \bu^{C_1,\dots,C_{\lfloor \frac{k}{2}\rfloor}}  \] for some constants \(a_{i,\bu}\).  Hence we can express \(\bv\) as a conical sum of elements from \(BASIS \)  as follows:
	\begin{gather*}
		[\bv_1,\dots,\bv_{\lfloor \frac{k}{2}\rfloor}]^{C_1,\dots,C_{\lfloor \frac{k}{2}\rfloor}}
		=
		\sum_{i=1}^{\lfloor\frac{k}{2}\rfloor} [\bv_1^i,\dots,\bv^i_{\lfloor \frac{k}{2}\rfloor}]^{C_1,\dots,C_{\lfloor \frac{k}{2}\rfloor}} 
		=
		\sum_{i=1}^{\lfloor\frac{k}{2}\rfloor} \sum_{\bu\in BASIS_i} a_{i,\bu}\cdot \bu^{C_1,\dots,C_{\lfloor \frac{k}{2}\rfloor}}
		=\\
		\sum_{i=1}^{\lfloor\frac{k}{2}\rfloor} \big( \sum_{\bu\in BASIS_i; a_{i,\bu}\geq 0} a_{i,\bu}\cdot \bu^{C_1,\dots,C_{\lfloor \frac{k}{2}\rfloor}} + \sum_{\bu\in BASIS_i; a_{i,\bu}< 0} (-a_{i,\bu})\cdot (-\bu^{C_1,\dots,C_{\lfloor \frac{k}{2}\rfloor}}) \big)
	\end{gather*}
	%Notice that \(BASIS \) is closed under multiplication by \(-1\).
\end{proof}

\begin{lemma}\label{lemma-Rbk-vector-space}
	\(R^{B_\br,k}\) is a vector space.
\end{lemma}
\begin{proof}
	From  Lemma~\ref{lemma-addition-multicomponents}	\(R^{B_\br,k}\) is closed under addition and from  Lemma~\ref{lemma-multiplication-multicomponents}	\(R^{B_\br,k}\) is closed under multiplication by non-negative constant. It thus suffices to show that \(R^{B_\br,k} \) is also closed under multiplication by \(-1\).
	
	Let \(\bv=[\bv_1,\dots,\bv_{\lfloor \frac{k}{2}\rfloor}]^{C_1,\dots,C_{\lfloor \frac{k}{2}\rfloor}}\in R^{B_\br,k} \), we want to show that also \(-\bv=[-\bv_1,\dots,-\bv_{\lfloor \frac{k}{2}\rfloor}]^{C_1,\dots,C_{\lfloor \frac{k}{2}\rfloor}}\in R^{B_\br,k} \).

	From the definition of \(  R^{B_\br,k}\) we can write \[[\bv_1,\dots,\bv_{\lfloor \frac{k}{2}\rfloor}]^{C_1,\dots,C_{\lfloor \frac{k}{2}\rfloor}}=\sum_{i=1}^{\lfloor\frac{k}{2}\rfloor} [\bv_1^i,\dots,\bv^i_{\lfloor \frac{k}{2}\rfloor}]^{C_1,\dots,C_{\lfloor \frac{k}{2}\rfloor}} \] where \( [\vec{0},\dots,\vec{0},\bv^i_{i},\dots,\bv^i_{\lfloor \frac{k}{2}\rfloor}]^{C_1,\dots,C_{\lfloor \frac{k}{2}\rfloor}}=[\bv_1^i,\dots,\bv^i_{\lfloor \frac{k}{2}\rfloor}]^{C_1,\dots,C_{\lfloor \frac{k}{2}\rfloor}}\in \Delta^{C_1,\dots,C_{\lfloor \frac{k}{2}\rfloor}}(X^{B,k-2\cdot i+1}_{k-i,T_{k-i+1}}) \). Therefore for each \(1\leq i\leq \lfloor \frac{k}{2}\rfloor\) there exists a multi-component \(\bx_i\in X^{B,k-2\cdot i+1}_{k-i,T_{k-i+1}} \) such that \([\bv_1^i,\dots,\bv^i_{\lfloor \frac{k}{2}\rfloor}]^{C_1,\dots,C_{\lfloor \frac{k}{2}\rfloor}}= \Delta^{C_1,\dots,C_{\lfloor \frac{k}{2}\rfloor}}(x_i) \). Since \(X^{B,k-2\cdot i+1}_{k-i,T_{k-i+1}}\subseteq X_{k-i,T_{k-i+1}} \) it holds \(\bx_i\in Y_{k-i,T_{k-i+1}}  \). Since \(Y_{k-i,T_{k-i+1}}\) is closed under multiplication by \(-1\) there exists component \(\bx_i^-\in Y_{k-i,T_{k-i+1}}  \) such that \([-\bv_1^i,\dots,-\bv^i_{\lfloor \frac{k}{2}\rfloor}]^{C_1,\dots,C_{\lfloor \frac{k}{2}\rfloor}}=-\Delta^{C_1,\dots,C_{\lfloor \frac{k}{2}\rfloor}}(\bx_i)= \Delta^{C_1,\dots,C_{\lfloor \frac{k}{2}\rfloor}}(\bx_i^-) \)
	
	For each \(i \) let us define multi-components \(\by_i,\by_i^{\MEC_\br}, \bz_i \) as follows:
	
	\begin{itemize}
		\item for \(i=1\) we put: \(\by_1=\bx_1^- \), \(\by_1^{\MEC_\br}(t)=\by_1(t) \) if \(t\) is contained in the same MEC of \(\A_{k-1} \) as \(\MEC_\br\) and \(\by_1^{\MEC_\br}(t)=0 \) otherwise, and \(\bz_1=\by_1-\by_1^{\MEC_\br} \).
		\item for \(i>1\) we put: \(\by_i=\bx_i^- + \bz_{i-1} \),  \(\by_i^{\MEC_\br}(t)=\by_i(t) \) if \(t\) is contained in the same MEC of \(\A_{k-2\cdot i+1} \) as \(\MEC_\br\)  and \(\by_i^{\MEC_\br}(t)=0 \) otherwise, and  \(\bz_i=\by_i-\by_i^{\MEC_\br} \).
	\end{itemize}
	Note that \(\by_i^{\MEC_\br}\in X^{B,k-2\cdot i+1}_{k-i,T_{k-i+1}}\) for each \(1\leq i\leq \lfloor \frac{k}{2} \rfloor\).

	Let \(c\in C_j\) where \(1\leq j\leq \lfloor \frac{k}{2}\rfloor \), we will now show that \(\Delta(\bz_{\lfloor \frac{k}{2}\rfloor})(c)=0 \).
	
	If \(c \) is the local copy of a counter used in the MEC of \(\A_{k-j} \) containing \(\MEC_\br \) then \(\Delta(\bz_i)(c)=0 \) for each \(j\leq i\leq \lfloor \frac{k}{2}\rfloor \) since every transition of \(\A_{k-i} \) which modifies \(c \) and is included in \(\by_i \) is by the definition of \(\by^{\MEC_\br}_i \) included in \(\by^{\MEC_\br}_i \), and thus these transitions can no longer be included in \(\bz_i \). Hence in this case \(\Delta(\bz_{\lfloor\frac{k}{2}\rfloor})(c)=0 \)
	
	If \(c \) is a not the local copy of a counter used in the MEC of \(\A_{k-j} \) containing \(\MEC_\br \)  \(\Delta(\by)(c)=0 \) for each multi-component \(\by\in X^{B,k-2\cdot i+1}_{k-i,T_{k-i+1}} \) where \(1\leq i\leq j \) (this is because none of the transitions modify \(c\) as they modify a different local copy), and thus from the definition of \(  R^{B_\br,k} \) it holds \(\bv_i(c)=0 \) for all \(1\leq i\leq \lfloor \frac{k}{2}\rfloor \). Thus we can apply inducition on \(i\) to prove \(\bz_i(c)=0 \) for all \(1\leq i\leq \lfloor \frac{k}{2}\rfloor \). Base case \(i=1\) we obtain straightforward \(\Delta(\bz_1)(c)=\Delta(\by_1)(c)-\Delta(\by^{\MEC_\br}_1)(c)=\bv_1(c)-0=0 \). Let \(\kappa \) be such that \(c\) is the local copy used in the MEC of \(\A_{k-\kappa} \) containing \(\MEC_\br \). For the induction step for \(j<i< \kappa \)   we obtain  \(\Delta(\bz_i)(c)=\Delta(\by_i)(c)-\Delta(\by^{\MEC_\br}_i)(c)=\Delta(\bx_i^-)(c)+\Delta(\bz_{i-1})(c)-0=\bv_i(c)-0=0 \). And for the induction step for \(\kappa \leq i\leq \lfloor\frac{k}{2}\rfloor \) since every transition of \(\A_{k-i} \) which modifies \(c \) and is included in \(\by_i \) is by the definition of \(\by^{\MEC_\br}_i \) included in \(\by^{\MEC_\br}_i \), and thus these transitions can no longer be included in \(\bz_i \).
	Hence in this case \(\Delta(\bz_{\lfloor\frac{k}{2}\rfloor})(c)=0 \) as well.
	
	Therefore \(\Delta^{C_1,\dots,C_{\lfloor \frac{k}{2}\rfloor}}(\bz_{\lfloor \frac{k}{2}\rfloor})=\vec{0}\) and so it holds \begin{gather*}
		-\bv=\sum_{i=1}^{\lfloor \frac{k}{2} \rfloor} \Delta^{C_1,\dots,C_{\lfloor \frac{k}{2}\rfloor}}(\bx_i^-)
		=\\
		\Delta^{C_1,\dots,C_{\lfloor \frac{k}{2}\rfloor}}(\by_1)+  
		\sum_{i=2}^{\lfloor \frac{k}{2} \rfloor} \Delta^{C_1,\dots,C_{\lfloor \frac{k}{2}\rfloor}}(\by_i)-\sum_{i=2}^{\lfloor \frac{k}{2} \rfloor}\Delta^{C_1,\dots,C_{\lfloor \frac{k}{2}\rfloor}}(\bz_{i-1})
		=\\
		\sum_{i=1}^{\lfloor \frac{k}{2} \rfloor} \big(\Delta^{C_1,\dots,C_{\lfloor \frac{k}{2}\rfloor}}(\by_i^{\MEC_\br}) +\Delta^{C_1,\dots,C_{\lfloor \frac{k}{2}\rfloor}}(\bz_i) \big)
		-
		\sum_{i=2}^{\lfloor \frac{k}{2} \rfloor}\Delta^{C_1,\dots,C_{\lfloor \frac{k}{2}\rfloor}}(\bz_{i-1})
		=\\
		\sum_{i=1}^{\lfloor \frac{k}{2} \rfloor} \Delta^{C_1,\dots,C_{\lfloor \frac{k}{2}\rfloor}}(\by_i^{\MEC_\br}) +\Delta^{C_1,\dots,C_{\lfloor \frac{k}{2}\rfloor}}(\bz_{\lfloor \frac{k}{2}\rfloor})
		=
		\sum_{i=1}^{\lfloor \frac{k}{2} \rfloor} \Delta^{C_1,\dots,C_{\lfloor \frac{k}{2}\rfloor}}(\by_i^{\MEC_\br})
	\end{gather*} 
	but since \(\by_i^{\MEC_\br}\in X^{B,k-2\cdot i+1}_{k-i,T_{k-i+1}}\) for each \(1\leq i\leq \lfloor \frac{k}{2} \rfloor\), this implies that \(-\bv\in R^{B_\br,k} \).
\end{proof}

% \(BASIS\) is closed under multiplication by \(-1\)., it holds that 
Since \(R^{B_\br,k}\) is a vector space form Lemma~\ref{lemma-Rbk-vector-space} there exists an orthogonal basis \(BASIS_\bot \) of \(R^{B_\br,k}\). Since \(BASIS_\bot\subseteq R^{B_\br,k} \), from Lemma~\ref{lemma-bvuiuobvuiguolbj} we obtain that for each \(\bv\in BASIS_\bot \) and each \(\bu\in BASIS \) there exist constants \(b_{\bv,\bu} \) (not necessarily unique, but we shall fix some such constants thorough this Section) such that for any \(\bv\in BASIS_\bot \) it holds \(\sum_{\bu\in BASIS} b_{\bv,\bu}\cdot \bu = \bv \). This allows us to uniquely (for our fixed constants \(b_{\bv,\bu} \)) express each \(\bs\in R^{B_\br,k} \) as a linear sum \(\bs=\sum_{\bu\in BASIS} b_{\bs,\bu}\cdot \bu \) for some constants \(b_{\bs,\bu} \). Furthermore for each \(\bu\in BASIS \) there exists a linear transformation \(f_\bu \) such that for each \(\bs\in R^{B_\br,k} \) it holds  \(f_\bu(\bs)=b_{\bs,\bu} \).\footnote{This follows from \(\bs \) being expressible as a linear combination of elements from \(BASIS_\bot \), and each \(BASIS_\bot \) is expressible as a linear combination of elements from \(BASIS \) using our fixed constants \(b_{\bv,\bu} \).}

Let us define technical constants \(0<\epsilonr_1,\epsilonr_2,\dots  \). As their exact values are not important we leave the assignment of their exact values to Table~\ref{Tadsdsble-eps-section-anotdfdsfsdfdher-proof-k+1}, where we also show that our assignment satisfies all the assumptions we make on \(\epsilonr_{1},\epsilonr_{2},\dots\).

\begin{lemma}\label{lemma-rand-walk-double-exp-zero}
	Let \(\alpha_{\hat{\br}} \) be the random variable denoting the computation on \(\M_{\hat{\br}} \) initiated in the configuration \(p_{\br}\vec{0} \) under the {\cMD} strategy \(\sigma \) corresponding to \(\br\) until \(p_{\br}\) is revisited for the first time.
	Let \[\realeffect_{\hat{\br}}^{\epsilonr_{1}}=\begin{cases}
		\realeffect_{\hat{\br}} & \text{if }  \length(\alpha_{\hat{\br}})\leq n^{\epsilonr_{1}}  \\
		0 & \textit{else}
	\end{cases} 
	\]
	Then for each \(\bu\in BASIS \) it holds   \(\E_{p_{\br}\vec{0}}^\sigma[f_\bu(\realeffect_{\hat{\br}})]=0 \) as well as \(\big|\E_{p_{\br}\vec{0}}^\sigma[f_\bu(\realeffect_{\hat{\br}}^{\epsilonr_{1}})]\big|\in \bigO(a^{ n^{\epsilonr_{2}}}) \) for some \(a<1 \).
\end{lemma}
%\E_{p\vec{n}}^\sigma[f_\bu( \Delta^{C_1,\dots,C_{\lfloor\frac{k}{2}\rfloor}}(\alpha_{\hat{\br}}) )]=
\begin{proof}
	
	%	Assuming \begin{equation}\label{eq-epsbound-hgyciuvlhgculyvivyutvy}
		%		\epsilonr_{1}>\epsilonr_{2}
		%	\end{equation} \todo{put at proper spot}
	
	Since \(f_\bu \) is a linear transformation and \(\E_{p_{\br}\vec{0}}^\sigma[\realeffect_{\hat{\br}}]=\vec{0} \) (see definition of \(\hat{\br} \) in Section~\ref{sec-prelim}) it holds from the linearity of expectation that  \(\E_{p_{\br}\vec{0}}^\sigma[f_\bu(\realeffect_{\hat{\br}})]=0 \).
	We can write
	\begin{gather*}
		0=\E_{p_{\br}\vec{0}}^\sigma[f_\bu(\realeffect_{\hat{\br}})]
		=\\
		\E_{p_{\br}\vec{0}}^\sigma[f_\bu(\realeffect_{\hat{\br}})\mid \length(\alpha_{\hat{\br}})\leq n^{\epsilonr_{1}}]\cdot \prob_{p_{\br}\vec{0}}^\sigma[ \length(\alpha_{\hat{\br}})\leq n^{\epsilonr_{1}}]
		+
		\E_{p_{\br}\vec{0}}^\sigma[f_\bu(\realeffect_{\hat{\br}})\mid \length(\alpha_{\hat{\br}})> n^{\epsilonr_{1}}]\cdot \prob_{p_{\br}\vec{0}}^\sigma[ \length(\alpha_{\hat{\br}})> n^{\epsilonr_{1}}]
		=\\
		\E_{p_{\br}\vec{0}}^\sigma[f_\bu(\realeffect^{\epsilonr_{18}}_{\hat{\br}})\mid \length(\alpha_{\hat{\br}})\leq n^{\epsilonr_{1}}]\cdot \prob_{p_{\br}\vec{0}}^\sigma[ \length(\alpha_{\hat{\br}})\leq n^{\epsilonr_{1}}]
		+
		\E_{p_{\br}\vec{0}}^\sigma[f_\bu(\realeffect_{\hat{\br}})\mid \length(\alpha_{\hat{\br}})> n^{\epsilonr_{1}}]\cdot \prob_{p_{\br}\vec{0}}^\sigma[ \length(\alpha_{\hat{\br}})> n^{\epsilonr_{1}}]
	\end{gather*}
	
	Thus it suffices to show that \[\big|\frac{\E_{p_{\br}\vec{0}}^\sigma[f_\bu(\realeffect_{\hat{\br}})\mid \length(\alpha_{\hat{\br}})> n^{\epsilonr_{1}}]\cdot \prob_{p_{\br}\vec{0}}^\sigma[ \length(\alpha_{\hat{\br}})> n^{\epsilonr_{1}}]}{\prob_{p_{\br}\vec{0}}^\sigma[ \length(\alpha_{\hat{\br}})\leq n^{\epsilonr_{1}}]}\big| \in \bigO(a^{ n^{\epsilonr_{2}}}) \]

	For each for each \(i\in \mathbb{N}\) it holds  \(|(f_\bu(E_{\hat{\br}})\mid\length(\alpha)\leq i)|\leq  u\cdot i\) where \(u\) is the maximal change of a counter per single transition in \(\M_{\hat{\br}} \). Also it holds  \(\prob_{p_\br,\vec{0}}^\sigma[\len(\alpha)\geq i]\leq a^i \) for some constant \(a<1\) since every step there is a positive probability of reaching  \(p_\br\) within at most constantly many steps. Therefore \begin{gather*}
		\big|\E_{p_{\br}\vec{0}}^\sigma[f_\bu(\realeffect_{\hat{\br}})\mid \length(\alpha_{\hat{\br}})> n^{\epsilonr_{1}})]\big|
		\leq 
		\sum_{i=\lfloor n^{\epsilonr_{1}}\rfloor+1}^{\infty}  
		\big|\E_{p_{\br}\vec{0}}^\sigma[f_\bu(\realeffect_{\hat{\br}})\mid \length(\alpha_{\hat{\br}})=i ]\big|\cdot \prob_{p_{\br}\vec{0}}^\sigma[\length(\alpha)=i]
		\leq \\
		\sum_{i=\lfloor n^{\epsilonr_{1}}\rfloor+1}^{\infty}  
		u\cdot i\cdot a^i
		=
		u \cdot  \frac{a^{\lfloor n^{\epsilonr_{1}}\rfloor+1}\cdot (-a\cdot (\lfloor n^{\epsilonr_{1}}\rfloor+1)+\lfloor n^{\epsilonr_{1}}\rfloor+1)}{(a-1)^2}
		=\\
		\frac{u}{(a-1)^2}\cdot  a^{\lfloor n^{\epsilonr_{1}}\rfloor+1}\cdot (-a\cdot (\lfloor n^{\epsilonr_{1}}\rfloor+1)+\lfloor n^{\epsilonr_{1}}\rfloor+1)
		\leq 
		a^{n^{\epsilonr_{3}}}
	\end{gather*}
%	\todo{\(\epsilonr_3 \) was \(\epsilon_{03} \)}
	Where the last inequality holds for all sufficiently large \(n\) assuming  \begin{equation}\label{eq-epsbound-kgvkvubbvgvgvgppp}
		\epsilonr_{1}>\epsilonr_{3}
	\end{equation}
	
	Therefore assuming \begin{equation}\label{eq-epsbound-kgfdsfvkvubbvgvgvgppp}
		\epsilonr_{2}<\epsilonr_{3}+\epsilonr_{1} 
	\end{equation}  it holds \[\big|\frac{\E_{p_{\br}\vec{0}}^\sigma[f_\bu(\realeffect_{\hat{\br}})\mid \length(\alpha_{\hat{\br}})> n^{\epsilonr_{1}}]\cdot \prob_{p_{\br}\vec{0}}^\sigma[ \length(\alpha_{\hat{\br}})> n^{\epsilonr_{1}}]}{\prob_{p_{\br}\vec{0}}^\sigma[ \length(\alpha_{\hat{\br}})\leq n^{\epsilonr_{1}}]}\big| 
	\leq 
	\big|\frac{a^{n^{\epsilonr_{3}}}\cdot a^{n^{\epsilonr_{1}}}}{1-a^{n^{\epsilonr_1}}}\big|
	\in \bigO(a^{ n^{\epsilonr_{2}}}) \]
	Lemma holds.
\end{proof}

%	 Let \(X-BASIS=\bigcup_{i=1}^{\lfloor k/2\rfloor\rfloor } X_i \). Then since \(D_{1,\dots,k}\subseteq R^B \) it holds that each \(\bv\in D_{1,\dots,k} \) can be expressed as a conical combination of elements from \(\Delta^{{C_1,\dots,C_k}}(X-BASIS) \). 

%	We will now  describe a pointing strategy  that witnesses the lower estimate of \( n^{k+1}\) on \(\hat{B} \) in \(\A_{+\hat{B}} \).

%	For each  \(1\leq i\leq k \) we get from the induction assumption on point \ref{enum-main-4} of Lemma~\ref{lemma-main-for-all-k} for \(k-i \) that each component \(\by \) that contains only transitions from \(T_{i} \) has a lower asymptotic estimate of \(n^{i} \) for \(\calP_{\A_{+\hat{\by}}}[\M_{\hat{\by}}] \). For each \(\by\) that contains only transitions from \(T_i \) let \(k_\by \) be the highest value such that \(k_\by\leq k \) and there exists a transition \(t\) with \(\by(t)>0 \) and \(t\notin T_{k_\by+1} \) (here we use \(T_{k+1}=T_k \)), then there exists a pointing strategy  \(\sigma_{\epsilon_{50}}^{\hat \by} \) that runs on \(\A_{+\hat{\by}} \) and that points to \(\M_{\hat{\by}} \) at least \(n^{k_\by-\epsilonr_{3}} \) times  from the initial pointing configuration with counters vector \(\vec{n} \)   with probability \(p_{n}^{\hat \by} \) such that \(\lim_{n\rightarrow\infty}p_{n}^{\hat \by}=1 \).

Let  \(\epsilon>0 \) be arbitrary but further fixed, we will now describe a pointing strategy \(\sigma_\epsilon \) such that the pointing computation under \(\sigma_\epsilon \) on the pointing VASS \(\A_{+\hat{B}} \) from any initial pointing configuration with all counters set at \(\vec{n} \) will point to \(\M_{\hat{\br}} \) at least \(n^{k+1-\epsilon} \) times with probability \(p_n \) such that \(\lim_{n\rightarrow\infty}p_n=1 \), assuming \(\support^{C_1,\dots,C_{\lfloor \frac{k}{2}\rfloor}}(\hat\br)\subseteq R^{B_\br,k} \) holds. Note that this would prove Lemma~\ref{lemma-D-in-R-lower-estimate-nk+1-for-hatB}. We will give only a high level description of \(\sigma_\epsilon \). To ease the description we assume \(\sigma_\epsilon\)  remembers values \(a_\bu\in \mathbb{Q} \) for each \(\bu\in BASIS \) that are initialized to \(0\) (note that these can always be computed from the pointing history).
%\todo{check previous epislons do not depend on epsilon}
%\todo{\(\epsilonr_{4}\) was \(\epsilont_2 \)}
\textbf{Bins initialization:} From Lemma~\ref{counters-pumpable-all-at-once} combined with  the induction assumptions for all \(0\leq k'<k\) on points  \ref{enum-main-1} and  \ref{enum-main-2} of Lemma~\ref{lemma-main-for-all-k} we have a strategy \(\pi \) in \(\A\) that from any initial configuration with initial counters vector \(\vec{n}\) reaches a configuration with counters vector \(\bv_0 \) such that for each counter \(c\in C_i\) where \(i\leq k \) (here we put \(C_{k}=C_{k+} \)) it holds \(\bv_0(c)\geq n^{i-\epsilonr_{4}}\), and this is achieved with probability \(p_n'\) such that \(\lim_{n\rightarrow\infty} p_n'=1\). \(\sigma_\epsilon \) therefore starts by using \(\pi \) to reach such counters vector (remember that \(\pi \) can be seen as a pointing strategy on \(\tilde{\A} \) as per \cite{AKCONCUR23}). At this point, \(\sigma_\epsilon \)  divides the counters vector into \(m \) bins of equal size (here \(m\) is some properly chosen constant, important thing is that \(m\) does not depend on \(n\)), that is each of the \(m\) bins contains the counters vector \(\lfloor \frac{\bv_0}{m} \rfloor \).

%	Let \(B_1,B_2,\dots,B_{\lfloor k/2 \rfloor} \) be the MECs of \(\A_{k-1},\A_{k-2},\dots,\A_{k-\lfloor k/2 \rfloor} \) containing \(B\), respectively.  

\textbf{Moving to \(p_{\br} \):} At this point \(\sigma_\epsilon \) will start a computation on a \(p_{\br} \)-bin (that will never be used again) under a strategy that minimizes the expected number of steps needed to reach \(p_{\br}\) and \(\sigma_\epsilon \) will play according to this strategy until reaching \(p_{\br}\). (Again a strategy on \(\A \) can be seen as a pointing strategy on \(\tilde{\A} \). We say a state \(p\) is reached if the new state of the VASS Markov chain that was just pointed to is \(p\)).

%Let us define  \(\A_{k,\hat{T}_{k+1}}^{\bx_k} \) as the VASS MDP created from \(\A_{k,\hat{T}_{k+1}} \) by removing any transition \(t\) with \(\bx_k(t)=0 \), and let \(B_1,\dots,B_w \) be all the MECs of \(\A_{k,\hat{T}_{k+1}}^{\bx_k} \). For each \(1\leq i\leq w\) we fix some state \(p_w\) of \(B_w \). 
%\todo{from here}
%In the following, if we say that \(\sigma_\epsilon\) moves over \(\MEC_i \) while the current state is \(p \) (note that we will only do so when \(p\) is a state of \( \)), what we mean is that \(\sigma_\epsilon\) proceeds to play according to the strategy \(\sigma_i \) on the bin move-\(\MEC_i\)-bin until this computation visits every state of \(\MEC_i \) at least once and then returns to \(p\), where \(\sigma_i \) is a strategy that at each step chooses the next transition uniformly at random from all the transitions that are present in \(\MEC_i \).

% \todo{\(\epsilonr_5 \)was \(\epsilont_{3} \)}
At this point \(\sigma_\epsilon \) starts pointing at \(\M_{\hat{\br}} \) using the main-bin as well as an additional high-bin such that the effects on counters in \(C_1,\dots,C_{\lfloor \frac{k}{2} \rfloor} \) are added to the main-bin while the effects on the other counters are added to the high-bin. 
%	 In addition to the main bin, there shall also be another buffer bin\michal{do we still need a buffer???} that is not used anywhere, but it allows us to reach values of up to \(-\lfloor \frac{\vec{\bv_0}}{x}\rfloor \) on the main bin without risking termination. (it uses bins as in the effect of the computation is added to the given bin). 
Every time \(\M_{\hat{\br}}  \) revisits \(p_B\), let \(\alpha \) denote the path taken in \(\M_{\hat{\br}} \) since its last visit of \(p_B \), and let \(\bv\) be the effect of \(\alpha \) on counters \(C_1 \cup \dots\cup C_{\lfloor \frac{k}{2}\rfloor}  \) (i.e., the effect in the main-bin). If \(length(\alpha)> n^{\epsilonr_{1}} \) then \(\sigma_\epsilon \) further behaves arbitrarily without changing the counter vectors in the bins. Otherwise \(\sigma_{\epsilon} \) adds the value \(f_\bu(\bv^{C_1,\dots,C_{\lfloor \frac{k}{2} \rfloor}}) \) to \(a_\bu \) for each \(\bu\in BASIS \) (note that \(\bv^{C_1,\dots,C_{\lfloor \frac{k}{2} \rfloor}}\in \support^{C_1,\dots,C_{\lfloor\frac{k}{2}\rfloor}}(\hat{\br}) \) and thus from our assumption that \(\support^{C_1,\dots,C_{\lfloor\frac{k}{2}\rfloor}}(\hat{\br})\subseteq R^{B_\br,k} \) the value \(f_\bu(\bv^{C_1,\dots,C_{\lfloor \frac{k}{2} \rfloor}}) \) is defined). Next, \(\sigma_\epsilon \) virtually ``distributes'' the vector \(\bv^{C_1,\dots,C_{\lfloor \frac{k}{2} \rfloor}} \) from the main-bin into individual \(\bu\)-bins as follows: First it adds \(-\bv^{C_1,\dots,C_{\lfloor \frac{k}{2} \rfloor}} \) to the main-bin, and then for each \(\bu\in BASIS \) it adds \(f_\bu(\bv^{C_1,\dots,C_{\lfloor \frac{k}{2} \rfloor}})\cdot \bu \) to the \(\bu\)-bin. At this point \(\sigma_\epsilon \) checks whether there exists \(\bu\in BASIS \) such that \(|a_\bu|\geq n^{i-\epsilonr_{5}} \) where \(i\) is such that \(\bu\in BASIS_i^{C_1, \dots, C_{\lfloor\frac{k}{2}\rfloor} }\cup (-BASIS_i^{C_1, \dots, C_{\lfloor\frac{k}{2}\rfloor} })\) (note that the \(i\) is uniquely given by \(\bu\) as discussed above). If such \(\bu\) does not exist then \(\sigma_\epsilon \) continues by again pointing a \(\M_{\hat{\by}} \) in the same way as above (thus doing so until either such \(\bu\) exists or termination). If such \(\bu\in BASIS\) exists, then \(\sigma_\epsilon \) temporarily pauses the pointing \(\M_{\hat{\br}} \) on the main-bin and performs a reset procedure. After this reset procedure finishes, \(\sigma_\epsilon \) unpauses the pointing \(\M_{\hat{\br}} \) on the main-bin and moves right to the ``\(\sigma_\epsilon \) checks whether there exists \(\bu\in BASIS \) such that...'', thus looping in this way till termination. 

Which reset procedure is chosen depends on whether \(a_\bu\) is positive or negative. \(\sigma_\epsilon \) performs the \(\bs\)-reset-procedure where \(\bs=\bu\) if \(a_\bu<0 \) and \(\bs=-\bu\) if  \(a_\bu>0 \). The \(\bs\)-reset-procedure is fully described below.

\subsubsection{\(\bs \)-reset-procedure}\label{section--sub-reseting-subprocedure}
 Let \(1\leq i\leq \frac{k}{2}\) and  \( \bv\in BASIS_i\) be such that either  \(\bs=\bv^{C_1,\dots,C_{\lfloor \frac{k}{2}\rfloor}}\) or \(\bs=-\bv^{C_1,\dots,C_{\lfloor \frac{k}{2}\rfloor}}  \). If \(\bs=\bv^{C_1,\dots,C_{\lfloor \frac{k}{2}\rfloor}}\) then let \(\bx=\bx_{i,\bv}^+ \) and if \(\bs=-\bv^{C_1,\dots,C_{\lfloor \frac{k}{2}\rfloor}}  \) then let \(\bx=\bx_{i,\bv}^- \).

When we say that the \(\bs\)-reset-procedure uses an \emph{interface transition} then \(\sigma_\epsilon \) doesn't actually do anything in \(\A_{+\hat{\br}} \), instead it only virtually simulates the interface transition (whose effect is \(-\Delta(\bx) \) without changing the state) by adding \(\Delta^{C_1,\dots,C_{\lfloor \frac{k}{2}\rfloor}}(\bx)\) to the \(\bu\)-bin and \(\bs^{C_{\lfloor \frac{k}{2}\rfloor+1},\dots,C_k}\) to the high-\(\bu\)-bin while adding \(-\bs \) to the bin on which the \(\bs\)-reset-procedure performs the interface transition (thus effectively simulating it for the procedure), in addition to this, every time the interface transition gets used \(\sigma_\epsilon \) also adds \(+1 \) to \(a_\bu \) if \(\bs=-\bu \) or it adds \(-1\) to \(a_\bu \) if \(\bs=\bu \) (i.e., it moves it by \(1\) towards \(0\)). Note that if the \(\bs\)-reset-procedure selects the interface transition while \(|a_\bu|<1\) then \(\sigma_\epsilon \) ``pauses the \(\bs\)-reset-procedure halfway during the interface transition'', that is \(\sigma_\epsilon \) changes the counters in the bins by only \(\bs\cdot|a_\bu| \) and sets \(a_\bu \) to \(0\), then pauses the \(\bs\)-reset-procedure and the next time it gets unpaused it performs the remaining portion of the interface transition (while changing the counters and modifying \(a_\bu \) correspondingly)\footnote{Technically, this can cause the counter values in these bins to no longer be integers. Note that this is not an issue as all the individual strategies simulated on these bins ignore the counter values. Also note that the sum of all the counters vectors in all of the bins (i.e. the actual counters vector) still sums up to an integer vector.} thus from the point of view of the \(\bs\)-reset-procedure the entire interface transition was performed (assuming \(\bs\)-reset-procedure gets used again later). \(\sigma_\epsilon \) keeps playing as per the \(\bs\)-reset-procedure until \(a_\bu\neq 0 \), at which point the \(\bs\)-reset-procedure gets paused (finished).  If again at some point later \(\sigma_\epsilon \) performs the  \(\bs\)-reset-procedure again then it unpauses it from the last time (thus there is only one \(\bs\)-reset  procedure for each \(\bs\in BASIS \) running for the entire computation).

We will describe the strategy  \(\sigma_\bs^{\epsilonr_{6}} \) on the VASS MDP \(\A \) modified in such a way that we also add a new loop  over \(p_{\br} \) that has the effect \(-\bs \), we use \(\A^{+p_B(-\bs)} \) to denote this new VASS MDP. We use \(\by_B \) to denote a component corresponding to this loop (i.e., \(\by_B(t)=0 \) for each \(t \) in \(\A \) and \(\by(t)=1 \) for the newly added transition \(t\)). We will call this new added transition the "interface transition" whose behavior is described above.

From Lemma~\ref{lemma-decompose-multicomponents-into-components} we can express \(\bx \) as a conical sum of components, that is \(\bx=\sum_{j=1}^w a_j\cdot \by_j \) where \(w \) is the number of components on \(\A_{k-i} \), \(\by_j \) is the \(j\)-th component and \(a_j\geq 0 \). Additionally, for the multi-component \(\bx'=\bx+\by_B \) on \(\A_{k-i}^{+p_B(-\bs)} \) it holds \(\Delta^{C_1,\dots,C_{k-i}}(\bx')=\Delta^{C_1,\dots,C_{k-i}}(\bx)+\Delta^{C_1,\dots,C_{k-i}}(\by_B)=\bs-\bs=\vec{0} \) and it can be expressed as a conical sum of components as \(\bx'=a_B\cdot \by_B + \sum_{j=1}^w a_j\cdot \by_j\) for \(a_B=1 \).

Given two components \(\by_1\) and \(\by_2 \) on \(\A_{k-i}^{+p_B(-\bs)} \), we define the \emph{level-distance} of \( \by_1\) and \(\by_2 \) as the smallest value \(l\in \mathbb{N}_0 \) such that both \(\by_1 \) and \(\by_2 \) are contained in the same MEC of \(\A_{k-i-l}^{+p_B(-\bs)}  \). (note that the maximal possible level-distance is \(k-i-1 \))

% 	 	Here we describe a strategy that on \(\sigma_\bu^\epsilonr_{6} \) that \(\A^{+p_B(-\bu)} \) iterates a multi-component \(\bx \) of \(A_{k-i,T_{k-i+1}})\) at least \((n^{1-\epsilon})^{k-i} \) times.

%	From  the induction assumption on \ref{enum-main-5} for \(k-i \), there exists a strategy \(\sigma_{k-i}^{\epsilonr_{8}} \) on \(\A\) that satisfies  \(\lim_{n\rightarrow \infty} \prob[\sigma_{k-i}^{\epsilont_17} \textit{ produces a }(k-i)\textit{-level }(n^{1-\epsilont_17})\textit{-cyclic path}]=1 \).

From  the induction assumption on point \ref{enum-main-4} of Lemma~\ref{lemma-main-for-all-k} for \(k-i \), there exists for each component \(\by \) of \(\A_{k-i,T_{k-i+1}}\) a pointing strategy \(\sigma_{\hat{\by}}^{\epsilonr_{7}} \) such that \(\lim_{n\rightarrow\infty}\prob_{p\vec{n}}^{\sigma_{\hat{\by}}^{\epsilonr_{7}}}[\calP_{\A_{+\hat{\by}}}[\M_{\hat{\by}}]\geq n^{k+1+\epsilonr_{7}}]=1 \) . In addition to this, since \(\hat\by_B \) never changes any counter we define \(\sigma_{\hat{\by}_\MEC}^{\epsilonr_{7}}  \) as a pointing strategy for \(\A^{+p_B(-\bs)}_{+\hat\by_B} \) that does nothing but keeps pointing to \(\M_{\hat{\by_B}} \), clearly this pointing strategy will result in termination when initiated in \(p_\br \).

From the induction assumption on point \ref{enum-main-5} of Lemma~\ref{lemma-main-for-all-k} for \(k-i \) there exists a strategy \(\sigma_{k-i}^{\epsilonr_{8}}  \) which for \(h_{\epsilonr_{8}}(n)=\lfloor n^{1-\epsilonr_{8}}\rfloor \)  produces with probability \(\lim_{n\rightarrow\infty} p_n=1 \) a \((k-i)\)-level \(h_{\epsilonr_{8}}(n) \)-cyclic computation when initialized with counters vector \(n\).

We will now describe the strategy \(\sigma_\bs^{\epsilonr_{6}} \) on \(\A^{+p_B(-\bs)} \) that starts in \(p_B \) and that we shall call the \(\bs\)-reset-procedure (i.e., the \(\bs\)-reset-procedure consists of playing as \(\sigma_\bs^{\epsilonr_{6}} \)).

The strategy \(\sigma_\bs^{\epsilonr_{6}} \) works as follows, once again we only give a high level description of the strategy: 

It remembers (it can always compute these from the history) values \(b_B,b_1,\dots,b_w \) that are all initialized to \(0 \).

It plays \(\sigma_{k-i}^{\epsilonr_{8}} \) on a \((\bs,k,m)\)-tree-bin, except that occasionally this computation can be temporarily ``paused'' in favor of playing something else. Every time this computation under \(\sigma_{k-i}^{\epsilonr_{8}} \) reaches some state \(p\) then \(\sigma_\bs^{\epsilonr_{6}} \)   ``processes'' every component \(\by_a \) of \(\A_{k-i}^{+p_B(-\bs)} \) such that \(p=p_{\by_a} \) and \(a_a\neq 0 \) before resuming playing as per \(\sigma_{k-i}^{\epsilonr{8}} \), where the ``processing'' of \(\by_a \) looks as follows: \(\sigma_\bs^{\epsilonr_{6}} \) asks whether there exists another component \(\by_b \) of \(\A_{k-i}^{+p_B(-\bs)} \) such that \(a_a,a_b\neq 0 \) and the level-distance of \(\by_a \) and \(\by_b \) is \(l \) such that it holds \(b_a - \frac{a_a}{a_b}\cdot (n^{1-\epsilonr_{6}})^{l+1} \geq b_b    \). If such \(\by_b \) exists then \(\by_a \) is  declared ``processed'' without \(\sigma_\bs^{\epsilonr_{6}} \) doing anything. 

If such \(\by_b \) does not exist, then \(\sigma_\bs^{\epsilonr_{6}} \) starts playing according to the strategy \(\sigma_{\hat\by_a}^{\epsilonr_{7}} \) using a \((\bs,\hat\by_a)\)-bin (technically, we initialize this bin with the counter vector \(\lfloor \frac{n^{1-\epsilonr_{5}}}{m}\rfloor \) instead of \(\lfloor \frac{\bv_0}{m}\rfloor \))  in such a way that whenever \(\sigma_{\hat\by_a}^{\epsilonr_{7}} \) points at \(\M_{ \hat\by_a} \) then \(\sigma_\bs^{\epsilonr_{6}} \) actually points at \(\M_{\by_a} \) while adding the resulting effect to the  \((\bs,\hat\by_a)\)-bin, and if this pointing results in reaching the state \(p_{\by_a} \) then additionally the vector \(\Delta(\by_a) \) is added to the \(\bs\)-main-bin (that is shared among all the components \(\by_a \)) and the same vector is subtracted from the  \((\bs,\hat\by_a)\)-bin (these are done virtually, so in effect the computation on the \((\bs,\hat\by_a)\)-bin corresponds to a computation on \(\A_{+\hat{\by}_a} \)).  This proceeds until this computation returns to  \(p_{\by_a} \) after \(\sigma_{\hat\by_a}^{\epsilonr_{7}} \) pointed at \(\M_{\hat{\by}_a} \) (if \(p_{\by_a} \) is reached after \(\sigma_{\hat\by_a}^{\epsilonr_{7}} \) pointed at different VASS Markov chain then it does not count) for a total of  \(n^{1-\epsilonr_{6}}\) times (note that if this number is not an integer then what the strategy effectively does is it reaches \(p_{\by_a} \) a total of \(\lceil n^{1-\epsilonr_{6}} \rceil\) times while remembering a value \(X\) equal to the difference from the number of times it reached \(p_{\by_a} \) and  \(n^{1-\epsilonr_{6}}\), and the next time it starts iterating \(\sigma_{\hat\by_a}^{\epsilonr_{7}} \) then before any transition is even chosen it assumes that it already reached \(p_{\by_a} \) a total of \(X \) times, while these "extra" iterations happen on a completely separate bin that is not used anywhere else (thus this cannot ever kill the computation)). Once this happens \(\sigma_\bs^{\epsilonr_{6}} \) increases the value \(b_a \) by \(n^{1-\epsilonr_{6}} \) and then declares \(\by_a \) to be ``processed'' (note that each component can be ``proccessed'' multiple times before unpausing \(\sigma_{k-i}^{\epsilonr_{8}} \), thus at this point \(\by_a \) gets checked for processing again).

\subsubsection{Analysis of \(\bs \)-reset-procedure}

In this Section we prove some Lemmas pertaining the behavior of the \(\bs \)-reset-procedure that will help us prove the main statement about \(\sigma_{\epsilon} \) later on.

\begin{lemma}\label{lemma-reset-procedure-main-bin-nonterminating}

For each counter \(c\in C_j \) for \(1\leq j \leq k-i \) (here we consider the counters on \(\A_{k-i} \)) the counter \(c\) can never become negative in the \(\bs\)-main bin for all sufficiently large \(n\).

%before there exists some \(a\in \{B,1,\dots,w \}\) with \(b_a\geq n^{k-i+1-\epsilon_{84}} \) is \(0\).\michal{this paragraph ought to be restated}

\end{lemma}

\begin{proof}

% Assume this is not the case, then there exists some computation under \(\sigma_\br^{\epsilonr_{9}} \) which runs in a counter \(c\in C_j \) with \(1\leq j \leq \lfloor \frac{k}{2}\rfloor \) out on the main bin before any counter running out in any other bin. 

Notice that at any point it holds that the counter vector in the \(\bs\)-main bin is equal to \(\lfloor\frac{\bv_0}{m}  \rfloor + b_B\cdot \Delta(\by_B) + \sum_{j=1}^{w}b_j\cdot \Delta(\by_j) \). Furthermore, notice that for each \(c\in C_j \) it holds for all \(\by_a,\by_b \) that if \(\Delta(\by_a)(c)\neq 0 \) and the level-distance of \(\by_a,\by_b \) is at least \(j \) then \(\Delta(\by_b)(c)=0 \) (this is because in such case \(\by_b\) lies in a different MEC of \(\A_{k-i-j}\) from \(\by_a \) and thus both consider different local copy of the counter corresponding to \(c\)). Let \(I_c=\{i\in \mathbb{N}\cup\{B\}\mid \textit{there exists a}\textit{ transition }t\textit{ of }\A_{k-i};\by_i(t)>0 \textit{ and }\bu_t(c)\neq 0\} \) be the indexes of all the components of \(\A_{k-i}^{+p_B(-\bs)} \) that do modify \(c\) in any way (note that the interface transition modifies the local copies of counters that are present in the MEC containing \(p_B \)), then the maximal level-distance between any \(\by_a,\by_b \) for \(a,b\in I_c \) is \(j-1 \), and thus at any point it holds for any \(a\in I_c \) that for any \(b\in I_c \) we  have \(b_a - \frac{a_a}{a_b}\cdot (n^{1-\epsilonr_{9}})^{j} \leq b_b + n^{1-\epsilonr_{9}}   \). Let \(\kappa \) be the largest value such that \(b_x-\kappa\cdot  a_x\geq 0 \) for all \(x\in I_c \), and for each \(a\in I_c \) let  \(b_a'=b_a-\kappa\cdot  a_a\). Notice that from from the definition of \(\kappa \) there exists at least one \(b\in I_c \) such that \(b_b'=0 \). Thus by substituting \(b_x'+\kappa\cdot  a_x \) for \(b_x \) we get 

\(b_a'+\kappa\cdot  a_a - \frac{a_a}{a_b}\cdot (n^{1-\epsilonr_{9}})^{j} \leq b_b'+\kappa \cdot a_b + n^{1-\epsilonr_{9}}    \)

which can be simplified into

\(b_a'+\kappa\cdot  a_a - \frac{a_a}{a_b}\cdot (n^{1-\epsilonr_{9}})^{j} \leq 0+\kappa\cdot  a_b + n^{1-\epsilon\cdot _{4}}  \)

\(b_a' \leq \kappa\cdot  a_b-\kappa \cdot a_a + n^{1-\epsilonr_{9}}+ \frac{a_a}{a_b}\cdot (n^{1-\epsilonr_{9}})^{j}   \)

At any point it holds that the counter \(c\) in the \(\bs\)-main bin is equal to \begin{multline*}
	\lfloor\frac{\bv_0(c)}{m}  \rfloor + b_B\cdot \Delta(\by_B)(c) + \sum_{j=1}^{w}b_j\cdot \Delta(\by_j)(c)
	=
	\lfloor\frac{\bv_0(c)}{m}  \rfloor   + \sum_{j\in I_c}b_j\cdot \Delta(\by_j)(c)
	=  \\
	\lfloor\frac{\bv_0(c)}{m}  \rfloor + \sum_{j\in I_c}(b_j'+\kappa\cdot  a_j)\cdot \Delta(\by_j)(c)
	=
	\lfloor\frac{\bv_0(c)}{m}  \rfloor  + \sum_{j\in I_c}\kappa\cdot  a_j\cdot \Delta(\by_j)(c) + \sum_{j\in I_c}b_j'\cdot \Delta(\by_j)(c)
\end{multline*}

since \( \sum_{j\in I_c}\kappa \cdot a_j\cdot  \by_j\) is a multi-component that matches \(\kappa\cdot \bx' \) on the MEC of \(\A_{k-i-j+1} \) that contains \(c\), it holds \(\sum_{j\in I_c}\kappa a_j\cdot \Delta(\by_j)(c)=\kappa \cdot \Delta(\bx')(c)=0 \) thus we can write for  \(u=\max_{i\in\{B,1,\dots,w\}}\{|\Delta(\by_i)(c)| \}\)

\begin{gather*}
	\lfloor\frac{\bv_0(c)}{m}  \rfloor  + \sum_{j\in I_c}\kappa\cdot  a_j\cdot \Delta(\by_j)(c) + \sum_{j\in I_c}b_j'\cdot \Delta(\by_j)(c)
	=
	\lfloor\frac{\bv_0(c)}{m}  \rfloor + 0 + \sum_{j\in I_c}b_j'\cdot \Delta(\by_j)(c)
	\geq   \\ 
	\lfloor\frac{n^{j-\epsilonr_{5}}}{m} \rfloor - \sum_{j\in I_c}b_j'\cdot u
	\geq 
	\lfloor\frac{n^{j-\epsilonr_{5}}}{m} \rfloor - \sum_{x\in I_c}u(\kappa \cdot a_b-\kappa\cdot  a_x + n^{1-\epsilonr_{9}}+ \frac{a_x}{a_b}\cdot (n^{1-\epsilonr_{9}})^{j})
	=\\
	\lfloor\frac{n^{j-\epsilonr_{5}}}{m} \rfloor - \sum_{x\in I_c}u\cdot \kappa\cdot  a_b+\sum_{x\in I_c}u\cdot \kappa \cdot a_x -\sum_{x\in I_c}u \cdot n^{1-\epsilonr_{9}}-\sum_{x\in I_c}u \cdot \frac{a_x}{a_b}\cdot (n^{1-\epsilonr_{9}})^{j}
	=\\
	\lfloor\frac{n^{j-\epsilonr_{5}}}{m} \rfloor - |I_c|\cdot u\cdot \kappa\cdot  a_b+u\cdot \kappa\cdot \sum_{x\in I_c} a_x -|I_c|\cdot u\cdot  \cdot n^{1-\epsilonr_{9}}-\frac{u}{a_b}\cdot (n^{1-\epsilonr_{9}})^{j}\cdot \sum_{x\in I_c} a_x
	\geq\\
	\lfloor\frac{n^{j-\epsilonr_{5}}}{m} \rfloor - |I_c|\cdot u\cdot \kappa \cdot a_b -|I_c|\cdot u\cdot  n^{1-\epsilonr_{9}}-\frac{u}{a_b}\cdot (n^{1-\epsilonr_{9}})^{j}\cdot \sum_{x\in I_c} a_x
	\geq\\
	\lfloor\frac{n^{j-\epsilonr_{5}}}{m} \rfloor - |I_c|\cdot u\cdot \kappa \cdot a_b -|I_c|\cdot u\cdot  n^{1-\epsilonr_{9}}-\frac{u}{a_b}\cdot (n^{1-\epsilonr_{9}})^{j}\cdot |I_c|\cdot \max_{x\in I_c} a_x
	=\\
	\lfloor\frac{n^{j-\epsilonr_{5}}}{m} \rfloor - |I_c|\cdot u\cdot \kappa \cdot a_b -|I_c|\cdot u\cdot  n^{1-\epsilonr_{9}}-\frac{u}{a_b}\cdot n^{j-j\epsilonr_{9}}\cdot |I_c|\cdot \max_{x\in I_c} a_x
\end{gather*}

Which is always non-negative for all sufficiently large \(n\) assuming \begin{equation}\label{eq-epsbound-bvcyuviuib}
	\epsilonr_{5}< k\cdot \epsilonr_{9}
\end{equation}  since in such case \(\lfloor\frac{n^{j-\epsilonr_{5}}}{m} \rfloor\) grows asymptotically faster than each of \(|I_c|\cdot u\cdot \kappa\cdot  a_b\), \(|I_c|\cdot u\cdot  n^{1-\epsilonr_{9}} \), and \(\frac{u}{a_b}\cdot n^{j-j\epsilonr_{9}}\cdot |I_c|\cdot \max_{j\in I_c} a_x\).  
\end{proof}

\begin{lemma}\label{lemma-reset-procedure-interface-transition-nk-i-eps101}

Let \(I=\{a\in\{B,1,\dots,w\}\mid a_a\neq 0  \} \).
The probability of \(\sigma_\bs^{\epsilonr_{9}} \) terminating while there exists \(a\in I \) such that \(b_a\leq n^{k-i+1-\epsilonr_{10}} \) is \(p_\bs\) such that \(\lim_{n\rightarrow\infty} p_\bs=0 \).
\end{lemma}

\begin{proof}
\(\sigma_\bs^{\epsilonr_{9}} \) can only terminate if at least one of the individual bins it uses reaches negative value on at least one counter. Hence it suffices to show that no counter gets depleted in any of those bins while there exists \(a\in I \) such that \(b_a\leq n^{k-i+1-\epsilonr_{10}} \) with probability \(p_\bs'\) such that \(\lim_{n\rightarrow\infty} p_\bs'=0 \). The bins on which \(\sigma_\bs^{\epsilonr_{9}} \) can terminate are the \((\bs,k,m)\)-tree-bin while using a \(\sigma_{k-i}^{\epsilonr_{14}} \) strategy, on a \((\bs,\hat\by_a)\)-bin under  \(\sigma_{\hat\by_a}^{\epsilonr_{13}} \), or in the \(\bs\)-main bin. From Lemma~\ref{lemma-reset-procedure-main-bin-nonterminating} the \(\bs\)-main bin can never reach negative values on counters from \(C_1,\dots,C_{k-i} \), thus the computation on this bin is safe until at least one \(b_a \) reaches a minimum of \(n^{k-i+1-\epsilonr_{4}-\epsilonr_{11}} \) (this is due to the the smallest counter it can deplete being \(\lfloor\frac{n^{k-i+1-\epsilonr_{4}}}{m}\rfloor \) and at any point counter vector in the \(\bs\)-main bin is equal to \(\lfloor\frac{\bv_0}{m}  \rfloor + b_B\cdot \Delta(\by_B) + \sum_{j=1}^{w}b_j\cdot \Delta(\by_j) \)).
The probability \(p_{n}^{\hat{\by_a}} \) of terminating in some \((\br,\hat\by_a)\)-bin before \(b_a\geq n^{k-i+1-\epsilonr_{13}-\epsilonr_{11}} \) satisfies \(\lim_{n\rightarrow\infty} p_n^{\hat{\by_a}}=0 \) from how we chose the strategy \(\sigma_{\hat\by_a}^{\epsilonr_{13}} \).

As any two \(b_a,b_b \) always satisfy \(b_a - \frac{a_a}{a_b}\cdot (n^{1-\epsilonr_{9}})^{l+1} \leq b_b + n^{1-\epsilonr_{9}}   \), where \(l\) is the level distance of \(\by_a\) and \(\by_b \), and the maximal level distance is \(k-i-1 \), it holds \(b_a - \frac{a_a}{a_b}\cdot (n^{1-\epsilonr_{9}})^{k-i-1+1} \leq b_b + n^{1-\epsilonr_{9}}   \).

Thus if it were to hold that there exists \(b_b\leq n^{k-i+1-\epsilonr_{12}} \) with \begin{equation}\label{eq-epsbound-bvcyuviuib2}
	\epsilonr_{12}>\epsilonr_{10}
\end{equation} then for each \(b_a \) it holds \(b_a - \frac{a_a}{a_b}\cdot n^{k-i-(k-i)\cdot \epsilonr_{9}} \leq b_b + n^{1-\epsilonr_{9}} \leq n^{k-i+1-\epsilonr_{12}} + n^{1-\epsilonr_{9}} \)
which gives us that \(b_a\leq \frac{a_a}{a_b}\cdot n^{k-i-(k-i)\epsilonr_{9}}+ n^{k-i+1-\epsilonr_{12}} + n^{1-\epsilonr_{9}} \leq n^{k-i+1-\epsilonr_{10}} \) assuming that \begin{equation}\label{eq-epsbound-bvcyuviuib3}
	\epsilonr_{10}<\epsilonr_{12}<1
\end{equation} Therefore assuming \begin{equation}\label{eq-epsbound-bvcyuviuib4}
\epsilonr_{4}+\epsilonr_{11}<\epsilonr_{12}
\end{equation}  it holds that the probability of \(\sigma_\bs \) depleting a counter on any \((\bs,\hat\by_a)\)-bin or \(\bs\)-main-bin while there exists \(a\in \{B,1,\dots,w \} \) with \(a_a\neq 0 \) and such that \(b_a\leq n^{k-i+1-\epsilonr_{10}} \) is \(p_\bs'\) such that \(\lim_{n\rightarrow\infty} p_\bs'=0 \).

Thus it remains to show that the computation will with high enough probability not die too early on the \((\bs,k,m)\)-tree-bin.

We will now show that if the computation under  \(\sigma_{k-i}^{\epsilonr_{14}} \) on the \((\bs,k,m)\)-tree-bin produces a \((k-i)\)-level \(
h_{\epsilonr_{14}}
(\lfloor \frac{n^{1-\epsilonr_{4}}}{m} \rfloor)\)-cyclic computation \(\alpha\) then assuming \(\sigma_\bs \) does not reach negative value on any counter in a \((\bs,\hat\by_a)\)-bin before \(b_a\geq n^{k-i+1-\epsilonr_{13}-\epsilonr_{11}} \) for some \(a\in I \) then \(\sigma_\bs\)  will reach \(n^{k-i+1-\epsilonr_{10}} \) in \(b_a\) for all \(a\in I\).

Assume towards contradiction that this does not happen.

Let \(G=(V,E,f,g) \) be the \((k-i,(\lfloor \frac{n^{1-\epsilonr_{4}}}{m} \rfloor)^{1-\epsilonr_{14}})\)-tree of \(\alpha \). 
Notice that the concatenation of the computations \(f(v) \) in the leaves of \(G \) is exactly \(\alpha \), thus to each \(\alpha_{..j} \) we can assign the leaves of \(G \) for whose the concatenation of \(f(v) \) corresponds to \(\alpha_{..j} \). 

For each \(a\in I \) let \(c_a\) be a function such that  \(c_a(\alpha_{..j})=\kappa_{\by_a}(\alpha_{..j}) n^{1-\epsilonr_{9}}  \) where \(\kappa_{\by_a}(\alpha_{..j}) \) is the number of leaves \(v\) of \(G\) for which \(g(v) \) is the MEC of \(\A_{k-i} \) that contains \(p_{\by_a} \) and that is included in the leaves corresponding to \(\alpha_{..j} \). Furthermore, let \(b_a(\alpha_{..j}) \) be the value of \(b_a \) right after \(\alpha_{..j} \).

We will now show the following Lemma.

\begin{lemma}\label{lemma-some-some-something7}
	It holds \(c_a(\alpha_{..j})\leq b_a(\alpha_{..j})  \) for all \(j\).
\end{lemma}

Note that this finishes the proof of the Lemma~\ref{lemma-reset-procedure-interface-transition-nk-i-eps101} since if \(\alpha \) is \((k-i,
h_{\epsilonr_{14}}
(\lfloor \frac{n^{1-\epsilonr_{4}}}{m} \rfloor))\)-cyclic then \begin{gather*}
	c_a(\alpha)=
	(h_{\epsilonr_{14}}
	(\lfloor \frac{n^{1-\epsilonr_{4}}}{m} \rfloor))^{k-i}\cdot n^{1-\epsilonr_{9}}=
	(\lfloor(\lfloor \frac{n^{1-\epsilonr_{4}}}{m} \rfloor)^{1-\epsilonr_{14}}\rfloor)^{k-i}\cdot n^{1-\epsilonr_{9}}
	\leq \\
	((n^{1-\epsilonr_{4}})^{1-\epsilonr_{14}})^{k-i}\cdot n^{1-\epsilonr_{9}}
	=
	(n^{1-\epsilonr_{14}-(1-\epsilonr_{14})\cdot \epsilonr_{4}})^{k-i}\cdot n^{1-\epsilonr_{9}}
	=
	n^{k-i-(k-i)\cdot (\epsilonr_{14}-(1-\epsilonr_{14})\cdot \epsilonr_{4})}\cdot n^{1-\epsilonr_{9}}
	=\\
	n^{k-i+1-\epsilonr_{9}-(k-i)\cdot (\epsilonr_{14}-(1-\epsilonr_{14})\cdot \epsilonr_{4})}
	\leq n^{k-i+1-\epsilonr_{10}} 
\end{gather*}

for all sufficiently large \(n\) assuming that \begin{equation}\label{eq-epsbound-bvcyuviuib5}
	\epsilonr_{9}+(\epsilonr_{14}-(1-\epsilonr_{14})\epsilonr_{4})>\epsilonr_{10}
\end{equation} 

We prove Lemma~\ref{lemma-some-some-something7} by induction on \(j\). 

Base case \(j=0 \): it holds \(c_a(\epsilon)=b_a(\epsilon)=0 \), thus base case holds. (here \(\epsilon \) denotes an empty computation)

Induction step: assume Lemma~\ref{lemma-some-some-something7} holds for all \(j\in\{0,\dots,j-1 \} \), we show it hold for \( j\) as well. 

Notice that when \(\sigma_{k-i}^{\epsilonr_{14}} \) reaches some state \(p_{\by_a}\) for \(a\in I \) after \(\alpha_{..j} \) then \(b_{\by_a} \) is either increased by at least \(n^{1-\epsilonr_{9}} \) or there exists some other \(\by_b \) with \(b\in I \) such that  \(b_a(\alpha_{..j}) - \frac{a_a}{a_b}(n^{1-\epsilonr_{9}})^{l+1} \geq b_b(\alpha_{..j})    \) with level distance between \(\by_a \) and \(\by_b \) being \(l \). In the former case \(c_a(\alpha_{..j})\leq  c_a(\alpha_{..j-1})+n^{1-\epsilonr_{9}}\leq b_a(\alpha_{..j-1})+n^{1-\epsilonr_{9}}=b_a(\alpha_{..j})\) and thus the induction step holds. Assume therefore the later case.
Then both \(p_{\by_a} \) and \(p_{\by_b} \) are part of the same MEC \(B_{a,b} \) in \(\A_{k-i-l} \), and thus \(\kappa_{\by_a} \) and \(\kappa_{\by_b} \) both count only leaves that belong to a sub-tree for vertices \(v\) of \(G\) that have distance of \(k-i-l\) from the root and with \(g(v)=B_{a,b} \). But for each such \(v\) the entire sub-tree of \(v\) cycles between \(p_{\by_a} \) and \(p_{\by_b} \) at least \((h_{\epsilonr_{14}}(\lfloor\frac{n^{1-\epsilonr_{4}}}{m}\rfloor))^l \) times, while containing at most \((h_{\epsilonr_{14}}(\lfloor\frac{n^{1-\epsilonr_{4}}}{m}\rfloor))^l \) leaves that count towards   \(\kappa_{\by_a} \) as well as at most \((h_{\epsilonr_{14}}(\lfloor\frac{n^{1-\epsilonr_{4}}}{m}\rfloor))^l \) leaves that count towards \(\kappa_{\by_b}\). Hence it holds

\begin{multline*}
	|c_a(\alpha_{..j})-c_b(\alpha_{..j})|\leq 2(h_{\epsilonr_{14}}(\lfloor\frac{n^{1-\epsilonr_{4}}}{m}\rfloor))^l n^{1-\epsilonr_{9}}
	=
	2(\lfloor(\lfloor\frac{n^{1-\epsilonr_{4}}}{m}\rfloor)^{1-\epsilonr_{14}}\rfloor)^l n^{1-\epsilonr_{9}}
	\leq \\
	2((n^{1-\epsilonr_{4}})^{1-\epsilonr_{14}})^l n^{1-\epsilonr_{9}}
	=
	2(n^{1-\epsilonr_{14}-(1-\epsilonr_{14})\epsilonr_{4}})^l n^{1-\epsilonr_{9}}
	=
	2n^{l-l\epsilonr_{14}-l(1-\epsilonr_{14})\epsilonr_{4}} n^{1-\epsilonr_{9}}
	=\\
	2n^{l+1-\epsilonr_{9}-l\epsilonr_{14}-l(1-\epsilonr_{14})\epsilonr_{4}}
	\leq
	n^{l+1-\epsilonr_{31}} \end{multline*} assuming \begin{equation}\label{eq-epsbound-bvcyuviuib6}
	-\epsilonr_{31}>-\epsilonr_{9}-k\epsilonr_{14}-k(1-\epsilonr_{14})\epsilonr_{4}
	\end{equation}

If \(l=0\) then either there is some \(b'\) at level distance at least \(1\) from \(a \) that is also preventing \(b_a\) from being increased, or \(b_b(\alpha_{..j})>b_b(\alpha_{..j-1}) \) which then contradicts that \(b_a(\alpha_{..j})=b_a(\alpha_{..j-1}) \) (since in such case both \(b_a\) and \(b_b \) have the same constraints on when they are being increased). Thus we can wlog. assume that \(p_{\by_a}\neq p_{\by_b} \) and thus also \(b_{b}(\alpha_{..j})=b_b(\alpha_{..j-1}) \).

And from the induction assumption, we have that \(b_b(\alpha_{..j})=b_b(\alpha_{..j-1}) \geq c_b(\alpha_{..j-1}) \) as well as \(b_a(\alpha_{..j}) \geq c_a(\alpha_{..j-1}) \)
% 		 and the fact that \(p_{\by_a}\neq p_{\by_b} \)

If \(c_a(\alpha_{..j})=c_a(\alpha_{..j-1}) \) then it holds \(b_a(\alpha_{..j})=b_a(\alpha_{..j-1})\leq c_a(\alpha_{..j-1})=c_a(\alpha_{..j}) \). Thus we can assume that \(c_a(\alpha_{..j})=c_a(\alpha_{..j-1})+n^{1-\epsilonr_{9}} \). But this means that \(c_a(\alpha_{..j-1})=c_a(\alpha_{..j})\) (since both \(\kappa_{\by_a} \) and \(\kappa_{\by_b} \) count different leaves).

Thus it holds
\begin{multline*}
	b_a(\alpha_{..j}) - \frac{a_a}{a_b}(n^{1-\epsilonr_{9}})^{l+1}=b_a(\alpha_{..{j-1}}) - \frac{a_a}{a_b}(n^{1-\epsilonr_{9}})^{l+1} \geq b_b(\alpha_{..j})
	\geq  
	c_b(\alpha_{..j-1})  
	=
	c_b(\alpha_{..j})  
\end{multline*}

%		 and    (if \(c_a(\alpha_{..j})=c_a(\alpha_{..j-1}) \) then we already have the induction step from the induction assumption).
%		From the induction assumption \(b_a(\alpha_{..j-1})\geq c_a(\alpha_{..j-1}) \)

which gives us \[b_a(\alpha_{..j}) 
\geq  
c_b(\alpha_{..j}) + \frac{a_a}{a_b}(n^{1-\epsilonr_{9}})^{l+1}
=
c_b(\alpha_{..j}) + \frac{a_a}{a_b}n^{l+1-(l+1)\epsilonr_{9}}
\]
%		Thus it holds that \(b_a(\alpha_{..j}) - \frac{a_a}{a_b}(n^{1-\epsilonr_{9}})^{l+1} \geq b_b(\alpha_{..j}) \geq c_b(\alpha_{..j})\)

but from 
\(|c_a(\alpha_{..j})-c_b(\alpha_{..j})|\leq n^{l+1-\epsilonr_{31}} \)
we get  \[c_a(\alpha_{..j})
\leq 
c_b(\alpha_{..j}) +n^{l+1-\epsilonr_{31}} 
\leq 
c_b(\alpha_{..j}) + \frac{a_a}{a_b}n^{l+1-(l+1)\epsilonr_{9}}
\leq 
b_a(\alpha_{..j})   \]

assuming \begin{equation}\label{eq-epsbound-bvcyuviuib7}
	-\epsilonr_{31}<-(k+1)\epsilonr_{9} 
\end{equation} Thus the induction step holds for all sufficiently large \(n\).
\end{proof}

\subsubsection{Analysis of the strategy \(\sigma_{\epsilon} \) from Lemma~\ref{lemma-rand-walk-double-exp-zero}}

In this Section we show that the strategy \(\sigma_{\epsilon} \) as described above when initiated in a pointing configuration with counters set to \(\vec{n} \) points to \(\M_{\hat{\br}} \) at least \(n^{k+1-\epsilon} \) times with probability \(p_n\) such that \(\lim_{n\rightarrow\infty}p_n=1 \).

%	We claim that such \(\sigma_\epsilon \) has our desired property. 

%	 To simplify the proof we also additionally assume that if \(\sigma_\epsilon \) ever produces a path not in \(R_{\epsilon_{90}} \) \michal{point where this is defined, its the \(R_{\epsilonr_{4}} \) set} then the computation simply terminates instantly, clearly this can only shorten the computation and thus we can wlog. assume this.

We will show that probability of any bin depleting any counter before \(n^{k+1-\epsilon} \) pointings at \(\M_{\hat{\br}} \) goes to \(0 \) as \(n \) goes to \(\infty \), which would prove our claim as the whole computation surely cannot terminate sooner then any of the individual bins (once \(\bv_0 \) is reached and the bins are formed).

First note that the probability of \(\sigma_\epsilon \) dying before reaching \(\bv_0 \), and subsequently forming the bins, using \(\pi \) is \(p_{n,1} \) satisfying \(\lim_{n \rightarrow \infty}p_{n,1}=0 \).

The probability that moving to \(p_B \) needs more than \(n^{\epsilonr_{30}} \) steps is \(p_{n,2} \) such that \(\lim_{n\rightarrow\infty} p_{n,2}=0 \), this follows from Markov inequality since the expected number of steps to reach the destination state is a constant independent of \(n\).

Hence probability of \(\sigma_\epsilon \) dying before reaching \(p_B \) with all the bins initiated to \(\lfloor \frac{\bv_0}{m} \rfloor \) is at most \(p_{n,1}+p_{n,2} \) and \(\lim_{n\rightarrow\infty}p_{n,1}+p_{n,2}=0 \).

\textbf{High-bin:}
Note that the probability \(p_{n,3}\) of any counter in the high-bin running out before \(n^{k+1-\epsilon} \) pointings at \(\M_{\hat{\br}} \) also goes to \(0\) as \(n\rightarrow\infty \), since this bin effectively simulates a computation on the VASS Markov chain \(\M_{\hat{\br}} \) restricted only to counters that are initialized to at least \(\lfloor \frac{n^{\lceil \frac{k+1}{2}\rceil-\epsilonr_{4}}}{m} \rfloor \) while the effect of \(\M_{\hat{\br}} \) on all these counters is either zero-bounded or zero-unbounded. And from theorem~\ref{thm-VASS-Markov-chain-analysis} we have that this computation has a quadratic lower asymptotic estimate on it's length before terminating. And for all sufficiently large \(n\) it holds 
\begin{gather*}
(\lfloor \frac{n^{\lceil \frac{k+1}{2}\rceil-\epsilonr_{4}}}{m} \rfloor)^{2-\epsilonr_{15}}  
\geq
(\lfloor \frac{n^{\frac{k+1}{2}-\epsilonr_{4}}}{m} \rfloor)^{2-\epsilonr_{15}} 
\geq 
( \frac{n^{\frac{k+1}{2}-\epsilonr_{4}}}{m} -1)^{2-\epsilonr_{15}}
\geq 
( \frac{n^{\frac{k+1}{2}-\epsilonr_{4}}}{2m})^{2-\epsilonr_{15}} 
\geq \\
\frac{n^{(2-\epsilonr_{15})\frac{k+1}{2}-(2-\epsilonr_{15})\epsilonr_{4}}}{(2m)^{2-\epsilonr_{15}}}
=
\frac{n^{k+1-\epsilonr_{15}\frac{k+1}{2}-(2-\epsilonr_{15})\epsilonr_{4}}}{(2m)^{2-\epsilonr_{15}}}
\geq  
n^{k+1-\epsilon} \end{gather*}

assuming \begin{equation}\label{eq-epsbound-bvcyuviuib8}
	-\epsilon<-\epsilonr_{15}\frac{k+1}{2}-(2-\epsilonr_{15})\epsilonr_{4} 
\end{equation}  Thus probability \(p_{n,3}\) of any counter in the high-bin becoming negative before \(n^{k+1-\epsilon} \) pointings at \(\M_{\hat{\br}} \) satisfies \(\lim_{n\rightarrow\infty} p_{n,3}=0 \).

\textbf{Main-bin:}
The only way for the main-bin to drop below \(0\) in any counter \(c\) is if while iterating \(\M_{\hat{\br}} \) the effect on c between two visits of \(p_\br\) is at most \(-n^{1-\epsilonr_{16}} \), for \begin{equation}\label{eq-epsbound-bvcyuviuib9}
\epsilonr_{16}>\epsilonr_{4}
\end{equation} which is only possible if this computation is longer than \(n^{1-\epsilonr_{17}} \) for \begin{equation}\label{eq-epsbound-bvcyuviuib10}
\epsilonr_{17}>\epsilonr_{16}
\end{equation} (since each step can change the counters by at most a constant value). And the probability of the iterations of \(\M_{\hat{\br}} \) producing such a path within first \( n^{k+1-\epsilon}\) iterations goes to \(0\) as \(n\rightarrow \infty \) since each at most constant number of steps the probability of having returned to \(p_\br\) is lower bounded by a positive constant, and hence if \(X\) denotes the numbers of such sub-computations on \(\M_{\hat{\br}}\) that do not contain \(p_{\br}\) and are of length \(n^{1-\epsilonr_{17}}\) withing the first \(n^{k+1-\epsilon} \) pointings has expectation \(\E_{p_\br\vec{n}}^{\sigma_\br}[X]\len^{k+1-\epsilon}a^{n^{1-\epsilonr_{17}}}  \) for some \(a<1 \), and thus from the Markov inequality it holds \(\E_{p_\br\vec{n}}^{\sigma_\br}[X\geq 1]\leq n^{k+1-\epsilon}a^{n^{1-\epsilonr_{17}}} \)  and \(\lim_{n\rightarrow\infty}n^{k+1-\epsilon}a^{n^{1-\epsilonr_{17}}}=0 \).

\textbf{\(\bu\)-bin:}
Note that whenever \(\sigma_\epsilon \) points at \(\M_{\hat{\br}} \) while in \( p_\br\) it holds \(|a_\bu|< n^{i-\epsilonr_{5}} \) for each \(\bu\in BASIS_i \). Hence the minimal value of each counter \(c\in C_j \) for \(j\leq k-i \)  in any \(\bu\)-bin at this point is \(\lfloor \frac{n^{j-\epsilonr_{4}}}{m}\rfloor-n^{i-\epsilonr_{18}} \) for \begin{equation}\label{eq-epsbound-bvcyuviuib11}
	\epsilonr_{18}<\epsilonr_{5}
\end{equation}
 which for \begin{equation}\label{eq-epsbound-bvcyuviuib12}
 	\epsilonr_{4}< \epsilonr_{18}
 \end{equation} is larger than \(0\) for all sufficiently large \(n\).

Thus if \(\sigma_{\epsilon} \) does not point at \(\M_{\hat{\br}} \) at least \(n^{k+1-\epsilon} \) times with high enough probability then it must die while performing some \(\bs \)-reset-procedure. 

We shall divide the following proof into three parts, first we will show that with high enough probability \(\sigma_{\epsilon} \) cannot deplete any counter from \(C_1,\dots,C_{\lfloor \frac{k}{2} \rfloor} \) during a \(\bs\)-reset for \(\bs\in BASIS_i\cup -BASIS_i \) before pointing to \(\M_{\br} \) at least \(n^{k+1-\epsilon} \) times, and then we shall show the same separately for the remaining counters from \(C_{\lfloor \frac{k}{2} \rfloor +1},\dots,C_{k-i} \) and then finally for the counters from \(C_{k-i+1},\dots,C_{k-1}, C_{k+} \).

\textbf{For the \(C_1,\dots,C_{\lfloor \frac{k}{2} \rfloor}  \) part:} 
First let us show that the \(\bu\)-bin corresponding to either \(\bs \) or \(-\bs \) (only one of these exists) cannot get depleted on these counters during a \(\bs \)-reset-procedure. At each point it holds that the counters vector in the \(\bu\)-bin is equal to \(\lfloor \frac{\bv_0}{m} \rfloor + a_\bu \bu \), and since during the \(\bs\)-reset-procedure the value \(|a_\bu| \) can only decrease no counter from \(C_1,\dots,C_{\lfloor \frac{k}{2} \rfloor}  \) can ever become negative during any reset procedure in the \(\bu\)-bin unless this counter was negative before this reset procedure even began.

Thus from Lemmas \ref{lemma-reset-procedure-main-bin-nonterminating} and \ref{lemma-reset-procedure-interface-transition-nk-i-eps101} the probability \(p_{n,\bs} \) of any counter \(c\in C_1,\dots,C_{\lfloor \frac{k}{2} \rfloor}   \) becoming negative in any bin modified during the \(\bs\)-reset-procedure before \(\bs\)-reset-procedure takes the ``interface transition'' at least \(n^{k-i+1-\epsilonr_{10}} \) times satisfies \(\lim_{n\rightarrow\infty} p_{n,\bs}=0 \). And since the \(\bs\)-reset-procedure gets paused every at least	   \(n^{i-\epsilonr_{5}}+n^{\epsilonr_{17}}\leq n^{i-\epsilonr_{19}}  \) times it uses the ``interface transition'', assuming \begin{equation}\label{eq-epsbound-bvcyuviuib13}
\epsilonr_{19}<\epsilonr_{5}
\end{equation} \( \) and \begin{equation}\label{eq-epsbound-bvcyuviuib14}
\epsilonr_{17}<1-\epsilonr_{5} 
\end{equation}  it holds that \( p_{n,\bs}\) is also an upper bound on the probability of  of any counter \(c\in C_1,\dots,C_{\lfloor \frac{k}{2} \rfloor}   \) becoming negative in any bin used by the \(\bs\)-reset-procedure before this \(\bs\)-reset-procedure  procedure is paused and unpaused at least \(\frac{n^{k-i+1-\epsilonr_{10}}}{n^{i-\epsilonr_{19}} }=n^{k-2i+1-\epsilonr_{10}+\epsilonr_{19}}\leq n^{k-2i+1-\epsilonr_{20}} \) times, where \begin{equation}\label{eq-epsbound-bvcyuviuib15}
\epsilonr_{20}<\epsilonr_{10}-\epsilonr_{19}
\end{equation} And from Lemma~\ref{lemma-resets-enough-for-us} there the probability of there being at least \(n^{k-2i+1-\epsilonr_{20}} \) such resets converges to \(1\) as \(n\) goes to \(\infty\).

\textbf{For the \(C_{\lfloor \frac{k}{2} \rfloor+1} ,\dots,C_{k-i}\)  part:}
From Lemma~\ref{lemma-reset-procedure-interface-transition-nk-i-eps101} we have that the probability of any counter from  \(C_{\lfloor \frac{k}{2} \rfloor+1} ,\dots,C_{k-i}\) being depleted in either bin directly modified by \(\sigma_\bs^{\epsilonr_{9}} \) in a \(\bs\)-reset-procedure is \(p_{n,\br,1} \) with \(\lim_{n\rightarrow\infty}p_{n,\br,1}=0 \). The only remaining bin where any counter from \(C_{\lfloor \frac{k}{2} \rfloor+1} ,\dots,C_{k-i}\) is modified during a \(\bs\)-reset-procedure is the high-\(\bu\)-bin, and it holds that the total effect of a \(\bs\)-reset-procedure from the moment it gets unpaused until it gets paused again on the counter \(c\in C_j\) for \(j\in \{\lfloor \frac{k}{2} \rfloor+1 ,\dots,C_{k-1}\} \) in the high-\(\bu\)-bin is \(a_\bu\bu(c) \) where \(\bu\in BASIS_i\cap \{\br,-\br \} \) and \(a_\bu \) is the value of \(a_\bu \) at the moment the \(\bs\)-reset-procedure is unpaused. And let \(X_\bu^\kappa \) be the random variable corresponding to the value of \(a_\bu \) at the moment either \(\bs\)-reset-procedure or the \((-\bs)\)-reset-procedure is unpaused for the \(\kappa\)-th time. Furthremore, let \(X_\kappa \) be the value of the counter \(c\) in the high-\(\bu\)-bin at the moment of the \(\kappa\)-th pausing of either \(\bu\)-reset-procedure or \((-\bu)\)-reset-procedure (i.e., either one counts).

Then since from the Lemma~\ref{lemma-a-bu-zero-exp} we have \(\E_n(X_\bu)=0 \) and it also holds that \(|X^\kappa_\bu|\leq n^{i-\epsilonr_{5}}+n^{\epsilonr_{17}} \) as well as  \(X_\kappa=X_{\kappa-1}+X^\kappa_\bu \) it holds that \(X_0,X_1,X_2,\dots \) is a
martingale with \(|X_{\kappa}-X_{\kappa+1}|\leq n^{i-\epsilonr_{5}}+n^{\epsilonr_{17}} \). 
Therefore we can apply Azuma's inequality to obtain that \[\prob_{p\vec{n}}^{\sigma_{\epsilon}}(|X_{\kappa}-X_0|\geq n^{j-\epsilonr_{23}})\leq 2 \exp(\frac{-(n^{j-\epsilonr_{23}})^2 }{2\sum_{r=1}^{\kappa} (n^{i-\epsilonr_{5}}+n^{\epsilonr_{17}})^2 }) \] Which for \(\kappa=n^{k-2i+1-\epsilonr_{22}} \) gives us 
\begin{gather*}
\prob(|X_{n^{k-2i+1-\epsilonr_{22}}}-X_0|\geq n^{j-\epsilonr_{23}})\leq 2 \exp(\frac{-n^{2j-2\epsilonr_{23}} }{2\sum_{r=1}^{n^{k-2i+1-\epsilonr_{22}}} (n^{i-\epsilonr_{5}}+n^{\epsilonr_{17}})^2 }) 
=\\
2 \exp(\frac{-n^{2j-2\epsilonr_{23}} }{2n^{k-2i+1-\epsilonr_{22}} (n^{i-\epsilonr_{5}}+n^{\epsilonr_{17}})^2 }) 
=
2 \exp(\frac{-n^{2j-k+2i-1+\epsilonr_{22}-2\epsilonr_{23}} }{2 (n^{i-\epsilonr_{5}}+n^{\epsilonr_{17}})^2 }) 
\leq \\
2 \exp(\frac{-n^{2j-k+2i-1+\epsilonr_{22}-2\epsilonr_{23}} }{ (n^{i-\epsilonr_{21}})^2 }) 
\end{gather*}

where we assume \begin{equation}\label{eq-epsbound-bvcyuviuib16}
	\epsilonr_{21}<\epsilonr_{5}
\end{equation}    and \begin{equation}\label{eq-epsbound-bvcyuviuib17}
\epsilonr_{17}<1-\epsilonr_{5} 
\end{equation}  We can further rewrite 

\begin{gather*}
2 \exp(\frac{-n^{2j-k+2i-1+\epsilonr_{22}-2\epsilonr_{23}} }{ (n^{i-\epsilonr_{21}})^2 }) 
=
2 \exp(\frac{-n^{2j-k+2i-1+\epsilonr_{22}-2\epsilonr_{23}} }{ n^{2i-2\epsilonr_{21}} }) 
=\\
2 \exp(-n^{2j-k+2i-2i-1+2\epsilonr_{21}+\epsilonr_{22}-2\epsilonr_{23}}  ) 
=
2 \exp(-n^{2j-k-1+2\epsilonr_{21}+\epsilonr_{22}-2\epsilonr_{23}} ) 
\end{gather*}

and since \(j\geq \lfloor\frac{k}{2} \rfloor+1 \) it holds \(2j-k-1\geq 0 \), thus 

\[
2 \exp(-n^{2j-k-1+2\epsilonr_{21}+\epsilonr_{22}-2\epsilonr_{23}} ) 
\leq 
2 \exp(-n^{0+2\epsilonr_{21}+\epsilonr_{22}-2\epsilonr_{23}} ) 
\]

which for \begin{equation}\label{eq-epsbound-bvcyuviuib171}
2\epsilonr_{21}+\epsilonr_{22}-2\epsilonr_{23}>0  
\end{equation} satisfies 
\(\lim_{n\rightarrow\infty}	2 \exp(-n^{0+2\epsilonr_{21}+\epsilonr_{22}-2\epsilonr_{23}} ) =0 \).  And from Lemma~\ref{lemma-resets-enough-for-us} there the probability of there being at least \(n^{k-2i+1-\epsilonr_{22}} \) such "resets" converges to \(1\) as \(n\) goes to \(\infty\).

\textbf{For the \(C_{k+}\) part:}
From Lemma~\ref{lemma-reset-procedure-interface-transition-nk-i-eps101} we have that the probability of any counter from  \(C_{k+}\) being depleted in any bin directly modified by \(\sigma_\bs^{\epsilonr_{9}} \) in a \(\bs\)-reset-procedure is \(p_{n,\bs,1} \) with \(\lim_{n\rightarrow\infty}p_{n,\bs,1}=0 \). The only remaining bin where any counter from \(C_{k+}\) is modified during a \(\bs\)-reset-procedure is the high-\(\bu\)-bin. But each counter  \(c\in C_{k+} \) starts at at least \(\lfloor \frac{n^{k-\epsilonr_{4}}}{m} \rfloor \) in the high-\(\bu\)-bin and during every \(\bu\)-reset or \((-\bu)\)-reset it gets decreased by at most \(u(n^{i-\epsilonr_{5}}+n^{\epsilonr_{24}}) \), where \(\bu\in BASIS_i \) and \(u=\max \{|\Delta(\bx_{i,\bu}^+)(c)|,|\Delta(\bx_{i,\bu}^-)(c)| \} \). Thus the counter \(c\) cannot get depleted in the high-\(\bu\)-bin until a \(\bu\)-reset or \((-\bu)\)-reset is performed at least \[
\frac{\lfloor \frac{n^{k-\epsilonr_{4}}}{m} \rfloor}{ u(n^{i-\epsilonr_{5}}+n^{\epsilonr_{24}}) } 
\geq 
\frac{\frac{n^{k-\epsilonr_{4}}}{m}-1 }{ u(n^{i-\epsilonr_{5}}+n^{\epsilonr_{24}}) }
\geq 
\frac{\frac{n^{k-\epsilonr_{4}}}{m}-1 }{ n^{i-\epsilonr_{25}} }
\]
assuming \begin{equation}\label{eq-epsbound-bvcyuviuib18}
	\epsilonr_{25}<\epsilonr_{5}
\end{equation}  and\begin{equation}\label{eq-epsbound-bvcyuviuib19}
1-\epsilonr_{25}>\epsilonr_{17}
\end{equation} 

\[
\frac{\frac{n^{k-\epsilonr_{4}}}{m}-1 }{ n^{i-\epsilonr_{25}} }
\geq
\frac{\frac{n^{k-\epsilonr_{4}}}{2m} }{ n^{i-\epsilonr_{25}} }
=
\frac{n^{k-i+\epsilonr_{25}-\epsilonr_{4}} }{ 2m }
\geq 
n^{k-i-\epsilonr_{26}}	
\]
assuming  \begin{equation}\label{eq-epsbound-bvcyuviuib20}
	\epsilonr_{26}\leq \epsilonr_{4}-\epsilonr_{25}
\end{equation}  and\begin{equation}\label{eq-epsbound-bvcyuviuib21}
	\epsilonr_{17}<1
\end{equation}  
And from Lemma~\ref{lemma-resets-enough-for-us} there the probability of there being at least \(n^{k-i-\epsilonr_{26}}	 \) such "resets" converges to \(1\) as \(n\) goes to \(\infty\).

\begin{lemma}\label{lemma-a-bu-zero-exp}
let \(X_\bu \) be the random variable corresponding to the value of \(a_\bu \) at the moment either \(\bu\)-reset-procedure or the \((-\bu)\)-reset-procedure is unpaused. Then it holds \(\E_n(X_\bu)=0 \).
\end{lemma}
\begin{proof}
(note that this is not exactly true, but we can show that it does hold if we modify \(\sigma_{\epsilon} \) to always subtract a very small constant equal to \(\E_{p_{\br}\vec{0}}^\sigma[f_\bu(\realeffect_{\hat{\br}})] \) from \(a_\bu \) whenever it modifies it upon reaching \(p_\br \) from the iterating on \(\M_{\hat{\br}} \), and over the first \(n^{k+1-\epsilon} \) pointings at \(\M_{\hat{\br}} \) the total effect of this addition will not even be enough to add up to \(1\), thus we can safely ignore it. We leave it out of the definition of \(\sigma_\epsilon \) as it would considerably increase the technical difficulty of all the above proofs)

Notice that \(X_\bu \) (with the aforementioned modification) corresponds to the random variable \(S_\tau \) defined as follows:
\(S_0=0\),  \(S_i=S_{i-1}+f_\bu(\realeffect_{\hat{\br}}^{\epsilonr_{17}})-\E_{p_{\br}\vec{0}}^\sigma[f_\bu(\realeffect_{\hat{\br}})] \), where \(f_\bu(\realeffect_{\hat{\br}}^{\epsilonr_{17}}) \) is as defined in Lemma~\ref{lemma-rand-walk-double-exp-zero},  and \(\tau\) is a stopping time such that \(|S_\tau|\geq n^{i-\epsilonr_{5}}  \). From 	Lemma~\ref{lemma-rand-walk-double-exp-zero} we have that 	\(\big|\E_{p_{\br}\vec{0}}^\sigma[f_\bu(\realeffect_{\hat{\br}}^{\epsilonr_{17}})]\big|\in \bigO(a^{ n^{\epsilonr_{27}}}) \) for some \(a<1 \) and \begin{equation}\label{eq-epsbound-bvcyuviuib22}
	\epsilonr_{17}>\epsilonr_{27}
\end{equation}  and thus \(\big|\E_{p_{\br}\vec{0}}^\sigma[f_\bu(\realeffect_{\hat{\br}}^{\epsilonr_{17}})]\leq ba^{n^{\epsilonr_{27}}}\) for some constant \(b \) and for all sufficiently large \(n\).

%		Let us consider the random variables 
%		\(S_0'=0\),  \(S_i=S_{i-1}+f_\bu(\realeffect_{\hat{\br}}) \).

%		Then from Lemma~\ref{lemma-rand-walk-double-exp-zero} we have that \(\E_{p_{\br}\vec{0}}^\sigma[f_\bu(\realeffect_{\hat{\br}})]=0  \) and thus \(S_1',S_2',\dots \) is a martingale. 

Clearly it holds \(\E_{p_{\br}\vec{0}}^\sigma[S_i]=0 \) and therefore from the optional stopping theorem it holds \(0=\E_{p_{\br}\vec{0}}^\sigma[S_0]=\E_{p_{\br}\vec{0}}^\sigma[S_\tau] \).  If \(\tau>n^{k+1-\epsilon} \) then \(\sigma_{\epsilon} \) had to already point at \(\M_{\hat{\br}} \) at least \(n^{k+1-\epsilon}\) times, and thus we do not have to consider such case. And if \(\tau\leq n^{k+1-\epsilon} \) then it holds that \(\tau|\E_{p_{\br}\vec{0}}^\sigma[f_\bu(\realeffect_{\hat{\br}})]|\leq n^{k+1-\epsilon}ba^{n^{\epsilonr_{27}}}\) and it holds \(\lim_{n\rightarrow\infty}n^{k+1-\epsilon}ba^{n^{\epsilonr_{27}}}=0 \), thus the difference caused by our modification of \(\sigma_\epsilon \) is negligible.		
\end{proof}

\textbf{Number of resets:}
Thus it now suffices to show that the probability of there being at least \(n^{k+1-\epsilon} \) pointings at \(\M_{\hat{\br}} \) before there are at least \(\min(n^{k-i-\epsilonr_{26}}, n^{k-2i+1-\epsilonr_{20}}, n^{k-2i+1-\epsilonr_{22}}	)\geq  \lfloor n^{k-2i+1-\epsilonr_{28}}\rfloor \)  \(\bu\)-resets or \((-\bu)\)-resets for each \(\bu\in BASIS_i \) goes to \(1\) as \(n \) goes to infinity, where \begin{equation}\label{eq-epsbound-bvcyuviuib23}
	\epsilonr_{28}>\epsilonr_{20}
\end{equation}
\begin{equation}\label{eq-epsbound-bvcyuviuib24}
\epsilonr_{28}>\epsilonr_{22}
\end{equation}
\begin{equation}\label{eq-epsbound-bvcyuviuib25}
	\epsilonr_{28}-1<-\epsilonr_{26}
\end{equation} 

%	Therefore the probability of any bin reaching negative value for any counter before there are at least \(\min(n^{k-i-\epsilonr_{26}}, n^{k-2i+1-\epsilonr_{20}}, n^{k-2i+1-\epsilonr_{22}}	)\geq  n^{k-2i+1-\epsilonr_{28}} \)  \(\bu\)-"resets" or \((-\bu)\)-resets for each \(\bu\in BASIS_i \) 

%	is \(p_{n,5}\leq \sum_{\bs\in BASIS\cup -BASIS} (2 \exp(-n^{2\epsilonr_{21}+\epsilonr_{22}-2\epsilonr_{23}} ) + p_{n,\bs}) + p_{n,1}+ p_{n,2}+p_{n,3} \) and it holds \(\lim_{n\rightarrow\infty} p_{n,5}\leq \lim_{n\rightarrow\infty}2 \exp(-n^{2\epsilonr_{21}+\epsilonr_{22}-2\epsilonr_{23}} ) + \sum_{\br\in BASIS\cup -BASIS}p_{n,\br} + p_{n,1}+ p_{n,2}+p_{n,3} =0 \), 

\begin{lemma}\label{lemma-resets-enough-for-us}
For each \(\bu\in BASIS_i \) the probability of \(\sigma_{\epsilon} \) making at least \(n^{k+1-\epsilon} \) poinitngs at \(\M_{\hat{\br}} \) before there are \(\lfloor n^{k-2i+1-\epsilonr_{28}}\rfloor\) \(\bu\)-resets or "\(-\bu\)"-"resets" goes to \(1\) as \(n\) goes to \(\infty \). 
\end{lemma}
\begin{proof}
Let \(Z_n^\bu \) denote the number of \(\bu\)-"resets" or \((-\bu)\)-"resets" for \(\bu\in BASIS_i \) that do not give us at least \(n^{2i-\epsilonr_{30}}\)  pointings at \(\M_{\hat{\br}} \) among the first  \(\lfloor n^{k-2i+1-\epsilonr_{28}}\rfloor\) such "resets", conditioned there are at least  \(n^{k-2i+1-\epsilonr_{28}} \) such resets (i.e., every time \(\sigma_{\epsilon} \) pointed at \(\M_{\hat{\br}} \) less than \(n^{2i-\epsilonr_{30}} \) times before \(|a_\bu|\geq n^{i-\epsilonr_{5}} \) since the last time a \(\bu\)-reset or \((-\bu) \)-reset happened then we add \(+1\) to \(Z_n^u \), but only if the total number of times these rests happened is no larger than \(n^{k-2i+1-\epsilonr_{28}} \)). We denote by * the event of there being at least \(n^{k-2i+1-\epsilonr_{28}} \) such reset.  It holds \(\E_{p\vec{n}}^{\sigma_{\epsilon}}(Z_n^\bu\mid *)= q_n^\bu \lfloor n^{k-2i+1-\epsilonr_{28}}\rfloor \), where \(q_n^\bu \) is the probability that a single \(\bu\)-reset or \((-\bu)\)-reset does not give at least \(n^{2i-\epsilonr_{30}}\)  pointings at \(\M_{\hat{\br}} \). 

We can upper bound \(q_n^\bu \) using Azuma inequality as follows: Let \(S_0=0 \) and \(S_i=S_{i-1}+f_\bu(\realeffect_{\hat{\br}}^{\epsilonr_{17}})-\E_{p_{\br}\vec{0}}^\sigma[f_\bu(\realeffect_{\hat{\br}})] \) (again here we consider the \(\sigma_{\epsilon} \) modified as described under Lemma~\ref{lemma-a-bu-zero-exp}), then \(S_0,S_1,\dots  \) is a martingale with \(|S_i-S_{i-1}|\leq un^{\epsilonr_{17}} \) for some constant \(u\), and it holds that the number of pointings  at \(\M_{\hat{\br}} \) per single  \(\bu\)-reset or \((-\bu)\)-reset is at least \(\tau \) where \(\tau\) is the stopping time such that \(|S_\tau|\geq n^{i-\epsilonr_{5}} \). Thus from Azuma inequality we get \begin{multline*}
	q_n^\bu\leq \prob[|S_{n^{2i-\epsilonr_{30}}}-S_0|\geq n^{i-\epsilonr_{5}}]\leq 2\exp(\frac{-(n^{i-\epsilonr_{5}})^2}{2\sum_{j=1}^{n^{2i-\epsilonr_{30}}}(un^{\epsilonr_{17}})^2 }) 
	= \\
	2\exp(\frac{-n^{2i-2\epsilonr_{5}}}{2n^{2i-\epsilonr_{30}}u^2n^{2\epsilonr_{17}} })
	=
	2\exp(\frac{-n^{2i-2i+\epsilonr_{30}-2\epsilonr_{5}-2\epsilonr_{17}}}{2u^2 })
	=
	2\exp(\frac{-n^{\epsilonr_{30}-2\epsilonr_{5}-2\epsilonr_{17}}}{2u^2})
\end{multline*}

and for \begin{equation}\label{eq-epsbound-bvcyuviuib26}
	0\leq \epsilonr_{30}-2\epsilonr_{5}-2\epsilonr_{17}=\epsilonr_{29}
\end{equation}   it holds that \[\lim_{n\rightarrow\infty}q_n^\bu\leq \lim_{n\rightarrow\infty} 2\exp(\frac{-n^{\epsilonr_{29}}}{2u^2})=0 \]

Therefore from Markov inequality we obtain \[
\prob[Z_n^\bu\geq \frac{1}{2}\lfloor n^{k-2i+1-\epsilonr_{28}}\rfloor\mid * ]
\leq
\frac{2q_n^\bu \lfloor n^{k-2i+1-\epsilonr_{28}}\rfloor}{\lfloor n^{k-2i+1-\epsilonr_{28}}\rfloor}
=
2q_n^\bu
\] Thus it holds it holds  \(\lim_{n\rightarrow\infty}\prob[Z_n^\bu\geq \frac{1}{2}\lfloor n^{k-2i+1-\epsilonr_{28}}\rfloor\mid * ]\leq \lim_{n\rightarrow\infty}2q_n^\bu=0 \).

Thus the probability of there being at least \(n^{2i-\epsilonr_{30}} \lfloor n^{k-2i+1-\epsilonr_{28}}\rfloor\geq \frac{1}{2}n^{k+1-\epsilonr_{30}-\epsilonr_{28}} \) pointings at \(\M_{\hat{\br}} \) before there are at least \( n^{k-2i+1-\epsilonr_{28}} \) \(\bu\)-"resets" or \((-\bu)\)-"resets" for each \(\bu\in BASIS_i \) goes to \(1\) as \(n \) goes to infinity. And assuming  \begin{equation}\label{eq-epsbound-bvcyuviuib27}
	\epsilonr_{30}+\epsilonr_{28}<\epsilon
\end{equation}  it holds \(n^{k+1-\epsilon}\leq \frac{1}{2}n^{k+1-\epsilonr_{30}-\epsilonr_{28}}  \) for all sufficiently large \(n\).
\end{proof}

Hence \(\sigma_\epsilon \) satisfies \(\lim_{n\rightarrow\infty}\prob_{p\vec{n}}^{\sigma_\epsilon}[\calP_{\A_{+\hat{br}}}[\M_{\hat{\br}}]]=1 \).

It remains to show there exist values for \(\epsilon_1,\epsilon_2,\dots \) that satisfy all of our assumptions. We do this in Table~\ref{Tadsdsble-eps-section-anotdfdsfsdfdher-proof-k+1}.

\begin{table*}[h]
	\caption{Values of \(\epsilon_1,\epsilon_2,\dots \) for Section~\ref{sec-pointingupperbound-app}. The equations after substitution can be found in Table~\ref{table-subst}}
	\centering
	%	\begin{center}
		\begin{tabular}{|l|| c c|} 
			\hline
			\(\epsilon\) assignment &  \multicolumn{2}{|c|}{restrictions}\\ 
			\hline\hline
			%		\multirow{10}{*}{ \begin{tabular}{c}
					%			 \(\epsilon_1=\frac{9}{10} \)
					%				\\ \(\epsilon_3=\frac{1}{8} \)
					%				\\ \(\epsilon_{4}=\frac{1}{10} \)
					%				\\ \(\epsilon_{5}=\frac{5}{12} \)
					%				\\ \(\epsilon_{6}=\frac{1}{2} \)
					%				\\ \(\epsilon_7=\frac{1}{7} \)
					%		\end{tabular} }
			\(\epsilon_1=\epsilon  \) & \(0<\epsilon_1,\epsilon_2,\dots  \)&\\ 
			\hline
			\(\epsilon_2=\epsilon  \) & \(\epsilonr_{5}< k\epsilonr_{9}\) & \eqref{eq-epsbound-bvcyuviuib}  \\
			\hline
			\(\epsilon_{3}=\epsilon  \) & 	\(\epsilonr_{12}>\epsilonr_{10}\) & \eqref{eq-epsbound-bvcyuviuib2}  \\
			\hline
			\(\epsilon_{4}=\frac{\min(\epsilon,\nicefrac{1}{2})}{10^{40}} \) & \(\epsilonr_{10}<\epsilonr_{12}<1\) & \eqref{eq-epsbound-bvcyuviuib3}  \\
			\hline
			\(\epsilon_{5}=\frac{\min(\epsilon,\nicefrac{1}{2})}{10^{29}}\) & \(\epsilonr_{4}+\epsilonr_{11}<\epsilonr_{12} \) & \eqref{eq-epsbound-bvcyuviuib4}   \\ 
			\hline
			\(\epsilon_{6}=\epsilon \) & \(\epsilonr_{9}+(\epsilonr_{14}-(1-\epsilonr_{14})\epsilonr_{4})>\epsilonr_{10} \) & \eqref{eq-epsbound-bvcyuviuib5}  \\
			\hline
			\(\epsilon_{7}=\epsilon \)	& \(	-\epsilonr_{31}>-\epsilonr_{9}-k\epsilonr_{14}-k(1-\epsilonr_{14})\epsilonr_{4}
			\) & \eqref{eq-epsbound-bvcyuviuib6}  \\
			\hline
			\(\epsilon_{8}=\epsilon \) & \(-\epsilonr_{31}<-(k+1)\epsilonr_{9}  \) & \eqref{eq-epsbound-bvcyuviuib7} \\
			\hline
			\(\epsilon_{9}=\frac{\min(\epsilon,\nicefrac{1}{2})}{10^{25}} \) & \(-\epsilon<-\epsilonr_{15}\frac{k+1}{2}-(2-\epsilonr_{15})\epsilonr_{4}  \) & \eqref{eq-epsbound-bvcyuviuib8}  \\
			\hline
			\(\epsilon_{10}=\frac{\min(\epsilon,\nicefrac{1}{2})}{10^{25}} \) & \(\epsilonr_{16}>\epsilonr_{4} \) & \eqref{eq-epsbound-bvcyuviuib9} \\
			\hline
			\(\epsilon_{11}=\frac{\min(\epsilon,\nicefrac{1}{2})}{10^{90}} \) & \(\epsilonr_{17}>\epsilonr_{16} \) & \eqref{eq-epsbound-bvcyuviuib10}  \\
			\hline
			\(\epsilon_{12}=\nicefrac{9}{10}  \) & \(\epsilonr_{18}<\epsilonr_{5} \) & \eqref{eq-epsbound-bvcyuviuib11}\\
			\hline
			\(\epsilon_{13}=\epsilon \) & \(\epsilonr_{4}< \epsilonr_{18} \) & \eqref{eq-epsbound-bvcyuviuib12} \\
			\hline
			\(\epsilon_{14}=\frac{ \min(\epsilon,\nicefrac{1}{2})}{10^{20}} \) & \(\epsilonr_{19}<\epsilonr_{5} \) & \eqref{eq-epsbound-bvcyuviuib13} \\
			\hline
			\(\epsilon_{15}=\frac{\min(\epsilon,\nicefrac{1}{2})}{(k+1)\cdot 10^{20}} \) & \(\epsilonr_{17}<1-\epsilonr_{5} \) & \eqref{eq-epsbound-bvcyuviuib14}  \\
			\hline
			\(\epsilon_{16}=\frac{\min(\epsilon,\nicefrac{1}{2})}{10^{90}} \) & \(\epsilonr_{20}<\epsilonr_{10}-\epsilonr_{19} \) & \eqref{eq-epsbound-bvcyuviuib15} \\
			\hline
			\(\epsilon_{17}=\frac{\min(\epsilon,\nicefrac{1}{2})}{10^{80}} \) & \(\epsilonr_{21}<\epsilonr_{5} \) & \eqref{eq-epsbound-bvcyuviuib16} \\
			\hline
			\(\epsilon_{18}=\frac{\min(\epsilon,\nicefrac{1}{2})}{10^{30}} \) & \(\epsilonr_{17}<1-\epsilonr_{5}  \) & \eqref{eq-epsbound-bvcyuviuib17}  \\
			\hline
			\(\epsilon_{19}=\frac{\min(\epsilon,\nicefrac{1}{2})}{10^{100}} \) & \(2\epsilonr_{21}+\epsilonr_{22}-2\epsilonr_{23}>0  \) & \eqref{eq-epsbound-bvcyuviuib171} \\
			\hline
			\(\epsilon_{20}=\frac{\min(\epsilon,\nicefrac{1}{2})}{10^{30}} \) & \(\epsilonr_{25}<\epsilonr_{5} \) & \eqref{eq-epsbound-bvcyuviuib18}  \\
			\hline
			\(\epsilon_{21}=\frac{\min(\epsilon,\nicefrac{1}{2})}{10^{30}} \) & \(1-\epsilonr_{25}>\epsilonr_{17} \) & \eqref{eq-epsbound-bvcyuviuib19}  \\
			\hline
			\(\epsilon_{22}=\frac{\min(\epsilon,\nicefrac{1}{2})}{10^{100}} \) & \(\epsilonr_{26}\leq \epsilonr_{4}-\epsilonr_{25} \) & \eqref{eq-epsbound-bvcyuviuib20}  \\
			\hline
			\(\epsilonr_{23}=\frac{\min(\epsilon,\nicefrac{1}{2})}{10^{100}} \) & \(\epsilonr_{17}<1 \) & \eqref{eq-epsbound-bvcyuviuib21}  \\
			\hline
			\(\epsilon_{24}=\epsilon \) & \(\epsilonr_{17}>\epsilonr_{27} \) & \eqref{eq-epsbound-bvcyuviuib22} \\
			\hline
			\(\epsilon_{25}=\frac{\min(\epsilon,\nicefrac{1}{2})}{10^{100}} \) & \(\epsilonr_{28}>\epsilonr_{20} \) & \eqref{eq-epsbound-bvcyuviuib23}  \\
			\hline
			\(\epsilon_{26}=\frac{\min(\epsilon,\nicefrac{1}{2})}{10^{100}} \) & \(\epsilonr_{28}>\epsilonr_{22} \) & \eqref{eq-epsbound-bvcyuviuib24}  \\
			\hline
			\(\epsilon_{27}=\frac{\min(\epsilon,\nicefrac{1}{2})}{10^{100}} \) & \(\epsilonr_{28}-1<-\epsilonr_{26} \) & \eqref{eq-epsbound-bvcyuviuib25}\\
			\hline
			\(\epsilon_{28}=\frac{\min(\epsilon,\nicefrac{1}{2})}{10^{20}} \) & \(0\leq \epsilonr_{30}-2\epsilonr_{5}-2\epsilonr_{17}=\epsilonr_{29} \) & \eqref{eq-epsbound-bvcyuviuib26}  \\
			\hline
			\(\epsilon_{29}=\frac{\min(\epsilon,\nicefrac{1}{2})}{10}-2\frac{\min(\epsilon,\nicefrac{1}{2})}{10^{29}}-2\frac{\min(\epsilon,\nicefrac{1}{2})}{10^{80}} \) & \(\epsilonr_{30}+\epsilonr_{28}<\epsilon \) & \eqref{eq-epsbound-bvcyuviuib27}  \\
			\hline
			\(\epsilon_{30}=\frac{\min(\epsilon,\nicefrac{1}{2})}{10} \) &   &      \\
			\hline
			\(\epsilon_{31}=(k+1)\frac{\min(\epsilon,\nicefrac{1}{2})}{10^{24}} \) &   &      \\
			%			\cline{2-4}
			%				& \(2\epsilon_{5}-2\epsilon_{7}-\epsilon_{6}>0 \) & 18744 & 7560 \\
			%				\cline{2-4}
			%				& \(\epsilon_{1}>\epsilon \) & 18744 & 7560 \\
			%				\cline{2-4}
			%				& 		 \( 0<\epsilon<\frac{1}{10}\) & 18744 & 7560 \\
			\hline
		\end{tabular}
		\par
		\label{Tadsdsble-eps-section-anotdfdsfsdfdher-proof-k+1}	
		%	\end{center}
\end{table*}

\begin{table*}[h]
	\caption{Values of \(\epsilon_1,\epsilon_2,\dots \) for Section~\ref{sec-pointingupperbound-app} after substitution.}
	\centering
	%	\begin{center}
		\begin{tabular}{| c c|} 
			\hline
			\multicolumn{2}{|c|}{ After substitution}  \\ 
			\hline\hline
			%		\multirow{10}{*}{ \begin{tabular}{c}
					%			 \(\epsilon_1=\frac{9}{10} \)
					%				\\ \(\epsilon_3=\frac{1}{8} \)
					%				\\ \(\epsilon_{4}=\frac{1}{10} \)
					%				\\ \(\epsilon_{5}=\frac{5}{12} \)
					%				\\ \(\epsilon_{6}=\frac{1}{2} \)
					%				\\ \(\epsilon_7=\frac{1}{7} \)
					%		\end{tabular} }
			\(\frac{\min(\epsilon,\nicefrac{1}{2})}{10^{29}}< k\frac{\min(\epsilon,\nicefrac{1}{2})}{10^{25}} \) & \eqref{eq-epsbound-bvcyuviuib} \\
			\hline
			\(\nicefrac{9}{10}>\frac{\min(\epsilon,\nicefrac{1}{2})}{10^{25}}\)& \eqref{eq-epsbound-bvcyuviuib2} \\
			\hline
			\(\frac{\min(\epsilon,\nicefrac{1}{2})}{10^{25}}<\nicefrac{9}{10}<1\)&  \eqref{eq-epsbound-bvcyuviuib3} \\
			\hline
			\(\frac{\min(\epsilon,\nicefrac{1}{2})}{10^{40}}+\frac{\min(\epsilon,\nicefrac{1}{2})}{10^{90}}<\nicefrac{9}{10} \) &\eqref{eq-epsbound-bvcyuviuib4} \\ 
			\hline
			\(\frac{\min(\epsilon,\nicefrac{1}{2})}{10^{25}}+(\frac{ \min(\epsilon,\nicefrac{1}{2})}{10^{20}} -(1-\frac{ \min(\epsilon,\nicefrac{1}{2})}{10^{20}} )\frac{\min(\epsilon,\nicefrac{1}{2})}{10^{40}})>\frac{\min(\epsilon,\nicefrac{1}{2})}{10^{25}} \) & \eqref{eq-epsbound-bvcyuviuib5} \\
			\hline
			\(-(k+1)\frac{\min(\epsilon,\nicefrac{1}{2})}{10^{25}}>-\frac{\min(\epsilon,\nicefrac{1}{2})}{10^{25}}-k\frac{ \min(\epsilon,\nicefrac{1}{2})}{10^{20}} -k(1-\frac{ \min(\epsilon,\nicefrac{1}{2})}{10^{20}} )\frac{\min(\epsilon,\nicefrac{1}{2})}{10^{40}}
			\) & \eqref{eq-epsbound-bvcyuviuib6}\\
			\hline
			\(-(k+1)\frac{\min(\epsilon,\nicefrac{1}{2})}{10^{24}}<-(k+1)\frac{\min(\epsilon,\nicefrac{1}{2})}{10^{25}}  \) & \eqref{eq-epsbound-bvcyuviuib7}\\
			\hline
			\(-\epsilon<-\frac{\min(\epsilon,\nicefrac{1}{2})}{(k+1)\cdot 10^{20}}\frac{k+1}{2}-(2-\frac{\min(\epsilon,\nicefrac{1}{2})}{(k+1)\cdot 10^{20}})\frac{\min(\epsilon,\nicefrac{1}{2})}{10^{40}}  \)&\eqref{eq-epsbound-bvcyuviuib8} \\
			\hline
			\(\frac{\min(\epsilon,\nicefrac{1}{2})}{10^{90}}>\frac{\min(\epsilon,\nicefrac{1}{2})}{10^{40}} \)&\eqref{eq-epsbound-bvcyuviuib9} \\
			\hline
			\(\frac{\min(\epsilon,\nicefrac{1}{2})}{10^{80}}>\frac{\min(\epsilon,\nicefrac{1}{2})}{10^{90}} \)&  \eqref{eq-epsbound-bvcyuviuib10} \\
			\hline
			\(\frac{\min(\epsilon,\nicefrac{1}{2})}{10^{30}}<\frac{\min(\epsilon,\nicefrac{1}{2})}{10^{29}} \) & \eqref{eq-epsbound-bvcyuviuib11} \\
			\hline
			\(\frac{\min(\epsilon,\nicefrac{1}{2})}{10^{40}}< \frac{\min(\epsilon,\nicefrac{1}{2})}{10^{30}} \) &\eqref{eq-epsbound-bvcyuviuib12} \\
			\hline
			\(\frac{\min(\epsilon,\nicefrac{1}{2})}{10^{100}}<\frac{\min(\epsilon,\nicefrac{1}{2})}{10^{29}} \)& \eqref{eq-epsbound-bvcyuviuib13} \\
			\hline
			\(\frac{\min(\epsilon,\nicefrac{1}{2})}{10^{80}}<1-\frac{\min(\epsilon,\nicefrac{1}{2})}{10^{29}} \)& \eqref{eq-epsbound-bvcyuviuib14} \\
			\hline
			\(\frac{\min(\epsilon,\nicefrac{1}{2})}{10^{30}}<\frac{\min(\epsilon,\nicefrac{1}{2})}{10^{25}}-\frac{\min(\epsilon,\nicefrac{1}{2})}{10^{100}} \)& \eqref{eq-epsbound-bvcyuviuib15} \\
			\hline
			\(\frac{\min(\epsilon,\nicefrac{1}{2})}{10^{30}}<\frac{\min(\epsilon,\nicefrac{1}{2})}{10^{29}} \)& \eqref{eq-epsbound-bvcyuviuib16} \\
			\hline
			\(\frac{\min(\epsilon,\nicefrac{1}{2})}{10^{80}}<1-\frac{\min(\epsilon,\nicefrac{1}{2})}{10^{29}}  \)& \eqref{eq-epsbound-bvcyuviuib17} \\
			\hline
			\(2\frac{\min(\epsilon,\nicefrac{1}{2})}{10^{30}}+\frac{\min(\epsilon,\nicefrac{1}{2})}{10^{100}}-2\frac{\min(\epsilon,\nicefrac{1}{2})}{10^{100}}>0  \) &\eqref{eq-epsbound-bvcyuviuib171}\\
			\hline
			\(\frac{\min(\epsilon,\nicefrac{1}{2})}{10^{100}}<\frac{\min(\epsilon,\nicefrac{1}{2})}{10^{29}} \) & \eqref{eq-epsbound-bvcyuviuib18} \\
			\hline
			\(1-\frac{\min(\epsilon,\nicefrac{1}{2})}{10^{100}}>\frac{\min(\epsilon,\nicefrac{1}{2})}{10^{80}} \) &\eqref{eq-epsbound-bvcyuviuib19} \\
			\hline
			\(\frac{\min(\epsilon,\nicefrac{1}{2})}{10^{100}}\leq \frac{\min(\epsilon,\nicefrac{1}{2})}{10^{40}}-\frac{\min(\epsilon,\nicefrac{1}{2})}{10^{100}} \) & \eqref{eq-epsbound-bvcyuviuib20}\\
			\hline
			\(\frac{\min(\epsilon,\nicefrac{1}{2})}{10^{80}}<1 \) &  \eqref{eq-epsbound-bvcyuviuib21} \\
			\hline
			\(\frac{\min(\epsilon,\nicefrac{1}{2})}{10^{80}}>\frac{\min(\epsilon,\nicefrac{1}{2})}{10^{100}} \) &\eqref{eq-epsbound-bvcyuviuib22}\\
			\hline
			\(\frac{\min(\epsilon,\nicefrac{1}{2})}{10^{20}}>\frac{\min(\epsilon,\nicefrac{1}{2})}{10^{30}} \) & \eqref{eq-epsbound-bvcyuviuib23}\\
			\hline
			\(\frac{\min(\epsilon,\nicefrac{1}{2})}{10^{20}}>\frac{\min(\epsilon,\nicefrac{1}{2})}{10^{100}} \) &\eqref{eq-epsbound-bvcyuviuib24}\\
			\hline
			\(\frac{\min(\epsilon,\nicefrac{1}{2})}{10^{20}}-1<-\frac{\min(\epsilon,\nicefrac{1}{2})}{10^{100}} \) &\eqref{eq-epsbound-bvcyuviuib25}\\
			\hline
			\(0\leq \frac{\min(\epsilon,\nicefrac{1}{2})}{10}-2\frac{\min(\epsilon,\nicefrac{1}{2})}{10^{29}}-2\frac{\min(\epsilon,\nicefrac{1}{2})}{10^{80}}=\frac{\min(\epsilon,\nicefrac{1}{2})}{10}-2\frac{\min(\epsilon,\nicefrac{1}{2})}{10^{29}}-2\frac{\min(\epsilon,\nicefrac{1}{2})}{10^{80}} \)&\eqref{eq-epsbound-bvcyuviuib26} \\
			\hline
			\(\frac{\min(\epsilon,\nicefrac{1}{2})}{10}+\frac{\min(\epsilon,\nicefrac{1}{2})}{10^{20}}<\epsilon \) & \eqref{eq-epsbound-bvcyuviuib27}  \\
			%			\cline{2-4}
			%				& \(2\epsilon_{5}-2\epsilon_{7}-\epsilon_{6}>0 \) & 18744 & 7560 \\
			%				\cline{2-4}
			%				& \(\epsilon_{1}>\epsilon \) & 18744 & 7560 \\
			%				\cline{2-4}
			%				& 		 \( 0<\epsilon<\frac{1}{10}\) & 18744 & 7560 \\
			\hline
		\end{tabular}
		\par
		\label{table-subst}	
		%	\end{center}
\end{table*}
	
	Hence Lemma~\ref{lemma-D-in-R-lower-estimate-nk+1-for-hatB} holds.
	
	\flushbottom
	\flushbottom

%\(P(Z_n\geq q_n n^{1-\frac{\epsilonr_{16}}{2}} )\leq \frac{q_n n^{1-\epsilonr_{16}}}{q_n n^{1-\frac{\epsilonr_{16}}{2}}} = \frac{1}{n^{\frac{\epsilonr_{16}}{2}}} \). Thherefore with probability at least \(1-\frac{1}{n^{\frac{\epsilonr_{16}}{2}}} \) we obtain at least \(n^{1-\epsilonr_{16}} - q_n n^{1-\frac{\epsilonr_{16}}{2}} \)

\subsubsection{Proof of Lemma~\ref{lemma-D-in-Rl-for-all-l-imply-D-in-R} and Lemma~\ref{lemma-D-not-in-Rl-implies-B-zero-unbounded-rankl-dsada}}

\begin{lemma}\label{lemma-D-in-Rl-for-all-l-imply-D-in-R}
	Let \(\br \) be a component  of \(\A\).	If it holds \(\support^{C_1,\dots,C_l}(\hat{\br})\subseteq  \Delta^{C_1,\dots,C_l}(X_{k-l,T_{k+1-l}}) \) for all \(1\leq l\leq \lfloor\frac{k}{2}\rfloor \), then \(\support^{C_1,\dots,C_{\lfloor\frac{k}{2}\rfloor}}(\hat{\br})\subseteq R^{B_\br,k}  \).
\end{lemma}
\begin{proof}	 Assume towards contradiction that \(\support^{C_1,\dots,C_l}(\hat{\br})\subseteq  \Delta^{C_1,\dots,C_l}(X_{k-l,T_{k+1-l}}) \) holds for for all \(1\leq l\leq \lfloor \frac{k}{2}\rfloor \), and that there exists \([\bv_1,\dots,\bv_{\lfloor \frac{k}{2}\rfloor}]^{C_1,\dots,C_{\lfloor \frac{k}{2}\rfloor}}\in \support^{C_1,\dots,C_{\lfloor\frac{k}{2}\rfloor}}(\hat{\br}) \) such that \([\bv_1,\dots,\bv_{\lfloor \frac{k}{2}\rfloor}]^{C_1,\dots,C_{\lfloor \frac{k}{2}\rfloor}}\notin R^{B_\br,k} \).   
	
	%	 From \(\sup^{C_1,\dots,C_l}(\hat{\br})\subseteq \Delta^{C_1,\dots,C_l}(X_{k-l,T_{k+1-l}}) \) for each \(1\leq l<\frac{k}{2} \), we have that there exist vectors	 \[ [\bu_1^1,\dots,\bv^1_{\lfloor k/2\rfloor}]^{C_1,\dots,C_{\lfloor k/2\rfloor}}  \in \Delta^{C_1,\dots,C_{\lfloor k/2\rfloor}}(X^B_{k-1,T_k}), \dots,[\bu_1^{\lfloor k/2\rfloor},\dots,\bv^{\lfloor k/2\rfloor}_{\lfloor k/2\rfloor}]^{C_1,\dots,C_{\lfloor k/2\rfloor}}  \in \Delta^{C_1,\dots,C_{\lfloor k/2\rfloor}}(X^B_{k-{\lfloor k/2\rfloor}-1,T_{\lfloor k/2\rfloor}})\] 	 such that for each \(1\leq i \leq  \lfloor k/2\rfloor \) and \(1\leq j\leq i \) it holds \(\bv_j^i=\bv_j\). 
	
	For each  \(1\leq l\leq \lfloor \frac{k}{2}\rfloor \) let
	\begin{multline*}	R^{B_\br,k}_{l}=\{[\bv_1,\dots,\bv_l]^{C_1,\dots,C_{l}} \mid \exists [\bv_1^1,\dots,\bv^1_{l}]^{C_1,\dots,C_{l}}  \in \Delta^{C_1,\dots,C_{l}}(X^{B_\br,k-2+1}_{k-1,T_k}),	\\ [\bv_1^2,\dots,\bv^2_{l}]^{C_1,\dots,C_{l}}  \in \Delta^{C_1,\dots,C_{l}}(X^{B_\br,k-2\cdot 2+1}_{k-2,T_{k-1}}),\dots,	[\bv_1^l,\dots,\bv^{l}_{l}]^{C_1,\dots,C_{l} } \in \Delta^{C_1,\dots,C_{l}}(X^{B_\br,k-2\cdot l+1}_{k-l,T_{k-l+1}});\\ \forall j\in \{1,\dots,l \}: \sum_{i=1}^{{l}} \bv_j^i=\bv_j; \forall m<j: \bv_m^j=\vec{0}  \} 
	\end{multline*}

	We will now do an induction over \(1\leq l \lfloor \frac{k}{2}\rfloor  \) to show that for each such \(i\) it holds  \( [\bv_1,\dots,\bv_{i}]^{C_1,\dots,C_{i}}\in R^{B_\br,k}_i \). Note that if this were to hold then we are done, since \(R^{B_\br,k}_{\lfloor \frac{k}{2}\rfloor}=R^{B_\br,k} \), thus we would obtain a contradiction with our assumption that \([\bv_1,\dots,\bv_{\lfloor \frac{k}{2}\rfloor}]^{C_1,\dots,C_{\lfloor \frac{k}{2} \rfloor}}\notin R^{B_\br,k}\).
	
	\textbf{Base case \(i=1\):} 
	\begin{gather*}
		R^{B_\br,k}_1=\{[\bv_1]^{C_1} \mid \exists [\bv_1^1]^{C_1}  \in \Delta^{C_1}(X^{B_\br,k-2+1}_{k-1,T_k}); \forall j\in \{1 \}: \sum_{i=1}^{{1}} \bv_j^i=\bv_j; \forall m<j: \bv_m^j=\vec{0}  \} =\\ \{[\bv_1^1]^{C_1} \in \Delta^{C_1}(X^{B_\br,k-1}_{k-1,T_k}) \} \end{gather*}
	Since \([\bv_1]^{C_1}\in \support^{C_1}(\hat{\br})\subseteq \Delta^{C_1}(X_{k-1,T_{k}}) \) there exists a multi-component \(\bx\in X_{k-1,T_k} \) with \(\Delta^{C_1}(	\bx)=\bv_1 \). Let \(\bx' \) be a multi-component of \(\A_{k-1,T_k} \) created from \(\bx\) by setting \(\bx'(t)=0 \) for each transition \(t\) that is not included in the same MEC of \(\A_{k-1} \) as \(\MEC_{\hat{\br}} \) (i.e., \(\bx'(t)=\bx(t) \) if \(t\) is in the MEC of \(\A_{k-1} \) containing \(p_\br \) and \(\bx'(t)=0 \) otherwise). Notice that each MEC of \(\A_{k-1} \) has different local copies of counters from \(C_1\), and as \(\bv_1\in \support^{C_1}(\hat{\br}) \) every counter different from \(0\) in \(\bv_1 \) must be the local copy of some counter in \(\A_{k-1} \) in the same MEC of \(\A_{k-1} \) as \(\MEC_\br \). It thus holds  \(\Delta^{C_1}(\bx)=\bv_1=\Delta^{C_1}(\bx') \) and since \(\bx'\in X_{k-1,T_k}^{\MEC_\br,k-1}\) this implies \(\bv_1\in R_1^{B_\br,k} \). Base case holds.
	
	\textbf{Induction step:} Assume we have proven this statement for \(i-1\) and we want to prove it for \(i\leq \lfloor \frac{k}{2}\rfloor\) as well. From the induction assumption we have \([\bv_1,\dots,\bv_{i-1}]^{C_1,\dots,C_{i-1}}  \in R^{B_\br,k}_{i-1} \).	 Therefore from the definition of \(R^B_{i-1}\) there exist	\begin{multline*}	[\bv_1^1,\dots,\bv^1_{i-1}]^{C_1,\dots,C_{i-1}}  \in \Delta^{C_1,\dots,C_{i-1}}(X^{B_\br,k-2+1}_{k-1,T_{k}}),\dots\\
		\dots,[\bv_1^{i-1},\dots,\bv^{i-1}_{i-1}]^{C_1,\dots,C_{i-1}}  \in \Delta^{C_1,\dots,C_{i-1}}(X^{B_\br,k-2\cdot i+2+1}_{k-i+1,T_{{k-i+2}}}) \end{multline*}
	such that \[ \forall j\in \{1,\dots,i-1 \}: \sum_{i=1}^{{i-1}} \bv_j^i=\bv_j; \forall m<j: \bv_m^j=\vec{0} \]  		 	Therefore there also exist multi-components \(\bx^1,\dots,\bx^{i-1} \) with \(\bx^j\in X^{B_\br,k-2\cdot j+1}_{k-j,T_{k-j+1}}  \) and \(\Delta^{C_1,\dots,C_{i-1}}(\bx^j)=[\bv_1^j,\dots,\bv_{i-1}^j] \). Since the MEC containing \(\MEC_\br\) in \(\A_{k-2\cdot j+1} \) is included in the MEC containing \(\MEC_\br\) in \(\A_{k-2\cdot i+1} \) for each \(j<i \), it holds  \(\bx^1,\dots,\bx^{i-1}\in X^{B_\br,k-2\cdot i+1}_{k-i,T_{k-i+1}} \).
	
	It holds \(\sum_{j=1}^{i-1}\Delta^{C_1,\dots,C_i}(\bx^j)=[\bv_1,\dots,\bv_{i-1},\hat\bv_i]^{C_1,\dots,C_{i}} \) for some \(\hat\bv_i \). If it were to hold  \(\hat\bv_i=\bv_i \) then it holds \([\bv_1,\dots,\bv_{i}]^{C_1,\dots,C_{i}}  \in R^{B_\br,k}_{i}\) since \(\vec{0}\in X_{k-i,T_{k-i+1}}^{B_\br,k-2\cdot i+1} \). Assume therefore that \(\hat\bv_i\neq \bv_i \).

	Since \([\bv_1,\dots,\bv_{i}]^{C_1,\dots,C_{i}}\in \support^{C_1,\dots,C_i}(\hat{\br})\subseteq\Delta^{C_1,\dots,C_i}(X_{k-i,T_{k-i+1}}) \), there exists a multi-component \(\hat\bx^i \) of \(\A_{k-i,T_{k-i+1}} \) with \(\Delta^{C_1,\dots,C_{i}} (\hat\bx^i)=[\bv_1,\dots,\bv_{i}]^{C_1,\dots,C_{i}} \). Let \(\bx^i \) be created from \(\hat\bx^i\) by setting \(\bx^i(t)=0 \) for each \(t\) that is not included in the MEC of \(\A_{k-2\cdot i+1} \) that contains \(\MEC_\br\).	Since \([\bv_1,\dots,\bv_i]^{C_1,\dots,C_i} \) is effect of \(\hat{\br}\), it holds that \(\bv_j(c)=0 \) for each counter \(c\) that is not a local copy of some counter from \(C_1,\dots,C_i\) in the MEC containing \(\MEC_\br\) in \(\A_{k-2\cdot i+1} \). Therefore it holds \(\Delta^{C_1,\dots,C_i}(\bx^i)=\Delta^{C_1,\dots,C_i}(\hat\bx^i)=[\bv_1,\dots,\bv_{i}]^{C_1,\dots,C_{i}} \).

	From the induction assumption on \(k-i \) for point~\ref{enum-main-3} of Lemma~\ref{lemma-main-for-all-k}  there exists a multi-component \(\hat\by \) on \(\A_{k-i} \) such that \(\hat\by(t)>0 \) iff \(t\in T_{k-i+1} \), and \(\Delta^{C_1,\dots,C_{k-i}}(\hat\by )=\vec{0} \). Let \(\by \) be a multi-component created from \(\hat\by \) by setting \(\by(t)=0 \) for each \(t\) not contained in MEC of \(\A_{k-2\cdot i+1} \) containing  \(\MEC_\br\). Using the same argument as above it holds \(\Delta^{C_1,\dots,C_i}(\by)=\Delta^{C_1,\dots,C_i}(\hat\by)=\vec{0} \).
	
	Since it holds \(\bx^i(t)>0 \) implies \(\by(t)>0 \), there exists some \(a>0\) such that \(\bx^{i}_-=a\cdot \by-\bx^i\geq \vec{0} \) and thus from Lemma~\ref{lemma-subtraction-multicomponents} \(\bx^{i}_-\) is a multi-component in \(X^{B_\br,k-2\cdot i+1}_{k-i,T_{k-i+1}}\) with \(\Delta^{C_1,\dots,C_i}(\bx^{i}_-)=a\cdot \Delta^{C_1,\dots,C_i}(\by)-\Delta^{C_1,\dots,C_i}(\bx^i)=a\cdot \vec{0}-[\bv_1,\dots,\bv_{i}]^{C_1,\dots,C_{i}}=[-\bv_1,\dots,-\bv_{i}]^{C_1,\dots,C_{i}}  \). 
	
	Thus \(\bx=\bx^{i}_-+\sum_{j=1}^{i-1}\bx^j \) is a multi-component satisfying \(\bx\in X^{B_\br,k-2\cdot i+1}_{k-i,T_{k-i+1}} \) and \[\Delta^{C_1,\dots,C_i}(\bx)=[-\bv_1,\dots,-\bv_{i-1},-\bv_i]^{C_1,\dots,C_{i}}+[\bv_1,\dots,\bv_{i-1},\hat\bv_i]^{C_1,\dots,C_{i}} = [\vec{0},\dots,\vec{0},\hat\bv_i-\bv_i]^{C_1,\dots,C_{i}} \] 
	
	But once again, since \(\bx(t)>0 \) implies \(\by(t)>0 \) there exists \(b>0\) such that \(\bz=b\cdot \by-\bx\geq \vec{0} \) and \[\Delta^{C_1,\dots,C_i}(\bz)=b\cdot \Delta^{C_1,\dots,C_i}(\by)-\Delta^{C_1,\dots,C_i}(\bx)=b\cdot \vec{0}-[\vec{0},\dots,\vec{0},\hat\bv_i-\bv_i]^{C_1,\dots,C_{i}}=[\vec{0},\dots,\vec{0},-\hat\bv_i+\bv_i]^{C_1,\dots,C_{i}} \] Thus \([\vec{0},\dots,\vec{0},-\hat\bv_i+\bv_i]^{C_1,\dots,C_{i}}\in \Delta^{C_1,\dots,C_i}(X^{B_\br,k-2\cdot i+1}_{k-i,T_{k-i+1}})\). But it holds
	\begin{gather*}
		\sum_{j=1}^{i-1}\Delta^{C_1,\dots,C_i}(\bx^j) + [\vec{0},\dots,\vec{0},-\hat\bv_i+\bv_i]^{C_1,\dots,C_{i}} =\\ [\bv_1,\dots,\bv_{i-1},\hat\bv_i]^{C_1,\dots,C_{i}} + [\vec{0},\dots,\vec{0},-\hat\bv_i+\bv_i]^{C_1,\dots,C_{i}} 
		=
		[\bv_1,\dots,\bv_{i-1},\bv_i]^{C_1,\dots,C_{i}}  \end{gather*}  and thus \([\bv_1,\dots,\bv_{i-1},\bv_i]^{C_1,\dots,C_{i}}\in R^{B_\br,k}_i \). Lemma holds.
\end{proof}

\begin{lemma}\label{lemma-D-not-in-Rl-implies-B-zero-unbounded-rankl-dsada}

	Let \(1\leq l < \lfloor \frac{k}{2}\rfloor \), and let \(\by^{C_1,\dots,C_l}_{k-l,T_{k-l+1}},\bz^{C_1,\dots,C_l}_{k-l,T_{k-l+1}} \) be maximal solutions to \hyperref[fig-systems]{(II)} for \(\A^{C_1,\dots,C_l}_{k-l,T_{k-l+1}} \). Let \(rank^{C_1,\dots,C_l}_{k-l,T_{k-l+1}} \) be the resulting ranking function defined by \(\by^{C_1,\dots,C_l}_{k-l,T_{k-l+1}},\bz^{C_1,\dots,C_l}_{k-l,T_{k-l+1}} \). If there exists \(\bv\in \support^{C_1,\dots,C_l}(\hat{\br})\) such  that \(\bv\notin \Delta^{C_1,\dots,C_l}(X_{k-l,T_{k+1-l}}) \) then \(\hat{\br}\) is not zero-bounded on   \(rank^{C_1,\dots,C_l}_{k-l,T_{k-l+1}} \).
\end{lemma}
%\todo{semi chekced}
\begin{proof}
	Assume towards contradiction this is not the case. Let \(1\leq l\leq \frac{k}{2} \), let there exist \(\bv\in \support^{C_1,\dots,C_l}(\hat{\br})\) such  that \(\bv\notin \Delta^{C_1,\dots,C_l}(X_{k-l,T_{k+1-l}}) \), and assume \(\hat{\br}\) is zero-bounded on \(rank^{C_1,\dots,C_l}_{k-l,T_{k-l+1}}\).
	
	From Lemma~\ref{lemma-v-in-XBiTi+1-implies-minus-v-also-in-XBiTi+1sss} the set \(\Delta^{C_1,\dots,C_l}(X_{{k-l},T_{k-l+1}}) \) is a closed under multiplication by \(-1\), from Lemma~\ref{lemma-addition-multicomponents} it is closed under addition, and from Lemma~\ref{lemma-multiplication-multicomponents} it is closed under multiplication by non-negative constant. Thus \(\Delta^{C_1,\dots,C_l}(X_{{k-l},T_{k-l+1}}) \) is a vector space. Since \(\bv\in \mathbb{Q}^{C_1,\dots,C_l}\setminus \Delta^{C_1,\dots,C_l}(X_{{k-l},T_{k-l+1}})  \) there exists a vector \(\bu\in \mathbb{Q}^{C_1,\dots,C_l} \) such that \(\bu \) is orthogonal to  \(\Delta^{C_1,\dots,C_l}(X_{{k-l},T_{k-l+1}}) \) but is not orthogonal to \(\bv\). Let \(\Pi_\bu(\bs) \) denote the size of the orthogonal projection of \(\bs \) onto \(\bu \). Note that  \(\Pi_\bu(\Delta^{C_1,\dots,C_l}(\bx_{{k-l},T_{k-l+1}}))= 0 \) for each \(\bx_{{k-l},T_{k-l+1}}\in X_{{k-l},T_{k-l+1}} \) and \(\Pi_\bu(\bs)\neq 0 \). 
	
	Let \(\A_{{k-l},T_{k-l+1}}^\bu\) be a one-counter VASS MDP created from \(\A_{{k-l},T_{k-l+1}} \) by replacing all the counters with the single counter \(\Pi_\bu \) (i.e. if \((p,\bs,q) \) is a transition of \(\A_{{k-l},T_{k-l+1}} \) then \((p,\Pi_\bu(\bs),q) \) is a transition of \(\A_{{k-l},T_{k-l+1}}^\bu \)). We use \(c^\bu\) to denote the counter of \(\A_{{k-l},T_{k-l+1}}^\bu \). Since every multi-component of \(\A_{{k-l},T_{k-l+1}}^\bu \) is in \(X_{{k-l},T_{k-l+1}}\)  it holds that every component of \(\A_{{k-l},T_{k-l+1}}^\bu \) is either zero-bounded or zero-unbounded on \(c^\bu\). Thus the maximal solution \( \bx^\bu,\by^\bu,\bz^\bu \) of \hyperref[fig-systems]{(I)} and \hyperref[fig-systems]{(II)} for 
	\(\A_{{k-l},T_{k-l+1}}^\bu \) satisfies that  \(\bx^\bu(t)>0 \) for each \(t\in T_{k-l+1} \) (since a multi-component obtained as a sum of all components of \(\A_{{k-l},T_{k-l+1}}^\bu \) is a valid solution of \hyperref[fig-systems]{(I)} for \(\A_{{k-l},T_{k-l+1}}^\bu \)) and from Lemma~\ref{lemma:dichotomy} it holds \(\by(c^\bu)>0 \) (since \(\Delta(\bx^\bu)(c^\bu)=0 \)). Thus from  Lemma~\ref{lemma:dichotomy} it holds for each transition \(t= (p,\Pi_\bu(\bs),q)\) of \(\A_{{k-l},T_{k-l+1}}^\bu \) that 
	\begin{itemize}
		\item if $p\in Q_n$ then $\bz^\bu(q)-\bz^\bu(p)+ \Pi_\bu(\bs)\cdot \by^\bu(c^\bu)=0$;
		\item if $p\in Q_p$ then \[\sum_{t'=(p,\Pi_\bu(\bs'),q') \in \tout(p)}P(t')\cdot \big(\bz^\bu(q')-\bz^\bu(p)+\Pi_\bu(\bs')\cdot \by^\bu(c^\bu)\big)= 0\]
	\end{itemize}
	We now define a new ranking function on \(\A_{{k-l},T_{k-l+1}} \) as \(rank^\bu_{{k-l},T_{k-l+1}}(p\bs)=\bz^\bu(p)+  \by^\bu(c^\bu)\cdot \Pi_\bu(\bs)=\bz^\bu(p)+ \sum_{c\in \countersset} \by^\bu(c^\bu)\cdot \Pi_\bu(\bs_c)  \) where \(\bs_c(c)=\bs(c) \) and \(\bs_c(c')=0 \) for \(c'\neq c \). 
	
	%	  Note that the effect of \(\Delta(X_{k-l},) \) each action in \(\A_{{k-l},T_{k-l+1}} \)  on \(rank^\bu_{{k-l},T_{k-l+1}} \) is zero while \(rank^\bu_{{k-l},T_{k-l+1}}(\bv)\neq 0\).

	Consider now the ranking functions \(rank_{a}(p\bs)=a \cdot rank^{C_1,\dots,C_l}_{k-l,T_{k-l+1}}(p\bs)+ rank^\bu_{{k-l},T_{k-l+1}}(p\bs) \) defined for each  \(a\geq 0 \). We can write
	\begin{gather*}
		rank_{a}(p\bs)=a \cdot rank^{C_1,\dots,C_l}_{k-l,T_{k-l+1}}(p\bs)+ rank^\bu_{{k-l},T_{k-l+1}}(p\bs)
		=\\
		a \cdot \big(\bz^{C_1,\dots,C_l}_{k-l,T_{k-l+1}}(p)+ \sum_{c\in \countersset}\by^{C_1,\dots,C_l}_{k-l,T_{k-l+1}}(i)\cdot \bs(c)\big)+\bz^\bu(p)+ \sum_{c\in \countersset} \by^\bu(c^\bu)\cdot \Pi_\bu(\bs_c)
		=\\
		a\cdot \bz^{C_1,\dots,C_l}_{k-l,T_{k-l+1}}(p)+ \bz^\bu(p)+  \sum_{c\in \countersset}\big(a\cdot \by^{C_1,\dots,C_l}_{k-l,T_{k-l+1}}(i)\cdot \bs(c)+ \by^\bu(c^\bu)\cdot \Pi_\bu(\bs_c)\big)
		=\\
		a\cdot \bz^{C_1,\dots,C_l}_{k-l,T_{k-l+1}}(p)+ \bz^\bu(p)+  \sum_{c\in \countersset}(a\cdot \by^{C_1,\dots,C_l}_{k-l,T_{k-l+1}}(c)+ \by^\bu(c^\bu)\cdot m_c^\bu)\cdot \bs(c) 
	\end{gather*}  where the last equality comes from the fact that we can write \(\Pi_\bu(\bs_c)=\bs(c)\cdot m_c^\bu \) for \(m_c^\bu=\Pi_\bu(\vec{1}_c) \). Therefore \(rank_{a} \) can be expressed as \[rank_{a}(p\bs)=\bz_{a}(p)+ \sum_{c\in \countersset}\by_{a}(c)\cdot \bs(c) \] where \(\bz_{a}(p)=a\cdot \bz^{C_1,\dots,C_l}_{k-l,T_{k-l+1}}(p)+ \bz^\bu(p) \) and \(\by_{a}(c)=a\cdot \by^{C_1,\dots,C_l}_{k-l,T_{k-l+1}}(c)+ \by^\bu(c^\bu)\cdot m_c^\bu \). Notice that \(\bz_{a}\geq \vec{0} \) and the only way for \(\by_a(c)< 0 \) to hold is if \(m_c^\bu<0\) and \begin{equation*}\label{eq-a-bound}
	a< \frac{- \by^\bu(c^\bu)\cdot m_c^\bu}{\by^{C_1,\dots,C_l}_{k-l,T_{k-l+1}}(c)}
\end{equation*} Therefore if \(a \) is so large that \[	a> \frac{- \by^\bu(c^\bu)\cdot m_c^\bu}{\by^{C_1,\dots,C_l}_{k-l,T_{k-l+1}}(c)}\] holds for each counter \(c \), we have that \(\bz_{a},\by_{a}\) are both non-negative. We will now show that for some sufficiently large \(a\) it holds that \(\bz_{a},\by_{a}\) is a more maximal solution of \hyperref[fig-systems]{(II)} for \(\A_{k-l,T_{k-l+1}}^{C_1,\dots,C_l} \) than \(\bz_{k-l,T_{k-l+1}}^{C_1,\dots,C_l},\by_{k-l,T_{k-l+1}}^{C_1,\dots,C_l} \) thus contradicting with \(\bz_{k-l,T_{k-l+1}}^{C_1,\dots,C_l},\by_{k-l,T_{k-l+1}}^{C_1,\dots,C_l} \) being a maximal such solution.
	
	First notice that for every transition \(t=(p,\bs,q)\) of \(\A_{k-l,T_{k-l+1}}^{C_1,\dots,C_l} \) it holds 
	\begin{itemize}
		\item if $p\in Q_n$ then $\bz_{a}(q)-\bz_{a}(p)+ \sum_{c\in \countersset} \bs(c)\cdot \by_{a}(c)\leq 0$;
		\item if $p\in Q_p$ then \[\sum_{t'=(p,\bs',q') \in \tout(p)}P(t')\cdot \big(\bz_{a}(q')-\bz_{a}(p)+\sum_{c\in \countersset}\bs'(c)\cdot \by_{a}(c)\big)\leq 0\]
	\end{itemize}
	This is due to both  \(rank_{k-l,T_{k-l+1}}\) and \(rank_{k-l,T_{k-l+1}}^\bu \) being non-increasing on average for each step (see \hyperref[fig-systems]{(II)}). Thus for all sufficiently large \(a\) it holds that \(\bz_{a},\by_{a}\) is a solution of \hyperref[fig-systems]{(II)} for \(\A_{k-l,T_{k-l+1}}^{C_1,\dots,C_l} \). It remains to show it is more maximal than \(\bz_{k-l,T_{k-l+1}}^{C_1,\dots,C_l},\by_{k-l,T_{k-l+1}}^{C_1,\dots,C_l} \). For the first maximization objective, if \(\bz_{k-l,T_{k-l+1}}^{C_1,\dots,C_l}(c)>0 \) then also \(\bz_{a}(c)=a\cdot \bz_{k-l,T_{k-l+1}}^{C_1,\dots,C_l}(c)+\bz^\bu(c)\geq a\cdot\bz_{k-l,T_{k-l+1}}^{C_1,\dots,C_l}(c)>0 \).

	For the second maximization objective of \hyperref[fig-systems]{(II)}, let $t=(p,\bs,q)$ be such that $p\in Q_n$ and \(\bz_{k-l,T_{k-l+1}}^{C_1,\dots,C_l}(q)-\bz_{k-l,T_{k-l+1}}^{C_1,\dots,C_l}(p)+ \sum_{c\in \countersset} \bs(c)\cdot \by_{k-l,T_{k-l+1}}^{C_1,\dots,C_l}(c)<0\). Then it holds 
	\begin{gather*}		
		\bz_{a}(q)-\bz_{a}(p)+ \sum_{c\in \countersset} \bs(c)\cdot \by_{a}(c)=\\
		a\cdot \big(\bz_{k-l,T_{k-l+1}}^{C_1,\dots,C_l}(q)-\bz_{k-l,T_{k-l+1}}^{C_1,\dots,C_l}(p)+ \sum_{c\in \countersset} \bs(c)\cdot \by_{k-l,T_{k-l+1}}^{C_1,\dots,C_l}(i)\big)+
		\bz^\bu(q)-\bz^\bu(p)+ \sum_{c\in \countersset} \bs(c)\cdot \by^\bu(c^\bu)\cdot m_c^\bu
		=\\
		a\cdot \big(\bz_{k-l,T_{k-l+1}}^{C_1,\dots,C_l}(q)-\bz_{k-l,T_{k-l+1}}^{C_1,\dots,C_l}(p)+ \sum_{c\in \countersset} \bs(c)\cdot \by_{k-l,T_{k-l+1}}^{C_1,\dots,C_l}(i)\big)+
		\bz^\bu(q)-\bz^\bu(p)+ \Pi_\bu(\bs) \by^\bu(c^\bu)
		\leq \\
		 a\cdot \big(\bz_{k-l,T_{k-l+1}}^{C_1,\dots,C_l}(q)-\bz_{k-l,T_{k-l+1}}^{C_1,\dots,C_l}(p)+ \sum_{c\in \countersset} \bs(c)\cdot\by_{k-l,T_{k-l+1}}^{C_1,\dots,C_l}(c)\big)+
		0
		<0
	\end{gather*} where the second to last inequality comes from \(\by^\bu,\bz^\bu \) being a solution of \hyperref[fig-systems]{(II)} for \(\A_{{k-l},T_{k-l+1}}^\bu \).
	
	For the third maximization objective of \hyperref[fig-systems]{(II)}, let $p\in Q_p$ be such that 
	\[\sum_{t=(p,\bs,q) \in \tout(p)}P(t)\cdot \big(\bz_{k-l,T_{k-l+1}}^{C_1,\dots,C_l}(q')-\bz_{k-l,T_{k-l+1}}^{C_1,\dots,C_l}(p)+\sum_{c\in \countersset}\bs(c)\cdot \by_{k-l,T_{k-l+1}}^{C_1,\dots,C_l}(c)\big)< 0\] Then it holds \begin{gather*}
		\sum_{t=(p,\bs,q) \in \tout(p)}P(t)\cdot \big(\bz_{a}(q')-\bz_{a}(p)+\sum_{c\in \countersset}\bs(c)\cdot \by_{a}(c)\big)
%		=\\
%		\sum_{t = (p,\bu,q) \in \tout(p)}\prob(t)( rank_{k-l,T_{k-l+1}}(t)+b\cdot rank^\bu_{k-l,T_{k-l+1}}(t))
		=\\
		a\cdot (\sum_{t=(p,\bs,q) \in \tout(p)}P(t)\cdot \big(\bz_{k-l,T_{k-l+1}}^{C_1,\dots,C_l}(q')-\bz_{k-l,T_{k-l+1}}^{C_1,\dots,C_l}(p)+\sum_{c\in \countersset}\bs(c)\cdot \by_{k-l,T_{k-l+1}}^{C_1,\dots,C_l}(i)\big)) +\\+ \sum_{t=(p,\bs,q) \in \tout(p)}P(t)\cdot \big(\bz^\bu(q)-\bz^\bu(p)+\sum_{c\in \countersset}\bs(c)\cdot \by^\bu(c^\bu)\cdot m_c^\bu\big)
		=\\
		a\cdot (\sum_{t=(p,\bs,q) \in \tout(p)}P(t)\cdot \big(\bz_{k-l,T_{k-l+1}}^{C_1,\dots,C_l}(q')-\bz_{k-l,T_{k-l+1}}^{C_1,\dots,C_l}(p)+\sum_{c\in \countersset}\bs(c)\cdot \by_{k-l,T_{k-l+1}}^{C_1,\dots,C_l}(c)\big)) +\\+ \sum_{t=(p,\bs,q) \in \tout(p)}P(t)\cdot \big(\bz^\bu(q)-\bz^\bu(p)+\Pi_\bu(\bs) \by^\bu(c^\bu)\big)
		\leq \\
		a\cdot (\sum_{t=(p,\bs,q) \in \tout(p)}P(t)\cdot \big(\bz_{k-l,T_{k-l+1}}^{C_1,\dots,C_l}(q')-\bz_{k-l,T_{k-l+1}}^{C_1,\dots,C_l}(p)+\sum_{c\in \countersset}\bs(c)\cdot \by_{k-l,T_{k-l+1}}^{C_1,\dots,C_l}(c)\big)) + 0
		<0 \end{gather*}
	where the second to last inequality comes from \(\by^\bu,\bz^\bu \) being a solution of \hyperref[fig-systems]{(II)} for \(\A_{{k-l},T_{k-l+1}}^\bu \).
	%	For the last maximization objective,
	%	Let $p\in Q_p$ be such that there exists a transition \(t\in \tout(p) \) such that
	%	\[rank_{k-l,T_{k-l+1}}(t)\neq 0\]. Then if it holds \( rank_{1,b}(t)=0\) it must hold that \(rank_{k-l,T_{k-l+1}}(t)=-b\cdot rank^\bu_{k-l,T_{k-l+1}}(t) \), but this equation can be satisfied by only a single value for \(b\), and since the number of such transitions \(t\) is finite, there exists at least one value for \(b\) such that  \(rank_{k-l,T_{k-l+1}}(t)=-b\cdot rank^\bu_{k-l,T_{k-l+1}}(t) \) is not true for any such \(t\). Hence for all sufficiently large values for \( b\) it holds \(rank_{1,b}(t)\neq 0\) as well.
	
	Therefore \(\bz_{a},\by_{a} \) is at least as maximal a solution of \hyperref[fig-systems]{(II)} for \(\A_{k-l,T_{k-l+1}}^{C_1,\dots,C_l} \) as \(\bz_{k-l,T_{k-l+1}}^{C_1,\dots,C_l},\by_{k-l,T_{k-l+1}}^{C_1,\dots,C_l} \). 
	
	And since \(\hat{\br} \) is zero-bounded  on \(rank_{k-l,T_{k-l+1}}^{C_1,\dots,C_l}\), it holds from \cite{AKCONCUR23} that every single transition \(t=(p,\bs,q)\)  with \(\br(t)>0 \) satisfies \(\bz_{k-l,T_{k-l+1}}^{C_1,\dots,C_l}(q)-\bz_{k-l,T_{k-l+1}}^{C_1,\dots,C_l}(p)+ \sum_{c\in \countersset} \bs(c)\cdot \by_{k-l,T_{k-l+1}}^{C_1,\dots,C_l}(c)=0 \). But since \(\bv\in \support^{C_1,\dots,C_l}(\hat{\br}) \) and  \(\Pi_\bu(\bv)\neq 0 \) it holds that \(\hat{\br} \) is not zero-bounded on   \(rank_{k-l,T_{k-l+1}}^\bu\), which from \cite{AKCONCUR23} gives that there exists a \(t=(p,\bs,q)\)  with \(\br(t)>0 \) such that \(0\neq \bz^\bu(q)-\bz^\bu(p)+ \sum_{c\in \countersset} \bs(c)\cdot \by^\bu(c^\bu)\cdot m_c^\bu\leq  0 \), and thus also  
	\begin{gather*}		
		\bz_{a}(q)-\bz_{a}(p)+ \sum_{c\in \countersset} \bs(c)\cdot \by_{a}(c)
		=\\
		a\cdot \big(\bz_{k-l,T_{k-l+1}}^{C_1,\dots,C_l}(q)-\bz_{k-l,T_{k-l+1}}^{C_1,\dots,C_l}(p)+ \sum_{c\in \countersset} \bs(c)\cdot \by_{k-l,T_{k-l+1}}^{C_1,\dots,C_l}(i)\big)+
		\bz^\bu(q)-\bz^\bu(p)+ \sum_{c\in \countersset} \bs(c)\cdot \by^\bu(c^\bu)\cdot m_c^\bu
		= \\
		 	0 +
		 \bz^\bu(q)-\bz^\bu(p)+ \sum_{c\in \countersset} \bs(c)\cdot \by^\bu(c^\bu)\cdot m_c^\bu
		<0
	\end{gather*}
	But this gives us that for a sufficiently large \(a\), \(\bz_{a},\by_{a} \) is a more maximal solution of \hyperref[fig-systems]{(II)} for \(\A_{k-l,T_{k-l+1}}^{C_1,\dots,C_l} \)  compared to \(\bz_{k-l,T_{k-l+1}}^{C_1,\dots,C_l},\by_{k-l,T_{k-l+1}}^{C_1,\dots,C_l} \), a contradiction with \(\bz_{k-l,T_{k-l+1}}^{C_1,\dots,C_l},\by_{k-l,T_{k-l+1}}^{C_1,\dots,C_l} \) being a maximal such solution. Lemma holds. 
%\(\bv\in \support^{C_1,\dots,C_l}(\hat{\br})\)  such  that \(\bv\notin \Delta^{C_1,\dots,C_l}(X_{k-l,T_{k+1-l}}) \)	
%	 But since there exists a cycle with effect \(\bv \) in \(\A_{k-l,T_{k-l+1}}^{C_1,\dots,C_l} \) with \(\Pi_\bu(\bv)\neq 0 \) it holds that \(\hat{\br} \)
\end{proof}

\subsubsection{Proof  of Lemma~\ref{lemma-hatB-zero-unbounded-rankl-give-upper-estimate-nk} and Lemma~\ref{lemma-I-dont-lknooow-but-it-is-somehting}}
\label{section-proof-lemma-hatB-zero-unbounded-rankl-give-upper-estimate-nk}

\begin{lemma*}[\textbf{\ref{lemma-hatB-zero-unbounded-rankl-give-upper-estimate-nk}}]	
			Let \(1\leq l \leq  \lfloor\frac{k}{2}\rfloor \). Let  \(\by^{C_1,\dots,C_l}_{k-l,T_{k-l+1}},\bz^{C_1,\dots,C_l}_{k-l,T_{k-l+1}} \) be   a maximal solution of \hyperref[fig-systems]{(II)} for \(\A^{C_1,\dots,C_l}_{k-l,T_{k-l+1}} \), and let \(rank^{C_1,\dots,C_l}_{k-l,T_{k-l+1}} \) be the resulting ranking function defined by \(\by^{C_1,\dots,C_l}_{k-l,T_{k-l+1}},\bz^{C_1,\dots,C_l}_{k-l,T_{k-l+1}} \). There exists a set of transitions \(R^{C_1,\dots,C_l}_{k-l} \) of \(\A^{C_1,\dots,C_l}_{k-l,T_{k-l+1}} \) such that both of these hold:
		\begin{itemize}
			\item for each component  \(\by \) of \(\A_{k-l,T_{k-l+1}} \) it holds that \(\hat{\by} \) is zero-bounded on \(rank^{C_1,\dots,C_l}_{k-l,T_{k-l+1}} \) iff \( R^{C_1,\dots,C_l}_{k-l}\cap \{t\mid \by(t)>0 \}= \emptyset \);
			\item \(\calT_\A[t] \) has an upper asymptotic estimate of \(n^k\) for each \(t\in R^{C_1,\dots,C_l}_{k-l}\).	
		\end{itemize}
		Furthermore, assuming we have a classification of \(\A\) up to \(k-1\),  \(R_{k-l}\) can be computed in time polynomial in \(\size{\A} \).
\end{lemma*}
%\todo{\(by\) was originally \(\bz \)?}

%\todo{also have to prove that :  If \(\hat{\by} \) is not zero-bounded on \(rank^{C_1,\dots,C_l}_{k-l,T_{k-l+1}} \) then \(\calP_\A[\M_\by] \) has an upper asymptotic estimate of \(n^{k}\).?}

%\begin{lemma}
%	Let \(\by \) be a component. If \(\hat{\by} \) is not zero-bounded on \(rank^{C_1,\dots,C_l}_{k-l,T_{k-l+1}} \) then \(\calP_\A[\M_\by] \) has an upper asymptotic estimate of \(n^{k}\).
%\end{lemma}

Towards the second part of this Lemma, any component of \(\A^{C_1,\dots,C_l}_{k-l,T_{k-l+1}}\) whose effect on \(rank^{C_1,\dots,C_l}_{k-l,T_{k-l+1}} \) is not zero-bounded is either decreasing or zero-unbounded, as increasing is not possible since \(rank^{C_1,\dots,C_l}_{k-l,T_{k-l+1}} \) is defined using a solution of \hyperref[fig-systems]{(II)} (see Section~\ref{sec-systems}). But any such component  that is decreasing must contain either a non-deterministic transition that decreases \(rank^{C_1,\dots,C_l}_{k-l,T_{k-l+1}}\) or a probabilistic state from which the value of \(rank^{C_1,\dots,C_l}_{k-l,T_{k-l+1}}\) is decreased on average in one computational step. Furthermore, from \cite{AKCONCUR23} any component  of \(\A^{C_1,\dots,C_l}_{k-l,T_{k-l+1}} \) that is zero-unbounded on \(rank^{C_1,\dots,C_l}_{k-l,T_{k-l+1}} \) must contain a transition that strictly decreases \(rank^{C_1,\dots,C_l}_{k-l,T_{k-l+1}} \). Therefore it suffices to find all transitions \(R^{C_1,\dots,C_l}_{k-l}  \) of \(\A^{C_1,\dots,C_l}_{k-l,T_{k-l+1}} \) whose effect on  \(rank^{C_1,\dots,C_l}_{k-l,T_{k-l+1}} \) is not \(0\), and we have that any component  \(\A^{C_1,\dots,C_l}_{k-l,T_{k-l+1}} \) that includes a transition from \(R^{C_1,\dots,C_l}_{k-l} \) is not zero-bounded on  \(rank^{C_1,\dots,C_l}_{k-l,T_{k-l+1}}\). Whereas each component of \(\A^{C_1,\dots,C_l}_{k-l,T_{k-l+1}} \) that contains no transition from \(R^{C_1,\dots,C_l}_{k-l} \) is zero-bounded on \(rank^{C_1,\dots,C_l}_{k-l,T_{k-l+1}}\). Furthermore from \cite{AKCONCUR23}, we can compute the set \(R^{C_1,\dots,C_l}_{k-l}  \) in polynomial time when there exists no component of \(\A^{C_1,\dots,C_l}_{k-l,T_{k-l+1}} \) that is increasing on \(rank^{C_1,\dots,C_l}_{k-l,T_{k-l+1}}\), which is our case.
% As components can be seen as a MEC of \(\A_\sigma \) for \(\sigma\in\stratsMD{\A^{C_1,\dots,C_l}_{k-l,T_{k-l+1}}} \) this extends to components as well.

The upper asymptotic estimates for transitions from \(R^{C_1,\dots,C_l}_{k-l} \) then follow from  Lemma~\ref{lemma-I-dont-lknooow-but-it-is-somehting} for \(s=k-l\) and \(r=l \). This follows from the fact that \(\calT_\A[t]\leq \sum_{\by}\calP_{\tilde{\A}}[\M_\by] \) where \(\by \) ranges over all the components of \(\A \) with \(\by(t)>0 \) (see Section~\ref{section-additional-definitions}). Note that this also gives us the following Lemma.

\begin{lemma*}[\textbf{\ref{lemma-I-dont-lknooow-but-it-is-somehtingmain}}]	
  	We can only obtain an upper asymptotic estimate of \(n^k\) from Lemma~\ref{lemma-hatB-zero-unbounded-rankl-give-upper-estimate-nk} for values of  \(k \) satisfying \(k=\max(s+r,2\cdot r)\) where \(r\in \Aset\) and \( s\in S_r \).
\end{lemma*}
% Note that \(\max(s+r,2r)=\max(k-l,2l)\leq \max(k-\lfloor\frac{k}{2}\rfloor,\lfloor\frac{2k}{2}\rfloor) \leq k \), and an upper asymptotic estimate of \(n^{a} \) implies an upper asymptotic estimate of \(n^{b} \) for each \(b>a\). 

\begin{lemma}\label{lemma-I-dont-lknooow-but-it-is-somehting}
	Let \(\by^{C_1,\dots,C_r}_{s,T_{s+1}},\bz^{C_1,\dots,C_r}_{s,T_{s+1}} \) be a maximal solution of \hyperref[fig-systems]{(II)} for \(\A^{C_1,\dots,C_r}_{s,T_{s+1}} \). Let \(rank^{C_1,\dots,C_r}_{s,T_{s+1}} \) be the resulting ranking function defined by \(\by^{C_1,\dots,C_r}_{s,T_{s+1}},\bz^{C_1,\dots,C_r}_{s,T_{s+1}} \). Let \(\by \) be a component of \(\A^{C_1,\dots,C_r}_{s,T_{s+1}}\). If \(\hat{\by} \) is not zero-bounded on \(rank^{C_1,\dots,C_l}_{k-l,T_{k-l+1}} \) then \(\calP_{\tilde{\A}}[\M_\by] \) has an upper asymptotic estimate of \(n^{k}\) for \(k=\max(s+r,2r)\).
\end{lemma}

\begin{proof}
	Assume towards contradiction this does not hold. Let \(\calP_{\tilde{\A}}[\M_\by] \) not have an upper asymptotic estimate of  \(n^k\), while \(\hat{\by} \) is not zero-bounded on \(rank^{C_1,\dots,C_r}_{s,T_{s+1}}\). 
	
	Let us fix \(\epsilon>0 \), and let us define technical constants \(0<\epsilonr_{1},\epsilonr_{2},\dots  \). As their exact values are not important we leave the assignment of their exact values to Table~\ref{Table-eps-section-finalpart} at the end of the section where we also show that our assignment satisfies all the assumptions we make on \(\epsilonr_{1},\epsilonr_{2},\dots\) thorough this section.
	
	Let \(R_{\epsilonr_{1}}\) be the set of all computations \(\alpha\) on \(\A \) for which all of the following holds:
	
	%	 MEC of \(\A_\sigma \) for \(\sigma\in\stratsMD{\A^{C_1,\dots,C_l}_{k-l,T_{k-l+1}}} \)

	\begin{itemize}
		\item each transition \(t\) with upper estimate \(n^i\) for \(i\in \{1,\dots,k \}  \) of \(\calT_\A[t]\) appears at most \(n^{i+\epsilonr_{1}} \) times, that is \(\calT_\A[t](\alpha)\leq n^{i+\epsilonr_{1}}\);
		\item each counter \(c \) with upper estimate \(n^i\) for \(i\in \{1,\dots,k \} \) of \(\calC_\A[c]\) never exceeds \(n^{i+\epsilonr_{1}} \), that is \(\calC_\A[c](\alpha)\leq n^{i+\epsilonr_{1}} \);
		\item  let \(\pi_{..n^{k+1}} \) be the pointing computation on \(\tilde{\A} \) corresponding to \(\alpha_{..n^{k+1}} \). For a pointing pair \((\M,p)\in \tilde{\A} \) let \(\pi_{..n^{k+1}} ^\M \) be the corresponding computation in \(\M \) produced by the steps of \(\pi_{..n^{k+1}} \) pointing at \((\M,p)\). Then for each \((\M,p)\in \tilde{\A} \)  it holds either that  \(\length(\pi_{..n^{k+1}} ^\M)\leq n^{\epsilonr_{1}}  \) or that \(\pi_{..n^{k+1}} ^\M \) reaches a MEC  of \(\M \) within at most \(n^{\epsilonr_{1}}\) steps. Additionally, if \(\pi_{..n^{k+1}} \) reaches a MEC \(\MEC\) of \(\M\) then \(\pi_{..n^{k+1}} \)  contains no sub-computation, that is also a computation on \(\MEC \), of length \(n^{\epsilonr_{1}}\) that does not contain some transition \(t\) of \(\MEC\).

%		 MEC \(\MEC \) of \(\A_\sigma \) where \(\sigma\in\stratsMD{\A}\), let \( \)for \(\alpha_{..n^{k+1}} \) within the first \(n^{k+1} \) steps of the computation, the paths generated by the steps pointing at pointing pairs corresponding to \(\MEC \) in \(\pointingVASScorrepsonding{\A} \), contain no sub-path of length \(n^{\epsilonr_{1}}\) that does not contain some transition \(t\) of \(\MEC\).
	\end{itemize}  Note that for every strategy \(\sigma\) and each initial state \(p\)  it holds \(\lim_{n\rightarrow\infty} \prob_{p\vec{n}}^\sigma[R_{\epsilonr_{1}}]=1 \). 	The limitation on counters and transitions gives this limit straight from the definition of upper asymptotic estimates. Whereas for the last point, the expected time to reach a MEC is constant in every VASS Markov chain, and let \(t=(p,\bu,q)\) be a transition of \(B\), then every step of \(\pi_{..n^{k+1}}^B \) the probability of not visiting \(p\) in the next \(n^{\epsilonr_{1}}\) steps is simialrly upper bounded by \(l^{-n^{\epsilonr_{1}}} \) for some constant \(l\), this is due to there being a non-zero, bounded from below probability of reaching \(p\) every at most constant number of steps. And every time \(p\) is reached, there is a constant probability \(P(t)\) that the next transition will be \(t \). Thus the probability of there being no \(t \) for \(n^{\epsilonr_{1}} \) steps when iterating \(B \) is upper bounded by \(g^{n^{\epsilonr_{1}}} \) for some \(g<1\).  Let \(Y_n \) be the random variable denoting the number of sub-paths generated on \(B \), within the first \(n^{k+1} \) steps, that are of length \(\lfloor n^{\epsilonr_{1}} \rfloor\) and that do not contain the transition \(t\). Then it holds \(\E_{p\vec{n}}^\sigma(Y_n)\leq n^{k+1} g^{n^{\epsilonr_{1}}} \), and thus from Markov inequality \(\prob_{p\vec{n}}^\sigma[Y_n\geq 1]\leq n^{k+1} g^{n^{\epsilonr_{1}}} \) which goes to \(0\) as \(n \) goes to infinity.

	Note that the only difference between \(\A_{{s},T_{s+1}} \) and \(\A_{1,T_{s+1}} \) is that \(\A_{s,T_{s+1}} \) contains local copies of counters from \(C_1,\dots,C_{s-1} \). We can thus extend \(rank^{C_1,\dots,C_r}_{s,T_{s+1}}\) onto \(\A_{{1},T_{s+1}} \) by defining \(rank^{C_1,\dots,C_r}_{{1},T_{s+1}}(p\bv_{{1},T_{s+1}})=\bz^{C_1,\dots,C_r}_{s,T_{s+1}}(p)+ \sum_{c\in \countersset_s} \by^{C_1,\dots,C_r}_{s,T_{s+1}}(c)\bv_{{1},T_{s+1}}^p(c)  \) where \(\countersset_s \) is are the coutners of \(\A_{{s},T_{s+1}} \),  and \(\bv_{{1},T_{s+1}}^p(c)=\bv_{{1},T_{s+1}}(c) \) if \(c\notin C_1\cup\dots\cup C_{s-1} \), and for \(c\in C_1\cup\dots\cup C_{s-1} \) we put  \(\bv_{{1},T_{s+1}}^p(c)=\bv_{{1},T_{s+1}}(c') \) if \(c \) is the local copy of \(c'\) in the MEC of \(\A_{s} \) containing \(p\), and \(\bv_{{1},T_{s+1}}^p(c)=0 \) otherwise. Notice that every transition of \(\A_{{1},T_{s+1}} \) has the same effect on \(rank^{C_1,\dots,C_r}_{{1},T_{s+1}}\) as the same transition on \(rank^{C_1,\dots,C_r}_{{s},T_{s+1}} \) in \(\A_{{s},T_{s+1}} \).

	Thus the only transitions that can increase \(rank^{C_1,\dots,C_r}_{{1},T_{s+1}}\) in \(\A\) are either those from from  \(T_1\setminus T_{s+1} \), or those \(t=(p,\bu,q)\) where \(p\in Q_p \) and the average effect of a single computational step from \(p\) on \(rank^{C_1,\dots,C_r}_{{1},T_{s+1}}\) is at most \(0\). 
	But for each transition \(t_i\in T_i\setminus T_{i+1} \) for \(i\leq s \)   the maximal increase of \(rank^{C_1,\dots,C_r}_{{1},T_{s+1}} \) from single iteration of \(t_i\) is at most a constant \(u\) from the effect of the transition itself on the counters, plus at most a constant multiple of the sum of all of the counters from \(C_1\cup \dots \cup C_{s-i}\) from the change in the weight of these counters in \(rank^{C_1,\dots,C_r}_{{1},T_{s+1}}\). And since \(t_i\) appears in any \(\alpha\in R_{\epsilonr_{1}} \) at most \(n^{i+\epsilonr_{1}} \) times, the maximal increase of \(rank^{C_1,\dots,C_r}_{{1},T_{s+1}} \) from all occurrences of \(t_i\) along \(\alpha\in R_{\epsilonr_{1}}  \) is upper bounded by  \(n^{i+\epsilonr_{1}}\cdot (u+n^{s-i+\epsilonr_{2}}) \) assuming 
	\begin{equation}\label{eq-epsbound-ojbfsdvugolyibvuotyib}
		\epsilonr_{2}>\epsilonr_{1}
	\end{equation}

	Thus along the entire computation on \(\A \), regardless of the initial state or the strategy chosen, the maximal increase of \(rank^{C_1,\dots,C_r}_{1,T_{s+1}}\) from the transitions from \(T_1\setminus T_{s+1}\) conditioned on \(R_{\epsilonr_{1}} \) is upper bounded by \[\sum_{i=1}^{s}|T_i\setminus T_{i+1}|\cdot n^{i+\epsilonr_{1}}\cdot (u+n^{s-i+\epsilonr_{2}})=u\cdot \sum_{i=1}^{s}|T_i\setminus T_{i+1}|\cdot n^{i+\epsilonr_{1}}+\sum_{i=1}^{s}n^{s+\epsilonr_{1}+\epsilonr_{2}}\leq n^{s+\epsilonr_{3}} \]  for all sufficiently large \(n\) assuming \begin{equation}\label{eq-epsbound-ojbvugolyibvuotyib}
		\epsilonr_{3}>\epsilonr_{1}+\epsilonr_{2}
	\end{equation} Remember that for each computation \(\alpha \) on \(\A\) we can assign to each step along \(\alpha \) a unique pointing pair \((\M,\stateVASSone)\in \tilde{\A} \) that this step points to (see Section~\ref{section-additional-definitions}). We say that the step points to a MEC \(\MEC\) of \(\A_\sigma \) for \(\sigma\in \stratsMD{\A} \) if the step points at a pointing pair \((\M,p) \) such that\(\MEC\) is entirely contained in \(\M \) and the current state of \(\M \) along the corresponding pointing computation is a state from \(\MEC \). 
	
	For any computation in \(R_{\epsilonr_{1}} \) the total effect on \(rank^{C_1,\dots,C_r}_{1,T_{s+1}}\) of all steps that take a transition \(t_i\in T_{i}\setminus T_{i+1} \) with \(s<i \) and that point at a given MEC \(\MEC\) of \(\A_\sigma \) where \(\sigma\in\stratsMD{\A} \), such that \(\MEC \) contains a transition from \(T_{1}\setminus T_{s+1} \), is at most \(u\cdot n^{s+2\cdot \epsilonr_{1}}\) for some constant \(u\). This is since to the effect of \(t_i\) on \(rank^{C_1,\dots,C_r}_{1,T_{s+1}}\) can be upper bounded by a constant, and the maximal number of times we can point at some \(\MEC \) which contains a transition from \(T_{1}\setminus T_{s+1}\), conditioned on \(R_{\epsilonr_{1}}\), is \(n^{\epsilonr_{1}}\cdot n^{s+\epsilonr_{1}} \), hence the maximal number of steps that take \( t_i\) and point at this \(\MEC \) is also \(n^{\epsilonr_{1}}\cdot n^{s+\epsilonr_{1}} \) (and there are only constantly many such MECs \(\MEC\)). This is due to every \( n^{\epsilonr_{1}}\) pointings at such \(\MEC \) every transition of \(\MEC \) has to appear at least once in \(R_{\epsilonr_{1}} \) and only steps taking a transition from \(\MEC \) can point at \(\MEC \). Thus the total effect of such pointings in \(R_{\epsilonr_{1}} \)  is upper bounded by \( u\cdot |T|\cdot |MECs|\cdot n^{s+2\cdot \epsilonr_{1}} \), where \(|MECs|  \) is the number of all MECs in \(\A \). But as \(u\cdot |T|\cdot |MECs|\) is a constant for a given \(\A \), it holds that this value is upper bounded for all sufficiently large  \(n\) by \(n^{s+\epsilonr_{3}} \) assuming \begin{equation}\label{eq-epsbound-nbvtuioqvdfdfdfd}
		\epsilonr_{3}>2\cdot \epsilonr_{1}
	\end{equation} Notice also that for every single configuration of any computation from \(R_{\epsilonr_{1}} \) the value of \(rank^{C_1,\dots,C_r}_{1,T_{s+1}}\)  is upper bounded by \(n^{r+\epsilonr_{3}} \) for all sufficiently large \(n\). 
	
	Let \(\B \) be the set of all MECs \(\MEC \) of \(\A_\sigma \) where \(\sigma\in\stratsMD{\A_{1,T_{s-1}} } \) for which the corresponding component is zero-unbounded on \(rank^{C_1,\dots,C_r}_{1,T_{s-1}} \). Note that the MEC corresponding to \(\by \) is in \(\B \).

	Since we assume \(\calP_{\tilde{\A}}[\M_\by] \) does not have an upper asymptotic estimate of  \(n^k\), there exists \(a>0 \), \(\epsilon>0 \), initial state \(p\), a pointing strategy \(\sigma\), and an infinite set \(N\subseteq \mathbb{N} \) such that \(\prob_{p\vec{n}}^\sigma[\pointcomplex_{\tilde{\A}}(\M_\by)>n^{k+\epsilon}]\geq 2\cdot a \) for all \(n\in N \). Given a pointing computation \(\pi\) on \(\tilde{\A} \) let \(\calP(\B)(\pi) \) denote the number of steps, of the  computation on \(\A \) corresponding to \(\pi \), that point at a MEC from \(\B \). Note that for all sufficiently large \(n\in N\) it holds \(\prob_{p\vec{n}}^\sigma[\calP(\B)\geq n^{k+\epsilon} \textit{ and } R_{\epsilonr_{1}}]\geq a \).

	Given a computation \(\alpha \) on \(\A\) we divide \(\alpha \) into segments as follows:
	\begin{itemize}
		\item The beginning of \(\alpha \) is in the first segment;
		\item Let \(m \) be the number of steps in the current segment that point at a MEC from \(\B\) and let \(b \) be the total effect of these steps on \(rank^{C_1,\dots,C_r}_{1,T_{s+1}}\). The segment lasts until at least one of the following conditions is satisfied: 
		\begin{itemize}
			\item \(b\geq n^{r+\epsilonr_{4}} \) OR,
			\item \(b\leq -n^{r+\epsilonr_{4}} \) OR,
			\item \(m\geq n^{2\cdot r+\epsilonr_{5}} \);
		\end{itemize}
		AND the next MEC pointed to is in \(\B \) (i.e. the segment ends at the moment at least one of the above conditions is satisfied and the next MEC being pointed at is in \(\B\). The moment a segment ends the next segment begins. The first step of the new segment points to a MEC from \(\B \)).
	\end{itemize} 
	
	Let \(\#_{segments}(\alpha) \) denote the number of segments on \(\alpha \). Since for all sufficiently large \(n\in N \) we have \(\prob_{p\vec{n}}^\sigma[\calP(\B)\geq n^{k+\epsilon} \textit{ and } R_{\epsilonr_{1}}]\geq a \) and each segment can point to \(\B \) at most \(n^{2\cdot r+\epsilonr_{5}}\) times, it holds that \[\prob_{p\vec{n}}^\sigma[\#_{segments} \geq \frac{n^{k+\epsilon}}{n^{2r+\epsilonr_{5}}} \textit{ and } R_{\epsilonr_{1}}]=\prob_{p\vec{n}}^\sigma[\#_{segments} \geq n^{k-2\cdot r+\epsilon-\epsilonr_{5}} \textit{ and } R_{\epsilonr_{1}}]\geq a \]

	%We cal a segment decreasing if \(rank_\bt(\Delta(\beta))<-n^{1+\epsilonr_{2}} \), where \( \beta\) is the path taken by \(\hat{B} \) in the given segment. 

	Assume there were no limit on the length of \(m \) in the segment definition and let \(\tau \) be the stopping time in the given segment such that \(|b|\geq |n^{r+\epsilonr_{4}}| \), then since the expected effect on \(rank_{1,T_{s+1}}^{C_1,\dots,C_r}\) of each of the positioning counting towards \(b\)   is \(0\), from the optional stopping theorem we would obtain \(\prob_{p\vec{n}}^\sigma[b\geq n^{r+\epsilonr_{4}}]\cdot (n^{r+\epsilonr_{4}}+u)+\prob_{p\vec{n}}^\sigma[b\leq -n^{r+\epsilonr_{4}}]\cdot (-n^{r+\epsilonr_{4}})\geq 0 \). 
	Thus it would hold \[\frac{\prob_{p\vec{n}}^\sigma[b\geq n^{r+\epsilonr_{4}}]}{\prob_{p\vec{n}}^\sigma[b\leq -n^{r+\epsilonr_{4}}]}\geq \frac{n^{r+\epsilonr_{4}}}{(n^{r+\epsilonr_{4}}+u)}  \]
	Since both these probabilities would sum to \(1\) it would hold \(\prob_{p\vec{n}}^\sigma[b\leq -n^{l+\epsilonr_{4}}]\geq \frac{1}{2}-\kappa_n \) where \(\lim_{n\rightarrow\infty}\kappa_n=0 \). 
	
	Notice that any computation that only points to MECs from \(\B \) can be seen as a computation on a 1-dim VASS MDP with the only counter corresponding to \(rank_{1,T_{s+1}}^{C_1,\dots,C_r}\), and that contains only zero-unbounded MECs. If we were to consider the initial counter value for this computation to be set to \(n^{r+\epsilonr_{4}} \), then from results about 1-dim VASS MDPs  from \cite{AKCONCUR23} we have an upper asymptotic estimate \(n^{2\cdot r+2\cdot \epsilonr_{4}} \) on the number of steps before this counter becomes negative. Therefore it holds for each single segment that \(\lim_{n\rightarrow\infty}\prob_{p\vec{n}}^\sigma[m\geq n^{2\cdot r+2\cdot \epsilonr_{4}+\epsilonr_{6}}]=0 \). Thus assuming \begin{equation}\label{eq--epsbound-bcfyhcdqqqqg}
		2\cdot \epsilonr_{4}+\epsilonr_{6}<\epsilonr_{5} 
	\end{equation}	 	 for any constant \(a'>0 \) it holds for all sufficiently large \(n\) that \(\prob_{p\vec{n}}^\sigma[m\geq n^{2r+\epsilonr_{5}}]\leq a' \).

	We call a segment decreasing if  \(b\leq -n^{r+\epsilonr_{4}} \) in the given segment.  From the above, the probability that a segment is decreasing can be lower bounded by \(\frac{1}{2}-\kappa_n-\prob_{p\vec{n}}^\sigma[m\geq n^{2\cdot r+\epsilonr_{5}}] \). If a segment is not decreasing then we call it non-decreasing.
	%, then the probability that a segment is non-decreasing can be upper bounded by say \(\frac{2}{3} \).
	%	 \michal{here 2/3 is used arbitrarily, we just need a constant upper bound} for all sufficiently large \(n\)
	
	Let \(\#_{non-decreasing}(\alpha) \) and \(\#_{decreasing}(\alpha) \) denote the number of non-decreasing and decreasing segments along the computation \(\alpha \), respectively.
	Let \(LONG \) be the set of all computations \(\alpha \) on \(\A \) such that \(\alpha \) contains exactly \(\lfloor n^{k-2\cdot r+\epsilon-\epsilonr_{5}} \rfloor \) segments, and let \(LONGER \) be the set of all computations that have a prefix in \(LONG \). Let \(BAD\subseteq LONG \) be the set of all computations  \(\alpha\in LONG \) such that \(\#_{decreasing}(\alpha)\leq n^{k-2\cdot r+\epsilon-\epsilonr_{5}-\epsilonr_{7}} \). For each \(\alpha\in LONG \) it holds \(\#_{decreasing}(\alpha)\leq n^{k-2r+\epsilon-\epsilonr_{5}-\epsilonr_{7}} \) iff \(\#_{non-decreasing}(\alpha)> \lfloor n^{k-2\cdot r+\epsilon-\epsilonr_{5}}\rfloor - n^{k-2\cdot r+\epsilon-\epsilonr_{5}-\epsilonr_{7}}\geq \frac{2}{3}\cdot \lfloor n^{k-2\cdot r+\epsilon-\epsilonr_{5}}\rfloor \) (here the last inequality holds for all sufficiently large \(n\)). 	Notice that whether a segment is decreasing or not is independent of all the other segments along \(\alpha\in LONG \), therefore this can be viewed as a binomial distribution with probability of success (i.e. being decreasing) being at least \(\frac{1}{2}-p_n \) where \(p_n=\kappa_n+\prob_{p\vec{n}}^\sigma[m\geq n^{2\cdot r+\epsilonr_{5}}] \) and the number of trials being \(\lfloor n^{k-2\cdot r+\epsilon-\epsilonr_{5}} \rfloor\). Thus from the Chernoff bound for Binomial distribution (Theorem~1 of \cite{Arratia1989}) it holds  for all sufficiently large \(n\)
	\begin{gather*}
		\prob_{p\vec{n}}^{\sigma}[\#_{decreasing}\leq n^{k-2\cdot r+\epsilon-\epsilonr_{5}-\epsilonr_{7}}\mid LONG]
		=\\		\prob_{p\vec{n}}^{\sigma}[\#_{non-decreasing}\geq \lfloor n^{k-2\cdot r+\epsilon-\epsilonr_{5}}\rfloor - n^{k-2\cdot r+\epsilon-\epsilonr_{5}-\epsilonr_{7}}\mid LONG] 	
		\leq \\
		\prob_{p\vec{n}}^{\sigma}[\#_{non-decreasing}\geq  \frac{2}{3}\cdot \lfloor n^{k-2\cdot r+\epsilon-\epsilonr_{5}}\rfloor\mid LONG] 	
		\leq\\
		\exp\big( -  \lfloor n^{k-2\cdot r+\epsilon-\epsilonr_{5}}\rfloor\cdot  (\frac{2}{3}\cdot  \log\frac{\frac{2}{3}}{\frac{1}{2}+p_n}+(1-\frac{2}{3})\cdot \log \frac{1-\frac{2}{3}}{\frac{1}{2}-p_n}   )  \big)
		%		=\\		
		%		\exp\big( -  \lfloor n^{k-2r+\epsilon-\epsilonr_{5}}\rfloor (\frac{2}{3} \log\frac{\frac{2}{3}}{\frac{1}{2}+p_n}+\frac{1}{3}\log \frac{\frac{1}{3}}{\frac{1}{2}-p_n}   )  \big)
		\leq 
		\exp\big( -  \lfloor n^{k-2\cdot r+\epsilon-\epsilonr_{5}-\epsilonr_{8}}\rfloor  \big)
	\end{gather*}
	where the last inequality follows from \(\log\frac{1}{\frac{1}{2}+p_n}>\frac{2}{3}\cdot  \log\frac{\frac{2}{3}}{\frac{1}{2}+p_n}+(1-\frac{2}{3})\cdot \log \frac{1-\frac{2}{3}}{\frac{1}{2}-p_n}       >0 \) which holds from \cite{Arratia1989}.

	As  \(\lim_{n\rightarrow\infty}\exp\big( -  \lfloor n^{k-2\cdot r+\epsilon-\epsilonr_{5}-\epsilonr_{8}}\rfloor  \big)=0 \) it holds for all sufficiently large \(n\) that \(\prob_{p\vec{n}}^{\sigma}[\#_{decreasing}\leq  n^{k-2\cdot r+\epsilon-\epsilonr_{5}-\epsilonr_{7}}\mid LONG] \leq \frac{a}{2} \). But since \(\prob_{p\vec{n}}^\sigma[LONGER  \textit{ and } R_{\epsilonr_{1}}]\geq a \) it must hold \(\prob_{p\vec{n}}^\sigma[\#_{decreasing}> n^{k-2\cdot r+\epsilon-\epsilonr_{5}-\epsilonr_{7}}  \textit{ and } R_{\epsilonr_{1}}]\geq a-\frac{a}{2}= \frac{a}{2} \). But notice that each decreasing segment must increase the value of \(rank^{C_1,\dots,C_r}_{1,T_{s-1}}\) by at least \(n^{r+\epsilonr_{9}} \)  with the steps that do not point at a MEC from \(\B \), since the section starts with \(rank^{C_1,\dots,C_r}_{1,T_{s-1}} \) being at most \(n^{r+\epsilonr_{3}} \) and  the effect of \(b\) on \(rank^{C_1,\dots,C_r}_{1,T_{s-1}} \) is at most \(-n^{r+\epsilonr_{4}} \), thus for \begin{equation}\label{eq-epsbound-kbhvjgcuivhkvcivs}
	\epsilonr_{4}>\epsilonr_{3}
	\end{equation} in order for the ranking function to remain positive it must be increased by at least \(n^{r+\epsilonr_{9}}\leq n^{r+\epsilonr_{4}}-n^{r+\epsilonr_{3}}  \) where we assume \begin{equation}\label{eq-epsbound-njbhioopihworbwo}
	\epsilonr_{9}<\epsilonr_{4}
	\end{equation} as the value of \(rank^{C_1,\dots,C_r}_{1,T_{s-1}} \) cannot be negative without a negative counter implies termination in \(\A\).
	
	 The steps pointing to a MEC of \(\A_\sigma \) where \(\sigma\in\stratsMD{\A_{1,T_{s-1}} } \) corresponding to a component whose effect on \(rank^{C_1,\dots,C_r}_{1,T_{s-1}}\) is zero-bounded can ever change \(rank^{C_1,\dots,C_r}_{1,T_{s-1}}\) by at most a constant in total. The steps pointing to a MEC of \(\A_\sigma \) where \(\sigma\in\stratsMD{\A_{1,T_{s-1}} } \) whose effect on \(rank^{C_1,\dots,C_r}_{1,T_{s-1}}\) is decreasing clearly cannot increase \(rank^{C_1,\dots,C_r}_{1,T_{s-1}}\) by more that \(n^\epsilon \) with high enough probability (i.e., the probability goes to \(0\) as \(n\rightarrow\infty\)). The maximal change of  \(rank^{C_1,\dots,C_r}_{1,T_{s-1}}\) by steps that do not point at any MEC of \(\A_\sigma \) where \(\sigma\in\stratsMD{\A} \) is upper bounded by \(u\cdot |MECs|\cdot n^{\epsilonr_1} \) for any computation from \(E_{\epsilonr_1} \). As no component of no \(\A_{1,T_{s+1}}\) is increasing on \(rank^{C_1,\dots,C_r}_{1,T_{s-1}}\) (see Section~\ref{sec-systems}) this leaves only steps pointing to a MEC of \(\A_\sigma \) where \(\sigma\in\stratsMD{\A} \) that contains a transition from \(T_1\setminus T_{s} \), but we already established that the maximal possible  total increase of \(rank^{C_1,\dots,C_r}_{1,T_{s-1}} \) using these transitions conditioned on \(R_{\epsilonr_{1}} \) is upper bounded by \(n^{s+\epsilonr_{3}}\). 
	 
	 Hence let \(INCREASE(\alpha) \) be the increase of \(rank^{C_1,\dots,C_r}_{1,T_{s-1}} \) from the steps pointing to a MEC not in \(\B \), then it holds \[
	 	\lim_{n\rightarrow\infty}\prob_{p\vec{n}}^{\sigma}[INCREASE\geq u + n^\epsilon + u\cdot|MECs|\cdot n^{\epsilonr_{1}}+n^{s+\epsilonr_{3}}\mid R_{\epsilonr_{1}}] =0
	 \]

	  But since as discussed above it holds \(INCREASE(\alpha)\geq   \#_{decreasing}(\alpha)\cdot n^{r+\epsilonr_{9}} \) we have for all sufficiently large \(n\in N\) that
	  \begin{gather*}
	  	\prob_{p\vec{n}}^\sigma[INCREASE\geq \#_{decreasing}\cdot  n^{r+\epsilonr_{9}} \textit{ and } R_{\epsilonr_{1}}] 
	  	\geq\\
	  	\prob_{p\vec{n}}^\sigma[INCREASE\geq n^{k-2\cdot r+\epsilon-\epsilonr_{5}-\epsilonr_{7}}\cdot n^{r+\epsilonr_{9}} \textit{ and } R_{\epsilonr_{1}}] 
	  	=\\
	  	\prob_{p\vec{n}}^\sigma[INCREASE\geq n^{k- r+\epsilon-\epsilonr_{5}-\epsilonr_{7}+\epsilonr_{9}} \textit{ and } R_{\epsilonr_{1}}] 
	  	\geq \frac{a}{2}
	  \end{gather*} 
	  
	  Thus if it were to holds  \(n^{k- r+\epsilon-\epsilonr_{5}-\epsilonr_{7}+\epsilonr_{9}}\geq u + n^\epsilon + u\cdot|MECs|\cdot n^{\epsilonr_{1}}+n^{s+\epsilonr_{3}} \) for all sufficiently large \(n\) we would have arrived at a contradiction. Remember that \(k=\max(s+r,2r) \), let us consider the two possibilities separately.
	  
	  If \(k=s+r \) then we can write \begin{align*}
	  	n^{k- r+\epsilon-\epsilonr_{5}-\epsilonr_{7}+\epsilonr_{9}}
	  	\geq
	  	 u + n^\epsilon + u\cdot|MECs|\cdot n^{\epsilonr_{1}}+n^{s+\epsilonr_{3}} 
	  	 \\
	  	 n^{s+r- r+\epsilon-\epsilonr_{5}-\epsilonr_{7}+\epsilonr_{9}}
	  	 \geq
	  	 u + n^\epsilon + u\cdot|MECs|\cdot n^{\epsilonr_{1}}+n^{s+\epsilonr_{3}} 
	  	 \\
	  	 n^{s+\epsilon-\epsilonr_{5}-\epsilonr_{7}+\epsilonr_{9}}
	  	 \geq
	  	 u + n^\epsilon + u\cdot|MECs|\cdot n^{\epsilonr_{1}}+n^{s+\epsilonr_{3}} 
	  \end{align*}
	  which holds for all sufficiently large \(n\) assuming  \begin{equation}\label{eq-epsbound-jbvhucivkhgckutfliygkc}
	  	\epsilon<\epsilon-\epsilonr_{5}-\epsilonr_{7}+\epsilonr_{9}
	  \end{equation}
	  \begin{equation}\label{eq-epsbound-jbvhucivkfdhgckutfliygkc}
	  	\epsilonr_1<\epsilon-\epsilonr_{5}-\epsilonr_{7}+\epsilonr_{9}
	  \end{equation}
	  \begin{equation}\label{eq-epsbound-jbfdssvhucivkhgckutfliygkc}
	  	\epsilonr_3<\epsilon-\epsilonr_{5}-\epsilonr_{7}+\epsilonr_{9}
	  \end{equation}

	   If \(k=2r \) then we can write \begin{align*}
	  	n^{k- r+\epsilon-\epsilonr_{5}-\epsilonr_{7}+\epsilonr_{9}}
	  	\geq
	  	u + n^\epsilon + u\cdot|MECs|\cdot n^{\epsilonr_{1}}+n^{s+\epsilonr_{3}} 
	  	\\
	  	n^{2r- r+\epsilon-\epsilonr_{5}-\epsilonr_{7}+\epsilonr_{9}}
	  	\geq
	  	u + n^\epsilon + u\cdot|MECs|\cdot n^{\epsilonr_{1}}+n^{s+\epsilonr_{3}} 
	  	\\
	  	n^{r+\epsilon-\epsilonr_{5}-\epsilonr_{7}+\epsilonr_{9}}
	  	\geq
	  	u + n^\epsilon + u\cdot|MECs|\cdot n^{\epsilonr_{1}}+n^{s+\epsilonr_{3}} 
	  \end{align*}
	 as in this case it holds \(r>s \), the above holds  for all sufficiently large \(n\) for the same assumptions as in the previous case.
	  
It remains to show there exist values for \(\epsilon_1,\epsilon_2,\dots \) that satisfy all of our assumptions. We do this in Table~\ref{Table-eps-section-finalpart}.

\begin{table*}[h]
	\caption{Values of \(\epsilon_1,\epsilon_2,\dots \) for Section~\ref{section-proof-lemma-hatB-zero-unbounded-rankl-give-upper-estimate-nk}}
	\centering
	%	\begin{center}
		\begin{tabular}{|l|| c c| c|} 
			\hline
			\(\epsilon\) assignment &  \multicolumn{2}{|c|}{restrictions}  & After substitution \\ 
			\hline\hline
			%		\multirow{10}{*}{ \begin{tabular}{c}
					%			 \(\epsilon_1=\frac{9}{10} \)
					%				\\ \(\epsilon_3=\frac{1}{8} \)
					%				\\ \(\epsilon_{4}=\frac{1}{10} \)
					%				\\ \(\epsilon_{5}=\frac{5}{12} \)
					%				\\ \(\epsilon_{6}=\frac{1}{2} \)
					%				\\ \(\epsilon_7=\frac{1}{7} \)
					%		\end{tabular} }
			\(\epsilonr_1=\nicefrac{\min(\frac{1}{3},\epsilon)}{1000} \) & \(0<\epsilon_1,\epsilon_2,\dots  \) & &  \\ 
			\hline
			\(\epsilonr_2=\nicefrac{\min(\frac{1}{3},\epsilon)}{100} \) & \(\epsilonr_{2}>\epsilonr_{1} \) & \eqref{eq-epsbound-ojbfsdvugolyibvuotyib} & \(\nicefrac{\min(\frac{1}{3},\epsilon)}{100}>\nicefrac{\min(\frac{1}{3},\epsilon)}{1000} \) \\
			\hline
			\(\epsilon_{3}=\nicefrac{\min(\frac{1}{3},\epsilon)}{40} \) & 	\(\epsilonr_{3}>\epsilonr_{1}+\epsilonr_{2}\) & \eqref{eq-epsbound-ojbvugolyibvuotyib} & \(\nicefrac{\min(\frac{1}{3},\epsilon)}{40}>\nicefrac{\min(\frac{1}{3},\epsilon)}{1000}+\nicefrac{\min(\frac{1}{3},\epsilon)}{100}\) \\
			\hline
			\(\epsilon_{4}=\nicefrac{\min(\frac{1}{3},\epsilon)}{25} \) & \(\epsilonr_{3}>2\cdot \epsilonr_{1}\) & \eqref{eq-epsbound-nbvtuioqvdfdfdfd} & \(\nicefrac{\min(\frac{1}{3},\epsilon)}{40}>2\cdot \nicefrac{\min(\frac{1}{3},\epsilon)}{1000}\) \\
			\hline
			\(\epsilon_{5}=\nicefrac{\min(\frac{1}{3},\epsilon)}{10} \) & \(2\cdot \epsilonr_{4}+\epsilonr_{6}<\epsilonr_{5}  \) & \eqref{eq--epsbound-bcfyhcdqqqqg} & \(2\cdot \nicefrac{\min(\frac{1}{3},\epsilon)}{25}+\nicefrac{\min(\frac{1}{3},\epsilon)}{100}<\nicefrac{\min(\frac{1}{3},\epsilon)}{10}  \)  \\ 
			\hline
			\(\epsilon_{6}=\nicefrac{\min(\frac{1}{3},\epsilon)}{100} \) & \(\epsilonr_{4}>\epsilonr_{3}\) & \eqref{eq-epsbound-kbhvjgcuivhkvcivs} & \(\nicefrac{\min(\frac{1}{3},\epsilon)}{25}>\nicefrac{\min(\frac{1}{3},\epsilon)}{40} \) \\
			\hline
		\(\epsilon_{7}=\nicefrac{\min(\frac{1}{3},\epsilon)}{10} \)	& \(	\epsilonr_{9}<\epsilonr_{4} \) & \eqref{eq-epsbound-njbhioopihworbwo} & \(\nicefrac{\min(\frac{1}{3},\epsilon)}{100}<\nicefrac{\min(\frac{1}{3},\epsilon)}{25} \) \\
			\hline
		\(\epsilon_{8}=\epsilon \)	& \(	\epsilon<\epsilon-\epsilonr_{5}-\epsilonr_{7}+\epsilonr_{9} \) & \eqref{eq-epsbound-jbvhucivkhgckutfliygkc} & \(\epsilon<\epsilon-\nicefrac{\min(\frac{1}{3},\epsilon)}{10}-\nicefrac{\min(\frac{1}{3},\epsilon)}{10}+\epsilon \) \\
			\hline
		\(\epsilon_{9}=\epsilon \)	& \(	\epsilonr_1<\epsilon-\epsilonr_{5}-\epsilonr_{7}+\epsilonr_{9}\) & \eqref{eq-epsbound-jbvhucivkfdhgckutfliygkc} & \(\nicefrac{\min(\frac{1}{3},\epsilon)}{1000}<\epsilon-\nicefrac{\min(\frac{1}{3},\epsilon)}{10}-\nicefrac{\min(\frac{1}{3},\epsilon)}{10}+\epsilon \) \\
			\cline{2-4}
				& \(	\epsilonr_3<\epsilon-\epsilonr_{5}-\epsilonr_{7}+\epsilonr_{9}\) & \eqref{eq-epsbound-jbfdssvhucivkhgckutfliygkc} & \(\nicefrac{\min(\frac{1}{3},\epsilon)}{40}<\epsilon-\nicefrac{\min(\frac{1}{3},\epsilon)}{10}-\nicefrac{\min(\frac{1}{3},\epsilon)}{10}+\epsilon\) \\
			%				& \(2\epsilon_{5}-2\epsilon_{7}-\epsilon_{6}>0 \) & 18744 & 7560 \\
			%				\cline{2-4}
			%				& \(\epsilon_{1}>\epsilon \) & 18744 & 7560 \\
			%				\cline{2-4}
			%				& 		 \( 0<\epsilon<\frac{1}{10}\) & 18744 & 7560 \\
			\hline
		\end{tabular}
		\par
		\label{Table-eps-section-finalpart}	
		%	\end{center}
\end{table*}

	Thus Lemma~\ref{lemma-I-dont-lknooow-but-it-is-somehting} holds.
\end{proof}

\end{document}